\documentclass[11pt]{article}
\usepackage[margin=1in]{geometry}

\usepackage[utf8]{inputenc} %
\usepackage[dvipsnames]{xcolor}
\usepackage{booktabs}       %
\usepackage{amsfonts}       %
\usepackage{nicefrac}       %
\usepackage{microtype}      %
\usepackage{parskip}        %
\usepackage{amsmath, amsthm, mathtools, dsfont, mathdots}
\usepackage{txfonts} %
\usepackage{color,graphicx}
\usepackage{enumerate}
\usepackage{algorithm}
\usepackage[noend]{algpseudocode}
\usepackage{fancybox}
\usepackage{xspace}
\usepackage{cuted,tcolorbox,lipsum}
\usepackage{multicol}
\usepackage{threeparttable}
\usepackage{tabularx}
\usepackage[T1]{fontenc}
\usepackage{tikz}
\usepackage{tcolorbox}
\usepackage{breakurl}
\usepackage{hyperref}
\hypersetup{breaklinks=true}
\usepackage[all]{hypcap}
\usetikzlibrary{matrix,positioning,calc,arrows,patterns,decorations.pathreplacing,decorations.markings,patterns.meta,arrows.meta,shapes,positioning}
\tcbuselibrary{skins}
\usepackage[colorinlistoftodos]{todonotes}
\usepackage{array,multirow}
\usepackage{csquotes}
\usepackage{subcaption}
\usepackage{cellspace} %

\usepackage{amsthm}
\usepackage[dvipsnames]{xcolor}

\newcommand*\samethanks[1][\value{footnote}]{\footnotemark[#1]}

\theoremstyle{definition}

\colorlet{shadecolor}{gray!20}

\newtheorem{theorem}{Theorem}
\newtheorem{corollary}[theorem]{Corollary}
\newtheorem{lemma}[theorem]{Lemma}

\newtheorem{definition}[theorem]{Definition}

\DeclareMathOperator*{\argmax}{argmax}

\newcommand{\appx}{\ensuremath{10^{-11}}}
\newcommand{\alfbound}{\appx}

\newcommand{\sumYS}{\ensuremath{\sum \ys}}

\newcommand{\main}{\textsc{Main}}
\newcommand{\streer}{\textsc{SteinerTree}}

\newcommand{\G}{\ensuremath{G}}
\newcommand{\V}{\ensuremath{V}}
\newcommand{\E}{\ensuremath{E}}
\newcommand{\cc}{\ensuremath{c}}
\newcommand{\ce}{\ensuremath{\cc_e}}

\newcommand{\inp}{in}
\newcommand{\out}{out}

\newcommand{\Ins}{\ensuremath{I}} %
\newcommand{\tm}{\ensuremath{t}}
\newcommand{\tmp}{\ensuremath{\tm^{*}}}

\newcommand{\GBGr}{\textsc{ShadowMoatGrowing}} %
\newcommand{\GBG}{\hyperref[proc:gbg]{\GBGr}}

\newcommand{\BBGr}{\textsc{LegacyMoatGrowing}}
\newcommand{\BBG}{\hyperref[proc:bbg]{\BBGr}} %

\newcommand{\currenttime}{\tau}

\newcommand{\Deltae}{\ensuremath{\Delta_e}}

\newcommand{\currentsets}{\ensuremath{FC}}
\newcommand{\activesets}{\ensuremath{ActS}}
\newcommand{\deactivesets}{\ensuremath{DS}}
\newcommand{\Fc}{\F_{contracted}}

\newcommand{\core}{\ensuremath{\textsc{Superactive}}}
\newcommand{\corep}{\ensuremath{\core^{\ep}}}
\newcommand{\corez}{\ensuremath{\core^{\ez}}}

\newcommand{\deltaS}{\ensuremath{\delta(S)}}

\newcommand{\rep}{\ensuremath{\sigma}}
\newcommand{\repd}{\ensuremath{\phi}}
\newcommand{\fard}{\ensuremath{f}}
\newcommand{\pare}{\ensuremath{b}}
\newcommand{\reptim}{\ensuremath{\currenttime}}
\newcommand{\vi}{\ensuremath{V_I}}
\newcommand{\acttree}{\ensuremath{T_S}}
\newcommand{\Root}{\textit{root}}
\newcommand{\RSET}{\{\Root\}}

\newcommand{\rstree}{\gr}
\newcommand{\sttm}{\currenttime}
\newcommand{\sttmp}{\currenttime'}
\newcommand{\stbstm}{\currenttime^*}
\newcommand{\stins}{\BIns}
\newcommand{\stinsp}{\BIns^*}
\newcommand{\deltaY}{\Delta_y}
\newcommand{\styi}{\ensuremath{z}}
\newcommand{\styo}{z^{*}}
\newcommand*\diff{\mathop{}\!\mathrm{d}}
\newcommand{\stlsttime}{\ensuremath{{\currenttime'}}}
\newcommand{\stacts}{K}
\newcommand{\terminals}{\mathcal{T}}
\newcommand{\stapx}{1+2\frac{\sqrt{2}}{3}}
\newcommand{\strgh}{1.943}
\newcommand{\stalfval}{\frac{1}{2} - \frac{\sqrt{2}}{3}}
\newcommand{\stbtaval}{\sqrt{2}-1}

\newcommand{\stx}{\ensuremath{\dis(q,u)}}
\newcommand{\sty}{\ensuremath{\dis(q,v)}}
\newcommand{\stz}{\ensuremath{\dis(q,w)}}
\newcommand{\stxp}{\ensuremath{\dis(q,\vone)}}
\newcommand{\styp}{\ensuremath{\dis(q,\vtwo)}}
\newcommand{\stzp}{\ensuremath{\dis(q,\vthree)}}

\newcommand{\Boost}{B}
\newcommand{\tmpp}{\ensuremath{\tau}}
\newcommand{\BoostExp}{\ensuremath{\Boost = (v, \tmpp)}}
\newcommand{\WithBoost}{\textsc{WithBoost}}

\newcommand{\BIns}{\ensuremath{\I}}
\newcommand{\BInsExp}{\ensuremath{\BIns = (\G, \tm, \tmp)}}

\newcommand{\BoostedModGWr}{\textsc{BoostedMoatGrowing}}
\newcommand{\BoostedModGW}{\hyperref[func:bmg]{\BoostedModGWr}}

\newcommand{\base}{base}
\newcommand{\add}{boost}

\newcommand{\Winr}{\textsc{Win}}
\newcommand{\Win}{\hyperref[func:win]{\Winr}}

\newcommand{\Lossr}{\textsc{Loss}}
\newcommand{\Loss}{\hyperref[func:loss]{\Lossr}}

\newcommand{\bta}{\beta}

\newcommand{\win}{\textit{win}}
\newcommand{\loss}{\textit{loss}}

\newcommand{\winone}{\ensuremath{\win_{1}}}
\newcommand{\wintwo}{\ensuremath{\win_{2}}}
\newcommand{\lossone}{\ensuremath{\loss_{1}}}
\newcommand{\losstwo}{\ensuremath{\loss_{2}}}

\newcommand{\LocalSearchr}{\textsc{LocalSearch}}
\newcommand{\LocalSearch}{\hyperref[func:ls]{\LocalSearchr}}

\newcommand{\YIn}{\y}
\newcommand{\YOut}{\y^{*}}

\newcommand{\activeattime}{\activesets^{\inp}_{\currenttime}}
\newcommand{\activeattimeprime}{\activesets^{\out}_{\currenttime}}

\newcommand{\ccattime}{\currentsets^{\inp}_{\currenttime}}
\newcommand{\ccattimeprime}{\currentsets^{\out}_{\currenttime}}

\newcommand{\y}{\ensuremath{y}}
\newcommand{\yb}{\ensuremath{\y^{b}}}
\newcommand{\ys}{\ensuremath{\y_S}}
\newcommand{\ysv}{\ensuremath{\gr_{S,v}}}
\newcommand{\yadd}{\ensuremath{\y_{\add}}}
\newcommand{\ybase}{\ensuremath{\y_{\base}}}

\newcommand{\yy}{\ensuremath{y^{\ee}}}
\newcommand{\yyb}{\ensuremath{y^{b\ee}}}
\newcommand{\yys}{\ensuremath{y^{\ee}_S}}
\newcommand{\yplus}{\ensuremath{y^{\eplus}}}
\newcommand{\yp}{\ensuremath{y^{\ep}}}
\newcommand{\yz}{\ensuremath{y^{\ez}}}

\newcommand{\refines}{\preceq}
\newcommand{\Fam}{\mathcal{F}}

\newcommand{\priority}{priority}

\newcommand{\priorityset}{\textsc{Reps}}

\newcommand{\ee}{+}
\newcommand{\eplus}{-}
\newcommand{\ep}{\prime}
\newcommand{\ez}{\prime\prime}

\newcommand{\BaseSet}{\textsc{Base}}

\newcommand{\Deputy}{\textsc{Deputy}}

\newcommand{\tv}[1]{\ensuremath{t_{#1}}}
\newcommand{\tV}{\tv{v}}
\newcommand{\Aa}{\activesets^{\ee}}
\newcommand{\Aplus}{\activesets^{\eplus}}
\newcommand{\Ap}{\activesets^{\ep}}
\newcommand{\Az}{\activesets^{\ez}}

\newcommand{\assignee}{\ensuremath{Assignee}}
\newcommand{\asp}{\assignee^{\ep}}

\newcommand{\gr}{\ensuremath{r}}

\newcommand{\rr}{\ensuremath{r^{\ee}}}
\newcommand{\rplus}{\ensuremath{r^{\eplus}}}
\newcommand{\rstar}{\ensuremath{\hat{r}^{\ee}}}
\newcommand{\rp}{\ensuremath{\hat{r}^{\ep}}}
\newcommand{\rpb}{\ensuremath{\hat{r}^{b\ep}}}
\newcommand{\rpp}{\ensuremath{\hat{r}^{p\ep}}}
\newcommand{\rz}{\ensuremath{\hat{r}^{\ez}}}
\newcommand{\rzb}{\ensuremath{\hat{r}^{b\ez}}}
\newcommand{\rzp}{\ensuremath{\hat{r}^{p\ez}}}

\newcommand{\gt}{\ensuremath{t}}

\newcommand{\ttt}{\ensuremath{t^{\ee}}}
\newcommand{\tplus}{\ensuremath{t^{\eplus}}}
\newcommand{\tp}{\ensuremath{t^{\ep}}}
\newcommand{\tz}{\ensuremath{t^{\ez}}}

\newcommand{\tin}{\ensuremath{t^{\inp}}}
\newcommand{\tout}{\ensuremath{t^{\out}}}

\newcommand{\rmax}{\ensuremath{\rr_{max}}} %
\newcommand{\largest}{\ensuremath{\textit{LG}}} %
\newcommand{\A}{\mathcal{A}}
\newcommand{\B}{\mathcal{B}}
\newcommand{\Bone}{\B_{1}}
\newcommand{\Btwo}{\B_{2}}

\newcommand{\OPT}{\text{OPT}}
\newcommand{\OPTA}{\text{OPT}_{\A}}
\newcommand{\OPTB}{\text{OPT}_{\B}}
\newcommand{\OPTBone}{\text{OPT}_{\Bone}}
\newcommand{\OPTBtwo}{\text{OPT}_{\Btwo}}

\newcommand{\opti}{\ensuremath{\mathcal{C}}}

\newcommand{\optcom}{\ensuremath{\textsc{comp}}}

\newcommand{\Extendr}{\textsc{Extend}}
\newcommand{\Extend}{\hyperref[func:extend]{\Extendr}}
\newcommand{\eps}{\ensuremath{\epsilon}}
\newcommand{\pot}{\ensuremath{\pi}}
\newcommand{\potcap}{\ensuremath{\rho}}
\newcommand{\setpot}{\ensuremath{\mathcal{P}}}
\newcommand{\EGWr}{\textsc{ExtendedMoatGrowing}}
\newcommand{\EGW}{\hyperref[func:egw]{\EGWr}}
\newcommand{\ExtendedGW}{\EGW}
\newcommand{\solext}{\ensuremath{\textit{SOL}_{\textit{XT}}}}
\newcommand{\otree}{\ensuremath{T_{\opti}}}
\newcommand{\xtree}{\ensuremath{T_{\largest(\opti)}}}
\newcommand{\X}{\ensuremath{X}}
\newcommand{\LG}{\ensuremath{\largest(\opti)}}
\newcommand{\NLG}{\ensuremath{\opti\setminus\largest(\opti)}}
\newcommand{\lgr}{\rstar_{\largest}}
\newcommand{\lgrz}{\rz_{\largest}}
\newcommand{\rmaxsum}{\ensuremath{\sum_{\opti\in\B}\rmax(\opti)}}

\newcommand{\XX}{X_{sum}}
\newcommand{\wvar}{\ensuremath{w}}
\newcommand{\wpvar}{\ensuremath{w'}}
\newcommand{\losscoeff}{\zeta}

\newcommand{\rex}{\gr}
\newcommand{\rxone}{\rex^{(1)}}
\newcommand{\rxtwo}{\rex^{(2)}}

\newcommand{\maxacts}{\ensuremath{\mathcal{M}}}
\newcommand{\maxactsp}{\ensuremath{\maxacts'}}
\newcommand{\maxactsz}{\ensuremath{\maxacts''}}

\newcommand{\pair}{\textsc{Pair}}
\newcommand{\pairv}{\ensuremath{\pair_v}}
\newcommand{\pairS}{\ensuremath{\pair(S)}}
\newcommand{\solone}{\ensuremath{\textit{SOL}_{\textit{LS}}}}
\newcommand{\solcnd}{\ensuremath{\textit{SOL}_{\textit{\cnd}}}}
\newcommand{\ALG}{\ensuremath{ALG}}

\newcommand{\cnd}{AP}
\newcommand{\candidater}{\textsc{AutarkicPairs}}
\newcommand{\candidate}{\hyperref[func:candid]{\candidater}}
\newcommand{\CInsExp}{\ensuremath{\G, \pair, \yy, \ccnd}}
\newcommand{\unsatisfiedr}{\textsc{Unsatisfied}}
\newcommand{\unsatisfied}{\hyperref[func:unsats]{\unsatisfiedr}}
\newcommand{\unsatisfiedS}{\ensuremath{\unsatisfied(S)}}
\newcommand{\Gp}{\G'}
\newcommand{\vap}{\ensuremath{\upsilon}}
\newcommand{\maxdis}{\ensuremath{\dis_{\max}}}
\newcommand{\maxdisS}{\ensuremath{\maxdis(S)}}

\newcommand{\vone}{\ensuremath{v_1}}
\newcommand{\vtwo}{\ensuremath{v_2}}
\newcommand{\vthree}{\ensuremath{v_3}}
\newcommand{\sone}{\ensuremath{S_1}}
\newcommand{\stwo}{\ensuremath{S_2}}
\newcommand{\cone}{\ensuremath{\lambda}}
\newcommand{\ctwo}{\ensuremath{\varkappa}}
\newcommand{\cfive}{\ensuremath{\gamma}}
\newcommand{\twoopt}{\ensuremath{g^{\textit{mix}}}}
\newcommand{\ccnd}{\ensuremath{\eta}}
\newcommand{\Y}{\ensuremath{Y}}
\newcommand{\YS}{\ensuremath{\Y_S}}
\newcommand{\YUS}{\ensuremath{\Y_{\unsatisfiedS}}}
\newcommand{\YPS}{\ensuremath{\Y_{\pairS}}}
\newcommand{\Ycnd}{\ensuremath{\widehat{\Y}}}
\newcommand{\dis}{\ensuremath{d}}
\newcommand{\tone}{\ensuremath{\Ycnd_{1}}}
\newcommand{\ttwo}{\ensuremath{\Ycnd_{2}}}

\newcommand{\sactsf}{\ensuremath{\textit{SActS}_F}}
\newcommand{\macts}{\ensuremath{\textit{MActS}}}
\newcommand{\mactsf}{\ensuremath{\textit{MActS}_F}}
\newcommand{\extra}{\ensuremath{\textit{UM}}}

\newcommand{\extralegacy}{\ensuremath{\extra^-}}
\newcommand{\extralegacyotree}{\extralegacy(\otree)}
\newcommand{\extraboosted}{\ensuremath{\extra^+}}
\newcommand{\extraboostedotree}{\extraboosted(\otree)}
\newcommand{\uncolored}{\ensuremath{\textit{UC}}}

\newcommand{\multicolored}{\ensuremath{\textit{MC}}}
\newcommand{\singlecolored}{\ensuremath{\textit{SC}}}

\newcommand{\solp}{\ensuremath{\textit{SOL}'}}
\newcommand{\EAP}{\ensuremath{\E_{\cnd}}}

\newcommand{\alf}{\alpha}

\newcommand{\I}{\ensuremath{I}}

\newcommand{\T}{\ensuremath{T}}
\newcommand{\F}{\ensuremath{F}}
\title{Breaking a Long-Standing Barrier:\\ 2-$\varepsilon$ Approximation for Steiner Forest}
\date{}
\author{
Ali Ahmadi\thanks{University of Maryland.}\\
\texttt{ahmadia@umd.edu}
\and
Iman Gholami\samethanks\\
\texttt{igholami@umd.edu}
\and
MohammadTaghi Hajiaghayi\samethanks\\
\texttt{hajiagha@umd.edu}
\and 
Peyman Jabbarzade\samethanks\\
\texttt{peymanj@umd.edu}
\and
Mohammad Mahdavi\samethanks\\
\texttt{mahdavi@umd.edu}
}

\begin{document}
\maketitle

\pagenumbering{gobble} 
\begin{abstract}
The Steiner Forest problem, also known as the Generalized Steiner Tree problem, is a fundamental optimization problem on edge-weighted graphs where, given a set of vertex pairs, the goal is to select a minimum-cost subgraph such that each pair is connected. 
This problem generalizes the Steiner Tree problem, first introduced in 1811, for which the best known approximation factor is 1.39 [Byrka, Grandoni, Rothvo{\ss}, and Sanità, 2010] (Best Paper award, STOC 2010).

The celebrated work of [Agrawal, Klein, and Ravi, 1989] (30-Year Test-of-Time award, STOC 2023), along with refinements by [Goemans and Williamson, 1992] (SICOMP'95), established a 2-approximation for Steiner Forest over 35 years ago. 
Jain's (FOCS'98) pioneering iterative rounding techniques later extended these results to higher connectivity settings. 
Despite the long-standing importance of this problem, breaking the approximation factor of 2 has remained a major challenge, raising suspicions that achieving a better factor---similar to Vertex Cover---might indeed be hard. Notably, fundamental works, including those by Gupta and Kumar (STOC'15) and Gro{\ss} et al. (ITCS'18), introduced 96- and 69-approximation algorithms, possibly with the hope of paving the way for a breakthrough in achieving a constant-factor approximation below 2 for the Steiner Forest problem. 

In this paper, we break the approximation barrier of 2 by designing a novel deterministic algorithm that achieves a $2 - \appx$ approximation for this fundamental problem.
As a key component of our approach, we also introduce a novel dual-based local search algorithm for the Steiner Tree problem with an approximation guarantee of $\strgh$, which is of independent interest.
\end{abstract}

\newpage
\tableofcontents
\thispagestyle{empty}
\newpage

\pagenumbering{arabic}  %
\setcounter{page}{1}
\section{Introduction}
The Steiner Forest problem, also known as the the Generalized Steiner Tree problem, is a fundamental NP-hard problem in computer science. Given a weighted undirected graph and a set of vertex pairs (demands), the goal is to select a minimum-cost subset of edges that connects each pair.
We can assume each vertex belongs to exactly one demand by duplicating vertices as needed, assigning each demand to a separate copy, and connecting copies with zero-cost edges. 
Additionally, vertices not in any demand can be treated as being paired with themselves.
Formally, given a weighted undirected graph $\G = (\V, \E, \cc)$ with edge costs $\cc : \E \to \mathbb{R}_{\geq 0}$ and an involution function $\pair : \V \to \V$ indicating that each vertex $v$ must be connected to $\pairv$, the objective is to find a minimum-cost subgraph that ensures $v$ and $\pairv$ are connected for all $v \in \V$.
In this work, we break the long-standing approximation barrier of 2 for Steiner Forest, marking a breakthrough after more than 35 years of efforts by the community.

The Steiner Forest problem generalizes the Steiner Tree problem, in which a set of terminals is given and the goal is to connect them using a tree of minimum total edge weight.
In our model, we assume that one terminal is designated as the root, and we create a duplicate of it for each of the remaining terminals, pairing each duplicate with its corresponding terminal.
All other vertices are paired with themselves.

The first studies on the Steiner Tree problem date back to 1811, and it was later recognized as one of Karp’s classic NP-complete problems~\cite{DBLP:conf/coco/Karp72}.  
Moreover, the problem is MAX SNP-hard~\cite{DBLP:journals/ipl/BernP89} and inapproximable within a factor of 96/95 unless P = NP~\cite{DBLP:journals/tcs/ChlebikC08}.  
Compared to the Steiner Tree problem, variants of the Steiner Forest problem appear to be more challenging. 
Examples include the Prize-Collecting Steiner Forest\footnote{While the Prize-Collecting Steiner Tree has had a better-than-2 approximation for a long time~\cite{DBLP:journals/siamcomp/ArcherBHK11,DBLP:conf/stoc/AhmadiGHJM24}, the Prize-Collecting Steiner Forest only recently achieved a 2-approximation~\cite{10.1145/3722551} (JACM'25).} and the even more difficult $k$-Steiner Forest\footnote{The $k$-Steiner Tree problem has long been known to admit a 2-approximation~\cite{DBLP:conf/stoc/Garg05,DBLP:journals/siamcomp/BateniHL18}, while the $k$-Steiner Forest problem is roughly as hard as the densest $k$-subgraph problem~\cite{DBLP:conf/soda/HajiaghayiJ06}, for which the best-known approximation is $O(n^{1/4})$~\cite{DBLP:conf/stoc/BhaskaraCCFV10}.}.  
This contrast highlights the inherent challenges in achieving better approximations for Steiner Forest, suggesting the possibility of a 2-approximation lower bound.

The Steiner Tree problem has been extensively studied in the context of approximation algorithms. 
A trivial 2-approximation algorithm was first proposed~\cite{DBLP:journals/acta/KouMB81}. 
Breaking this natural 2-approximation barrier became a major research focus, with Zelikovsky achieving an $11/6$-approximation~\cite{DBLP:journals/algorithmica/Zelikovsky93}, followed by Karpinski and Zelikovsky’s $1.65$-approximation~\cite{DBLP:journals/eccc/ECCC-TR95-030}. 
Further refinements led to a $1.55$-approximation by Robins and Zelikovsky~\cite{DBLP:journals/siamdm/RobinsZ05} and finally a $1.39$-approximation by Byrka, Grandoni, Rothvo{\ss}, and Sanità~\cite{DBLP:conf/stoc/ByrkaGRS10}.
While these improvements have been achieved for the Steiner Tree problem, determining whether a better-than-2 approximation exists for Steiner Forest remains one of the most prominent open problems in approximation algorithms (see e.g.~\cite{DBLP:books/daglib/0030297}\footnote{See Problem 6 of Shmoys and Williamson~\cite{DBLP:books/daglib/0030297}, among their ten key open problems in approximation algorithms, with Problems 1 and 2 on TSP and Asymmetric TSP. Notably, some of these problems like TSP~\cite{DBLP:conf/stoc/KarlinKG21} (Best Paper award, STOC 2021) and Asymmetric TSP~\cite{DBLP:conf/stoc/SvenssonTV18} (Best Paper award, STOC 2018) have seen progress since the book's 2011 publication.}).

The best-known approximation factor for Steiner Forest remains $2$, achieved in a landmark paper by Agrawal, Klein, and Ravi~\cite{AKRSTOC91,DBLP:journals/siamcomp/AgrawalKR95} using the primal-dual method.
This approach was later revisited as a general framework for constrained forest problems by Goemans and Williamson~\cite{GWSODA92,DBLP:journals/siamcomp/GoemansW95} and applied to other problems in this area, such as the Prize-Collecting Steiner Tree. Jain's  pioneering iterative rounding techniques~\cite{DBLP:conf/focs/Jain98,DBLP:journals/combinatorica/Jain01} later extended these results to higher connectivity settings. 
Despite many attempts, no improvement over the 2-approximation factor has been achieved. The problem is so significant that researchers have published papers in top conferences with much worse approximation factors, hoping to pave the way for a breakthrough below $2$ for the Steiner Forest problem.
Notably, a greedy algorithm achieved a $96$-approximation~\cite{DBLP:conf/stoc/Gupta015}, and a local search algorithm attained a $69$-approximation~\cite{DBLP:conf/innovations/0001G0MS0V18}.
These works were published in the hope that tighter analyses or refinements might eventually help break the barrier of $2$.
For the greedy algorithm with a $96$-approximation, we provide a counterexample in Appendix~\ref{eg:anupam83}, showing that its approximation factor is indeed worse than  $2$.  
Recently, even half-integral solutions to an LP relaxation for Steiner Forest have been studied~\cite{DBLP:journals/corr/abs-2412-06518}.
Needless to say, our approach takes a completely different path from these previous attempts to surpass the $2$-approximation barrier.

In this paper, we present the first deterministic algorithm to break the long-standing approximation barrier of 2 for the Steiner Forest problem.

\begin{theorem} 
\label{thm:main_steiner_forest}
    There exists a deterministic polynomial-time algorithm that achieves an approximation factor of $2 - \appx$ for the Steiner Forest problem. 
\end{theorem}

Along the way, we introduce a novel dual-based local search algorithm with a $\stapx\approx\strgh$ approximation factor for the Steiner Tree problem, complemented by a lower bound of $1.5$ on its approximation factor. 

\begin{theorem} 
\label{thm:main_steiner_tree}
    There exists a deterministic polynomial-time algorithm that achieves a $(\stapx)$-approximation for the Steiner Tree problem.
\end{theorem}

\paragraph{Other related works.}
A notable problem of similar caliber in combinatorial optimization to the Steiner Forest problem is the Traveling Salesman Problem (TSP). 
Recent advancements, such as the Best Paper at STOC 2021, which presented a randomized $(3/2 - \varepsilon)$-approximation algorithm for metric TSP~\cite{DBLP:conf/stoc/KarlinKG21} (later derandomized in~\cite{DBLP:conf/ipco/KarlinKG23}), highlight the progress in improving approximation algorithms for the classic TSP.
While this algorithm does not break the integrality gap of the natural linear programming (LP) relaxation for TSP~\cite{DBLP:conf/focs/KarlinKG22}, it represents a major step forward in surpassing a long-standing approximation barrier.
In contrast, our approach for the Steiner Forest problem provides a deterministic algorithm that breaks the long-standing $2$-approximation barrier by overcoming the integrality gap of its natural LP relaxation~\cite{DBLP:journals/siamcomp/AgrawalKR95,DBLP:journals/siamcomp/GoemansW95}.

Generalizations of the Steiner Tree and Steiner Forest problems have been extensively studied in the literature.
One notable extension is the Prize-Collecting Steiner Tree problem, which initially admitted a $2$-approximation algorithm~\cite{DBLP:journals/siamcomp/GoemansW95}, later improved to $1.96$~\cite{DBLP:journals/siamcomp/ArcherBHK11}, and more recently to a $1.79$-approximation~\cite{DBLP:conf/stoc/AhmadiGHJM24}.
In contrast, the Prize-Collecting Steiner Forest problem saw a breakthrough only recently, with the first $2$-approximation algorithm~\cite{10.1145/3722551} (JACM'25), following numerous $3$-approximation algorithms~\cite{DBLP:conf/soda/HajiaghayiJ06,DBLP:conf/soda/GuptaKLRS07,DBLP:conf/latin/HajiaghayiN10}. 
See~\cite{BateniH12,HajiaghayiKKN12,SharmaSW07} for additional work on special graph classes and generalizations of the Prize-Collecting Steiner Forest problem.

Similarly, the $k$-MST problem has undergone a long sequence of improvements~\cite{DBLP:conf/soda/RaviSMRR94,DBLP:conf/stoc/AwerbuchABV95,DBLP:journals/jcss/BlumRV99,DBLP:conf/focs/Garg96,DBLP:journals/ipl/AryaR98,DBLP:journals/mp/AroraK06} to eventually achieve a $2$-approximation algorithm~\cite{DBLP:conf/stoc/Garg05}. 
The $k$-Steiner Tree problem has been reduced to $k$-MST, benefiting from these advancements~\cite{DBLP:journals/siamcomp/BateniHL18}.
However, the $k$-Steiner Forest problem has been shown to be roughly as hard as the densest $k$-subgraph problem~\cite{DBLP:conf/soda/HajiaghayiJ06}, for which the current best approximation factor is $O(n^{1/4})$~\cite{DBLP:conf/stoc/BhaskaraCCFV10}.  

Finally, while no prior work has improved upon the $2$-approximation factor for general Steiner Forest instances, better approximations have been achieved for special graph classes.  
For example, a PTAS is known for planar graphs and, more generally, for graphs with bounded genus~\cite{DBLP:journals/jacm/BateniHM11}.  
Additionally, the Euclidean Steiner Forest problem also admits a PTAS~\cite{DBLP:journals/talg/BorradaileKM15}.  

\section{Algorithms and Intuition}
In this section, we present the intuition behind our algorithm, along with a more detailed explanation of its components and the techniques outlined in Section~\ref{sec:overview_techniques}. 

To begin, we introduce the concept of \emph{moat growing}, originally defined by Jünger and Pulleyblank in 1991~\cite{DBLP:journals/algorithmica/JungerP95}.  
We then describe the \emph{Legacy Moat Growing} algorithm, which is equivalent to the primal-dual algorithms of~\cite{DBLP:journals/siamcomp/AgrawalKR95,DBLP:journals/siamcomp/GoemansW95}.  
Moat growing algorithms maintain a forest, marking a subset of its connected components as \emph{active sets} and expanding these sets uniformly over time.  
As the active sets grow, they \emph{color} the uncolored portions of their adjacent edges. Once an edge is fully colored, it is added to the forest.
Since the cost of the forest produced by a moat growing algorithm is upper bounded by twice the total growth of active sets, we introduce the notion of \emph{assignments}, which distribute this growth among the vertices responsible for it.  
This allows us to bound the assigned values rather than the total growth, providing an upper bound on the solution cost.
In this work, we leverage the versatility of moat growing and coloring.  
See Section~\ref{sec:mod-goemans} for further details on moat growing algorithms, particularly the \BBG{} procedure.

Starting with Legacy Moat Growing as the base algorithm, we introduce a novel \emph{local search} method that iteratively attempts to increase the active duration of an active set.
This operation, called a \emph{boost action}, is applied only when increasing the growth of an active set significantly reduces the total growth across all active sets, thereby leading to a smaller upper bound on the solution cost.
Figure~\ref{fig:wheel} illustrates how a boost action can reduce the total growth of active sets, leading to a better solution.
For a detailed explanation of the \LocalSearch{} procedure and its properties, see Section~\ref{sec:local_search}.

The local search terminates when no boost action can further reduce the total growth of active sets.
At this point, we derive key structural properties of the resulting moat growing algorithm.
Most importantly, for any subset of \emph{actively connected} vertices—those that remain in active sets until reaching one another—their total assigned value, excluding the \emph{maximum assigned value}, is at most a fraction (less than one) of the cost of any tree spanning them.
We will define the maximum assigned value more formally later; for now, it can be viewed as roughly the time when these vertices become connected.
See Section~\ref{sec:steiner_tree} for details of this property.

This property leads to a $\strgh$-approximation algorithm for the Steiner Tree problem, where all terminals are actively connected and the optimal solution spans them. 
While this bound may not be tight—the best lower bound we found is $1.5$ (see Appendix~\ref{eg:binary15})—we did not pursue a tighter analysis, as the current guarantee already suffices for our main goal of improving the Steiner Forest approximation.
The algorithm satisfying Theorem~\ref{thm:main_steiner_tree} is presented in Algorithm~\ref{alg:steiner-tree}, with its proof provided in Section~\ref{apx:steiner_tree}.

\begin{algorithm}[H]
  \caption{Steiner Tree Algorithm}
  \label{alg:steiner-tree}
  \hspace*{\algorithmicindent} \textbf{Input:} A graph $\G=(\V, \E, \cc)$, with edge costs $\cc: \E \rightarrow \mathbb{R}_{\ge 0}$, a demand function $\pair : \V \to \V$, corresponds to a Steiner Tree instance, and algorithm parameter $0 < \bta < 1$.\\
  \hspace*{\algorithmicindent} \textbf{Output:} A tree $T$ satisfying \pair~demands.
  \begin{algorithmic}[1]
    \Procedure{\streer}{$\G, \pair, \bta$}
      \label{func:stree}
      \State $F, \tplus \gets \BBG(G, \pair)$
      \State $\solone, \ttt \gets \LocalSearch(G, \tplus, \bta)$    
      \State \Return $\solone$
    \EndProcedure
  \end{algorithmic}
\end{algorithm}

\begin{figure}[t]
    \centering
    \begin{subfigure}{0.25\textwidth}
        \centering
        \begin{tikzpicture}[scale=0.7]
            \draw[White] (0,-3) -- (0, 3);
            \def\dem{Red!70}
            \def\ter{White}
            \def\col{Black!70}
            \def\noder{0.1cm}
            \foreach \i in {0,1,2,3} {
                \coordinate (P\i) at (\i*90:2);
            }
        
            \coordinate (C) at (0, 0);

            \foreach \i in {0,1,2,3} {
                \pgfmathtruncatemacro{\j}{mod(\i+1,4)} %
                \coordinate (C\i) at ({(\i*90+45)}:3); %
                \draw[\col] (P\i) .. controls (C\i) .. (P\j);
                \draw[\col] (P\i) -- (C);
            }
            \node[\col, above left=-2pt] at (C) {$v$};
            \node[\col] at (C0) {$2$};
            \draw[\col] (P0) -- node[above]{$1+\xi$} (C);
            
            \foreach \i in {0,1,2,3} {
                \draw[\col, fill = \col] (P\i) circle (\noder);
            }
            \draw[\col, fill= \ter] (C) circle (\noder);
            
        \end{tikzpicture}
        \caption{}
        \label{fig:wheel-input}
    \end{subfigure}
    \begin{subfigure}{0.4\textwidth}
        \centering
        \begin{tikzpicture}[scale=0.7]
            \draw[White] (0,-3) -- (0, 3);
            \def\dem{Red!70}
            \def\ter{White}
            \def\col{Black!70}
            \def\noder{0.1cm}
            \foreach \i in {0,1,2,3} {
                \coordinate (P\i) at (\i*90:2);
            }
        
            \coordinate (C) at (0, 0);

            \foreach \i in {0,1,2,3} {
                \pgfmathtruncatemacro{\j}{mod(\i+1,4)} %
                \coordinate (C\i) at ({(\i*90+45)}:3); %
                \draw[\dem, line width=2pt] (P\i) .. controls (C\i) .. (P\j) coordinate[pos=0.5] (M\i);
                \draw[\col] (P\i) -- (C);
                \draw[\dem, line width=2pt] (P\i) -- ($ (P\i)!0.91!(C)$);
            }
            
            \foreach \i in {0,1,2,3} {
                \draw[\col, fill = \col] (P\i) circle (\noder);
            }
            \draw[\col, fill = \ter] (C) circle (\noder);
            
            \draw[\dem, line width=0.3pt, dash pattern=on 1.2pt off 0.4pt] ($(P0) + (180:1.85)$) to[] (M0)
            to[out=45, in=-45] (M3)
            to[] ($(P0) + (180:1.85)$);

            \draw[\dem, line width=0.3pt, dash pattern=on 1.2pt off 0.4pt] ($(P1) + (-90:1.85)$) to[] (M0)
            to[out=45, in=135] (M1)
            to[] ($(P1) + (-90:1.85)$);

            \draw[\dem, line width=0.3pt, dash pattern=on 1.2pt off 0.4pt] ($(P2) + (0:1.85)$) to[] (M1)
            to[out=135, in=-135] (M2)
            to[] ($(P2) + (0:1.85)$);

            \draw[\dem, line width=0.3pt, dash pattern=on 1.2pt off 0.4pt] ($(P3) + (90:1.85)$) to[] (M2)
            to[out=-135, in=-45] (M3)
            to[] ($(P3) + (90:1.85)$);

        \end{tikzpicture}
        \caption{}
        \label{fig:wheel-2}
    \end{subfigure}
    \begin{subfigure}{0.25\textwidth}
        \centering
        \begin{tikzpicture}[scale=0.7]
            \draw[White] (0,-3) -- (0, 3);
            \def\dem{Blue!70}
            \def\ter{White}
            \def\col{Black!70}
            \def\noder{0.1cm}
            \foreach \i in {0,1,2,3} {
                \coordinate (P\i) at (\i*90:2);
            }
        
            \coordinate (C) at (0, 0);

            \foreach \i in {0,1,2,3} {
                \pgfmathtruncatemacro{\j}{mod(\i+1,4)} %
                \coordinate (C\i) at ({(\i*90+45)}:3); %
                \draw[\col] (P\i) .. controls (C\i) .. (P\j) coordinate[pos=0.19] (M\i) coordinate[pos=0.81] (N\i);
                \draw[\dem, line width=2pt] (P\i) -- (C);
                \draw[\dem, line width=2pt] (P\i) -- ($ (P\i)!0.93!(C)$);
            }
            \foreach \i in {0,1,2,3} {
                \pgfmathtruncatemacro{\j}{mod(\i+1,4)}
                \draw[\dem, line width=2pt] (P\i) -- (M\i);
                \draw[\dem, line width=2pt] (P\j) -- (N\i);
            }
            \foreach \i in {0,1,2,3} {
                \draw[\col, fill = \col] (P\i) circle (\noder);
            }
            \draw[\col, fill = \ter] (C) circle (\noder);

            \foreach \i in {C, P0, P1, P2, P3} {
            \draw[\dem, line width=0.3pt, dash pattern=on 1.2pt off 0.4pt] (\i) circle (1);
            }
        \end{tikzpicture}
        \caption{}
        \label{fig:wheel-our}
    \end{subfigure}
    \caption{A Steiner Tree instance. (a) Initial configuration with four terminal vertices and a central vertex $v$, where $\xi > 0$ is sufficiently small. (b) After Legacy Moat Growing, moats centered on terminals expand, fully coloring the outer edges, resulting in a total growth of 4 and a solution cost of 8. (c) A boost action on $v$ connects moats earlier via central edges, reducing the total growth to $\frac{5}{2}(1+\xi)$ and the solution cost to $4(1+\xi)$.}
    \label{fig:wheel}
\end{figure}
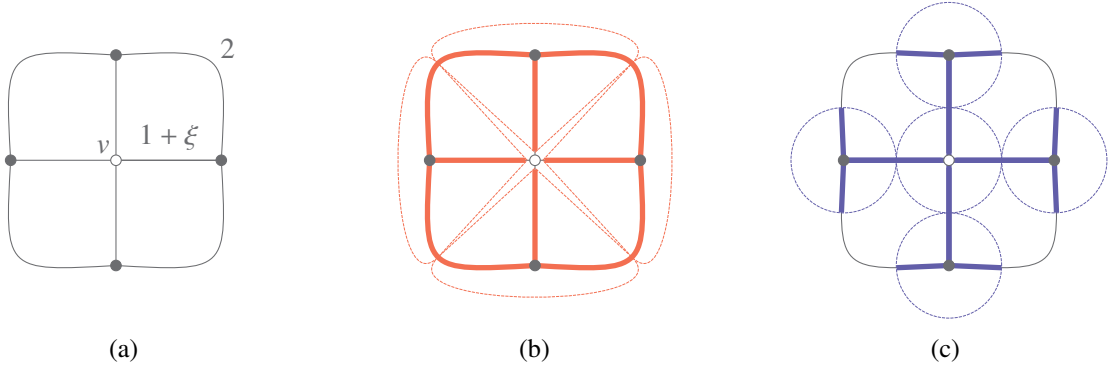 
Adapting our approach to the Steiner Forest problem introduces new challenges.  
Unlike in the Steiner Tree case, the vertices of each connected component in the optimal solution may not be actively connected in our algorithm, as some pairs can be satisfied earlier than others.  
This prevents us from directly applying the same bound developed for Steiner Tree.  
To address this, we introduce an \emph{extension} step and rerun the local search algorithm afterward.  
In the extension step, each active set receives a potential equal to $\eps$ times its growth time, allowing it to continue expanding using this potential before deactivation.  
This ensures, at a small cost, that the vertices of each connected component in the optimal solution become actively connected, enabling us to apply the same bounding technique as in the Steiner Tree case.  
If some components remain not actively connected after the extension, we provide an alternative upper bound to complete the analysis.  
For details on the extension step, including the \Extend{} procedure and its analysis, see Section~\ref{sec:extension}.

Finally, as mentioned earlier, the bounding argument for Steiner Tree does not account for the maximum assigned value within each connected component of the optimal solution.
This is not an issue if the connected components in our solution match those of the optimal solution.
However, problems arise when they do not align, and the maximum assigned value for a component becomes disproportionately large compared to the total assigned value of its vertices.
In such cases, vertices in those components often split into two groups that form quickly but take a long time to connect.
This leads to a large maximum assigned value, meaning these vertices remain active and continue growing for a long time relative to the cost of their tree in the optimal solution.
As a result, many edges are colored—edges that may later be used by pairs in other components—leading to inefficiencies.
See Figure~\ref{fig:ten-norm} for an example illustrating how these large values create challenges, preventing our local search and extension methods from achieving a better-than-2 approximation.

To address this, we introduce a novel technique called \emph{autarkic
pairs}, which refers to pairs of vertex subsets predicted to form the problematic structure described above.
These subsets are selected based on the observation that vertices in one subset correspond to pairs of vertices in the other.
While vertices within each subset connect quickly, it takes too long for the two groups to connect in our algorithm.
From each autarkic pair, we select one vertex pair, add the shortest path between them to the solution, and directly connect them with a zero-cost edge.
This ensures that all vertices in the autarkic pair use this shortcut to reach their counterparts, preventing moats from growing excessively and reducing the maximum assigned value for those vertices.
Once autarkic pairs are connected, we proceed with Legacy Moat Growing algorithm on the modified graph with these added links to satisfy all demands.
Figure~\ref{fig:ten} illustrates how this process works.
For details on how \candidate{} works and the proofs of its guarantees, see Section~\ref{sec:candidate}.

\tikzset{
  triangle/.style args={#1}{
    insert path={
        ++(90:1.2*#1)
      -- ++($(90:-1.2*#1)+(210:1.2*#1)$)
      -- ++($(210:-1.2*#1)+(330:1.2*#1)$)
      -- cycle %
    }
  }
}
\tikzset{
  pentagon/.style args={#1}{
    insert path={
      ++(90:1.05*#1)
      -- ++($(90:-1.05*#1)+(162:1.05*#1)$)
      -- ++($(162:-1.05*#1)+(234:1.05*#1)$)
      -- ++($(234:-1.05*#1)+(306:1.05*#1)$)
      -- ++($(306:-1.05*#1)+(18:1.05*#1)$)
      -- cycle
    }
  }
}
\tikzset{
  square/.style args={#1}{
    insert path={
        ++(45:1.1*#1)
      -- ++($(45:-1.1*#1)+(135:1.1*#1)$)
      -- ++($(135:-1.1*#1)+(225:1.1*#1)$)
      -- ++($(225:-1.1*#1)+(315:1.1*#1)$)
      -- cycle %
    }
  }
}
\tikzset{
  star/.style args={#1}{
    insert path={
      ++(90:1.1*#1)
      \foreach \a in {90,162,234,306,378} {
        -- ++($(\a:-1.1*#1) + (\a+36:0.5*#1)$)
        -- ++($(\a+36:-0.5*#1) + (\a+72:1.1*#1)$)
      }
      -- cycle
    }
  }
}
\begin{figure}[t]
    \centering
    \begin{subfigure}{0.22\textwidth}
    \centering
\begin{tikzpicture}[scale=0.7]

\draw[White] (0,-1.1) -- (0, 7.1);
\def\dem{Red!70}
\def\ter{White}
\def\col{Black!70}
\def\noder{0.1cm}

\def\n{3}
\pgfmathtruncatemacro{\top}{\n+1}

\def\tercl#1{%
  \ifcase#1
    White\or
    Green!20\or
    Sepia!20\or
    Plum!20%
  \fi
}

\def\k{1}

\def\angle{35}
\def\smallRotate{40}

\def\petalLength{0.5}
\def\forkLength{\petalLength/3}
\def\heightStep{1.5}
\def\widthStep{1.5}

\coordinate (A0) at (0,0);
\coordinate (A\top) at (0,\heightStep*\top);
\draw[\col] (A0) to[out=0, in=0] node[name=MID]{} (A\top);
\node[\col, right=-2pt] at (MID) {\small{2}};

\foreach \i in {1,...,\n} {
    \coordinate (B\i) at (-\widthStep,\heightStep*\i);
    \coordinate (A\i) at (0,\heightStep*\i);

    \draw[\col] (A\i) -- node[above=-2pt, name=L\i]{\small{1}} (B\i);
}
\foreach \i in {0,...,\n} {
    \pgfmathtruncatemacro{\nxt}{\i+1};
    \draw[\col] (A\i) -- node[name=LLL\i,left=-2pt]{\small{1}} (A\nxt);
}

\foreach \i in {1,...,\n} {
    \foreach \S/\d in {A\i/1, B\i/-1} {
        \foreach \j in {-\k,...,\k} {
            \def\currentAngle{\j*\angle}
            \coordinate (\S\j) at ($(\S) + (\currentAngle:\d*\petalLength)$);
            
            \draw[\col] (\S) -- node[name=LL\S\j]{} (\S\j);
        }
    }
    \node[\col, above=0pt] at (LLB\i-1){\small{$\xi$}};
    \node[\col, above=0pt] at (LLA\i1){\small{$\xi$}};
}
\node[\col, below=0pt] at (A0) {\small{$u$}};
\node[\col, above=0pt] at (A\top) {\small{$v$}};
\draw[\col, fill= \col] (A0) [triangle=\noder];
\draw[\col, fill= \col] (A\top) [triangle=\noder];
\foreach \i in {1,...,\n} {
    \foreach \S/\d in {A\i/1, B\i/-1} {
        \draw[\col, fill= \ter] (\S) circle(\noder); 
        \foreach \j in {-\k,...,\k} {
        }
    }
}
\foreach \S/\d in {A/1, B/-1} {
    \foreach \j in {-\k,...,\k} {
        \draw[\col, fill= \col] (\S1\j) [square=\noder];
        \draw[\col, fill= \col] (\S2\j) [star=\noder];
        \draw[\col, fill= \col] (\S3\j) [pentagon=\noder];
    }
}

\end{tikzpicture}
    \caption{}
    \label{fig:ten-inp}
    \end{subfigure}
    \hfill
    \begin{subfigure}{0.22\textwidth}
    \centering
\begin{tikzpicture}[scale=0.7]
\draw[White] (0,-1.1) -- (0, 7.1);
\def\dem{Red!70}
\def\ter{White}
\def\col{Black!70}
\def\noder{0.1cm}

\def\n{3}
\pgfmathtruncatemacro{\top}{\n+1}

\def\tercl#1{%
  \ifcase#1
    White\or
    Green!20\or
    Sepia!20\or
    Plum!20%
  \fi
}

\def\k{1}

\def\angle{35}
\def\smallRotate{40}

\def\petalLength{0.5}
\def\forkLength{\petalLength/3}
\def\heightStep{1.5}
\def\widthStep{1.5}

\coordinate (A0) at (0,0);
\coordinate (A\top) at (0,\heightStep*\top);
\draw[\col] (A0) to[out=0, in=0] (A\top);

\draw[\col, postaction={decorate, decoration={
    markings,%
    mark=between positions 0 and .096 step 0.1mm with {\draw[\dem, line width=1.5pt] (0,0) -- (0.2mm, 0);},
    mark=between positions 0.9 and 0.99 step 0.1mm with {\draw[\dem, line width=1.5pt] (0,0) -- (0.2mm, 0);}
}}] (A0) to[out=0, in=0] (A\top);

\foreach \i in {1,...,\n} {
    \coordinate (B\i) at (-\widthStep,\heightStep*\i);
    \coordinate (A\i) at (0,\heightStep*\i);

    \draw[\dem, line width=1.5pt] (A\i) -- (B\i);
}
\foreach \i in {0,...,\n} {
    \pgfmathtruncatemacro{\nxt}{\i+1};
    \draw[\dem, line width=1.5pt] (A\i) -- (A\nxt);
}

\draw[\dem, line width=0.3pt, dash pattern=on 1.2pt off 0.4pt] (A0) ellipse[x radius=\widthStep/2, y radius=\heightStep/2];
\draw[\dem, line width=0.3pt, dash pattern=on 1.2pt off 0.4pt] (A\top) ellipse[x radius=\widthStep/2, y radius=\heightStep/2];

\foreach \i in {1,...,\n} {
    \foreach \S/\d in {A\i/1, B\i/-1} {
        \draw[\dem, line width=0.3pt, dash pattern=on 1.2pt off 0.4pt] (\S) ellipse[x radius=\widthStep/2, y radius=\heightStep/2];
        \foreach \j in {-\k,...,\k} {
            \def\currentAngle{\j*\angle}
            \coordinate (\S\j) at ($(\S) + (\currentAngle:\d*\petalLength)$);
            
            \draw[\dem, line width=1.5pt] (\S) -- (\S\j);
        
            \draw [\dem, line width=0.2pt, dash pattern=on 1pt off 0.2pt] plot [smooth cycle] coordinates {
                ($(\S\j) + (\currentAngle+\smallRotate:\d*\forkLength)$)
                (\S)
                ($(\S\j) + (\currentAngle-\smallRotate:\d*\forkLength)$)
            };
        }
    }
}
\draw[\col, fill= \col] (A0) [triangle=\noder];
\draw[\col, fill= \col] (A\top) [triangle=\noder];
\foreach \i in {1,...,\n} {
    \foreach \S/\d in {A\i/1, B\i/-1} {
        \draw[\col, fill= \ter] (\S) circle(\noder); 
        \foreach \j in {-\k,...,\k} {
        }
    }
}
\foreach \S/\d in {A/1, B/-1} {
    \foreach \j in {-\k,...,\k} {
        \draw[\col, fill= \col] (\S1\j) [square=\noder];
        \draw[\col, fill= \col] (\S2\j) [star=\noder];
        \draw[\col, fill= \col] (\S3\j) [pentagon=\noder];
    }
}
\end{tikzpicture}
    \caption{}
    \label{fig:ten-norm}
    \end{subfigure}
    \hfill
    \begin{subfigure}{0.22\textwidth}
    \centering
\begin{tikzpicture}[scale=0.7]
\draw[White] (0,-1.1) -- (0, 7.1);
\def\dem{Red!70}
\def\ter{White}
\def\col{Black!70}
\def\noder{0.1cm}

\def\n{3}
\pgfmathtruncatemacro{\top}{\n+1}

\def\tercl#1{%
  \ifcase#1
    White\or
    Green!20\or
    Sepia!20\or
    Plum!20%
  \fi
}

\def\k{1}

\def\angle{35}
\def\smallRotate{40}

\def\petalLength{0.5}
\def\forkLength{\petalLength/3}
\def\heightStep{1.5}
\def\widthStep{1.5}

\coordinate (A0) at (0,0);
\coordinate (A\top) at (0,\heightStep*\top);
\draw[\col] (A0) to[out=0, in=0] (A\top);

\foreach \i in {1,...,\n} {
    \coordinate (B\i) at (-\widthStep,\heightStep*\i);
    \coordinate (A\i) at (0,\heightStep*\i);

    \draw[\col] (A\i) -- (B\i);
}
\foreach \i in {0,...,\n} {
    \pgfmathtruncatemacro{\nxt}{\i+1};
    \draw[\col] (A\i) -- (A\nxt);
}

\foreach \i in {1,...,\n} {
    \foreach \S/\d in {A\i/1, B\i/-1} {
        \foreach \j in {-\k,...,\k} {
            \def\currentAngle{\j*\angle}
            \coordinate (\S\j) at ($(\S) + (\currentAngle:\d*\petalLength)$);
            
            \draw[\col] (\S) -- (\S\j);
        }
    }
    \draw[\col, dash pattern=on 3pt off 0.5pt] (A\i-\k) to[out=-150, in=-30] node[below, name=M\i]{} (B\i\k);
}
\node[\col, below=-5pt] at (M1) {\small{0}};

\draw[\col, fill= \col] (A0) [triangle=\noder];
\draw[\col, fill= \col] (A\top) [triangle=\noder];
\foreach \i in {1,...,\n} {
    \foreach \S/\d in {A\i/1, B\i/-1} {
        \draw[\col, fill= \ter] (\S) circle(\noder); 
        \foreach \j in {-\k,...,\k} {
        }
    }
}
\foreach \S/\d in {A/1, B/-1} {
    \foreach \j in {-\k,...,\k} {
        \draw[\col, fill= \col] (\S1\j) [square=\noder];
        \draw[\col, fill= \col] (\S2\j) [star=\noder];
        \draw[\col, fill= \col] (\S3\j) [pentagon=\noder];
    }
}

\end{tikzpicture}
    \caption{}
    \label{fig:ten-zero-edge}
    \end{subfigure}
    \hfill
        \begin{subfigure}{0.22\textwidth}
        \centering
\begin{tikzpicture}[scale=0.7]
\draw[White] (0,-1.1) -- (0, 7.1);
\def\dem{Blue!70}
\def\ter{White}
\def\col{Black!70}
\def\noder{0.1cm}

\def\n{3}
\pgfmathtruncatemacro{\top}{\n+1}

\def\tercl#1{%
  \ifcase#1
    White\or
    Green!20\or
    Sepia!20\or
    Plum!20%
  \fi
}

\def\k{1}

\def\angle{35}
\def\smallRotate{40}

\def\petalLength{0.5}
\def\shortLength{0.3}
\def\forkLength{\petalLength/3}
\def\heightStep{1.5}
\def\widthStep{1.5}

\coordinate (A0) at (0,0);
\coordinate (A\top) at (0,\heightStep*\top);
\draw[\dem, line width=1.5pt] (A0) to[out=0, in=0] node[name=MID]{} (A\top);

\foreach \i in {1,...,\n} {
    \coordinate (B\i) at (-\widthStep,\heightStep*\i);
    \coordinate (A\i) at (0,\heightStep*\i);

    \draw[\col] (A\i) -- (B\i);
}
\foreach \i in {0,...,\n} {
    \pgfmathtruncatemacro{\nxt}{\i+1};
    \draw[\col] (A\i) -- (A\nxt);
}
\draw[\dem, line width=1.5pt] (A0) -- (A1);
\draw[\dem, line width=1.5pt] (A\n) -- (A\top);

\tikzset{
    center arc/.style args={#1:#2:#3}{
        insert path={+ (#1:#3) arc (#1:#1+#2:#3)}
    }
}

\draw[\dem, line width=0.3pt, dash pattern=on 1.2pt off 0.4pt] (A0) [center arc=-30:-165:2*\heightStep/3];
\draw[\dem, line width=0.3pt, dash pattern=on 1.2pt off 0.4pt] ($(A0) + (90:\heightStep/3)$) [center arc=-165:-110:2*\heightStep/3];
\draw[\dem, line width=0.3pt, dash pattern=on 1.2pt off 0.4pt] (A1) arc(-180:-263:2*\heightStep/3);
\draw[\dem, line width=0.3pt, dash pattern=on 1.2pt off 0.4pt] (MID) 
arc(-270:-209:2*\heightStep/3);
\draw[\dem, line width=0.3pt, dash pattern=on 1.2pt off 0.4pt] (MID) 
arc(-270:-400:2*\heightStep/3)
arc(-0:-60:2*\heightStep/3)
arc(-0:-100:2*\heightStep/3);

\draw[\dem, line width=0.3pt, dash pattern=on 1.2pt off 0.4pt] ($(A\top)$) [center arc=30:165:2*\heightStep/3];
\draw[\dem, line width=0.3pt, dash pattern=on 1.2pt off 0.4pt] ($(A\top) + (-90:\heightStep/3)$) [center arc=165:110:2*\heightStep/3];
\draw[\dem, line width=0.3pt, dash pattern=on 1.2pt off 0.4pt] (A\n) arc(180:263:2*\heightStep/3);
\draw[\dem, line width=0.3pt, dash pattern=on 1.2pt off 0.4pt] (MID) 
arc(270:209:2*\heightStep/3);
\draw[\dem, line width=0.3pt, dash pattern=on 1.2pt off 0.4pt] (MID) 
arc(270:400:2*\heightStep/3)
arc(0:60:2*\heightStep/3)
arc(0:100:2*\heightStep/3);

\foreach \i in {1,...,\n} {
    \foreach \S/\d in {A\i/1, B\i/-1} {
        \foreach \j in {-\k,...,\k} {
            \def\currentAngle{\j*\angle}
            \coordinate (\S\j) at ($(\S) + (\currentAngle:\d*\petalLength)$);
            
            \draw[\dem, line width=1.5pt] (\S) -- (\S\j);
        
            \draw [\dem, line width=0.2pt, dash pattern=on 1pt off 0.2pt] plot [smooth cycle] coordinates {
                ($(\S\j) + (\currentAngle+\smallRotate:\d*\forkLength)$)
                (\S)
                ($(\S\j) + (\currentAngle-\smallRotate:\d*\forkLength)$)
            };
        }
    }
    \draw[\dem, line width=1.5pt, dash pattern=on 3pt off 0.5pt] (A\i-\k) to[out=-150, in=-30] node[below, name=M\i]{} (B\i\k);
}
\draw[\col, fill= \col] (A0) [triangle=\noder];
\draw[\col, fill= \col] (A\top) [triangle=\noder];
\foreach \i in {1,...,\n} {
    \foreach \S/\d in {A\i/1, B\i/-1} {
        \draw[\col, fill= \ter] (\S) circle(\noder); 
        \foreach \j in {-\k,...,\k} {
        }
    }
}
\foreach \S/\d in {A/1, B/-1} {
    \foreach \j in {-\k,...,\k} {
        \draw[\col, fill= \col] (\S1\j) [square=\noder];
        \draw[\col, fill= \col] (\S2\j) [star=\noder];
        \draw[\col, fill= \col] (\S3\j) [pentagon=\noder];
    }
}
\end{tikzpicture}
    \caption{}
    \label{fig:ten-our}
    \end{subfigure}
    \caption{A Steiner Forest instance.
    (a) Initial configuration with $n$ rows (here, $n = 3$), and small $\xi \in (0, 1/n]$.
    Each row contains three demand pairs, and vertices $v$ and $u$ also form a required pair.
    A vertical path through the middle rows costs $n + 1$ and connects $v$ and $u$, while an edge of cost 2 offers an alternative.
    (b) After Legacy Moat Growing, vertices in each row quickly form two groups, but require significant growth to connect, resulting in coloring the vertical path.
    The total cost is $2n + 1 + O(1)$.
    (c) We detect vertices in each row as autarkic pairs and directly connect one vertex pair per row using their shortest paths, assuming a zero-cost edge connects them from this point on.
    (d) Running Legacy Moat Growing on the modified graph, the row vertices connect quickly via the new links, avoiding further growth and preventing coloring of the vertical path.
    Vertices $u$ and $v$ now connect via the edge of cost 2, and the total cost becomes $n + 2 + O(1)$.}
    \label{fig:ten}
\end{figure}
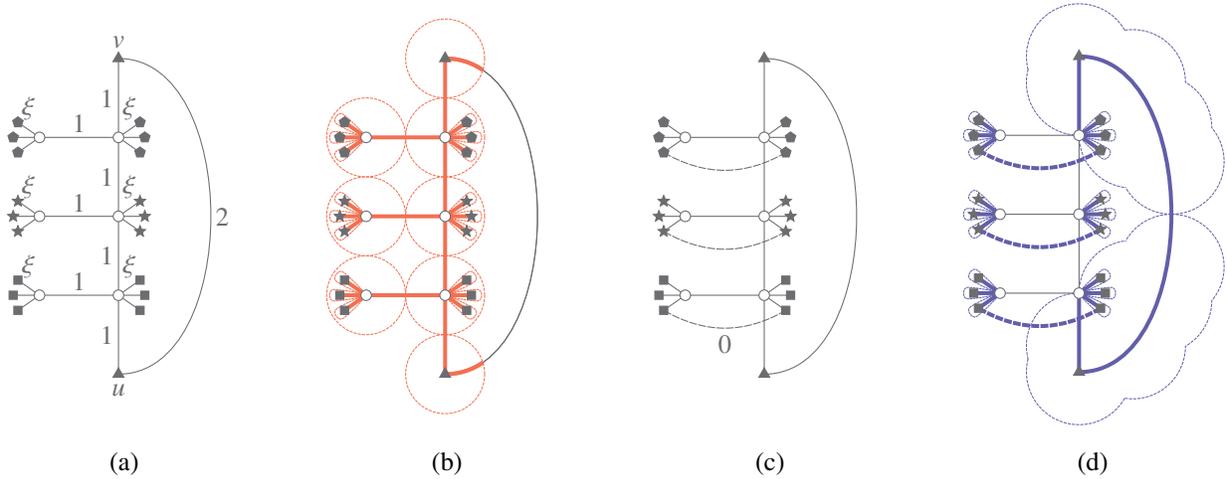 
The idea behind autarkic pairs is that, for any minimization problem with a known $\chi$-approximation algorithm, the approximation can be improved by identifying a structure whose addition allows part of the optimal solution to be removed without violating feasibility. If the cost of this structure is less than $\chi$ times the cost of the removed part, and the improvement accounts for a constant fraction of the optimal cost, applying the $\chi$-approximation algorithm to the remainder yields a better approximation. For Steiner Forest, we design autarkic pairs based on this natural idea and show that when the maximum assigned value for most optimal components is large enough, the resulting approximation is strictly better than 2.

\begin{algorithm}[b]
  \caption{Main Algorithm}
  \label{alg:main}
  \hspace*{\algorithmicindent} \textbf{Input:} A graph $\G = (\V, \E, \cc)$ with edge costs $\cc : \E \to \mathbb{R}_{\ge 0}$, a demand function $\pair : \V \to \V$, and algorithm parameters $0 < \bta, \eps, \ccnd < 1$.\\
  \hspace*{\algorithmicindent} \textbf{Output:} A forest $F$ satisfying \pair~demands.
  \begin{algorithmic}[1]
    \Procedure{\main}{$\G, \pair, \bta, \eps, \ccnd$}
      \State $F, \tplus \gets \BBG(\G, \pair)$
      \label{exe:base}
      \State $\solone, \ttt, \yyb \gets \LocalSearch(\G, \tplus, \bta)$
      
      \label{exe:localsearch}

      \State $\tp \gets \Extend(\G, \tplus, \ttt, \yyb,\eps)$ 
      \label{exe:egw}
      \State $\solext, \tz, \y^{b\prime\prime} \gets \LocalSearch(\G, \tp, \bta)$
      \label{exe:Xlocalsearch}
      
      \State $\solcnd \gets \candidate(\G, \pair, \yyb, \ccnd)$
      \label{exe:candid}
      
      \State \Return Best among \solone, \solcnd, 
      and \solext
      \label{line:best-sol}
    \EndProcedure
  \end{algorithmic}
\end{algorithm}

These approaches lead to a $(2 - 2\alf)$-approximation algorithm for the Steiner Forest problem, where $\alf \ge \frac{\alfbound}{2}$.
The algorithm satisfying Theorem~\ref{thm:main_steiner_forest} is presented in Algorithm~\ref{alg:main}.
It has three parameters $0 < \bta, \eps, \ccnd < 1$, each used by different modules of the algorithm.
The parameter values that achieve the desired approximation guarantee, along with the proof of the theorem, are provided in Section~\ref{sec:final}.

\subsection{Overview of Techniques}
\label{sec:overview_techniques}
In this section, we provide a comprehensive overview of the techniques that form the foundation of our algorithm.
We begin by discussing the classic 2-approximation algorithms introduced by~\cite{DBLP:journals/siamcomp/AgrawalKR95} and revisited by~\cite{DBLP:journals/siamcomp/GoemansW95}, which serve as the basis for our work.

\paragraph{Legacy Moat Growing.}
The Legacy Moat Growing algorithm starts with an empty forest $F$ and maintains its connected components, designating some as active sets.
An active set is any component that contains at least one \emph{unsatisfied} vertex, meaning a vertex that is not yet connected to its pair.

During the main phase of the algorithm, active sets grow uniformly, coloring their adjacent edges—those with exactly one endpoint in the set—at a constant rate.
Edges are viewed as curves with lengths equal to their costs.
Growth continues until an edge is fully colored, at which point it is added to $F$, merging the components at its endpoints.
Only edges between distinct components are colored and added to $F$.
Newly formed components become active or inactive depending on whether they contain any unsatisfied vertices.

The algorithm terminates when no active sets remain, ensuring all demand pairs are connected in the forest $F$.
Finally, in a pruning phase, any edge of $F$ that is the only edge crossing a set of vertices which was inactive at some point during the algorithm is removed, yielding the final solution.
An illustration of the algorithm is provided in Figure~\ref{fig:moats}, and a more detailed explanation appears in Section~\ref{subsec:def-legacy}.

\tikzset{
  triangle/.style args={#1}{
    insert path={
        ++(90:1.2*#1)
      -- ++($(90:-1.2*#1)+(210:1.2*#1)$)
      -- ++($(210:-1.2*#1)+(330:1.2*#1)$)
      -- cycle %
    }
  }
}
\tikzset{
  pentagon/.style args={#1}{
    insert path={
      ++(90:1.05*#1)
      -- ++($(90:-1.05*#1)+(162:1.05*#1)$)
      -- ++($(162:-1.05*#1)+(234:1.05*#1)$)
      -- ++($(234:-1.05*#1)+(306:1.05*#1)$)
      -- ++($(306:-1.05*#1)+(18:1.05*#1)$)
      -- cycle
    }
  }
}
\tikzset{
  square/.style args={#1}{
    insert path={
        ++(45:1.1*#1)
      -- ++($(45:-1.1*#1)+(135:1.1*#1)$)
      -- ++($(135:-1.1*#1)+(225:1.1*#1)$)
      -- ++($(225:-1.1*#1)+(315:1.1*#1)$)
      -- cycle %
    }
  }
}
\tikzset{
  star/.style args={#1}{
    insert path={
      ++(90:1.1*#1)
      \foreach \a in {90,162,234,306,378} {
        -- ++($(\a:-1.1*#1) + (\a+36:0.5*#1)$)
        -- ++($(\a+36:-0.5*#1) + (\a+72:1.1*#1)$)
      }
      -- cycle
    }
  }
}
\begin{figure}[t]
    \centering
    \begin{subfigure}{0.32\textwidth}
        \centering
        \begin{tikzpicture}[scale=0.7]

\draw[White] (0,-2.9) -- (0, 1.4);
\def\dem{Red!70}
\def\demf{Red!20}
\def\demi{Blue!70}
\def\demif{Blue!20}
\def\demii{Green!70}
\def\demiif{Green!20}
\def\ter{White}
\def\col{Black!70}
\def\treeone{Purple}
\def\treetwo{RubineRed!70!Black}
\def\noder{0.1cm}

\def\tercl#1{%
  \ifcase#1
    White\or
    Goldenrod!20\or
    olive!20\or
    Green!20\or
    Sepia!20\or
    Plum!20\or
  \fi
}

\def\len{4}

\coordinate (A) at (0, 0);
\coordinate (B) at ($(A) + (0:\len/1.5)$);
\coordinate (C) at ($(A) +(-70:\len/3)$);
\coordinate (D) at ($(C) +(-45:\len/4)$);

\foreach \i in {A, B, C, D} {
    \draw[\dem, 
    line width=0.3pt, 
    dash pattern=on 1.2pt off 0.4pt,
    pattern={
        Lines[angle=45, distance=2pt,  line width=0.3pt]%
    },
    pattern color=\demf
    ] (\i) circle (\len/8);
}
\draw[\col] (B) -- (A) -- (C) -- (D);
\draw[\dem, line width=1.5pt] (B) -- ($(B)!0.5cm!(A)$);
\draw[\dem, line width=1.5pt] (A) -- ($(A)!0.5cm!(B)$);
\draw[\dem, line width=1.5pt] (A) -- ($(A)!0.5cm!(C)$);
\draw[\dem, line width=1.5pt] (C) -- ($(C)!0.5cm!(A)$);
\draw[\dem, line width=1.5pt] (C) -- (D);
\foreach \i/\c in {A/star, B/star, C/square, D/square} {
    \draw[\col, fill= \col] (\i) [\c=\noder];
}
    \end{tikzpicture}
        \caption{}
        \label{fig:maots-first}
    \end{subfigure}
    \hfill
\begin{subfigure}{0.32\textwidth}
        \centering
        \begin{tikzpicture}[scale=0.7]

\draw[White] (0,-2.9) -- (0, 1.4);
\def\dem{Red!70}
\def\demf{Red!20}
\def\demi{Blue!70}
\def\demif{Blue!20}
\def\demii{Green!70}
\def\demiif{Green!20}
\def\ter{White}
\def\col{Black!70}
\def\treeone{Purple}
\def\treetwo{RubineRed!70!Black}
\def\noder{0.1cm}

\def\tercl#1{%
  \ifcase#1
    White\or
    Goldenrod!20\or
    olive!20\or
    Green!20\or
    Sepia!20\or
    Plum!20\or
  \fi
}

\def\len{4}

\tikzset{
    center arc/.style args={#1:#2:#3}{
        insert path={+ (#1:#3) arc (#1:#1+#2:#3)}
    }
}

\coordinate (A) at (0, 0);
\coordinate (B) at ($(A) + (0:\len/1.5)$);
\coordinate (C) at ($(A) +(-70:\len/3)$);
\coordinate (D) at ($(C) +(-45:\len/4)$);

\foreach \i in {A, B} {
    \draw[
    \demi, 
    line width=0.3pt, 
    dash pattern=on 1.2pt off 0.4pt,
    pattern={
        Lines[angle=45, distance=2pt,  line width=0.3pt]%
    },
    pattern color=\demif
    ] (\i) circle (\len/3 - \len/8);
}
\draw[\demi, 
    line width=0.3pt, 
    dash pattern=on 1.2pt off 0.4pt]
    (A);
\foreach \i in {A, B, C, D} {
    \draw[\dem, 
    line width=0.3pt, 
    dash pattern=on 1.2pt off 0.4pt,
    pattern={
        Lines[angle=45, distance=2pt,  line width=0.3pt]%
    },
    pattern color=\demf
    ] (\i) circle (\len/8);
}

\draw[\col] (B) -- (A) -- (C) -- (D);
\draw[\dem, line width=1.5pt] (B) -- ($(B)!0.5cm!(A)$);
\draw[\demi, line width=1.5pt] ($(B)!0.5cm!(A)$) -- ($(B)!0.833cm!(A)$);
\draw[\dem, line width=1.5pt] (A) -- ($(A)!0.5cm!(B)$);
\draw[\demi, line width=1.5pt] ($(A)!0.5cm!(B)$) -- ($(A)!0.833cm!(B)$);
\draw[\demi, line width=1.5pt] ($(A)!0.5cm!(C)$) -- ($(A)!0.833cm!(C)$);
\draw[\dem, line width=1.5pt] (A) -- ($(A)!0.5cm!(C)$);
\draw[\dem, line width=1.5pt] (C) -- ($(C)!0.5cm!(A)$);
\draw[\dem, line width=1.5pt] (C) -- (D);
\foreach \i/\c in {A/star, B/star, C/square, D/square} {
    \draw[\col, fill= \col] (\i) [\c=\noder];
}

    \end{tikzpicture}
     \caption{}
        \label{fig:maots-second}
    \end{subfigure}
    \hfill
\begin{subfigure}{0.32\textwidth}
        \centering
        \begin{tikzpicture}[scale=0.7]

\draw[White] (0,-2.9) -- (0, 1.4);
\def\dem{Red!70}
\def\demf{Red!20}
\def\demi{Blue!70}
\def\demif{Blue!20}
\def\demii{Green!70}
\def\demiif{Green!20}
\def\ter{White}
\def\col{Black!70}
\def\treeone{Purple}
\def\treetwo{RubineRed!70!Black}
\def\noder{0.1cm}

\def\tercl#1{%
  \ifcase#1
    White\or
    Goldenrod!20\or
    olive!20\or
    Green!20\or
    Sepia!20\or
    Plum!20\or
  \fi
}

\def\len{4}

\tikzset{
    center arc/.style args={#1:#2:#3}{
        insert path={+ (#1:#3) arc (#1:#1+#2:#3)}
    }
}

\coordinate (A) at (0, 0);
\coordinate (B) at ($(A) + (0:\len/1.5)$);
\coordinate (C) at ($(A) +(-70:\len/3)$);
\coordinate (D) at ($(C) +(-45:\len/4)$);

\foreach \i in {C, D} {
    \fill[
    \demii, 
    line width=0.1pt, 
    pattern={
        Lines[angle=45, distance=2pt,  line width=0.3pt]%
    },
    pattern color=\demiif
    ] (\i) circle (\len/4);
}
\foreach \i in {A, B} {
    \fill[
    \demii, 
    line width=0.1pt, 
    pattern={
        Lines[angle=45, distance=2pt,  line width=0.3pt]%
    },
    pattern color=\demiif
    ] (\i) circle (\len/3);
}
\draw[\demii, line width=0.3pt, dash pattern=on 1.2pt off 0.4pt] (B) circle(\len/3);
\draw[\demii, line width=0.3pt, dash pattern=on 1.2pt off 0.4pt] (A) [center arc=-26:272:\len/3];
\draw[\demii, line width=0.3pt, dash pattern=on 1.2pt off 0.4pt] (C) [center arc=15:27:\len/4];
\draw[\demii, line width=0.3pt, dash pattern=on 1.2pt off 0.4pt] (C) [center arc=179:76:\len/4];
\draw[\demii, line width=0.3pt, dash pattern=on 1.2pt off 0.4pt] (D) [center arc=195:240:\len/4];
\foreach \i in {A, B} {
    \draw[White, fill=White] (\i) circle (\len/3 - \len/8);
}
\foreach \i in {C, D} {
    \draw[White, fill=White] (\i) circle (\len/8);
}
\foreach \i in {A, B} {
    \draw[
    \demi, 
    line width=0.3pt, 
    dash pattern=on 1.2pt off 0.4pt,
    pattern={
        Lines[angle=45, distance=2pt,  line width=0.3pt]%
    },
    pattern color=\demif
    ] (\i) circle (\len/3 - \len/8);
}
\draw[\demi, 
    line width=0.3pt, 
    dash pattern=on 1.2pt off 0.4pt]
    (A);
\foreach \i in {A, B, C, D} {
    \draw[\dem, 
    line width=0.3pt, 
    dash pattern=on 1.2pt off 0.4pt,
    pattern={
        Lines[angle=45, distance=2pt,  line width=0.3pt]%
    },
    pattern color=\demf
    ] (\i) circle (\len/8);
}

\draw[\col] (B) -- (A) -- (C) -- (D);
\draw[\dem, line width=1.5pt] (B) -- ($(B)!0.5cm!(A)$);
\draw[\demi, line width=1.5pt] ($(B)!0.5cm!(A)$) -- ($(B)!0.833cm!(A)$);
\draw[\demii, line width=1.5pt] ($(B)!0.833cm!(A)$) -- ($(A)!0.833cm!(B)$);
\draw[\dem, line width=1.5pt] (A) -- ($(A)!0.5cm!(B)$);
\draw[\demi, line width=1.5pt] ($(A)!0.5cm!(B)$) -- ($(A)!0.833cm!(B)$);
\draw[\demi, line width=1.5pt] ($(A)!0.5cm!(C)$) -- ($(A)!0.833cm!(C)$);
\draw[\dem, line width=1.5pt] (A) -- ($(A)!0.5cm!(C)$);
\draw[\dem, line width=1.5pt] (C) -- ($(C)!0.5cm!(A)$);
\draw[\dem, line width=1.5pt] (C) -- (D);
\foreach \i/\c in {A/star, B/star, C/square, D/square} {
    \draw[\col, fill= \col] (\i) [\c=\noder];
}
    \end{tikzpicture}
    \caption{}
        \label{fig:maots-third}
    \end{subfigure}
    \caption{Here, there are two demand pairs: star pairs and square pairs.
    (a) Initially, all vertices form active sets and grow until the square pairs reach each other.
    (b) Once the connected component containing squares becomes inactive, the remaining connected components continue to grow until one of them reaches the square component.
    (c) At this point, there are two connected components that are both active and grow until they meet, causing the edge between the stars to become fully colored.
    Note that while all edges are added to $\F$ at this moment, the edge between a star and a square will be removed during the pruning phase since there was an inactive set containing squares that cut only this edge.}
    \label{fig:moats}
\end{figure}
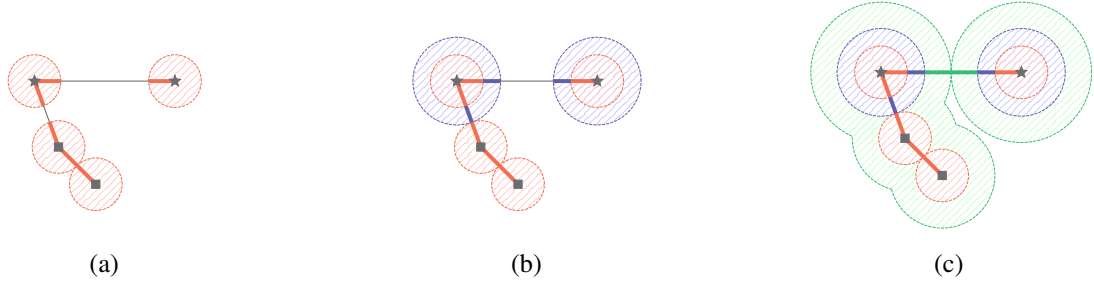
 
This algorithm achieves a 2-approximation because each portion of every edge is colored at most once, and each active set colors part of at least one edge from the optimal solution—specifically, from the paths connecting its unsatisfied vertices to their respective pairs.
Moreover, the forest $F$ is fully colored, with active sets coloring, on average, at most two edges in the final $F$.
Thus, the cost of the solution is at most twice that of the optimal solution.
Next, we take a closer look at the lower bound for the optimal solution and the upper bound for our solution.

As mentioned, the lower bound on the cost of the optimal solution is based on each active set coloring at least one edge of the optimal solution.
However, if active sets color at least two such edges, a higher lower bound is achieved, leading to a better approximation factor.
This idea was introduced by Ahmadi, Gholami, Hajiaghayi, Jabbarzade, and Mahdavi in their work on the Prize-Collecting Steiner Forest~\cite{10.1145/3722551} and Prize-Collecting Steiner Tree~\cite{DBLP:conf/stoc/AhmadiGHJM24}.
They classify active sets into \emph{single-edge sets} and \emph{multi-edge sets} based on how many edges they color from the optimal solution.
If most of the active sets are multi-edge sets, meaning they color more than one edge of the optimal solution, then the optimal cost must be significantly larger than the total growth of active sets, implying that the moat growing algorithm’s solution could be better than a 2-approximation.
Conversely, if most active sets are single-edge sets, this offers insight into the structure of the optimal solution and its relationship to active set growth.

We now highlight and make use of the fact that active sets color, on average, two edges in the solution. This holds because, at any moment during the algorithm, if we contract the final forest $F$ by its connected components at that time, all leaves in the resulting forest correspond to active sets.
Consequently, the average degree of active sets—that is, the number of edges in the final $F$ being colored at that moment—is at most 2.  
To maintain this property, we modify the growth rule so that the invariant of leaves being active always holds, ensuring the solution remains at most twice the total growth of active sets.
A key condition that preserves this structure is that an inactive set never becomes active unless it merges with another active set to form a larger component.  
Furthermore, by allowing vertex sets to remain active longer than in Legacy Moat Growing, connectivity between vertices increases, resulting in valid solutions.
Understanding the power of moat growing algorithms allows us to explore modifications that reduce the total growth of active sets.  
Since the solution cost is at most twice this value, such reductions lead to better upper bounds on the final cost.

\paragraph{Shadow Moat Growing.}
In Legacy Moat Growing algorithm, active sets are determined by unsatisfied demands, making it difficult to modify their behavior independently.
To address this, we introduce a modified version called Shadow Moat Growing.
This approach uses a \emph{fingerprint} from a moat growing algorithm, assigning a time value to each vertex to ensure that active sets containing those vertices remain active until that time.
Given this fingerprint, Shadow Moat Growing reruns a moat growing algorithm while preserving the intended active set behavior.
It also allows increasing the time value of any vertex, enabling certain active sets to remain active longer and adapt their behavior accordingly.
Details of this approach are provided in Section~\ref{subsec:def-shadow}.

\paragraph{Assignment.}
Since our solution is at most twice the total growth of active sets, our goal is to effectively bound this growth.
To do so, we define assignments, which distribute the growth of active sets across the vertices responsible for it.
For example, in Legacy Moat Growing algorithm, each active component is active due to specific unsatisfied vertices, and we assign the corresponding growth to those vertices.
However, modifying active set behavior complicates this process.
Our algorithm involves multiple executions of moat growing, each using a different method of assignment.
These assignments associate a value with each vertex such that the total assigned value captures the majority of the growth of active sets, allowing us to bound the solution cost by roughly twice this total.
We therefore focus on bounding the assigned values.

To define assignments, we prioritize vertices according to Legacy Moat Growing execution, giving higher priority to vertices that are satisfied later.
For each active set, we select the highest-priority vertex from each connected component of the optimal solution and decide how to distribute the growth among them.
The distribution method varies: in some assignments, the growth is partitioned; in others, it is fully allocated to each selected vertex.
Since these assignments are used solely for analysis and are not part of the actual algorithm, we assume access to the optimal solution for defining them.
The general definition of assignments and the one corresponding to Legacy Moat Growing are provided in Sections~\ref{subsec:monotonic_preliminary} and~\ref{subsec:legacy_execution}, with additional assignments introduced later for other moat growing variants.

Now that we have demonstrated the flexibility of the moat growing algorithm and introduced the concept of assignment, we turn to improving the approximation factor for the Steiner Forest problem.  
While no algorithm is currently known to achieve this, significant progress has been made on a well-known special case: the Steiner Tree problem.  
The 2-approximation algorithm for Steiner Forest naturally extends to Steiner Tree.  
Additionally, Steiner Tree admits a trivial 2-approximation via the minimum spanning tree on the metric closure of the terminals.  
Over the years, several algorithms have improved the approximation factor for this problem, many of which are based on local search~\cite{DBLP:journals/algorithmica/Zelikovsky93,DBLP:journals/jco/KarpinskiZ97,DBLP:journals/siamdm/RobinsZ05,DBLP:conf/soda/TraubZ22}.

\paragraph{Local Search.}  
Local search algorithms for Steiner Tree typically begin with an initial solution, often the 2-approximation, and iteratively add a constant-size subgraph.  
Consequently, some edges can be removed as long as all demands remain satisfied, with the goal of improving the solution.

However, despite being well studied, this approach has not led to improvements for Steiner Forest.  
We revisit this paradigm with a broader perspective.  
Traditional local search relies on adding a small subgraph to improve the solution, but such subgraphs may not exist in Steiner Forest instances.  
More precisely, the optimal solution may consist of many connected components, and improving the solution could require breaking these components and regrouping vertices differently (see Figure~\ref{fig:ten}).  
In Steiner Tree, knowing which vertices appear in the optimal solution is sufficient to reconstruct it.  
This gives hope that adding a small subset of vertices from the optimal solution could lead to improvement.  
In contrast, Steiner Forest requires knowing not only which vertices to include but also how they are grouped into components, making local improvements substantially more complex.  
The only known local search approach for Steiner Forest achieves a 69-approximation~\cite{DBLP:conf/innovations/0001G0MS0V18}.

Here, we propose a novel local search algorithm for Steiner Tree.  
Our approach is more general than previous ones, yielding a solution better than the 2-approximation for Steiner Tree and introducing structural properties that help improve the approximation factor for Steiner Forest.  
It builds on a deep understanding of moat growing algorithms, specifically why their solution cost is at most twice the total growth of active sets and how to preserve this relationship while modifying growth behavior.

Unlike traditional local search, our algorithm aims to improve the upper bound on the solution rather than the cost itself.  
We define a \emph{boost action}, which increases the fingerprint of a vertex, allowing its active set to remain active longer.  
In each iteration, we select a boost action that reduces the total growth of all active sets (see Figure~\ref{fig:wheel}).  
We define the \emph{win} of a boost action as the reduction in total growth of existing active sets, and the \emph{loss} as the additional growth introduced by the action.  
We apply a boost only when its win exceeds its loss by a factor of $1 + \bta$, for some $\bta \in (0, 1)$, calling it a \emph{valuable boost action}.  
Each valuable boost reduces the total growth of active sets, and since the solution cost is at most twice this total, the bound improves accordingly.  
If the total growth becomes small enough, we obtain a good solution.  
Details of the local search algorithm are provided in Section~\ref{sec:local_search}.

Notably, the cost of the solution may increase in a single iteration, but the upper bound always decreases.  
Over multiple iterations, the cost typically improves as well.  
An example in Appendix~\ref{eg:grid} demonstrates how repeated boosts can eventually yield the optimal solution, even if the first step does not reduce the cost.  
This illustrates the generality of our approach: unlike classical local search, which focuses on small local changes, our method can reshape the entire solution structure, even if it temporarily worsens.  
The key objective is to reduce the total growth of active sets, which upper bounds the solution cost, rather than minimizing the cost directly.  
This enables strategic decisions that ultimately lead to a good approximation.

\paragraph{Steiner Tree.}  
The development of a general local search technique provides several useful properties for the resulting moat growing algorithm.  
In particular, we observe structural features in moat growing algorithms where no valuable boost action can be applied.  
These properties hold for the algorithm produced by our local search.  
We now present a key result for such algorithms, which allows us to prove an approximation factor strictly better than 2 for the Steiner Tree problem.

The main result, which we call the \emph{claw property}, states that when no valuable boost action is available, the cost of any tree connecting three actively connected vertices can be lower bounded by a factor greater than one times the sum of their assigned values, excluding the maximum.
Here, actively connected means that the vertices remain in active sets until they all lie in the same connected component.
We formalize this property in Section~\ref{sec:claw_property} and extend it to larger sets of actively connected vertices in Section~\ref{sec:claw_property_extension}.
We then apply this result to the Steiner Tree problem in Section~\ref{apx:steiner_tree}, proving Theorem~\ref{thm:main_steiner_tree}.
Note that the claw property is also used in our autarkic pair method, which we explain later.

To extend the claw property to more vertices, we assume the tree connecting them is a binary tree with the relevant vertices placed as leaves.  
Any tree can be transformed into such a binary tree by duplicating vertices and adding zero-cost edges.  
We analyze the structure by selecting a vertex as a center point and fixing one leaf in each of its three directions.  
Applying the claw property at this center bounds the sum of the assigned values of the selected leaves, excluding the maximum, by their total distance to the center.  
This follows from the assumption that boosting the center point is not valuable.

By applying this process to each vertex of the tree and selecting the closest leaf in each direction, we generate a set of inequalities.  
Summing the left-hand sides gives a bound on the total assigned value of the vertices, excluding the maximum.  
To make this summation meaningful, we show that at any time $\tau$ during the moat growing algorithm, if $k$ active sets contain at least one of the relevant vertices (meaning at most $k$ assigned values exceed $\tau$), then there are $3(k - 1)$ center points whose inequalities contribute at that moment.  
We also show that the total of the right-hand sides is at most a constant factor times the cost of the tree.

Summing all inequalities, we find that the total assigned value to these vertices, minus the maximum, is at most $\frac{5 + \bta}{6}$ times the cost of the tree, where $\bta$ is the constant from our local search.  
Since all terminals in the Steiner Tree problem are actively connected in Legacy Moat Growing, and therefore also in the moat growing algorithm produced by our local search, we can apply this result to the terminals using the optimal solution as the connecting tree.  
This implies that our algorithm achieves a strictly better-than-2 approximation for Steiner Tree.

However, this analysis does not suffice to improve the approximation factor for Steiner Forest.  
The main issue is that we do not bound the maximum assigned value for each connected component of the optimal solution.  
This is not a problem when the connected components in our solution match those of the optimal solution, since a refined analysis shows that the cost of each component is bounded by twice the total assigned value minus the maximum.  
However, if our solution connects multiple components of the optimal solution, we need a bound on the maximum assigned value within each of those components (see Figure~\ref{fig:ten}).

Another challenge is that the claw property requires all vertices in a component to be actively connected.  
In Steiner Forest, demands may not enforce this, and some vertices may be satisfied early and become inactive.  
To apply the claw property, we must first ensure that most vertices in each optimal component remain actively connected.  
We begin by addressing this issue and then return to the problem of bounding the maximum assigned value for each component of the optimal solution.

\paragraph{Extension.}  
To address the issue of vertices within the same connected component of the optimal solution not being actively connected, we extend active sets by allocating an $\eps$-fraction of their growth time as potential. This potential allows active sets, or their supersets, to continue growing before deactivation, ensuring that more vertices become actively connected.  
More precisely, based on the growth of an active set before extension, we allocate $\eps$ times this growth as potential to one of its vertices that remained active from the beginning until the set was formed.  
We then run a new moat growing algorithm in which a connected component is considered active at time $\tau$ under two conditions. The first is that it contains a vertex that was active at time $\tau$ in the moat growing algorithm before the extension. If this condition is not met, the component remains active as long as it has remaining potential, which it consumes until exhausted.  
Whenever two active sets merge, their potentials are combined and transferred to the newly formed set.

After this extension step, which increases connectivity and merges some previously separate components, we hope that most vertices within each connected component of the optimal solution become actively connected.  
We then execute the local search to apply the same bounding technique used in the Steiner Tree case to the vertices of each optimal component that are actively connected.

On the other hand, if many vertices within a connected component remain not actively connected after the extension, we derive a stronger lower bound on the cost of the optimal solution.  
Since these vertices are not actively connected, those that become inactive earlier must, during their active sets' growth, color at least one edge of the optimal solution, as there must be a path connecting them to other vertices in the same component.  
Moreover, when these active sets grow solely due to potential, their vertices are already satisfied and remain active only because of the extension. 
In this case, they cannot color only one edge of the optimal solution, as such an edge would be redundant and contradict optimality.  
Therefore, these active sets must color at least two edges of the optimal solution and are classified as multi-edge sets.  
This provides a stronger lower bound than the previous assumption, where each active set was charged for coloring only one edge.  
With this improved lower bound, we can show that our solution is strictly better than twice the cost of the optimal solution.
More details on the algorithm and analysis of our extension approach are provided in Section~\ref{sec:extension}.

Note that while both the extension and the local search tend to increase the growth of active sets, they do so in different ways.
The local search focuses on incremental improvement by extending the active duration of a single vertex at a time, with the amount of increase varying across steps.
In contrast, the extension applies a small, uniform increase to the activity period of all active sets, relative to their previous growth.

To the best of our knowledge, the only prior work using a similar idea is by Bateni, Hajiaghayi, and Moharammi~\cite{DBLP:journals/jacm/BateniHM11}, who developed a PTAS for the Steiner Forest problem.  
They introduced a procedure called PC-Clustering, which, after constructing a solution, assigns each component a potential equal to $1/\eps$ times its cost and lets components grow and color adjacent edges in a separate phase.  
Like us, their goal is to connect vertices that belong to the same component of the optimal solution.
However, our approach differs in two key ways. 
First, in PC-Clustering, all components grow after the main algorithm has terminated, whereas we allow active sets to consume their potential immediately after deactivation, regardless of the state of other components. This ensures that the relevant vertices become actively connected, not just connected.  
Second, our potential is $\eps$ times the growth of each active set, a small value close to zero, while their potential is $1/\eps$ times the cost of the component’s tree, which is large.  
As a result, they rely on unconnected vertices being far apart after PC-Clustering, while we only need to ensure that a small amount of additional growth cannot connect such vertices. We also leverage the fact that this small growth results in coloring multiple edges of the optimal solution, leading to a stronger lower bound.

Next, we address the issue of bounding the maximum assigned value for each connected component of the optimal solution.

\paragraph{Autarkic Pairs.}  
As mentioned earlier, vertices in different connected components of the optimal solution can end up connected in our solution.
This creates a challenge, as we need to bound the maximum assigned value of vertices within each connected component of the optimal solution.
However, the bound derived from local search and used for the Steiner Tree case does not guarantee such a property.
If the total assigned value of the vertices in a component of the optimal solution (including the maximum assigned value) is a constant factor smaller than the cost of its tree, then that component already has low assigned value and can be ignored.
Likewise, if the maximum assigned value is small compared to the component’s cost, its contribution is negligible, and the lack of a bound does not matter.
The real difficulty arises when the total assigned value of a component is close to the cost of its optimal tree, and the maximum assigned value is large.

To address this, we try to identify such components despite not knowing the structure of the optimal solution.
We do this by selecting groups of vertices whose growth times are largely aligned, meaning that the majority of growth from their active sets comes from active sets whose unsatisfied vertices are exactly those vertices, and their pairs behave symmetrically.
We refer to these groups of vertices and their associated pairs as autarkic pairs.
We show that any connected component of the optimal solution whose total assigned value nearly matches its tree cost and has a large maximum assigned value must contain such autarkic pairs.

Typically, these components consist of two groups of vertices, each of which connects internally very quickly during the moat growing process, while connecting the two groups takes significantly longer.
This delay results in large assigned values for one vertex from each group.
Note that if a component contains three actively connected vertices with large assigned values, we can apply the claw property to conclude that the cost of connecting them in the optimal solution must be large, contradicting the assumption that the assigned value closely matches the tree cost.
This ensures that, in the moat growing algorithm produced by local search, such components must evolve as two early-formed groups that connect later.
Moreover, if the active sets corresponding to these groups contain unsatisfied vertices from other components of the optimal solution, they must be multi-edge sets, as they color multiple components.
This leads to a higher lower bound for the optimal solution, and similarly, the total assigned value no longer matches the tree cost.
Therefore, we may assume that these two groups mostly grow together and do not include vertices from other components of the optimal solution, which results in them being selected as autarkic pairs.
See Section~\ref{sec:candidate_properties} for further details.

After identifying autarkic pairs, we select the shortest path between an arbitrary pair of vertices from each pair and add it to our solution.
We then insert a zero-cost edge between the endpoints of these paths and run Legacy Moat Growing again on the modified graph.
In this new run, autarkic pairs no longer need to grow significantly to reach each other, as they are already connected.
This ensures that the maximum assigned value within each connected component remains small relative to the total assigned value, effectively resolving the issue.
See Figure~\ref{fig:ten} for further intuition and Section~\ref{sec:candidate} for a full description and analysis.

To analyze the solution obtained through this approach, we focus on active sets whose unsatisfied vertices belong to a group of an autarkic pair.  
If a constant fraction of the total growth of these active sets comes from multi-edge sets, then similar to the extension case, this case is already handled.  
Otherwise, if these active sets are single-edge sets, meaning they color exactly one edge of the optimal solution, we can safely remove those edges from the optimal solution.  
Since we have already added connecting paths between autarkic pairs, removing these edges does not violate any demands.  
This gives a valid solution for the remaining instance, with cost significantly smaller than the original optimal solution.  
As Legacy Moat Growing yields a 2-approximate solution for the new instance, removing part of the optimal solution improves the upper bound on our solution by approximately twice the cost of the removed edges.
While we do add some edges—the paths between autarkic pairs—their total cost is nearly equal to the cost of the removed edges.
Since the gain is roughly twice the cost and the added cost is only equal to it, the overall solution cost is reduced.
See Section~\ref{sec:candidate_analysis} for more details.

It is worth noting that the autarkic pair method, like our local search, is completely novel.  
However, it follows a natural idea for improving approximation guarantees in minimization problems.  
The idea is to find a structure such that, when selected and added to any solution, and its cost set to zero, the optimal solution improves by more than $1/\chi$ of its value, where $\chi$ is the approximation factor of the base algorithm.  
If the cost of the added structure is less than $\chi$ times the saved cost, then combining the $\chi$-approximate solution on the new instance with the added structure gives a better-than-$\chi$ approximation for the original instance.  
To achieve a constant-factor improvement, this saved cost must be a constant fraction of the optimal solution.  
For Steiner Forest, we design autarkic pairs to realize this idea in a problem-specific way.  
While this method does not always guarantee a constant-factor improvement, the extension and local search procedures complement it and handle the remaining cases.

\section{Preliminaries}
\label{sec:prelim}

\paragraph{Graph Notation.}  
We denote the cost of an edge $e$ by $\cc_e$. For convenience, we extend this to a cost function $\cc: 2^{\E} \to \mathbb{R}_{\geq 0}$ such that for any subset of edges $P \subseteq \E$, we define $\cc(P) = \sum_{e \in P} \cc_e$. This allows us to compute the cost of trees, forests, and other edge sets more easily. 

For any $S \subseteq \V$, we denote the adjacent edges of $S$ by $\delta(S)$, which is the set of edges with exactly one endpoint in $S$, i.e., the set of edges between $S$ and $\V \setminus S$ in $G$.  
Given sets $S, S' \subseteq \V$, we say that $S$ \emph{cuts} $S'$, denoted $S \odot S'$, if $S \cap S' \ne \emptyset$ and $S' \nsubseteq S$.  
For a tree $T$, we write $S \odot T$ to mean $S \odot V(T)$, and for a forest $F$, we write $S \odot F$ to mean that $S$ cuts at least one connected component of $F$.

We use $\dis(u, v)$ to denote the total cost of the shortest path between vertices $u$ and $v$ in the graph $\G$.

\paragraph{Set Theory Notation.}  
A family of sets is \emph{disjoint} if its elements are pairwise disjoint.  
For a set of vertices $\V$, we say a disjoint family $\mathcal{F}$ is a \emph{refinement} of another disjoint family $\mathcal{H}$ if, for every $A \in \mathcal{F}$, there exists $B \in \mathcal{H}$ such that $A \subseteq B$.  
A disjoint family is a \emph{partition} if the union of its sets equals the ground set.

\paragraph{Problem-Specific Notation.}  
Throughout the paper, we fix a specific optimal solution, which is a forest in $G$. Among all optimal solutions of minimum cost, we select one with the fewest edges; this is the solution we refer to as \emph{the} optimal solution.  
For simplicity, $\OPT$ refers both to the forest of the optimal solution and the partition of vertices it induces.  
Thus, we use $\opti \in \OPT$ to denote a connected component in $\OPT$, and $\cc(\OPT)$ to denote the cost of the optimal solution. For each component $\opti \in \OPT$, we denote its tree by $\T_{\opti}$. 
For any vertex $v$, $\optcom(v)$ denotes the connected component in $\OPT$ that contains $v$.

We denote the pair of a vertex $v$ by $\pairv$.  
We extend this notation and define a pair function $\pair: 2^{\V} \to 2^{\V}$ such that for any subset $S \subseteq \V$, we define $\pair(S) = \{\pairv \mid v \in S\}$.
A vertex is said to be \emph{satisfied} if it is connected to its pair, and \emph{unsatisfied} otherwise.  
We also define the following function:

\begin{definition}[$\unsatisfiedr$]
\label{def:unsatisfied}
For any subset of vertices $S \subseteq \V$, we define $\unsatisfiedr(S) \subseteq S$ as the subset of vertices in $S$ whose pair is not also in $S$:
$$
\unsatisfiedr(S) = \{v \in S \mid \pairv \notin S\}.
$$ \label{func:unsats}
\end{definition}

\section{Monotonic Moat Growing Algorithms}
\label{sec:mod-goemans}
Here, we provide a detailed exploration of monotonic moat growing algorithms.  
We begin in Section~\ref{subsec:def-legacy} by introducing the \BBG{} algorithm, which yields a 2-approximate solution.  
In Section~\ref{subsec:def-shadow}, we formalize the notion of monotonic moat growing and introduce the concept of fingerprints, which are used to implement Shadow Moat Growing, a technique capable of simulating any monotonic moat growing algorithm.  
Section~\ref{subsec:monotonic_preliminary} presents core definitions and properties relevant to monotonic moat growing algorithms.  
Finally, in Section~\ref{subsec:legacy_execution}, we analyze Legacy Execution, which corresponds to the call to \BBG{} in Line~\ref{exe:base} of our main algorithm, and examine its key properties.

\subsection{Legacy Moat Growing: The 2-Approximation Algorithm}
\label{subsec:def-legacy}

In this section, we present Legacy Moat Growing algorithm, a 2-approximation method introduced by Agrawal, Klein, and Ravi~\cite{DBLP:journals/siamcomp/AgrawalKR95}, and later used by Goemans and Williamson~\cite{DBLP:journals/siamcomp/GoemansW95}.  
This algorithm forms the foundation of our approach, and its pseudocode is presented in Algorithm~\ref{alg:bbg}.  
Although the original algorithm was introduced using a primal-dual approach, we omit the LP formulation here, as we can prove its approximation guarantee without it.

The algorithm maintains a forest $\F$, which is initially empty, along with a collection of its connected components, denoted as $\currentsets$.  
To track progress, the algorithm maintains a subset of $\currentsets$ called active sets, stored in $\activesets$.  
These active sets correspond to connected components that contain vertices whose pairs are not yet in the same connected component.  
We also refer to these active sets as moats.  
Additionally, the algorithm keeps track of connected components that were inactive at some point during the algorithm in $\deactivesets$.

The algorithm proceeds as a continuous process where active sets grow and color their adjacent edges.  
Edges are modeled as curves with a length equal to their cost, and each portion of an edge can be colored only once.  
Let $\ys$ represent the growth duration of the set $S \subseteq \V$.  
When an edge $e$ becomes fully colored, i.e., when  $\sum_{S:e\in \deltaS} \ys = \ce$, it is added to the forest $\F$, and the connected components of its endpoints are merged within $\currentsets$.  
The active sets $\activesets$ are then updated by removing the previous components and adding the newly merged one.  
If the new component no longer contains any unsatisfied vertices, it becomes inactive.  
The process continues until all demands are satisfied. 

After completing the main process, the algorithm enters a pruning phase.
During this phase, as long as there exists a subset of vertices $S \in \deactivesets$ that cuts exactly one edge of $\F$, that edge is removed from the forest.
This removal is safe because the pairs of all vertices inside $S$ are also contained within $S$, and removing the edge only disconnects vertices within $S$ from those outside it.
See Figure~\ref{fig:moats} for an illustration of the algorithm.

Since the algorithm operates as a continuous process, we define a \emph{moment} as a conceptual unit of time that increases continuously throughout the moat growing process.  
This notion of time reflects the progress of active sets as they grow and color edges.  
We denote the state of the algorithm at a specific moment as time $\currenttime$, with a small duration $\Delta$ such that no event occurs within the interval $(\currenttime, \currenttime + \Delta)$.
During this interval, the variables $\F$, $\currentsets$, and $\activesets$ remain unchanged, while the growth duration $\ys$ of active sets increases by $\Delta$, and the $\ys$ of other sets stay the same.

To make the algorithm efficient and polynomial-time, we run it in discrete steps.  
The current time $\currenttime$ is maintained as a discrete variable that increments only when a new event occurs.  
At each step, the algorithm calculates the next event time $\Delta_e$, which is the minimum duration needed to fully color an edge.  
The current time $\tau$ is then incremented by $\Delta_e$, and the growth duration $\ys$ of all active sets is updated by $\Delta_e$.  
Subsequently, all edges that become fully colored at this moment are added to $\F$, and the connected components and active sets are updated accordingly.  

Finally, the algorithm returns a function $t: V \rightarrow \mathbb{R}_{\ge 0}$ that records the first moment when vertices $v$ and $\pairv$ become connected.
This function is subsequently used in our main algorithm to facilitate running other moat growing algorithms on the given input.

\begin{algorithm}[ht]
  \caption{A 2-approximation Algorithm: Legacy Moat Growing}
  \label{alg:bbg}
  \hspace*{\algorithmicindent} \textbf{Input:} A graph $\G=(\V, \E, \cc)$ with edge costs $\cc: \E \rightarrow \mathbb{R}_{\ge 0}$, and demand function $\pair : \V \to \V$.\\
  \hspace*{\algorithmicindent} \textbf{Output:} A forest $F$ that satisfies all demands and a function $t: V \rightarrow \mathbb{R}_{\ge 0}$ indicating the earliest moment $v$ and \pairv~are connected.
  \begin{algorithmic}[1]
    \Procedure{\BBGr}{\G, $\pair$} \label{proc:bbg}
      \State $\tau \gets 0$
      \State $\F \gets \emptyset$
      \State $\currentsets \gets \{\{v\} \mid v \in \V\}$
      \State $\activesets \gets \{\{v\} \mid v \in \V, \pairv \neq v\}$ \label{line:legacy_define_acts}
      \State $\deactivesets \gets \{\{v\} \mid v\in \V, \pairv = v\}$
      \State $\tV \gets 0 \text{ for all } v\in\V$
      \State Implicitly set $\ys \gets 0$ for all $S \subseteq \V$
      
      \While{$\activesets \neq \emptyset$}  \Comment{While demands are not satisfied}
        \State $\Delta_e \gets \min_{e = uv \in E} \frac{c_e - \sum_{S \ni e} y_S}{|\{S_u, S_v\} \cap \activesets|}$, where $u\in S_u\in \currentsets$, $v\in S_v\in \currentsets$, and $S_u \ne S_v$
        \State $\tau \gets \tau + \Delta_e$
        \For{$S \in \activesets$}
            \State $\ys \gets \ys + \Delta_e$
        \EndFor
        \For{$e\in E$} 
          \State Let $S_v, S_u \in \currentsets$ be sets that contains each endpoint of $e$
          \If{$\sum_{S: e \in \deltaS} \ys = \ce$ \textbf{and} $S_v \neq S_u$}\Comment{Edge $(v, u)$ become fully colored}
            \State $\F \gets \F \cup \{e\}$
            \State $\currentsets \gets (\currentsets \setminus \{S_v, S_u\}) \cup \{S_v \cup S_u\}$
            
            \State $\activesets \gets (\activesets \setminus \{S_v, S_u\}) \cup \{S_v \cup S_u\}$
            \For{$w\in S_v \textbf{ such that } \pair_w \in S_u$}
                \State $t_w, t_{\pair_w} \gets \currenttime$ \label{line:legacy_define_t} \Comment{$w$ and $\pair_w$ just connected}
            \EndFor
          \EndIf
        \EndFor
        \For{$S \in \activesets$}
          \If{$\unsatisfiedS = \emptyset$} \label{line:legacy_deactivation_condition}\Comment{$S$ become inactive}
            \State $\activesets \gets \activesets \setminus \{S\}$    
            \State $\deactivesets \gets \deactivesets \cup \{S\}$
          \EndIf

        \EndFor
         
      \EndWhile      
      \While{$S\in \deactivesets \textbf{ exists such that } |\deltaS \cap \F|= 1$} \Comment{Remove unnecessary edges}
        \State $\F \gets \F \setminus \deltaS$
      \EndWhile
      \State \Return $\F, t$
    \EndProcedure
  \end{algorithmic}
\end{algorithm}

Note that our explanation of pruning based on deactivated sets differs slightly from the 2-approximation algorithm commonly used for the Steiner Forest problem, although our approach has been widely applied in related settings.
In the standard method for Steiner Forest, edges that are not required to satisfy any demands are removed.
The edges we remove form a subset of these, yet we still obtain a 2-approximate solution.
This distinction is important, as our goal is to design moat growing algorithms that operate independently of demand information during execution.

In each iteration of Algorithm~\ref{alg:bbg}, either an edge is added to the forest $F$, or an active set becomes inactive. Since both events can occur only a linear number of times in $\lvert \V \rvert$, the algorithm runs in polynomial time.
\begin{corollary}
    \label{cor:legacy_polynomial_time}
    The \BBG{} procedure runs in polynomial time.
\end{corollary}

\subsection{Generalize Moat Growing: Shadow Moat Growing Algorithm}
\label{subsec:def-shadow}

Here we propose a more abstract view of this algorithm, called monotonic moat growing algorithm.
Monotonic moat growing algorithm is a more general concept that Legacy Moat Growing algorithm falls within it.
This concept refer to set of algorithms similar to Legacy Moat Growing such that we will explore many other algorithms of this concept.

\begin{definition}[Monotonic Moat Growing Algorithm]
\label{def:monotonic}
A continuous algorithm where a forest $\F$, initially empty, is maintained.
A set of connected components of $\F$ are called active sets.
As time advance, $\ys$ of active sets increase uniformly and they color their adjacent edges uniformly.
When an edge become fully colored, it is added to $\F$, the connected components of its endpoints are merged, and its new component become active.
When an active set become inactive, it cannot activate again unless it merges with another active set, which form an active set that is union of those components.
After completing the main process, the algorithm enters a pruning phase where while there is a subset of vertices that were an inactive connected component at any moment of the algorithm and cuts exactly one edge of $\F$, that edge is removed from $\F$.
The remaining $\F$ is the final solution of the monotonic moat growing algorithm.
\end{definition}

Now, we want to define concept of fingerprint which then used by Shadow Moat Growing algorithm to simulate any monotonic moat growing algorithm on a given input.

\begin{definition}[Fingerprint]
    \label{def:fingerprint}
    Given a Steiner Forest instance and a monotonic moat growing algorithm, we define the \textit{fingerprint} of the execution of the algorithm on the given input as a function \(t: V \to \mathbb{R}_{\ge 0}\), which assigns to each vertex a value corresponds to time with the following properties:
    \begin{itemize}
        \item Every vertex $v$ should be in an active set from the beginning until $\tV$ in the algorithm.
        \item At any moment $\tau$, every active set $S$ must contain a vertex $v \in S$ such that $\tV \ge \tau$.
    \end{itemize}
\end{definition}

Note that a monotonic moat growing algorithm may admit infinitely many distinct fingerprints.

The procedure \BBG{} returns a function \( t: V \rightarrow \mathbb{R}_{\ge 0} \), where \( t_v \) is the earliest moment at which vertex \( v \) becomes connected to \( \pairv \) in the algorithm. 
We now verify that this function is indeed a fingerprint.

\begin{lemma}
\label{lm:tplus_is_fingerprint}
The output function \( t \) returned by $\BBG$ is a fingerprint.
\end{lemma}
\begin{proof}
We want to prove that  both conditions of Definition~\ref{def:fingerprint} holds.

\begin{itemize}
    \item For each vertex \( v \), it must be in  active sets from moment \( 0 \) up to moment \( \tV \):
    Initially, each vertex \( v \) which $\pairv\neq v$ is placed in an active singleton set in Line~\ref{line:legacy_define_acts}. 
    A connected components deactivated only if all its vertices are satisfied (see Line~\ref{line:legacy_deactivation_condition}). 
    Since vertex $v$ connect to $\pairv$ at time $t_v$ (see Line~\ref{line:legacy_define_t}), it is unsatisfied until that time.
    Therefore, every connected component containing $v$ before time $t_v$ has an unsatisfied vertex and should be active.

    \item At any moment \( \tau \), for each active set in that moment \( S \), there should be a vertex \( v \in S \) with \( \tV \ge \tau \):
    In \BBG{}, a connected component $S$ is active only if it contains at least one unsatisfied vertex \( v \in S \) (see Line~\ref{line:legacy_deactivation_condition}). Since \( v \) and \( \pairv \) are not yet connected at time \( \tau \) (see Line~\ref{line:legacy_define_t}), we have \( \tV > \tau \) which complete the proof.
\end{itemize}
    Since both conditions are satisfied, the function \( t \) returned by Algorithm~\ref{alg:bbg} is a fingerprint.
\end{proof}

Now we propose Shadow Moat Growing algorithm, which is a monotonic moat growing algorithm that utilizes fingerprints to simulate any monotonic moat growing algorithm.
The pseudocode for this algorithm can be found in Algorithm~\ref{alg:gbg}.
Although we never directly invoke this algorithm, it serves an essential purpose in demonstrating that monotonic moat growing algorithms, such as \BBG{}, can be interpreted through this framework.
This understanding allows us to modify the moat growing process by adjusting the fingerprint.
Furthermore, our algorithm invokes advanced versions of \GBG{}, such as Algorithms~\ref{alg:mod-gw-boost} and~\ref{alg:extend}.
These advanced algorithms take the fingerprint of a monotonic moat growing algorithm as input and assume that running \GBG{} with this fingerprint simulates the desired moat growing algorithm.
They then modify the fingerprint to obtain a new monotonic moat growing algorithm, allowing for flexibility and customization in the moat growing process.

The process of \GBG{} closely resembles that of \BBG{}.
The key difference lies in the criteria for marking a connected component as an active set.
In \BBG{}, a connected component is active if it contains an unsatisfied vertex.
In contrast, in \GBG{}, a connected component is considered an active set if it contains a vertex $v$ such that $t_v > \currenttime$, where $\currenttime$ denotes the current moment of the algorithm and $t$ represents the fingerprint given as input to \GBG.
In essence, the fingerprint $t_v$ enforces that any connected component containing vertex $v$ must remain active until time $t_v$.

\begin{algorithm}[ht]
  \caption{Shadow Moat Growing}
  \label{alg:gbg}
  \hspace*{\algorithmicindent} \textbf{Input:}
  A graph $G = (V, E, c)$ with edge costs $c: E \to \mathbb{R}_{\geq 0}$, and a function $t: V \to \mathbb{R}_{\geq 0}$ specifying the minimum time each vertex $v$ must be growing.\\
  \hspace*{\algorithmicindent} \textbf{Output:}  
$F$, the resulting forest of fingerprint $t$.
  \begin{algorithmic}[1]
    \Procedure{\GBGr}{$\G, t$} \label{proc:gbg}
      \State $\tau \gets 0$
      \State $\F \gets \emptyset$
      \State $\currentsets \gets \{\{v\} \mid v \in \V\}$
      
      \State $\activesets \gets \{\{v\} \mid v \in \V, \tV > 0\}$
      \State $\deactivesets \gets \{\{v\} \mid v\in \V, \tV = 0\}$
      
      \State Implicitly set $\ys \gets 0$ for $S \subseteq \V$
      \While{$\activesets \neq \emptyset$}  \Comment{While there exists an active set}
        \State $\Delta_e \gets \min_{e = uv \in E} \frac{c_e - \sum_{S \ni e} y_S}{|\{S_u, S_v\} \cap \activesets|}$, where $u\in S_u\in \currentsets$, $v\in S_v\in \currentsets$, and $S_u \ne S_v$

        \State $\Delta_t \gets \min_{v\in \V, \tV > \currenttime}(\tV - \currenttime)$
        \State $\Delta \gets \min(\Deltae, \Delta_t)$
        
        \For{$S \in \activesets$}
            \State $\ys \gets \ys + \Delta$ 
        \EndFor
        \State $\currenttime \gets \currenttime + \Delta$ 
        \For{$e\in E$}
          \State Let $S_v, S_u \in \currentsets$ be sets that contain each endpoint of $e$
          \If{$\sum_{S: e \in \deltaS} \ys = \ce$ \textbf{and} $S_v \neq S_u$}\Comment{Edge $(v, u)$ become fully colored}
            \State $\F \gets \F \cup \{e\}$
            \State $\currentsets \gets (\currentsets \setminus \{S_v, S_u\}) \cup \{S_v \cup S_u\}$
           
            \State $\activesets \gets (\activesets \setminus \{S_v, S_u\}) \cup \{S_v \cup S_u\}$
          \EndIf
        \EndFor
        \For{$S \in \activesets$}
          \If{$\tV \le \currenttime \textbf{ for all } v \in S$} \Comment{$S$ become inactive}
            \State $\activesets \gets \activesets \setminus \{S\}$    
            \State $\deactivesets \gets \deactivesets \cup \{S\}$
          \EndIf
        \EndFor
      \EndWhile
      
      \While{$S\in \deactivesets \textbf{ exists such that } |\deltaS \cap \F|= 1$}  \Comment{Remove unnecessary edges}
        \State $\F \gets \F \setminus \deltaS$
      \EndWhile
      
      \State \Return $F$
    \EndProcedure
  \end{algorithmic}
\end{algorithm}

Next, we aim to show that Shadow Moat Growing, given a fingerprint of a monotonic moat growing algorithm, accurately simulates that algorithm.

\begin{lemma}
\label{lm:ghost_equivalence}
    If $t$ is the fingerprint of a monotonic moat growing algorithm $A$ on a given graph $\G$, then for any moment during the process of $A$, the connected components, active sets, and the $\ys$ values in $A$ will be identical to those in $\GBG(\G, t)$ at that moment.
\end{lemma}
\begin{proof}
First, note that if at any moment until time $\currenttime$, the active sets in both algorithms are the same, then the $\ys$ values of sets $S \subseteq \V$ are also the same in both algorithms at moment $\currenttime$.
This is because $\ys$ represents the duration for which a component is active, and since active sets are identical in both algorithms, these values must also be the same.

Now, we show that at each moment $\currenttime$, the set of all connected components and active sets are identical in both algorithms. This directly implies the identicality of $\ys$, and consequently, the final forest $F$ in both algorithms $A$ and $\GBG$.

We prove this by induction on the events that change connected components and active sets in either of these algorithms. The induction base (at the start of the algorithm) is clear since all singleton sets are connected components in both algorithms.
Additionally, in $\GBG$, the singleton sets for vertices with $t_v > 0$ start as active sets, and $t_v>0$ for a vertex $v$ if and only if it starts as an active set in $A$.

Now, assume for contradiction that there exists a first moment $\tau$ at which the connected components or active sets differ between the two algorithms.
We consider three possible cases:

\begin{enumerate}
    \item An active set $S$ deactivates in $A$ but remains active in $\GBG$:
    Since $S$ is active in $\GBG$, there must be a vertex $v \in S$ such that $\tV > \tau$.
    By the first property in Definition~\ref{def:fingerprint}, $v$ must have remained in active sets of $A$ until $t_v$.
    Therefore, $S$ should still be active in $A$, which contradicts the assumption that $S$ has deactivated in $A$. Hence, this case is not possible.

    \item An active set $S$ deactivates in $\GBG$ but remains active in $A$:  
    According to the implementation of $\GBG$, an active set deactivates only when there is no vertex $v \in S$ such that $\tV > \tau$.  
    However, if $S$ is still active in $A$, then by the second property in Definition~\ref{def:fingerprint}, the fingerprint function $t$ must assign a value larger than $\tau$ to some vertex in $S$.  
    This contradicts the assumption that no vertex $v \in S$ has $\tV > \tau$.  
    Therefore, this case is also not possible.

    \item Two connected components merge in one algorithm but not in the other: 
    We have established that until $\tau$, the set of active sets, and therefore, $y_S$ values are identical in both algorithms. A merge occurs when an edge becomes fully colored, meaning $\sum_{S: e \in \deltaS} \ys = c_e$. Since being fully colored for an edge depends solely on $y_S$ values, which are identical in both algorithms, any merge that occurs in one algorithm must also occur in the other. Thus, this case is also impossible.
\end{enumerate}

Since all cases leading to a difference in connected components and active sets have been ruled out, it follows that these remain identical in both algorithms at all times. Consequently, the $y_S$ values are the same.
Additionally, the third point above shows that the same edge becomes fully colored in both algorithms, which ultimately proves that the forest $F$ remains identical in both algorithms throughout their execution.
Note that the pruning phase does not violate this consistency, as connected components and active sets remain identical at every moment in both algorithms, resulting in the same connected components becoming inactive at the same moments.  
Since the pruning phase removes edges based on inactive connected components, it follows that both algorithms prune the same set of edges.
\end{proof}

The next lemma is a classic result used in previous works that leverage a similar primal-dual approach for related problems. It shows that the final forest cost is at most twice the total growth of active sets.  
We prove this for any monotonic moat growing algorithm and then focus on designing new moat growing algorithms that reduce this total growth while satisfying all demands, aiming to produce improved solutions with a better upper bound.

\begin{lemma}
\label{lm:monotonic-forest-twice-growth}
    For any monotonic moat growing algorithm and any vertex $v \in \V$, the total cost of the resulting forest $F$ is at most twice the total growth of active sets not containing $v$, that is
    $$\cc(F) \le 2 \sum_{\substack{S\subseteq \V \\ v \notin S}}\ys.$$
\end{lemma}
\begin{proof}
    Since all edges in $F$ are fully colored, we can say
    \begin{align*}
    \cc(F)&=\sum_{e\in F} \cc_e \\
        &=\sum_{e\in F}\sum_{\substack{S\subseteq \V\\e\in \deltaS}} \ys\\
        &=\sum_{\substack{S\subseteq \V}}\lvert \deltaS\cap F\rvert\cdot \ys.
    \end{align*}
    Now, we show that at any moment $\currenttime$ of the monotonic moat growing, the increase in this value is at most the increase in \[2 \sum_{\substack{S\subseteq \V \\ v \notin S}}\ys.\]
    Let $\activesets_\currenttime$ be the active sets at moment $\currenttime$. Then, if $\ys$ of active sets increase by $\Delta$ at moment $\currenttime$, the increase in $\sum_{\substack{S\subseteq \V}}\lvert \deltaS\cap F\rvert\cdot \ys$ will be
    \begin{align*}
        \sum_{S\in \activesets_\currenttime} \lvert \deltaS\cap F\rvert\cdot \Delta,
    \end{align*}
    while the increase in 
    \(2 \sum_{S:v \notin S}\ys\)
    is
    \begin{align*}
    2 \sum_{\substack{S\in \activesets_\currenttime\\ v \notin S}}\Delta \geq 2(\lvert \activesets_\currenttime\rvert-1)\cdot \Delta.
    \end{align*}
    Therefore, it suffices to show that 
    \[
    \sum_{S\in \activesets_\currenttime} \lvert \deltaS\cap F\rvert \leq 2(\lvert \activesets_\currenttime\rvert-1)
    \]
    at any moment $\currenttime$.

    Given the components $\currentsets_\currenttime$ at moment $\currenttime$, define $\Fc$ as the graph obtained from $\F$ by contracting each set in $\currentsets_\currenttime$ into a single vertex and removing isolated vertices. Now, let $V_a$ represent the vertices in $\Fc$ corresponding to active sets $\activesets_\currenttime$, and $V_i$ the rest of the vertices corresponding to inactive components. Then, $\Fc$ has the following properties:
    \begin{itemize}
        \item $\Fc$ is a forest, since $\F$ is a forest and any contracted set cannot contain two vertices in one component of $\F$ without containing the path between them.
        \item For any vertex $v$ of $\Fc$, if $S$ is the corresponding set in $\currentsets_\currenttime$, $\deg_{\Fc}(v)=\lvert \deltaS\cap\F\rvert$. Additionally, isolated vertices excluded from $\Fc$ represent sets $S$ with $\lvert \deltaS\cap\F\rvert=0$.
        \item For any vertex $v$ in $V_i$, we have $\deg_{\Fc}(v)\geq 2$. $\Fc$ contains no isolated vertices, so $\deg_{\Fc}(v)\geq 1$. Now, assume otherwise that $\deg_{\Fc}(v)=1$. Then, consider the set $S$ corresponding to $v$. Since $v \in V_i$, $S$ must be in $\deactivesets$ at the end of the algorithm. However, the final check in the monotonic moat growing algorithm would remove the single edge attached to $S$ from $\F$ if this was the case. Hence, $\deg_{\Fc}(v)$ must be at least $2$.
    \end{itemize}
    Now, combining these properties, we can say
    \begin{align*}
        \sum_{S\in \activesets_\currenttime} \lvert \deltaS\cap F\rvert &= \sum_{v\in V_a} \deg_{\Fc}(v)\\
        &=\sum_{v\in V_a \cup V_i} \deg_{\Fc}(v) - \sum_{v\in V_i} \deg_{\Fc}(v)\\
        &\leq \sum_{v\in V_a \cup V_i} \deg_{\Fc}(v) - 2\lvert V_i \rvert \tag{$\deg_{\Fc}(v) \geq 2$ for all $v\in V_i$}\\
        &\leq 2(\lvert V_a \rvert + \lvert V_i\rvert - 1) - 2\lvert V_i \rvert
        \tag{$\Fc$ is a forest with at most $\lvert V_a \rvert + \lvert V_i\rvert - 1$ edges}\\
        &\leq 2(\lvert V_a\rvert-1)\\
        &\leq 2(\lvert \activesets_\currenttime\rvert - 1)
    \end{align*}
    which completes the proof.
\end{proof}

\subsection{Notation and Definitions for Moat Growing Algorithms}
\label{subsec:monotonic_preliminary}

In this section, we introduce notation and properties specific to moat growing algorithms.
We begin by examining how often active sets cut the optimal solution, using this to classify active sets and to partition the optimal solution based on which class of active sets colors each portion.
We then present a general definition of assignments, which serves as a key component of our analysis.

Since our approach focuses on variants of monotonic moat growing algorithms, we introduce common notation that will be used throughout the paper.  
For the remainder of the paper, we consider four specific executions of monotonic moat growing algorithms, each denoted by a distinct symbol.  
These symbols appear as superscripts on variables to indicate the corresponding execution.  
See Figure~\ref{fig:overview} for the list of symbols and their associated moat growing algorithms.

\begin{figure}[ht]
    \centering
\begin{tikzpicture}[
  scale=0.7,
  node distance=1.8cm and 2.3cm,
  round/.style={draw, circle, minimum size=8mm, inner sep=0pt},
  arrow/.style={-Stealth, thick},
  font=\small
]
\def\dblue{Black!90}
\def\dred{red}
\def\dyellow{Black!40}

\node[round] (a) at (0,0) {$\eplus$};
\node[left=3.2cm of a] (root) {};
\node[round, right=of a] (b) {$\ee$};
\node[round, above right=of b, yshift=-1cm] (c) {};
\node[round, below right=of b, yshift=1cm] (d) {$\ep$};
\node[round, right=3cm of d] (e) {$\ez$};

\draw[arrow] (root) -- node[above, yshift=1pt, xshift=-2pt] {Legacy Moat Growing} (a);
\draw[arrow] (a) -- node[above, yshift=1pt, xshift=-2pt] {Local Search} (b);
\draw[arrow] (b) -- node[above, yshift=1pt, xshift=-2pt, sloped] {Autarkic Pairs} (c);
\draw[arrow] (b) -- node[above, yshift=1pt, xshift=-2pt, sloped] {Extension} (d);
\draw[arrow] (d) -- node[above, yshift=1pt, xshift=-2pt] {Local Search} (e);

\draw[arrow, \dyellow, shorten <= 4pt] (a) -- ++(0,-2) node[\dblue, below, align=center] {Legacy Execution};
\draw[arrow, \dyellow, shorten <= 4pt] (b) -- ++(0,-2) node[\dblue, below, align=center] {Boosted Execution\\\textcolor{\dred}{$\solone$}};
\draw[arrow, \dyellow, shorten <= 4pt] (d) -- ++(0,-2) node[\dblue, below, align=center] {Extended Execution};
\draw[arrow, \dyellow, shorten <= 4pt] (e) -- ++(0,-2) node[\dblue, below, align=center] {Extended-Boosted Execution\\\textcolor{\dred}{$\solext$}};

\node[above=0cm of c] {\textcolor{\dred}{$\solcnd$}};

\end{tikzpicture}
    \caption{Flow of Algorithm~\ref{alg:main}, showing the different modules and their order.  
    Each module results in an ``execution'' of a monotonic moat growing algorithm, with the name of each execution shown below the corresponding circle.  
    The symbol inside each circle is used to refer to variables from that execution by placing it as a superscript.  
    The three solutions compared at the end of the algorithm are also identified, along with their positions in the flow and how they are obtained.}
    \label{fig:overview}
\end{figure}
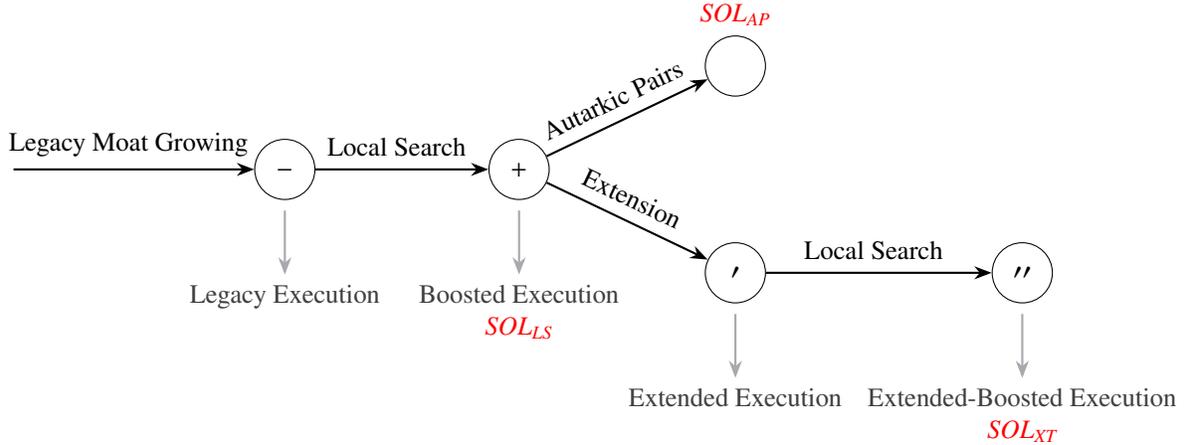 
\paragraph{Active and Superactive.}
We say that a vertex is \emph{active} at moment $\currenttime$ in a monotonic moat growing algorithm if it has not belonged to any inactive set up to that point.  
A set of vertices is said to be \emph{actively connected} if its members remain in active sets from the beginning of the algorithm until they all reach one another and lie in the same connected component.

We now introduce the notion of \emph{superactive sets}, which capture subsets of vertices that are actively connected.

\begin{definition}[Superactive Sets]
    For an active set in an execution of a monotonic moat growing algorithm, its corresponding superactive set is the subset of its vertices that are active at the moment the active set is formed.  
    We denote the superactive set of an active set $S$ by $\core(S)$, and for a particular execution $(*)$, we write $\core^*(S)$.  
    The term superactive sets refers to the collection of such subsets corresponding to all active sets throughout the algorithm.
\end{definition}

Although some vertices may be immediately deactivated and never included in any active set, we assume that $\{v\}$ is a superactive set for every vertex $v \in \V$.

In a monotonic moat growing algorithm, each connected component is formed by merging two existing components.  
As a result, the family of connected components that appear during the execution forms a \emph{laminar} structure: any two such subsets are either disjoint or one is a subset of the other.  
Since each active set corresponds to a connected component at some moment, we obtain the following corollary:

\begin{corollary}
\label{cor:active_sets_laminar}
    Active sets during a monotonic moat growing algorithm form a laminar family.
\end{corollary}

Similarly, the superactive sets in a monotonic moat growing execution also form a laminar family.  
Initially, superactive sets are disjoint. As the execution proceeds, when two active sets merge, their superactive sets merge as well.  
When an active set is deactivated, its superactive set becomes empty. These rules ensure that the laminar structure is preserved throughout the execution. 
This property allows us to define the following notion.

\begin{definition}[Maximal Superactive Sets]
    In an execution of a monotonic moat growing algorithm, a superactive set is called \emph{maximal} if it is not strictly contained in any other superactive set.  
    The family of maximal superactive sets, denoted by $\maxacts$, is obtained by filtering the collection of superactive sets to retain only the maximal ones.  
    This family forms a partition of the vertex set.
\end{definition}

For any maximal superactive set $S$, there exists an active set $S'$ whose corresponding superactive set is $S$, and which becomes inactive at some moment $\currenttime$.  
When we refer to the deactivation of a superactive set $S$, we mean the moment $\currenttime$ when the associated active set $S'$ is deactivated.  
We also refer to $S'$ as the active set from which $S$ is derived, although $S'$ may contain subsets whose superactive set is also $S$.

\paragraph{Comparison of Fingerprints and Refinement Families.}
We say that a fingerprint $\tout$ is \emph{larger} than a fingerprint $\tin$ if $\tout_v \ge \tin_v$ for all $v \in \V$. In this case, we also say that $\tin$ is \emph{smaller} than $\tout$.

We now compare pairs of monotonic moat growing executions whose fingerprints satisfy this relation. These comparisons play a central role in the remainder of the paper, as we will increase the value of $t$ for some vertices to generate larger fingerprints and design new moat growing algorithms with the goal of obtaining better solutions.
Figure~\ref{fig:ref} summarizes the key properties of such comparisons. The precise statements and their proofs are provided in the following lemmas.

The next lemma states that if the fingerprint of one execution is larger than that of another, then at any moment during the execution, any active set in the run with the smaller fingerprint is a subset of an active set in the run with the larger fingerprint at the same time.  
A similar statement holds for connected components.

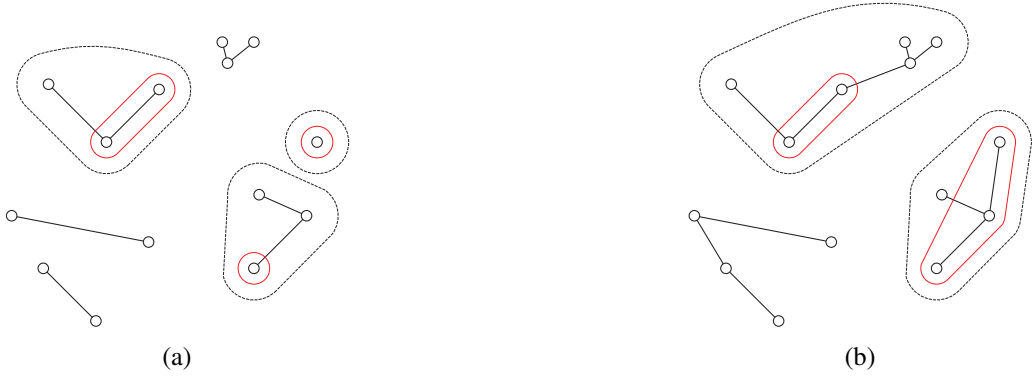
\begin{figure}[t]
    \centering
    \begin{subfigure}{0.45\textwidth}
    \centering
        \begin{tikzpicture}[scale=0.7]

\def\dem{Black}
\def\demf{Red!20}
\def\demi{Red}
\def\demif{Blue!20}
\def\demii{Green!70}
\def\demiif{Green!20}
\def\ter{White}
\def\col{Black}
\def\treeone{Purple}
\def\treetwo{RubineRed!70!Black}
\def\noder{0.1cm}

\def\tercl#1{%
  \ifcase#1
    White\or
    Goldenrod!20\or
    olive!20\or
    Green!20\or
    Sepia!20\or
    Plum!20\or
  \fi
}

\def\edge{1}
\def\r{1.044}

\coordinate (A) at (0, 0);
\coordinate (B) at (-\edge, \edge);
\coordinate (C) at (\edge, 1.5*\edge);
\coordinate (D) at (-1.6*\edge, 2*\edge);

\coordinate (E) at (0.2*\edge, 3.4*\edge);
\coordinate (F) at (1.2*\edge, 4.4*\edge);
\coordinate (G) at (-0.9*\edge, 4.5*\edge);

\coordinate (L) at (3*\edge, \edge);
\coordinate (M) at (4*\edge, 2*\edge);
\coordinate (N) at (3.1*\edge, 2.4*\edge);

\coordinate (K) at (4.2*\edge, 3.4*\edge);

\coordinate (H) at (2.5*\edge, 4.9*\edge);
\coordinate (I) at (3*\edge, 5.3*\edge);
\coordinate (J) at (2.4*\edge, 5.3*\edge);

\draw[\col] (A) -- (B);
\draw[\col] (C) -- (D);

\draw[\col] (E) -- (F);
\draw[\col] (E) -- (G);

\draw[\col] (L) -- (M);
\draw[\col] (N) -- (M);

\draw[\col] (H) -- (I);
\draw[\col] (H) -- (J);

\foreach \i in {A,...,N} {
    \draw[\col, fill= \ter] (\i) circle(\noder);
}

\def\mar{0.6}
\draw[\dem, line width=0.3pt, dash pattern=on 1.2pt off 0.4pt] (K) circle(\mar);

\def\ang{-45}
\draw[\dem, line width=0.3pt, dash pattern=on 1.2pt off 0.4pt] ($(E) + (\ang:\mar)$) arc(\ang:\ang-90:\mar) -- ($(G) + (\ang-90:\mar)$) arc(\ang-90:\ang-210:\mar) to[out=\ang-300, in=\ang+210] ($(F) + (\ang+120:\mar)$) arc(\ang+120:\ang:\mar) -- cycle;

\def\ang{-45}
\draw[\dem, line width=0.3pt, dash pattern=on 1.2pt off 0.4pt] ($(L) + (\ang:\mar)$) arc(\ang:\ang-120:\mar) -- ($(N) + (\ang-130:\mar)$) arc(\ang-130:\ang-240:\mar) -- ($(M) + (\ang-260:\mar)$) arc(\ang-260:\ang-350:\mar) -- cycle;

\def\mar{0.3}
\foreach \i in {K, L} {
    \draw[\demi, line width=0.3pt] (\i) circle(\mar); 
}

\def\ang{135}
\draw[\demi, line width=0.3pt] ($(E) + (\ang:\mar)$) arc(\ang:\ang+180:\mar) -- ($(F) + (\ang-180:\mar)$)  arc(\ang-180:\ang:\mar) -- cycle;

    \end{tikzpicture}
    \caption{}
    \label{fig:ref-one}
    \end{subfigure}
    \hfill
    \begin{subfigure}{0.45\textwidth}
    \centering
        \begin{tikzpicture}[scale=0.7]

\def\dem{Black}
\def\demf{Red!20}
\def\demi{Red}
\def\demif{Blue!20}
\def\demii{Green!70}
\def\demiif{Green!20}
\def\ter{White}
\def\col{Black}
\def\treeone{Purple}
\def\treetwo{RubineRed!70!Black}
\def\noder{0.1cm}

\def\tercl#1{%
  \ifcase#1
    White\or
    Goldenrod!20\or
    olive!20\or
    Green!20\or
    Sepia!20\or
    Plum!20\or
  \fi
}

\def\edge{1}
\def\r{1.044}

\coordinate (A) at (0, 0);
\coordinate (B) at (-\edge, \edge);
\coordinate (C) at (\edge, 1.5*\edge);
\coordinate (D) at (-1.6*\edge, 2*\edge);

\coordinate (E) at (0.2*\edge, 3.4*\edge);
\coordinate (F) at (1.2*\edge, 4.4*\edge);
\coordinate (G) at (-0.9*\edge, 4.5*\edge);

\coordinate (L) at (3*\edge, \edge);
\coordinate (M) at (4*\edge, 2*\edge);
\coordinate (N) at (3.1*\edge, 2.4*\edge);

\coordinate (K) at (4.2*\edge, 3.4*\edge);

\coordinate (H) at (2.5*\edge, 4.9*\edge);
\coordinate (I) at (3*\edge, 5.3*\edge);
\coordinate (J) at (2.4*\edge, 5.3*\edge);

\draw[\col] (A) -- (B);
\draw[\col] (B) -- (D);
\draw[\col] (C) -- (D);

\draw[\col] (E) -- (F);
\draw[\col] (E) -- (G);

\draw[\col] (L) -- (M);
\draw[\col] (N) -- (M);

\draw[\col] (H) -- (I);
\draw[\col] (H) -- (J);

\draw[\col] (M) -- (K);
\draw[\col] (F) -- (H);

\foreach \i in {A,...,N} {
    \draw[\col, fill= \ter] (\i) circle(\noder);
}

\def\mar{0.6}

\def\ang{-45}
\draw[\dem, line width=0.3pt, dash pattern=on 1.2pt off 0.4pt] ($(E) + (\ang-10:\mar)$) arc(\ang-10:\ang-90:\mar) -- ($(G) + (\ang-90:\mar)$) arc(\ang-90:\ang-200:\mar) to[out=\ang-290, in=\ang+210] ($(I) + (\ang+120:\mar)$) arc(\ang+120:\ang:\mar) -- cycle;

\def\ang{-45}
\draw[\dem, line width=0.3pt, dash pattern=on 1.2pt off 0.4pt] ($(L) + (\ang:\mar)$) arc(\ang:\ang-120:\mar) -- ($(N) + (\ang-130:\mar)$) arc(\ang-130:\ang-180:\mar) -- ($(K) + (\ang-180:\mar)$) arc(\ang-180:\ang-320:\mar) -- ($(M) + (\ang-320:\mar)$) arc(\ang-320:\ang-350:\mar) -- cycle;

\def\mar{0.3}

\def\ang{135}
\draw[\demi, line width=0.3pt] ($(E) + (\ang:\mar)$) arc(\ang:\ang+180:\mar) -- ($(F) + (\ang-180:\mar)$)  arc(\ang-180:\ang:\mar) -- cycle;

\def\ang{-45}
\draw[\demi, line width=0.3pt] ($(L) + (\ang:\mar)$) arc(\ang:\ang-150:\mar) -- ($(K) + (\ang-160:\mar)$) arc(\ang-160:\ang-320:\mar) -- ($(M) + (\ang-320:\mar)$) arc(\ang-320:\ang-350:\mar) -- cycle;

    \end{tikzpicture}
    \caption{}
    \label{fig:ref-two}
    \end{subfigure}
    \caption{Connected components, active sets (black dashed lines), and their superactive sets (solid red lines) at the same moment in two monotonic moat growing executions. The fingerprint in (b) is larger than in (a). As shown, every connected component, active set, and superactive set in (a) is a subset of its counterpart in (b). The execution in (b) exhibits greater connectivity: any two vertices connected in (a) are also connected in (b), some vertices that are inactive in (a) are active in (b), and some vertices not included in any superactive set in (a) belong to one in (b).}
    \label{fig:ref}
\end{figure} 
\begin{lemma}
\label{lm:large_fingerprint_refinement}  
Let $\tin$ and $\tout$ be fingerprints of two monotonic moat growing algorithms on a graph $\G$, where $\tout$ is larger than $\tin$.  
At any moment $\currenttime$ during these executions, the following hold:  
\begin{itemize}  
    \item The active sets in the execution with fingerprint $\tin$ form a \emph{refinement} of the active sets in the execution with fingerprint $\tout$.  
    \item The connected components in the execution with fingerprint $\tin$ form a \emph{refinement} of those in the execution with fingerprint $\tout$.  
\end{itemize}   
\end{lemma}
\begin{proof}
    For each moment $\currenttime$, let $\activeattime$ denote the active sets and $\ccattime$ denote the connected components in the execution on $\tin$, and let $\activeattimeprime$ and $\ccattimeprime$ denote the corresponding active sets and connected components in the execution on $\tout$.
    
    Initially, when $\currenttime = 0$, we have
    $\activesets^{\inp}_{0} \subseteq \activesets^{\out}_{0}$ since vertices with $\tin_v = 0$ may have $\tout_v > 0$.
    Also, $\currentsets^{\inp}_{0} = \currentsets^{\out}_{0}$.
    
    Assume $\currenttime$ is the minimum time at which neither $\activeattime \refines \activeattimeprime$ nor $\ccattime \refines \ccattimeprime$.

    The relation $\refines$ means the left-hand side is a refinement of the right-hand side. Additionally, we say a set $S$ of vertices covers another set $S'$ if $S' \subseteq S$.
    
    There are four possible cases for the last event that could lead to this violation:
    \begin{enumerate}
        \item \emph{Two connected components $C_1$ and $C_2$ in the execution on $\tout$ are merged:} Then $C_1 \cup C_2$ covers everything $C_1$ and $C_2$ independently covered. Therefore, this case cannot break refinement.

        \item \emph{Two connected components $C_1$ and $C_2$ in the execution on $\tin$ are merged:}
        \begin{itemize}
            \item If both $C_1$ and $C_2$ were covered by the same active set $C$ in the execution on $\tin$ beforehand, then $C$ also covers $C_1 \cup C_2$ at this moment. Thus, this case preserves refinement.
            \item If $C_1$ and $C_2$ were covered by different sets beforehand, let the merging edge be $e = (u, v)$. At any earlier time, whenever $u$ and $v$ belonged to active sets in the execution on $\tin$,
            they have been in active sets cutting $e$ on both executions (until now active sets hold the refinement condition). Thus, $\sum_{S: e \in \delta(S)} y_S = c_e$ must hold in both runs. This contradicts $C_1$ and $C_2$ being subsets of different active sets at moment~$\currenttime$. Hence, this case is also not possible.
        \end{itemize}

        \item \emph{One active set $C'$ in the execution on $\tin$ becomes deactivated:} 
        In this case, first, there is no change in $\ccattime$ and $\ccattimeprime$. Second, since only one set is removed from $\activeattime$, we still have $\activeattime \refines \activeattimeprime$. Thus, this case cannot violate refinement.

        \item \emph{One active set $C'$ in the execution on $\tout$ becomes deactivated:}
        Again, there is no change in $\ccattime$ and $\ccattimeprime$. Since $C'$ becomes deactivated, for all $v \in C'$, we have $\tout_v \le \tau$. This condition also holds for any $C \subseteq C'$ as $\tin_v \le \out_v$. Therefore, there is no active $C \subseteq C'$ in execution on $\tin$. 

    \end{enumerate}
    
    Since none of these cases can lead to a violation, such a time $\currenttime$ cannot exist. Hence, for all $\currenttime$, $\activeattime$ is a refinement of $\activeattimeprime$, and $\ccattime$ is a refinement of $\ccattimeprime$.
\end{proof}

The next three lemmas establish direct consequences of the above lemma.
\begin{lemma}
\label{lm:large_fingerprint_vertex_remains_active}
    Let $\tin$ and $\tout$ be fingerprints of two monotonic moat growing algorithms on a graph $\G$, where $\tout$ is larger than $\tin$.  
    If a vertex $v$ is active at moment $\currenttime$ in the execution with fingerprint $\tin$, then $v$ is also active at moment $\currenttime$ in the execution with fingerprint $\tout$.
\end{lemma}
\begin{proof}
    Assume for contradiction that $v$ is active at moment $\currenttime$ in the execution on $\tin$ but not in the execution on $\tout$.  
    Then $v$ must belong to an active set $S$ in the execution on $\tin$ at time $\currenttime$, but not to any active set in the execution on $\tout$ at that time.  
    However, by Lemma~\ref{lm:large_fingerprint_refinement}, $S$ must be contained in an active set $S'$ in the execution on $\tout$ at the same moment.  
    This implies $v \in S'$, so $v$ is active in the execution on $\tout$, contradicting the assumption.  
\end{proof}

\begin{lemma}
\label{lm:large_fingerprint_superactive}
Let $\tin$ and $\tout$ be fingerprints of two monotonic moat growing algorithms on a graph $\G$, where $\tout$ is larger than $\tin$.  
If two vertices are actively connected in the execution on $\tin$, then they are also actively connected in the execution on $\tout$.
\end{lemma}
\begin{proof}
Assume $u$ and $v$ are actively connected in the execution on $\tin$.  
This means that at some moment $\currenttime$, they lie in the same active set $S$ and both are still active.  

By Lemma~\ref{lm:large_fingerprint_refinement}, in the execution on $\tout$, $u$ and $v$ must belong to an active set $S' \supseteq S$ at the same time $\currenttime$.  
Furthermore, by Lemma~\ref{lm:large_fingerprint_vertex_remains_active}, both $u$ and $v$ are active at time $\currenttime$ in the execution on $\tout$.  
Therefore, $u$ and $v$ are active and lie in the same connected component in the execution on $\tout$, so they are actively connected.
\end{proof}

\begin{lemma}
\label{lm:large_fingerprint_superactive_refinement}
    Let $\tin$ and $\tout$ be fingerprints of two monotonic moat growing algorithms on a graph $\G$, where $\tout$ is larger than $\tin$.
    At any time $\currenttime$ during both executions, the superactive sets derived from the active sets in the execution with fingerprint $\tin$ form a refinement of the superactive sets derived from the active sets in the execution with fingerprint $\tout$.
\end{lemma}

\begin{proof}
    Let $\core^{\inp}$ and $\core^{\out}$ be the functions that return the superactive set of a given active set in the executions on $\tin$ and $\tout$, respectively.
    
    At time $\currenttime$, we claim that for any active set $S$ in the execution on $\tin$, $\core^{\inp}(S) \subseteq \core^{\out}(S')$ for the active set $S' \supseteq S$ in the execution on $\tout$, whose existence is guaranteed by Lemma~\ref{lm:large_fingerprint_refinement}. 
    This is sufficient to show that the superactive sets under $\tin$ form a refinement of those under $\tout$.
    
    Take any vertex $v \in \core^{\inp}(S)$. By definition, $v$ is active at the time $S$ is formed in the execution on $\tin$, and remains active up to $\currenttime$. By Lemma~\ref{lm:large_fingerprint_vertex_remains_active}, $v$ must also be active at $\currenttime$ in the execution on $\tout$. Since $v \in S \subseteq S'$, it follows that $v \in \core^{\out}(S')$.
    This completes the proof.
\end{proof}

\paragraph{Cut-Based Classification.}
Next, we begin by defining the notions of \emph{single-edge set} and \emph{multi-edge set}, originally introduced by~\cite{10.1145/3722551} and later used by~\cite{DBLP:conf/stoc/AhmadiGHJM24}.
These concepts are fundamental to our analysis.

\begin{definition}[Single-Edge Set and Multi-Edge Set]
    For any active set $S \subseteq \V$, we call $S$ a \emph{single-edge set} if it cuts exactly one edge of the optimal solution, i.e., $|\deltaS \cap \OPT| = 1$, and a \emph{multi-edge set} if it cuts more than one edge of the optimal solution, i.e., $|\deltaS \cap \OPT| > 1$.
\end{definition}

While the works of~\cite{10.1145/3722551,DBLP:conf/stoc/AhmadiGHJM24} use these terms to classify sets based on how many edges of the optimal solution they cut, we require a more refined usage. In our setting, we consider specific forests—typically subgraphs of the optimal solution—and classify active sets based on how many edges of that subgraph they cut. Furthermore, we partition the forest based on the types of active sets that color its edges.

\begin{definition}
\label{def:uncolored_single_multi_colored}
    Let $\activesets$ denote the family of all active sets that appear throughout a given monotonic moat growing algorithm.  
    For a forest $F$, define:
    \begin{align*}
        \sactsf &= \{S \in \activesets \mid |\deltaS \cap F| = 1\}, \tag{active sets that cut exactly one edge of $F$} \\
        \mactsf &= \{S \in \activesets \mid |\deltaS \cap F| > 1\}. \tag{active sets that cut more than one edge of $F$}
    \end{align*}
    We then define:
    \begin{itemize}
        \item $\uncolored(F)$ as the total portion of $F$ not colored by any active set,
        \item $\singlecolored(F) = \sum_{S \in \sactsf} \ys$ as the total portion of $F$ colored by sets in $\sactsf$,
        \item $\multicolored(F) = \sum_{S \in \mactsf} |\deltaS \cap F| \cdot \ys$ as the total portion of $F$ colored by sets in $\mactsf$,
        \item $\extra(F) = \uncolored(F) + \multicolored(F)$.
    \end{itemize}
\end{definition}

For a particular execution $(*)$, we denote these quantities by $\uncolored^*$, $\multicolored^*$, and $\extra^*$.

The following corollary follows directly from the above definitions:

\begin{corollary}
\label{cor:F_USM_SE}
    For any monotonic moat growing algorithm and any forest $F$, we can partition $F$ into uncolored, single-colored, and multi-colored portions as follows:
    \[
        \cc(F) = \uncolored(F) + \singlecolored(F) + \multicolored(F) = \singlecolored(F) + \extra(F).
    \]
\end{corollary}

Next, we prove the monotonicity property of $\extra$.

\begin{lemma}
\label{lm:extra_subgraph}
    For any monotonic moat growing algorithm, and any forests $F$ and $F'$ such that $F'$ is a subgraph of $F$, we have
    \[
        \extra(F') \le \extra(F).
    \]
\end{lemma}
\begin{proof}
    First, note that $\uncolored(F') \le \uncolored(F)$ since every edge in $F'$ is also in $F$, and any uncolored portion in $F'$ must also be uncolored in $F$.

    Furthermore, we have $\macts_{F'} \subseteq \mactsf$, because if an active set cuts more than one edge in $F'$, it must also cut more than one edge in $F$. Therefore, any portion of $F'$ colored by $\macts_{F'}$ is also part of the portion of $F$ colored by $\mactsf$.
    It follows that $\multicolored(F') \le \multicolored(F)$. 
    
    Combining both bounds, we obtain:
    \[
        \extra(F') = \uncolored(F') + \multicolored(F') \le \uncolored(F) + \multicolored(F) = \extra(F). \qedhere
    \]
\end{proof}

The next lemma shows that $\extra(F)$ provides a lower bound on the cost of the optimal solution that exceeds the total growth of the active sets involved in coloring it.

\begin{lemma}
\label{lm:ys_le_T_extra}
    For any monotonic moat growing algorithm and any forest $F$, we have
    $$\sum_{\substack{S\subseteq \V \\ S \odot F}} \ys \le \cc(F) - \dfrac{\extra(F)}{2}.$$
\end{lemma}
\begin{proof}
    We partition the active sets that cut $F$ into those that cut exactly one edge and those that cut at least two edges:
    \begin{align*}
        \sum_{\substack{S\subseteq \V \\ S \odot F}} \ys &= \sum_{\substack{S\subseteq \V \\ |\deltaS \cap F| = 1}} \ys + \sum_{\substack{S\subseteq \V \\ |\deltaS \cap F| \ge 2}} \ys\\
        &\le \sum_{S\in \sactsf} \ys + \sum_{S\in \mactsf} \dfrac{|\deltaS \cap F|}{2}\cdot\ys \\
        &= \singlecolored(F) + \dfrac{\multicolored(F)}{2} \\
        &\le \singlecolored(F) + \dfrac{\extra(F)}{2} \tag{Definition~\ref{def:uncolored_single_multi_colored}} \\
        &= \cc(F) - \dfrac{\extra(F)}{2}. \tag{Corollary~\ref{cor:F_USM_SE}}
    \end{align*}
\end{proof}

Next, we prove that for any tree $T$, every active set contributing to $\singlecolored(T)$ must cut the set of leaves of $T$ (see Figure~\ref{fig:tree-cuts}).
We then use this fact in Lemma~\ref{lm:coloring_exactly_one_edge} to derive a lower bound on $\multicolored(T)$.
This observation is crucial, as it implies that any active set that does not cut the set of leaves of a tree in the optimal solution must color more than one of its edges.
Identifying such active sets enables us to establish a lower bound on the cost of the optimal solution that exceeds the total growth of the active sets involved in coloring it, similar in spirit to the above lemma.

\begin{lemma}
\label{lm:cut_one_edge_cut_leaves}
    Let $T$ be a tree with leaf set $S' \subseteq \V$, and let $S \subseteq \V$ be a set such that $|T \cap \deltaS| = 1$.
    Then, $S \odot S'$.
\end{lemma}
\begin{proof}
    Let $e = \{u, v\} \in T \cap \deltaS$ be the unique edge of $T$ cut by $S$.
    Since $S$ cuts $T$, it must contain some, but not all, of the vertices of $T$.
    Removing $e$ splits $T$ into two connected components, each a subtree of $T$, such that one component lies entirely inside $S$, and the other lies entirely outside.
    Suppose this were not the case: that is, if one of the components contained vertices both inside and outside of $S$, then $S$ would cut at least one additional edge of $T$, contradicting the assumption that $|\deltaS \cap T| = 1$.
    On the other hand, if both components were on the same side of $S$ (both fully inside or both fully outside), then $S$ would not cut $e$ at all.
    Therefore, the two components must lie entirely on opposite sides of $S$.

    Now consider the two components created by removing $e$.
    If a component contains more than one vertex, it must have at least two leaves, and at least one of them is different from $u$ and $v$.
    Such a vertex is also a leaf in the original tree $T$, i.e., a member of $S'$.
    If a component consists of a single vertex, then that vertex is necessarily a leaf in $T$ and thus also belongs to $S'$.
    Therefore, one vertex from $S'$ lies inside $S$ and another lies outside $S$, implying $S \odot S'$.
\end{proof}

\tikzset{
  triangle/.style args={#1}{
    insert path={
        ++(90:1.2*#1)
      -- ++($(90:-1.2*#1)+(210:1.2*#1)$)
      -- ++($(210:-1.2*#1)+(330:1.2*#1)$)
      -- cycle %
    }
  }
}
\tikzset{
  hexagon/.style args={#1}{
    insert path={
      ++(90:#1)
      -- ++($(90:-#1)+(30:#1)$)
      -- ++($(30:-#1)+(-30:#1)$)
      -- ++($(-30:-#1)+(-90:#1)$)
      -- ++($(-90:-#1)+(-150:#1)$)
      -- ++($(-150:-#1)+(150:#1)$)
      -- cycle
    }
  }
}
\tikzset{
  pentagon/.style args={#1}{
    insert path={
      ++(90:1.05*#1)
      -- ++($(90:-1.05*#1)+(162:1.05*#1)$)
      -- ++($(162:-1.05*#1)+(234:1.05*#1)$)
      -- ++($(234:-1.05*#1)+(306:1.05*#1)$)
      -- ++($(306:-1.05*#1)+(18:1.05*#1)$)
      -- cycle
    }
  }
}
\tikzset{
  square/.style args={#1}{
    insert path={
        ++(45:1.1*#1)
      -- ++($(45:-1.1*#1)+(135:1.1*#1)$)
      -- ++($(135:-1.1*#1)+(225:1.1*#1)$)
      -- ++($(225:-1.1*#1)+(315:1.1*#1)$)
      -- cycle %
    }
  }
}
\tikzset{
  star/.style args={#1}{
    insert path={
      ++(90:1.1*#1)
      \foreach \a in {90,162,234,306,378} {
        -- ++($(\a:-1.1*#1) + (\a+36:0.5*#1)$)
        -- ++($(\a+36:-0.5*#1) + (\a+72:1.1*#1)$)
      }
      -- cycle
    }
  }
}

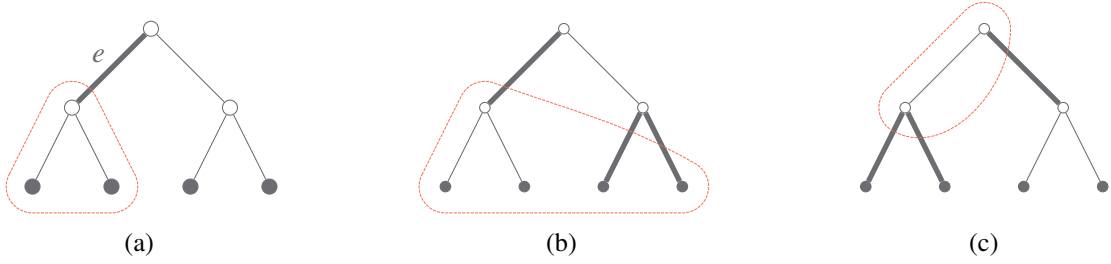
\begin{figure}[t]
    \centering
    \begin{subfigure}{0.32\textwidth}
        \centering
        \begin{tikzpicture}[
  scale=0.7,
  level distance=1.5cm, 
  level 1/.style={sibling distance=3cm},
  level 2/.style={sibling distance=1.5cm},
  level 3/.style={sibling distance=3cm},
  edge from parent/.style={draw=\col} 
        ]
\draw[White] (0,-3.5) -- (0, 0.5);
\def\dem{Red!70}
\def\ter{White}
\def\col{Black!70}
\def\noder{0.1cm}

\def\n{4}
\def\m{4}
\pgfmathtruncatemacro{\last}{\n+1}

\def\tercl#1{%
  \ifcase#1
    White\or
    olive!20\or
    Green!20\or
    Sepia!20\or
    Plum!20\or
    Goldenrod!20\or
  \fi
}
\def\tersh#1{%
  \ifcase#1
    none\or
    star\or
    square\or
    pentagon\or
    triangle\or
    hexagon\or
  \fi
}

\tikzset{
    center arc/.style args={#1:#2:#3}{
        insert path={+ (#1:#3) arc (#1:#1+#2:#3)}
    }
}
\node[circle,draw=\col,fill=\ter,inner sep=0pt,minimum size=2*\noder] (root) {}
    child {node[circle,draw=\col,fill=\ter,inner sep=0pt,minimum size=2*\noder] (A) {}
      child {node[circle,draw=\col,fill=\col,inner sep=0pt,minimum size=2*\noder] (B) {}}
      child {node[circle,draw=\col,fill=\col,inner sep=0pt,minimum size=2*\noder] (C) {}}
    }
    child {node[circle,draw=\col,fill=\ter,inner sep=0pt,minimum size=2*\noder] (D) {}
      child {node[circle,draw=\col,fill=\col,inner sep=0pt,minimum size=2*\noder] (E) {}}
      child {node[circle,draw=\col,fill=\col,inner sep=0pt,minimum size=2*\noder] (F) {}}
    };
\def\thick{2pt}
\draw[\col, line width=\thick] (A) -- (root);

\def\margin{0.5cm}
\draw[\dem, line width=0.3pt, 
    dash pattern=on 1.2pt off 0.4pt] ($(B) + (150:\margin)$) arc(150:270:\margin) -- ($(C) + (270:\margin)$) arc(-90:30:\margin) -- ($(A) + (30:\margin)$) arc(30:150:\margin) -- cycle;

\node[\col, above left=-2pt] at ($(A)!0.5!(root)$) {$e$};

\end{tikzpicture}
\caption{}
\label{fig:cut-only-one}
\end{subfigure}
\hfill
\begin{subfigure}{0.32\textwidth}
        \centering
        \begin{tikzpicture}[
  scale=0.7,
  level distance=1.5cm, 
  level 1/.style={sibling distance=3cm},
  level 2/.style={sibling distance=1.5cm},
  level 3/.style={sibling distance=3cm},
  edge from parent/.style={draw=\col} 
        ]
\draw[White] (0,-3.5) -- (0, 0.5);
\def\dem{Red!70}
\def\ter{White}
\def\col{Black!70}
\def\noder{0.07cm}

\def\n{4}
\def\m{4}
\pgfmathtruncatemacro{\last}{\n+1}

\def\tercl#1{%
  \ifcase#1
    White\or
    olive!20\or
    Green!20\or
    Sepia!20\or
    Plum!20\or
    Goldenrod!20\or
  \fi
}
\def\tersh#1{%
  \ifcase#1
    none\or
    star\or
    square\or
    pentagon\or
    triangle\or
    hexagon\or
  \fi
}

\tikzset{
    center arc/.style args={#1:#2:#3}{
        insert path={+ (#1:#3) arc (#1:#1+#2:#3)}
    }
}
\node[circle,draw=\col,fill=\ter,inner sep=0pt,minimum size=2*\noder] (root) {}
    child {node[circle,draw=\col,fill=\ter,inner sep=0pt,minimum size=2*\noder] (A) {}
      child {node[circle,draw=\col,fill=\col,inner sep=0pt,minimum size=2*\noder] (B) {}}
      child {node[circle,draw=\col,fill=\col,inner sep=0pt,minimum size=2*\noder] (C) {}}
    }
    child {node[circle,draw=\col,fill=\ter,inner sep=0pt,minimum size=2*\noder] (D) {}
      child {node[circle,draw=\col,fill=\col,inner sep=0pt,minimum size=2*\noder] (E) {}}
      child {node[circle,draw=\col,fill=\col,inner sep=0pt,minimum size=2*\noder] (F) {}}
    };
\def\thick{2pt}
\draw[\col, line width=\thick] (A) -- (root);
\draw[\col, line width=\thick] (E) -- (D);
\draw[\col, line width=\thick] (D) -- (F);

\def\margin{0.5cm}
\draw[\dem, line width=0.3pt, 
    dash pattern=on 1.2pt off 0.4pt] ($(B) + (150:\margin)$) arc(150:270:\margin) -- ($(F) + (270:\margin)$) arc(-90:60:\margin) to[out=150, in=-20] ($(A) + (70:\margin)$) arc(70:150:\margin) -- cycle;
\end{tikzpicture}
\caption{}
\label{fig:cut-all-leaves}
\end{subfigure}
\hfill
\begin{subfigure}{0.32\textwidth}
        \centering
        \begin{tikzpicture}[
  scale=0.7,
  level distance=1.5cm, 
  level 1/.style={sibling distance=3cm},
  level 2/.style={sibling distance=1.5cm},
  level 3/.style={sibling distance=3cm},
  edge from parent/.style={draw=\col} 
        ]
\draw[White] (0,-3.5) -- (0, 0.5);
\def\dem{Red!70}
\def\ter{White}
\def\col{Black!70}
\def\noder{0.07cm}

\def\n{4}
\def\m{4}
\pgfmathtruncatemacro{\last}{\n+1}

\def\tercl#1{%
  \ifcase#1
    White\or
    olive!20\or
    Green!20\or
    Sepia!20\or
    Plum!20\or
    Goldenrod!20\or
  \fi
}
\def\tersh#1{%
  \ifcase#1
    none\or
    star\or
    square\or
    pentagon\or
    triangle\or
    hexagon\or
  \fi
}

\tikzset{
    center arc/.style args={#1:#2:#3}{
        insert path={+ (#1:#3) arc (#1:#1+#2:#3)}
    }
}
\node[circle,draw=\col,fill=\ter,inner sep=0pt,minimum size=2*\noder] (root) {}
    child {node[circle,draw=\col,fill=\ter,inner sep=0pt,minimum size=2*\noder] (A) {}
      child {node[circle,draw=\col,fill=\col,inner sep=0pt,minimum size=2*\noder] (B) {}}
      child {node[circle,draw=\col,fill=\col,inner sep=0pt,minimum size=2*\noder] (C) {}}
    }
    child {node[circle,draw=\col,fill=\ter,inner sep=0pt,minimum size=2*\noder] (D) {}
      child {node[circle,draw=\col,fill=\col,inner sep=0pt,minimum size=2*\noder] (E) {}}
      child {node[circle,draw=\col,fill=\col,inner sep=0pt,minimum size=2*\noder] (F) {}}
    };
\def\thick{2pt}
\draw[\col, line width=\thick] (A) -- (C);
\draw[\col, line width=\thick] (A) -- (B);
\draw[\col, line width=\thick] (root) -- (D);

\def\margin{0.5cm}
\draw[\dem, line width=0.3pt, 
    dash pattern=on 1.2pt off 0.4pt] ($(A) + (135:\margin)$) arc(135:240:\margin) to[out=330, in=-90] ($(root) + (0:\margin)$) arc(0:135:\margin) -- cycle;

\end{tikzpicture}
\caption{}
\label{fig:cut-no-leaf}
\end{subfigure}
\caption{Tree $T$ with three subsets of vertices.  
(a) A set that cuts the leaves of $T$ can cut only one edge of $T$.  
(b, c) A set that contains all leaves of $T$ (b) or none of them (c) cannot cut only one edge of $T$.
}
\label{fig:tree-cuts}
\end{figure} 
\begin{lemma}
\label{lm:coloring_exactly_one_edge}
    For any monotonic moat growing algorithm and any tree $T$ with leaves $S' \subseteq \V$, we have
    \[
        \extra(T) \ge \cc(T) - \sum_{\substack{S \subseteq \V \\ S \odot S'}} \ys.
    \]
\end{lemma}
\begin{proof}
    We first show that 
    \[
        \singlecolored(T) \le \sum_{\substack{S \subseteq \V \\ S \odot S'}} \ys.
    \]
    By Lemma~\ref{lm:cut_one_edge_cut_leaves}, for any active set $S \in \sactsf$, we have $S \odot S'$.
    Since $\singlecolored(T)$ is the total portion of $T$ colored by active sets in $\sactsf$, and each such active set colors exactly one edge of $T$, we can conclude the above inequality.

    Now, using Corollary~\ref{cor:F_USM_SE}, we have
    $$\cc(T) = \extra(T) + \singlecolored(T) \le \extra(T) +  \sum_{\substack{S\subseteq \V \\ S\odot S'}} \ys,$$
    which completes the proof.
\end{proof}

The next lemma shows that if an active set cuts exactly one edge of a tree in the optimal solution, then removing that edge only disconnects those pairs that are cut by the active set itself.  
We use this property for two purposes.  
First, we show that if an active set does not contain any unsatisfied vertices, then it cannot cut exactly one edge of the optimal solution; otherwise, we could remove the edge it cuts and obtain a more optimal solution that satisfies all demands, leading to a contradiction.  
Second, we use this lemma to argue that if we remove edges cut by certain special single-edge active sets, then the only unsatisfied pairs are those cut by the same active set.  
This is useful because if we can find an alternative structure to satisfy those pairs, we can remove the corresponding edges from the optimal solution and obtain a lower-cost solution for the remaining pairs.

\begin{lemma}
\label{lm:remove_one_edge}
    Let $S \subseteq \V$ be a subset of vertices such that, in a connected component of the optimal solution $\opti$ (with tree $\otree$), $S$ cuts exactly one edge of that component, denoted by $e$.
    That is, $\deltaS \cap \otree = \{e\}$.
    Then, removing $e$ only violates pair demands for those pairs that have one endpoint in $\unsatisfiedS$.
\end{lemma}
\begin{proof}
    Similar to the proof of Lemma~\ref{lm:cut_one_edge_cut_leaves}, removing $e$ splits $\otree$ into two subtrees, one entirely inside $S$ and the other entirely outside.
    This means a pair $\{v, \pairv\}$ is disconnected by the removal of $e$ only if $v$ and $\pairv$ lie on opposite sides of the cut. Without loss of generality, assume $v \in S$. Then $\pairv \notin S$, and so $v \in \unsatisfiedS$, completing the proof.
\end{proof}

The next lemma provides a lower bound on the cost of the optimal solution based on the total portion of it colored by active sets.

\begin{lemma}
\label{lm:tree-bound-by-activeset-num-cuts}
    For any monotonic moat growing algorithm and any forest $F$, the cost of $F$ can be bounded in terms of its coloring by the active sets of the algorithm as follows:
    $$\cc(F) \ge \sum_{S\subseteq \V} |\deltaS\cap F| \cdot \ys.$$
\end{lemma}
\begin{proof}
    For an active set $S \subseteq \V$, the edges of $F$ it colors are precisely those in $\deltaS \cap F$.
    If $S$ grows for a duration of $\ys$, then it contributes $|\deltaS \cap F| \cdot \ys$ units to the total colored portion of $F$.
    Since each portion of $F$ can be colored at most once, summing over all active sets gives a lower bound on $\cc(F)$, completing the proof.
\end{proof}

\paragraph{Assignment.}
Another important concept that we introduce is \emph{assignment}.
The idea behind assignment is to allocate the growth of an active set to one or more of its vertices.
This approach allows us to bound the total growth of all active sets using the assigned values.
An assignment specifies how the growth of each active set is distributed among its vertices. 
Specifically, consider a particular moment during the execution of a monotonic moat growing algorithm. 
At this moment, each active set \(S\) assigns a fraction (possibly zero) of its growth to each of its vertices. 
It is worth noting that these fractions lie between 0 and 1, but the total fraction assigned by an active set is not necessarily bounded by 1.

\begin{definition} 
\label{def:assignment}
    We define \(\gr: 2^{\V} \times \V \rightarrow \mathbb{R}_{\ge 0}\) as an assignment where 
    \(\ysv\) represents the portion of growth assigned by the set \(S\) to a vertex \(v \in \core(S)\). 
    We abuse the notation and write \(\gr_v\) to represent the total growth assigned to vertex \(v\). Consequently, 
    $$
        \gr_v = \sum_{S \subseteq \V} \ysv.
    $$

    We further use the same notation and define \(\gr_{\max}(S)\) for a subset of vertices \(S\) to denote the maximum value of \(\gr_v\) for any \(v \in S\), i.e.,
    $$
        \gr_{\max}(S) = \max_{v \in S} \gr_v.
    $$
\end{definition}

Generally, we have two types of assignments, defined as follows:

\begin{definition}[Exclusive Assignment]
\label{def:ex-assignment}
An assignment is \emph{exclusive} if the total fractional assignment across all vertices at any given moment for an active set is at most $1$.
\end{definition}

\begin{definition}[Prefix-Time Assignment]
\label{def:prefix-time}
A prefix-time assignment is an assignment satisfying the following conditions:
\begin{itemize}
    \item At each moment in time, an active set assigns to each vertex a growth fraction of either $0$ or $1$.
    \item If, at some time \( \currenttime \), an active set assigns a growth fraction of $1$ to a vertex \( v \), then at any earlier time, the active set containing \( v \) must also assign its entire growth fraction of $1$ to \( v \).
\end{itemize}
\end{definition}

The exclusive assignment roughly matches the total assigned value with the total growth of active sets, allowing us to bound the total assigned value instead of the total growth.
It is used in the analysis throughout all subsequent sections and methods.
On the other hand, the prefix-time assignment better captures temporal aspects and provides an upper bound on nearly the total growth of active sets, although this bound can be significantly weak.
We define two types of prefix-time assignment, whose sole purpose is to analyze the autarkic pairs approach in Section~\ref{sec:candidate}.
The first prefix-time assignment will be formally introduced in the next subsection.

\subsection{Legacy Execution}
\label{subsec:legacy_execution}

Here, we discuss the execution of \BBG{} at Line~\ref{exe:base} of our main algorithm.
We refer to this execution as \emph{Legacy Execution}, and for any variable $x$ associated with this execution, we denote it as $x^-$.
Let us formally define Legacy Execution and then provide some of its properties that will be useful in the subsequent sections.

\begin{definition}[Legacy Execution]
\label{def:legacy_execution}
    We define Legacy Execution as the execution of \BBG{} in Line~\ref{exe:base} of Algorithm~\ref{alg:main} on the given problem input. In this context, we denote:
    \begin{itemize}
        \item $\yplus_S$ as the total growth of set $S \subseteq \V$,
        \item $\Aplus_{\currenttime}$ as the collection of active sets at time $\currenttime$ during this execution, and
        \item $\tplus$ as the fingerprint returned by this execution, which records the time at which each vertex connects to its pair (see Line~\ref{line:legacy_define_t}).
    \end{itemize}
    Note that, given Lemma~\ref{lm:tplus_is_fingerprint}, $\tplus$ is a fingerprint for this execution. According to Lemma~\ref{lm:ghost_equivalence}, executing \GBG{} on $\tplus$ results in an equivalent moat growing process. Therefore, all these variables can be considered as belonging to such a run of \GBG{} as well.
\end{definition}

Next, we prove that monotonic moat growing algorithms whose fingerprint has a higher value for every vertex compared to their $\tplus$ satisfy all demand pairs.  
This property is crucial for demonstrating that the various solutions introduced in the remainder of the paper are feasible for the problem.

\begin{lemma}
\label{lm:large_fingerprint_satisfies}
    Let $\gt$ be a fingerprint of a monotonic moat growing algorithm on a graph $\G$, where $\gt$ is larger than $\tplus$.  
    Then, the forest produced by the execution with fingerprint $\gt$ satisfies all demand pairs.
\end{lemma}
\begin{proof}
    First note that we can assume a Shadow Moat Growing algorithm has been run on $\gt$ and given Lemma~\ref{lm:ghost_equivalence}, it is equivalent to the given monotonic moat growing algorithm.
    
    Consider any vertex $v$.  
    By Definition~\ref{def:legacy_execution}, $\tplus_v$ is the moment when $v$ and its pair $\pairv$ reach each other in Legacy Execution, and both are active until this moment. Consequently, at this moment, there exists an active set $S$ such that $S$ contains both $v$ and $\pairv$ (they may become deactivated immediately afterward). 

    According to Lemma~\ref{lm:large_fingerprint_refinement}, at any moment until $\tplus_v$, $v$ should be in an active set in Shadow Moat Growing execution on $\gt$, and at moment $\tplus_v$, there exists an active set $S'$ such that $S \subseteq S'$. Therefore, $v$ and its pair reach each other while both are in active set until then. 

    As a result, even after removing unnecessary edges, they remain in the same component as no inactive set cuts their path.
\end{proof}

Now, based on the fingerprint \(\tplus\) obtained from Legacy Execution, we define the notion of the \emph{base phase} for any monotonic moat growing process as follows.

\begin{definition}[Base Phase]  
\label{def:base-phase}
    In an execution of a monotonic moat growing process, at moment \(\currenttime\), a vertex is said to be in the base phase if  
    \[
        \currenttime < \tplus_v.
    \]  
    Similarly, an active set \(S\) is in the base phase if it contains a vertex in the base phase.  
    Additionally, for an active set \(S\) and a given time $\currenttime$, we define \(\BaseSet(S, \currenttime)\) as the set of vertices in the base phase.  
\end{definition}

Intuitively, $\BaseSet(S, \currenttime)$ is the set of vertices in $S$ that have not yet reached their initial fingerprint, which indicates the time they reach their pair in Legacy Execution. 
We use this notation to define different assignments for different monotonic moat growing algorithms.
Now, a basic observation regarding the base phase is the following.

\begin{lemma}
\label{lm:base-set-subset}
    Let $\currenttime$ be any moment during the execution of a monotonic moat growing process. For any pair of vertex subsets $S$ and $S'$ such that $S \subseteq S'$, it holds that
    $$
        \BaseSet(S, \currenttime) \subseteq \BaseSet(S', \currenttime),
    $$
\end{lemma}
\begin{proof}
    Suppose $v \in \BaseSet(S, \currenttime)$. By Definition~\ref{def:base-phase}, we have $\currenttime < \tplus_v$. Since $v \in S$ and $S \subseteq S'$, it follows that $v \in S'$. Therefore, by the same definition, $v \in \BaseSet(S', \currenttime)$.
\end{proof}

Next, we introduce an assignment that maps the growth occurring in Legacy Execution to individual vertices, with respect to the optimal solution. Similar assignments will be defined for other executions in our algorithms, allowing us to compare them and ultimately derive the desired approximation factor. 

To make this assignment meaningful and consistent, we introduce a notion of \emph{priority}—a way to rank vertices based on when their demands are satisfied in Legacy Execution. Priority helps determine which vertices in a subset are most influential in driving further growth in the algorithm.

\begin{definition}[Priority]
\label{def:priority}
    The priority of a vertex $v$ is determined by the time at which its associated demand is satisfied during Legacy Execution. Vertices satisfied later (larger $\tplus$) are assigned higher priority.

    In the event of a tie, priorities are resolved using a fixed total ordering over all vertices. This ordering is constructed by first ordering all demand pairs, and then, for each pair involving distinct vertices, arbitrarily ordering the two vertices within the pair.

    We denote the priority of a vertex $v$ by $\priority_v$, and write $\priority_v > \priority_u$ to indicate that $v$ has a higher priority than $u$.
\end{definition}

Immediately, we observe the following relationship between priority and the fingerprint of Legacy Execution.

\begin{corollary}
\label{cor:priority-tplus}
    For any pair of vertices $v$ and $u$,
    $$
        \tplus_v > \tplus_u \;\Rightarrow\; \priority_v > \priority_u,
    $$
    and
    $$
        \priority_v > \priority_u \;\Rightarrow\; \tplus_v \ge \tplus_u.
    $$
\end{corollary}

Here, we define representatives for any subset of vertices based on priority, selecting vertices with the highest priority from each connected component within the subset.
This notation is also crucial for defining assignments.

\begin{definition}[Representatives]
\label{def:priority-set}
    For a subset $S$ of vertices, the set of \emph{representatives} of $S$ consists of the vertices in $S$ that have the highest priority within their respective connected components of the optimal solution. Formally,
    $$
        \priorityset(S) = \Big\{v \in S \;\Big|\; v = \argmax_{u \in S \cap \optcom(v)} \priority_u \Big\},
    $$
    where $\optcom(v)$ denotes the connected component of $v$ in the optimal solution.
\end{definition}

The following lemma can easily be concluded from the above definition.

\begin{lemma}
    \label{lm:opti_cap_priorityset}
    In a moat growing algorithm, for any time $\currenttime$, any subset of vertices $S\subseteq \V$, and any connected component of the optimal solution $\opti$, at most one vertex of $\opti$ is in $\priorityset(\BaseSet(S, \currenttime))$, which is the vertex with maximum priority. More precisely,
    \begin{align*}
        &\text{If } \BaseSet(S, \currenttime) \cap \opti = \emptyset \text{, then } & &\priorityset(\BaseSet(S, \currenttime)) \cap \opti = \emptyset \\
        &\text{Otherwise, } & &\priorityset(\BaseSet(S, \currenttime)) \cap \opti = \{\argmax_{v\in S \cap \opti}\priority_v\}
    \end{align*}
\end{lemma}
\begin{proof}
    By Definition~\ref{def:priority-set}, $\priorityset(\BaseSet(S, \currenttime))$ can contain at most one vertex from $\opti$.
    If $\tplus_v \le \currenttime$ for the vertex $v = \argmax_{u \in S\cap \opti} \priority_u$, which also has the maximum value of $\tplus$ among vertices in $S \cap \opti$ (based on Corollary~\ref{cor:priority-tplus}), then $\BaseSet(S, \currenttime) \cap \opti = \emptyset$, so the statement holds.
    Otherwise, since $\tplus_v > \currenttime$, $v$ is in $\BaseSet(S, \currenttime)$ and, as it has the maximum priority among vertices in $\BaseSet(S, \currenttime) \cap \opti$, the claim follows.
\end{proof}

Next, we prove a consistency property of representatives throughout the algorithm.

\begin{lemma}
\label{lm:smaller-priorityset}
    Let $S$ be a subset of vertices. If $v \in \priorityset(S)$, then for any subset $S' \subseteq S$ with $v \in S'$, it follows that
    $$
    v \in \priorityset(S').
    $$
\end{lemma}
\begin{proof}
    Suppose, for the sake of contradiction, that $v \notin \priorityset(S')$. Then there exists a vertex $u \in S' \cap \optcom(v)$ such that $\priority_u \ge \priority_v$. Since $S' \subseteq S$, we also have $u \in S$, which contradicts the assumption that $v \in \priorityset(S)$.
\end{proof}

Next, we demonstrate that $\priorityset$ exhibits cardinality monotonicity.

\begin{lemma}  
\label{lm:smaller-priorityset-size}  
    Let $S\subseteq \V$ be a subset of vertices and $S' \subseteq S$. Then,  
    $$
    |\priorityset(S')| \leq |\priorityset(S)|.
    $$  
\end{lemma}  
\begin{proof}
    For any set $S$, $\priorityset(S)$ contains exactly one vertex from each connected component of the optimal solution that intersects with $S$. Consequently, a subset $S' \subseteq S$ can only contain vertices from a subset of the connected components that $S$ intersects with, ensuring that $\priorityset(S')$ cannot have more elements than $\priorityset(S)$.  
\end{proof}

Now we can define our first assignment, $\rplus$, which corresponds to Legacy Execution, using the above definitions.

\begin{definition}
\label{def:rplus}
    Consider the moment $\currenttime$ of Legacy Execution.
    For each active set $S \in \Aplus_{\currenttime}$
    in the \emph{base phase},
    we assign the growth of this moment of $S$ to all vertices in $\priorityset(\BaseSet(S, \currenttime))$ with fraction $1$. We denote the total growth assigned to vertex \(v\) by set \(S\) as \(\rplus_{S,v}\), and the total accumulated growth assigned to \(v\) (over all sets) as \(\rplus_v\).
\end{definition}

Note that for every active set $S$ at time $\currenttime$, we have  $\BaseSet(S, \currenttime) \neq \emptyset$,  otherwise, all vertices in $S$ would already be satisfied, and $S$ would become inactive.  
This implies that $\yplus_s$ is assigned to at least one vertex, leading to the following corollary.

\begin{corollary}
\label{cor:yys_rplusv}
    The total assigned value of $\rplus$ is an upper bound for the total growth of active sets in Legacy Execution. That means,  
    $$
        \sum_{S\subseteq \V} \yplus_s \le \sum_{v\in \V} \rplus_v.
    $$
\end{corollary}

Another observation for this assignment shows that the assignment $\rplus$ to a vertex only happens when the vertex is unsatisfied.

\begin{lemma}
\label{lm:rplus_assigning_condition}
    The growth of a subset of vertices $S\subseteq \V$ can be assigned to $\rplus_v$ of vertex $v$ only if $v\in S$ and $\pairv \notin S$.
\end{lemma}
\begin{proof}
    Let the growth of the active set $S$ in Legacy Execution at moment $\currenttime$ be assigned to $v$.
    Since $\priorityset(\BaseSet(S,\currenttime)) \subseteq S$ for any moment $\currenttime$, the growth of $S$ can only be assigned to vertices in $S$, meaning we must have $v\in S$.
    Additionally, $v\in\priorityset(\BaseSet(S,\currenttime)) \subseteq \BaseSet(S,\currenttime)$ if it is in the base phase, meaning $\currenttime < \tplus_v$.
    Since $v$ and $\pairv$ reach each other at time $\tplus_v$, we have $\pairv \notin S$.
\end{proof}

As a reminder, we have introduced two general frameworks for assignments. We can now show that $\rplus$ satisfies the conditions of a prefix-time assignment as defined in Definition~\ref{def:prefix-time}.

\begin{lemma}
\label{lm:rplus_prefix_time_assignment}
    $\rplus$ is a prefix-time assignment.
\end{lemma}
\begin{proof}
    Consider a moment $\currenttime$ when an active set $S$ exists and assigns its growth to $v$. By Definition~\ref{def:rplus}, 
    $$v \in \priorityset(\BaseSet(S, \currenttime)).$$ 
    This implies $\currenttime < \tplus_v$, meaning that $v$ was in active sets until this moment.

    For any earlier time $\currenttime' < \currenttime$, there exists an active set $S'$ that contains $v$. 
    Given Corollary~\ref{cor:active_sets_laminar}, we know that $S' \subseteq S$.
    By way of contradiction, assume $v \not\in \priorityset(\BaseSet(S', \currenttime'))$. This implies the existence of some $u \in \BaseSet(S', \currenttime') \cap \optcom(v)$ such that $\priority_u > \priority_v$, or equivalently, $\tplus_u \geq \tplus_v$ (by Corollary~\ref{cor:priority-tplus}). Furthermore, since
    \begin{itemize} 
        \item $\currenttime < \tplus_v \leq \tplus_u$, and
        \item $v$ and $u$ remain in the same active sets after $\currenttime'$,
    \end{itemize}
    it follows that $u \in \BaseSet(S, \currenttime) \cap \optcom(v)$, contradicting the fact that $v \in \priorityset(\BaseSet(S, \currenttime))$.
        
    Finally, by Definition~\ref{def:rplus}, an active set assigns a fraction of either 0 or 1 to each vertex, satisfying the two criteria of prefix-time assignment (Definition~\ref{def:prefix-time}).
\end{proof}

As a consequence of the above lemma, we can prove that all vertices $v \in \opti \in \OPT$ are connected to their pair at moment $\rplus_{\max}(\opti)$ in Legacy Execution.

\begin{lemma}
    \label{lm:legacy_connected_at_rmax}
    For any connected component of the optimal solution $\opti$ and any vertex $v\in \opti$, vertex $v$ and $\pairv$ are in the same connected component at time $\rplus_{\max}(\opti)$ in Legacy Execution.
\end{lemma}
\begin{proof}
    Since for every vertex $v$, $\tplus_v$ is the time that $v$ and $\pairv$ connect together, we want to prove that $\tplus_v \le \rplus_{\max}(\opti)$ for any vertex $v \in \opti$.
    
    Let vertex $u\in \opti$ have the highest priority among vertices in $\opti$, meaning 
    $
    u = \argmax_{u' \in \opti} \priority_{u'}
    $.
    Given Lemma~\ref{lm:opti_cap_priorityset}, until $u$ is in the base phase, any active set containing $u$ should assign its growth to $u$.
    Therefore, since based on Lemma~\ref{lm:rplus_prefix_time_assignment} we know $\rplus$ is a prefix-time assignment, we have 
    $$
    \tplus_u \le \rplus_u.
    $$
    
    Additionally, using Corollary~\ref{cor:priority-tplus} we know that 
    $$
    u = \argmax_{u' \in \opti} \tplus_{u'}.
    $$
    Therefore, for any vertex $v \in \opti$, we have
    $$
    \tplus_v \le \tplus_u \le \rplus_u \le \rplus_{\max}(\opti),
    $$
    which completes the proof.
\end{proof}

The next lemma gives a lower bound for the cost of any connected component of the optimal solution based on the total assigned value to its vertices in $\rplus$.

\begin{lemma}
\label{lm:upper_bound_sum_rplus}
Let \(\opti\) be a connected component in the optimal solution and \(\otree\) be the tree of the optimal solution spanning the vertices in \(\opti\). We have,
\[
\sum_{v \in \opti} \rplus_v \le \cc(\otree).
\]
\end{lemma} 
\begin{proof}
    Recall from Definition~\ref{def:rplus}, by Lemma~\ref{lm:opti_cap_priorityset}, each moment of growth of active set $S$ in Legacy Execution is assigned to $\rplus$ of at most one vertex in $\opti$.
    Therefore,
    \begin{align}
    \label{eq:rplus_s_v_yplus_s}
        \sum_{v\in \opti}\rplus_{S,v} \le \yplus_s.
    \end{align}
    
    Additionally, based on Lemma~\ref{lm:rplus_assigning_condition}, the growth of $S$ can be assigned to $v \in \opti$ only if $v \in S$ and $\pairv \notin S$. Since $\pairv \in \opti$, we can conclude that $S \odot \opti$.
    Therefore, we can write a refined version of Definition~\ref{def:assignment} as follows:
    \begin{align*}
        \sum_{v\in \opti} \rplus_v &= \sum_{v\in \opti} \sum_{\substack{S\subseteq \V \\ S\odot \opti}} \rplus_{S,v} \\
        &= \sum_{\substack{S\subseteq \V \\ S\odot \opti}} \sum_{v\in \opti} \rplus_{S,v} \\
        &\le \sum_{\substack{S\subseteq \V \\ S\odot \opti}} \yplus_S \tag{Equation~\ref{eq:rplus_s_v_yplus_s}}\\
        &\le  \sum_{\substack{S\subseteq \V \\ S\odot \opti}} |\deltaS \cap \otree|\cdot\yplus_S \tag{If $S\odot \opti$ then $|\deltaS\cap \otree| \ge 1$}\\
        &\le \sum_{S \subseteq \V} |\deltaS \cap \otree|\cdot \yplus_S \\
        &\le \cc(\otree).\tag{Lemma~\ref{lm:tree-bound-by-activeset-num-cuts}}
    \end{align*}
\end{proof}

The following lemma, giving a lower bound for the cost of the optimal solution, can easily be derived from the above lemma.

\begin{lemma}
\label{lm:yys_opt}
    The total growth of active sets in Legacy Execution is a lower bound for the cost of the optimal solution.
    $$
    \sum_{S \subseteq \V} \yplus_S \le \cc(\OPT).
    $$
\end{lemma}
\begin{proof}
    The proof can be derived by giving a lower bound for the cost of each connected component of the optimal solution as follows:
    \begin{align*}
        \sum_{S \subseteq \V} \yplus_S &\le \sum_{v\in \V} \rplus_v \tag{Corollary~\ref{cor:yys_rplusv}}\\
        &= \sum_{\opti \in \OPT} \sum_{v\in \opti} \rplus_v \\
        &\le \sum_{\opti \in \OPT} \cc(\otree) \tag{Lemma~\ref{lm:upper_bound_sum_rplus}} \\
        &= \cc(\OPT).
    \end{align*}
\end{proof}

\section{Local Search}
\label{sec:local_search}
In this section, we introduce a novel \emph{local search} approach that incrementally improves a given solution by modifying an initial monotonic moat growing algorithm.
The key idea is to delay the deactivation of certain active sets, allowing them to remain active for longer.
This can reduce the total growth of active sets by enabling them to connect and grow together more quickly.
As a result, the upper bound on the solution becomes smaller, and the actual cost may also decrease.

More precisely, we select a vertex $v$ and increase its fingerprint value $\gt_v$, an operation we call a \emph{boost action}.
We then check whether the moat growing algorithm corresponding to the new fingerprint yields a smaller total growth of active sets.
If it does, we apply the boost and repeat the process until no further improvement is possible.
This method not only reduces the upper bound and potentially lowers the cost, but also ensures useful structural properties in the resulting moat growing algorithm.

In the following subsections, we formally introduce boost actions, describe how the local search is carried out using a modified version of \GBG{}, analyze the resulting solution, and present key properties that are useful for the remainder of the paper.

\subsection{Algorithm}
Here, we introduce some necessary preliminary notation that follows the definition of our local search algorithm. We also present a modified version of Shadow Moat Growing algorithm, which is used by the local search.

From Section~\ref{subsec:def-shadow}, we know that in \GBG{}, the input consists of a graph $\G$ and a fingerprint $\tm$, where each $\tm_v$ specifies a time until which vertex $v$ is required to grow.
Since we consider increasing some of these values, we define a \emph{boosted instance} as follows:

\begin{definition}[Boosted Instance]
\label{def:boosted-instance}
    In a boosted instance $\BInsExp$, $\G$ represents the graph, $\tm$ denotes the base fingerprint, and $\tmp$, referred to as the \emph{boosted fingerprint}, represents the required growth after boosting, where $\tmp_v \ge \tm_v$ for all $v \in \V$.
\end{definition}

Initially, we can convert an input instance $(G, \tm)$ of \GBG{} to a boosted input instance $\BInsExp$ by letting $\tmp_v = \tm_v$ for all $v \in \V$. Now, we can apply boost actions to this instance to adjust the required growth time. We formally define \emph{boost action} below:

\begin{definition}[Boost]
    For a boosted instance $\BInsExp$, a \emph{boost action} $\BoostExp$, where $\tmpp \ge \tmp_v$, applied to a vertex $v$ updates $\tmp_v$ to $\tmpp$. 
    This results in a new boosted instance where the value of $\tmp_v$ is replaced by $\tmpp$.
    We denote this updated instance by $\WithBoost(\BIns, \Boost)$.
\end{definition}

To evaluate whether a boost action is beneficial, we first need to modify \GBG{}.
Therefore, we introduce \BoostedModGW{}, which is essentially the same algorithm as \GBG{} described in Section~\ref{subsec:def-shadow}, as both grow moats until they reach the times specified in the fingerprint. 
The difference lies in additional measurements used to assess whether a boost action is beneficial.  
Additionally, \BoostedModGW{} maintains a connected component $S$ active until no vertex $v$ within $S$ satisfies $\tmp_v > \currenttime$, where $\currenttime$ denotes the current moment in the algorithm. 

We define two parameters, $\ybase$ and $\yadd$, along with a function $\yb: 2^\V \to \mathbb{R}_{\ge 0}$.
The value $\ybase$ denotes the total growth required to reach $\tm_v$ for all $v \in \V$, and $\yadd$ represents the additional growth needed to reach $\tmp_v$.
More precisely, when an active set grows at time $\currenttime$, its growth contributes to $\ybase$ if it contains at least one vertex $v$ with $\tm_v > \currenttime$.
If instead all vertices in the set satisfy $\tmp_v > \currenttime \ge \tm_v$, the growth contributes to $\yadd$.
It follows that the total growth of active sets is $\ybase + \yadd$.
We use $\yb_S$ to denote the portion of growth from an active set $S$ that contributes to $\ybase$.

In Algorithm~\ref{alg:mod-gw-boost}, similar to Algorithm~\ref{alg:gbg}, we continue expanding active sets until the time constraints are met. 
The main difference is that each vertex must remain active until $\tmp_v$ instead of $\tm_v$ (see Line~\ref{line:BoostedMG_deactivation_condition}).  
We also split the total growth $\sum_{S \subseteq \V} \ys$ into two quantities, $\ybase$ and $\yadd$, and compute the function $\yb$ during execution (see Lines~\ref{line:BoostedMG_check_base_boost}--\ref{line:BoostedMG_compute_boost}).

Now, we measure the benefit of a boost action against the additional cost it introduces.  
This is done by comparing the results of \BoostedModGW{} on two instances, the original boosted instance $\BIns$ and the instance $\WithBoost(\BIns, \Boost)$ obtained after applying a boost action $\Boost$.  
Based on this comparison, we define \Winr{} and \Lossr{} for a boost action as follows.

\begin{definition}[Win, Loss, and Valuable Boost]  
\label{def:win-loss}  
For an instance $\BIns$, let the boost action $\BoostExp$ produce a new instance $\BIns' = \WithBoost(\BIns, \Boost)$. Suppose the procedure \BoostedModGW{} returns $(\ybase, \yadd)$ on $\BIns$ and $(\ybase', \yadd')$ on $\BIns'$.  

The \emph{win} of the boost action is the total decrease in $\ybase$, defined by  
$$
\Winr(\BIns, \Boost) = \ybase - \ybase'.
$$  \label{func:win}
The \emph{loss} is the total increase in $\yadd$ caused by applying $\Boost$, given by  
$$
\Lossr(\BIns, \Boost) = \yadd' - \yadd.
$$  \label{func:loss}
We say that the boost action $\Boost$ on instance $\BIns$ is \emph{valuable} if  
$$
\Winr(\BIns, \Boost) \ge (1+\bta) \cdot \Lossr(\BIns, \Boost),
$$  
where $\bta \in (0, 1)$ is a fixed constant. 
\end{definition}  

It is important to note that we will present the appropriate value of $\beta$ in Section~\ref{sec:final}, although this value is passed to \LocalSearch{} whenever we use it in our algorithm.

\begin{algorithm}[H]
  \caption{Boosted Moat Growing}
  \label{alg:mod-gw-boost}
  \hspace*{\algorithmicindent} \textbf{Input:} 
  A boosted instance $\BInsExp$ including a graph $\G = (\V, \E, \cc)$ with edge costs $\cc: \E \to \mathbb{R}_{\geq 0}$, a function $\tm: \V \to \mathbb{R}_{\geq 0}$ specifying the original fingerprint, and a function $\tmp: \V \to \mathbb{R}_{\geq 0}$ specifying the boosted fingerprint. \\
  \hspace*{\algorithmicindent} \textbf{Output:} 
  $\F$ is the resulting forest of fingerprint $\tmp$,
  $\ybase$ is the amount of growth needed to satisfy fingerprint $\tm$, $\yadd$ is
  the additional amount of growth to reach fingerprint $\tmp$,
  $\yb: 2^\V \to \mathbb{R}_{\ge 0}$ is a function that maps each set $S$ to the amount of its growth counted toward $\ybase$.
  
  \begin{algorithmic}[1]
    \Procedure{\BoostedModGWr}{$\BInsExp$}
       \label{func:bmg}
      \State $\currenttime \gets 0$
      \State $\F \gets \emptyset$
      \State $\currentsets \gets \{\{v\} \mid v \in \V\}$
      
      \State $\activesets \gets \{\{v\} \mid v \in \V, \tV > 0\}$
      \State $\deactivesets \gets \{\{v\} \mid v\in \V, \tV = 0\}$
      
      \State Implicitly set $\ys \gets 0$ for $S \subseteq \V$
      \State Implicitly set $\yb_S \gets 0$ for $S \subseteq \V$
      \While{$\activesets \neq \emptyset$} 
      \Comment{While there exists an active set}
        \State $\Deltae \gets \min_{e = uv \in \E} \frac{\cc_e - \sum_{S \ni e} \ys}{|\{S_u, S_v\} \cap \activesets|}$, where $u\in S_u\in \currentsets$, $v\in S_v\in \currentsets$, and $S_u \ne S_v$
        \State $\Delta_{\tm} \gets \min_{v \in \V, \tm_v > \currenttime}(\tm_v - \currenttime)$
        \State $\Delta_{\tmp} \gets \min_{v \in \V, \tmp_v > \currenttime}(\tmp_v - \currenttime)$
        \State $\Delta \gets \min(\Deltae, \Delta_{\tm}, \Delta_{\tmp})$
        \For{$S \in \activesets$}
        \label{line:start-add-delta}
            \If{$v \in S$\textbf{ exists  such that } $\tm_v > \currenttime$} \label{line:BoostedMG_check_base_boost}
                \State $\ybase \gets \ybase + \Delta$
                \State $\yb_S \gets \yb_S + \Delta$
            \Else
                \State $\yadd \gets \yadd + \Delta$ \label{line:BoostedMG_compute_boost}
            \EndIf
            \State $\ys \gets \ys + \Delta$
        \EndFor
        \label{line:end-add-delta}
        \State $\currenttime \gets \currenttime + \Delta$
        \For{$e\in E$}
          \State Let $S_v, S_u \in \currentsets$ be sets that contain each endpoint of $e$
          \If{$\sum_{S: e \in \deltaS} \ys = \ce$ \textbf{and} $S_v \neq S_u$}
          \Comment{Edge $(v, u)$ become fully colored}
            \State $\F \gets \F \cup \{e\}$
            \State $\currentsets \gets (\currentsets \setminus \{S_v, S_u\}) \cup \{S_v \cup S_u\}$
           
            \State $\activesets \gets (\activesets \setminus \{S_v, S_u\}) \cup \{S_v \cup S_u\}$
          \EndIf
        \EndFor
        \For{$S \in \activesets$}
          \If{$\tmp_v \le \currenttime \textbf{ for all } v \in S$}  \label{line:BoostedMG_deactivation_condition}
          \Comment{$S$ become inactive}
            \State $\activesets \gets \activesets \setminus \{S\}$    
            \State $\deactivesets \gets \deactivesets \cup \{S\}$
          \EndIf
        \EndFor
      \EndWhile
      
      \While{$S\in \deactivesets \textbf{ exists such that } |\deltaS \cap \F|= 1$}
      \Comment{Remove unnecessary edges}
        \State $\F \gets \F \setminus \deltaS$
      \EndWhile
      
      \State \Return $\F, \ybase, \yadd, \yb$
    \EndProcedure
  \end{algorithmic}
\end{algorithm}

Using valuable boost actions, we design our local search algorithm. We start with a boosted instance where $\tmp_v = \tm_v$ for all $v \in V$. In each iteration, we identify a valuable boost action and apply it to update the instance. This process is repeated until no further valuable boost actions remain.

\paragraph{Boost Actions Space.} One natural question that arises is: given the infinite number of possible boost actions, how can we determine whether a valuable boost exists? To address this, we define a polynomially bounded set of boost actions (with size polynomial in the input), and later show that this restricted subspace is sufficient for the analysis in this paper (see Lemma~\ref{lm:xyz}).

Our method proceeds as follows. We consider all vertices as potential candidates for the boost action. For a given vertex \( v \), we identify at most \( n \) specific time points such that checking only these is sufficient to determine whether a valuable boost action exists involving \( v \).

To find these time points, we simulate Shadow Moat Growing algorithm with a modified parameter \(\tmp\), where we set \(\tmp_v = \infty\). Here, \(\infty\) is chosen to be large enough so that the final moat grown around \( v \) includes all vertices while \( v \) is active. During this run, each time the active set containing \( v \) merges with another active set, we record the current time \(\currenttime\), but only if \(\currenttime\) exceeds the previous value of \(\tmp_v\).  
Note that we only capture times when \( v \) becomes connected to another active set, not to a component that has already become inactive.  

Let the resulting sequence of such time points be \(\currenttime_1, \currenttime_2, \dots, \currenttime_k\).  
We then define the candidate boost actions for vertex \( v \) as \(\Boost \in \{(v, \currenttime_i)\}_{i=1}^k\).  
Since each recorded time corresponds to a merge event between two sets, and there are at most \( n \) such events, we have \( k \leq n \).  
Repeating this process for all \( n \) vertices yields at most \( n^2 \) candidate boost actions to check.

\begin{corollary}
\label{cor:find-boost-polynomial}
    A valuable boost action, if one exists in the defined subspace, can be found in polynomial time or confirmed not to exist.
\end{corollary}

At the end of this process, we obtain a boosted instance \( \BInsExp \) in which no further valuable boost actions within the defined space can be applied.  
The final output includes the boosted fingerprint and the forest produced by running \BoostedModGW{} on it.  
This approach is summarized in Algorithm~\ref{alg:local-search}.  
We refer to the execution of \BoostedModGW{} in Line~\ref{line:corres-mg} as the \emph{corresponding moat growing} in the local search, a term that will be further discussed.

\begin{algorithm}[H]
  \caption{Local Search}
  \label{alg:local-search}
  \hspace*{\algorithmicindent} \textbf{Input:} 
  Graph $G = (V, E, c)$ with edge costs $c: E \to \mathbb{R}_{\geq 0}$, function $t: V \to \mathbb{R}_{\geq 0}$ specifying a fingerprint, and parameter $0 < \bta < 1$.\\
  \hspace*{\algorithmicindent} \textbf{Output:} 
  A forest $\F$, a boosted fingerprint $\tmp$, and a function $y: 2^{\V} \to \mathbb{R}_{\geq 0}$ specifying the growth of active sets, all with respect to the corresponding moat growing algorithm.
  \begin{algorithmic}[1]
    \Procedure{\LocalSearchr}{$\G, \tm, \bta$}
    \label{func:ls}
        \State $\tmp \gets \tm$
        \State Initialize $\BIns \gets (\G, \tm, \tmp)$, representing the instance before any boost action is applied.
        \While {$\BoostExp \textbf{ exists such that } \Win(\BIns, \Boost) \ge (1+\bta) \cdot \Loss(\BIns, \Boost)$}
        \label{line:local-search-iteration}
            \State $\tmp_v \gets \tmpp$
            \State $\BIns \gets \WithBoost(\BIns, \Boost)$
        \EndWhile
        \State $\F, \ybase, \yadd, \yb \gets \BoostedModGW(\G, \tm, \tmp)$
        \label{line:corres-mg}
        \State \Return $\F, \tmp, \yb$
    \EndProcedure
  \end{algorithmic}
\end{algorithm}

Now, in \LocalSearch{}, we apply multiple boost actions. In our analysis, we are interested in the total win and loss across all these actions, leading to the following definition.

\begin{definition}[Total Win and Loss]
\label{def:total-win-loss}
    For a \LocalSearch{} on instance $\Ins$, we define $\win$ and $\loss$ as the total \Win{} and \Loss{}, respectively, resulting from all boost actions applied during the \LocalSearch{} to the initial instance. Consequently, we have
    $$
        \win \ge (1+\bta) \cdot \loss.
    $$
\end{definition}

\subsection{Analysis of Local Search}

In this section, we present key properties and analyses concerning boost actions, valuable boost actions, and the \LocalSearch{} procedure.

We first show that $\ybase$ and $\yadd$ together account for the total growth of active sets.

\begin{lemma}
\label{lm:ys-sum-ybase-yadd}
    For an instance $\BIns$, at the end of running \BoostedModGW{} on $\BIns$, we have
    $$
    \sum_{S \subseteq \V} \ys = \ybase + \yadd.
    $$
\end{lemma}
\begin{proof}
    Based on the Lines~\ref{line:start-add-delta}-~\ref{line:end-add-delta} of Algorithm~\ref{alg:mod-gw-boost}, anytime we are adding $\Delta$ to one $\ys$, we add the same amount to one of $\ybase$ or $\yadd$, so the equality holds.
\end{proof}
 
Since each boost action produces a larger fingerprint, we can derive the following two corollaries from Lemma~\ref{lm:large_fingerprint_refinement}.

\begin{corollary}
\label{cor:boost_activeset_refinement}
    For any instance $\BIns$ and boost action $\Boost$, 
    the active sets at any moment $\currenttime$ during the execution of \BoostedModGW{} on $\BIns$ form a {\em refinement} of the active sets at the same moment during the execution of \BoostedModGW{} on $\WithBoost(\BIns, \Boost)$.
    Similarly, the connected components at the same moment in the execution on $\BIns$ form a refinement of the connected components in the other execution.
\end{corollary}

\begin{corollary}
\label{cor:localsearch_activeset_refinement}
    For any instance $(\G, \tm, \bta)$, let $\tmp$ be the boosted fingerprint produced by \LocalSearch{} on this instance. Then, at any moment $\currenttime$, the active sets in the execution of \GBG{} on $(\G, \tm)$ form a \emph{refinement} of the active sets at the same moment in the execution on $(\G, \tmp)$.  
    Similarly, the connected components at the same moment in the execution on $(\G, \tm)$ form a refinement of those in the execution on $(\G, \tmp)$.
\end{corollary}

\paragraph{Boost Action Properties.}
We now focus on the properties of a boost action. Using Corollary~\ref{cor:boost_activeset_refinement}, we derive bounds for $\Win$ and $\Loss$ for any boost action. We begin by showing that $\Win$ is always nonnegative.

\begin{lemma}
\label{lm:positive_win}
    For any instance $\BInsExp$ and boost action $\Boost$, we have
    $$
        \Win(\BIns, \Boost) \ge 0.
    $$
\end{lemma}
\begin{proof}
    By Definition~\ref{def:win-loss}, let $\BIns' = \WithBoost(\BIns, \Boost)$, and suppose that running \BoostedModGW{} on $\BIns$ yields $(\ybase, \yadd)$ while running it on $\BIns'$ yields $(\ybase', \yadd')$. Since $\Win(\BIns, \Boost) = \ybase - \ybase'$, it suffices to show that $\ybase \ge \ybase'$.

    Consider an active set $S'$ at moment $\currenttime$ in the execution on $\BIns'$ that contributes to $\ybase'$. By definition, this means there exists a vertex $v \in S'$ such that $\currenttime < \tm_v$. At the same moment in the execution on $\BIns$, vertex $v$ must lie in a connected component $S$, and since $\currenttime < \tm_v$, its growth also contributes to $\ybase$.
    By Corollary~\ref{cor:boost_activeset_refinement}, we have $S \subseteq S'$. Hence, for every active set contributing to $\ybase'$ in $\BIns'$, there exists a distinct subset in $\BIns$ that contributes to $\ybase$. Since the active sets in the original execution are disjoint, their corresponding subsets in the boosted execution are also disjoint and distinct.

    Therefore, each contribution to $\ybase'$ is matched by a corresponding (distinct) contribution to $\ybase$, implying $\ybase \ge \ybase'$. This completes the proof.
\end{proof}

We also use Corollary~\ref{cor:boost_activeset_refinement} to establish an upper bound on $\Loss$.

\begin{lemma}  
\label{lm:loss-less-than-new-time}  
For any instance $\BInsExp$ and boost action $\BoostExp$, we have 
\[
    \Loss(\BIns, \Boost) \le \tmpp.
\]  
\end{lemma}  

\begin{proof}  
By Definition~\ref{def:win-loss}, let $\BIns' = \WithBoost(\BIns, \Boost)$, and suppose that running \BoostedModGW{} on $\BIns$ yields $(\ybase, \yadd)$, while on $\BIns'$ it yields $(\ybase', \yadd')$. It suffices to show that  
\[
    \yadd' - \yadd \le \tmpp.
\]  

Consider any moment $\currenttime'$ during the execution on $\BIns'$, and let $S'$ be an active set at $\currenttime'$ that contributes to $\yadd'$. Suppose there exists a vertex $u \ne v$ in $S'$ such that $\tm_u \le \currenttime' < \tmp_u$. Then, in the execution on $\BIns$, the same vertex $u$ is in an active set $S$ at $\currenttime'$.  

By Corollary~\ref{cor:boost_activeset_refinement}, we have $S \subseteq S'$, which ensures that no vertex $w$ in $S$ satisfies $\currenttime' < \tm_w$. Thus, $S$ contributes to $\yadd$.  

The only case where an active set contributes to $\yadd'$ at time $\currenttime'$ but a corresponding subset does not contribute to $\yadd$ is when $v$ is the only vertex satisfying $\tm_v \le \currenttime' < \tmpp$. That is, the moment lies between the original and updated boosted fingerprint values for $v$. The total growth over such moments is at most $\tmpp$.  

In summary, since every contribution to $\yadd'$ either:
\begin{itemize}
    \item Contains a subset that contributes to $\yadd$, or
    \item Is solely due to $v$ during its extension period, contributing at most $\tmpp$,
\end{itemize}
we conclude that  
\[
    \yadd' \le \yadd + \tmpp,
\]
as desired.
\end{proof}  

Next, we state a property which shows that for any active set after a boost action, if it is independent of that boost action, then there exists a subset of it that was active before the boost action.

\begin{lemma} \label{lm:boost-not-involved-active-set}
Consider any instance $\BIns$ and boost $\Boost = (v, \tmpp)$. Then, for any moment $\currenttime_{0}$ and any active set $S$ at this moment of the execution of \BoostedModGW{} on $\WithBoost(\BIns,\Boost)$, if either $S$ does not include $v$ or $\tmpp \le \currenttime_0$, then there exists an active set $S'$ at moment $\currenttime_{0}$ of the execution of \BoostedModGW{} on $\BIns$ such that $S' \subseteq S$. 
\end{lemma}
\begin{proof}
    Let $\BIns'$ be $\WithBoost(\BIns, \Boost)$, We denote $\tm$ as the fingerprint of $\BIns$. We investigate two cases for $S$ in the execution on $\BIns'$ at moment $\currenttime_0$.

    First, assume $v \not \in S$. In this case, a vertex $u \in S$ must exist where $\tm_u > \currenttime_0$. Consider the same moment in the execution on $\BIns$. Since $\tm_u > \currenttime_0$, $u$ must belong to an active set $S'$. As both $S$ and $S'$ have $u$, by Corollary~\ref{cor:boost_activeset_refinement}, $S' \subseteq S$. 
    Second, assume $v \in S$ and $\currenttime_0 \ge \tmpp$. Since 
    $\currenttime_0 \ge \tmpp$, $S$ is not active due to the boost $\Boost$. 
    Therefore, there again must exist a vertex $u$ such that $\tm_u < \currenttime_0$. Similar to the first case, we can prove that active set $S'$ exists such that $S' \subseteq S$.
    Thus, the proof is complete.
 \end{proof}

\paragraph{Valuable Boost Action Properties.} 
We now focus on valuable boost actions.  
To begin, we observe that for any valuable boost action, $\Win$ is greater than or equal to $\Loss$.

\begin{lemma}
\label{lm:win-minus-loss-positive}
    For any instance $\BIns$ and valuable boost action $\Boost$, we have
    $$
        \Win(\BIns, \Boost) - \Loss(\BIns, \Boost) \ge 0.
    $$
\end{lemma}
\begin{proof}
    We consider two cases based on the sign of $\Loss(\BIns, \Boost)$. 
    If $\Loss(\BIns, \Boost) < 0$, According to Lemma~\ref{lm:positive_win} we can conclude that $\Win(\BIns, \Boost) - \Loss(\BIns, \Boost) \ge 0$.
    Otherwise, based on Definition~\ref{def:win-loss} for a valuable boost action we have
    \begin{align*}
        \Win(\BIns, \Boost) - \Loss(\BIns, \Boost) &\ge \Win(\BIns, \Boost) - (1+\bta) \Loss(\BIns, \Boost) \tag{$\bta\Loss(\BIns, \Boost) > 0$}\\
        &\ge 0 \tag{Definition~\ref{def:win-loss}}.
    \end{align*}
\end{proof}

Next, we show how $\Win$ and $\Loss$ capture the relationship between the total growth of active sets before and after a valuable boost action.

\begin{lemma}
\label{lm:new-ys-based-on-win-loss}
    For an instance $\BIns$, let $\BIns' = \WithBoost(\BIns, \Boost)$ be the result of applying a \emph{valuable} boost action $\Boost$.  
    Let $\ys$ and $\ys'$ denote the growth duration of an active set $S \subseteq \V$ in the executions of \BoostedModGW{} on $\BIns$ and $\BIns'$, respectively.  
    Then,
    $$
        \sum_{S \subseteq \V} \ys' = \sum_{S \subseteq \V} \ys - \Win(\BIns, \Boost) + \Loss(\BIns, \Boost).
    $$
\end{lemma}
\begin{proof}
    Let $(\ybase, \yadd)$ and $(\ybase', \yadd')$ be the output of \BoostedModGW{} on $\BIns$ and $\BIns'$, respectively.  
    By applying Lemma~\ref{lm:ys-sum-ybase-yadd} and Definition~\ref{def:win-loss}, we have
    \begin{align*}
        \sum_{S \subseteq \V}\ys' &= \ybase' + \yadd' \tag{Lemma~\ref{lm:ys-sum-ybase-yadd}}\\
        &= \ybase' + \yadd' - \sum_{S \subseteq \V}\ys + \sum_{S \subseteq \V}\ys \\
        &= \ybase' + \yadd' - \ybase - \yadd + \sum_{S \subseteq \V}\ys\tag{Lemma~\ref{lm:ys-sum-ybase-yadd}}\\
        &= \sum_{S \subseteq \V}\ys - \Win(\BIns, \Boost) + \Loss(\BIns, \Boost). \tag{Definition~\ref{def:win-loss}}
    \end{align*}
\end{proof}

Consequently, we can conclude that after a valuable boost action, the total growth of active sets is reduced.  
We use this perspective in the appendix, where we provide instances and argue that if no boost action reduces the total growth, then no boost is valuable.

\begin{lemma}
    \label{lm:valuable-decrease-total}
    For an instance $\BIns$, let $\BIns' = \WithBoost(\BIns, \Boost)$ be the result of applying a \emph{valuable} boost action $\Boost$ to $\BIns$.
    Let $\ys$ and $\ys'$ denote the growth duration of an active set $S \subseteq \V$ in the executions of \BoostedModGW{} on $\BIns$ and $\BIns'$, respectively.  
    Then,
    $$
        \sum_{S \subseteq \V}\ys' \le \sum_{S \subseteq \V}\ys.
    $$
\end{lemma}
\begin{proof}
    According to Lemmas~\ref{lm:win-minus-loss-positive} and~\ref{lm:new-ys-based-on-win-loss}, we have
    \begin{align*}
        \sum_{S \subseteq \V}\ys' 
        &= \sum_{S \subseteq \V}\ys - \Win(\BIns, \Boost) + \Loss(\BIns, \Boost)\tag{Lemma~\ref{lm:new-ys-based-on-win-loss}}\\
        &\le \sum_{S \subseteq \V}\ys. \tag{Lemma~\ref{lm:win-minus-loss-positive}}
    \end{align*}
\end{proof}

\paragraph{Local Search Properties.}
Having established key properties of a single boost action, we now turn to the properties of the local search procedure, which applies multiple boost actions.
First, from Lemma~\ref{lm:new-ys-based-on-win-loss}, we derive the following corollary by noting that, as stated in Definition~\ref{def:total-win-loss}, $\win$ and $\loss$ are the total \Win{} and \Loss{} accumulated over all boost actions during the local search.

\begin{corollary}
\label{cl:total-ys-based-on-win-loss}
    For an instance $(\G, \tm, \bta)$, let $\tmp$ be the output fingerprint of $\LocalSearch(\G, \tm, \bta)$.  
    Consider the execution of \GBG{} with both $\tm$ and $\tmp$.  
    Let $\YIn_S$ and $\YOut_S$ denote the growth of $S$ during the respective runs of \GBG{}.  
    Recalling that the total win and loss of this \LocalSearch{} are denoted by $\win$ and $\loss$, we have
    $$
        \sum_{S \subseteq \V}\YOut_S = \sum_{S \subseteq \V}\YIn_S - \win + \loss.
    $$
\end{corollary}

The next lemma shows that $\loss$ is equal to the total growth of active sets accumulated in $\yadd$.

\begin{lemma}
\label{lm:loss-equal-yadd}
    For an instance $(\G, \tm, \bta)$, let $\tmp$ be the output fingerprint of $\LocalSearch(\G, \tm, \bta)$.  
    Recall that $\loss$ denotes the total loss accumulated during this \LocalSearch{}.  
    Let $\yadd$ be the value returned by \BoostedModGW{} on $(\G, \tm, \tmp)$. Then,
    $$
        \loss = \yadd.
    $$
\end{lemma}
\begin{proof}
    First, let \BoostedModGW{} on $(\G, \tm, \tmp)$ return $\yadd^{(0)}$. We observe that $\yadd^{(0)} = 0$ since active sets become deactivated when they contain no vertices $v$ with $\tm_v > \currenttime$.

    Next, let $\yadd^{(i)}$ denote the output of \BoostedModGW{} after the $i$-th valuable boost action, and let $Loss^{(i)}$ be the corresponding loss for this action. Suppose there are $k$ valuable boost actions. By the lemma's assumptions, we have $\yadd = \yadd^{(k)}$. Then, using Definitions~\ref{def:win-loss} and~\ref{def:total-win-loss}, we have
    \begin{align*}
        \loss &= \sum_{i=1}^{k} Loss^{(i)} \tag{Definition~\ref{def:total-win-loss}}\\
        &= \sum_{i=1}^{k} (\yadd^{(i)} - \yadd^{(i-1)})\tag{Definition~\ref{def:win-loss}}\\
        &= \yadd^{(k)} - \yadd^{(0)} \tag{Eliminating opposite terms}\\
        &= \yadd.
    \end{align*}
\end{proof}

We also show how the value of $\win$ can be computed from the total growth of active sets in the initial and final moat growing executions of the local search.

\begin{lemma}
\label{lm:win-equal-ybase-diff}
    For an instance $(\G, \tm, \bta)$, let $\tmp$ be the output fingerprint of $\LocalSearch(\G, \tm, \bta)$.  
    Recall that $\win$ denotes the total win of this \LocalSearch{}.  
    Let $\YIn_S$ represent the growth of $S$ in $\GBG(\G, \tm)$, and let $(\ybase^*, \yadd^*)$ be the output of $\BoostedModGW(\G, \tm, \tmp)$, where $\YOut_S$ denotes the growth of $S$ in this execution. Then
    $$
        \win = \sum_{S \subseteq \V} \YIn_S - \ybase^*.
    $$
\end{lemma}
\begin{proof}
    It can be concluded from the results so far that:
    \begin{align*}
        \win &= \sum_{S \subseteq \V}\YIn_S - \sum_{S \subseteq \V}\YOut_S + \loss \tag{Corollary~\ref{cl:total-ys-based-on-win-loss}}\\
        &= \sum_{S \subseteq \V}\YIn_S - \ybase^* - \yadd^* + \loss \tag{Lemma~\ref{lm:ys-sum-ybase-yadd}}\\
        &= \sum_{S \subseteq \V}\YIn_S - \ybase^*\tag{Lemma~\ref{lm:loss-equal-yadd}}
    \end{align*}
\end{proof}

Finally, we show that our local search algorithm runs in polynomial time by proving that the number of its iterations is polynomial and using the fact that each iteration considers only a polynomial number of boost actions.

\begin{lemma}
\label{lm:localsearch_polynomial_time}
    The \LocalSearch{} procedure runs in polynomial time.
\end{lemma}
\begin{proof}
    By Corollary~\ref{cor:find-boost-polynomial}, each iteration of Line~\ref{line:local-search-iteration} in \LocalSearch{} examines a polynomial number of candidate boost actions.
    For each evaluation of boost action, we run one instance of \BoostedModGW{}, which itself runs in polynomial time.
    
    To complete the proof, we show that the algorithm applies only a polynomial number of valuable boost actions. Specifically, after applying a polynomial number of such actions, no further valuable boost exists within the defined subspace.

    In each valuable boost action \(\BoostExp\), the chosen time \(\tmpp\) corresponds to the moment when an active set \(S\) reaches another active set \(S'\), which it had not previously reached while both were active. Additionally, since $S'$ is active, there exists a vertex \(u \in S'\) that has not yet reached its required growth time.

    This means that $v$ (from $S$) and $u$ (from $S'$) become a new pair of actively connected vertices, which was not the case before. Since there are at most $\binom{n}{2}$ such pairs, the number of valuable boost actions that can be applied is at most $\binom{n}{2}$, which is polynomial in $n$.

    Therefore, the total number of iterations is polynomial, and each iteration takes polynomial time, implying that the entire local search algorithm runs in polynomial time.
\end{proof}

\subsection{Boosted Execution}
\label{subsec:boosted_execution}
Up to this point, we have described the general structure and properties of \LocalSearch{}.
We now turn our attention to its first execution in Algorithm~\ref{alg:main}.
To support this part of the analysis, we introduce a set of definitions tailored to this specific run.
While these definitions may resemble earlier ones, they are adapted to the current context and will be used throughout the remainder of this section.

Recall that in Line~\ref{line:corres-mg}, we invoke \BoostedModGW{} on the fingerprint produced by the first call to \LocalSearch{}.
We now define the corresponding moat growing execution associated with this invocation.

\begin{definition}[Boosted Execution]
We refer to the execution of \BoostedModGW{} at Line~\ref{line:corres-mg} of the \LocalSearch{} call made in Line~\ref{exe:localsearch} of Algorithm~\ref{alg:main} as the \emph{Boosted Execution}~$(\ee)$.  
The following notation is associated with this execution:
\begin{itemize}
    \item $\solone$ denotes the output forest,
    \item $\yy_S$ denotes the growth of subset $S$,
    \item $\lossone$ denotes the $\loss$ of this \LocalSearch{},
    \item $\winone$ denotes the $\win$ of this \LocalSearch{},
    \item $\Aa_\currenttime$ denotes the family of active sets at time $\currenttime$, and
    \item $\ttt$ is the boosted fingerprint output by this \LocalSearch{}.
\end{itemize}
\end{definition}

Given the above definition, we observe that Boosted Execution corresponds exactly to the run of \BoostedModGW$(\G, \tplus, \ttt)$.

Recall that multiple boost actions are performed during the execution of \LocalSearch{}. Since each boost action increases the required growth time of a vertex, the following corollary holds.

\begin{corollary}
\label{cor:tt-ge-tplus}
For all vertices $v$, we have
\[
    \ttt_v \ge \tplus_v.
\]
\end{corollary}

Now, given $\ttt$, we can define the boost phase.

\begin{definition}[Boost Phase]
\label{def:boost-phase}
At time $\currenttime$ during the execution of a monotonic moat growing algorithm, a vertex $v$ is said to be in the \emph{boost phase} if 
$$
\tplus_v \le \currenttime < \ttt_v.
$$
An active set $S$ is in the \emph{boost phase} at time $\currenttime$ if:
\begin{itemize}
    \item no vertex in $S$ is in the base phase, and
    \item at least one vertex in $S$ is in the boost phase.
\end{itemize}
\end{definition}

\begin{definition}
\label{def:yy-base-boost}
We define the following quantities for Boosted Execution:
\begin{itemize}
    \item $\yy_{\base}$ denotes the total growth of active sets during the time they were in their base phase,
    \item $\yyb_{S}$ denotes the total growth of active set $S$ during the time it was in the base phase, and
    \item $\yy_{\add}$ denotes the total growth of active sets during the time they were in their boost phase.
\end{itemize}
As mentioned earlier, these quantities correspond to the output of $\BoostedModGW(\G, \tplus, \ttt)$.
\end{definition}

Since $\yyb_{S}$ represents the total time set $S$ is in the base phase, and $\yys$ represents the total time $S$ is an active set, we have the following corollary.

\begin{corollary}
\label{cor:yyb-smaller-yy}
    For any $S \subseteq V$,
    $$
        \yyb_{S} \le \yys.
    $$
\end{corollary}

\paragraph{A Prefix-Time Assignment.}
Next, we define a new assignment based on Boosted Execution, show that it satisfies the properties of a prefix-time assignment, and compare it with the previously defined one.

\begin{definition}
\label{def:rr}
    Consider the moment $\currenttime$ of Boosted Execution. For each active set $S \in \Aa_{\currenttime}$ in the \emph{base phase} ($\exists v \in S: \currenttime < \tplus_v)$, we assign this moment of growth of this active set to all vertices in $\priorityset(\BaseSet(S, \currenttime))$ with fraction 1 and 0 for other vertices. We refer to the total growth assigned to vertex $v$ by $\rr_v$.
\end{definition}

Similar to $\rplus$, we show that $\rr$ is also a prefix-time assignment.

\begin{lemma}
\label{lm:rr_prefix_time_assignment}
    Assignment $\rr$ is a prefix-time assignment.
\end{lemma}
\begin{proof}
Given Definition~\ref{def:prefix-time}, the first criterion holds since the assigned fraction is either 0 or 1.

Now, suppose a vertex $v$ is assigned a fraction of 1 at moment $\currenttime$ but not at an earlier moment $\currenttime'$. According to Definition~\ref{def:rr}, we have $\currenttime < \tplus_v$, which implies $\currenttime' < \tplus_v$. Therefore, at time $\currenttime'$, vertex $v$ belongs to the base set of some active set $S'$, i.e., $v \in \BaseSet(S', \currenttime')$.

Since $v$ was assigned a fraction of 0 at $\currenttime'$, there must exist another vertex $u \in \BaseSet(S', \currenttime')$ such that $\optcom(u) = \optcom(v)$ and $\priority_u > \priority_v$. This implies $\tplus_u \ge \tplus_v$ by Corollary~\ref{cor:priority-tplus}.

Given that we only merge sets in our process, it follows that $S' \subseteq S$ and $u \in S$. Since $\currenttime < \tplus_v \le \tplus_u$, it follows that $u \in \BaseSet(S, \currenttime)$, and hence $v \notin \priorityset(\BaseSet(S, \currenttime))$, which contradicts the assumption that $v$ was assigned a fraction of 1 at time $\currenttime$.
\end{proof}

The next lemma shows that the minimum assigned value $\rr$ among two vertices in the same connected component of the optimal solution is a lower bound on the time at which they become connected in Boosted Execution.

\begin{lemma}
\label{lm:connect_after_min_r}
    For any two vertices $u, v$ in a same connected component of the optimal solution, these vertices are not connected until time $\min(\rr_u, \rr_v)$ in Boosted Execution. 
\end{lemma}
\begin{proof}
    Without loss of generality, assume $\priority_u > \priority_v$.
    Let $\currenttime$ be the time at which these two vertices become connected in Boosted Execution.
    After time $\currenttime$, they remain in the same connected component for the remainder of the execution.
    For any set $S$ containing both $u$ and $v$, Lemma~\ref{lm:opti_cap_priorityset} implies that $\yys$ cannot be assigned to $v$, as $u$ has higher priority.
    Hence, $\currenttime \ge \rr_v$ because assigning to $v$ is only possible before or at time $\rr_v$.
    By symmetry, if $\priority_u < \priority_v$, we conclude that $\currenttime \ge \rr_u$.
    Therefore, $\currenttime \ge \min(\rr_u, \rr_v)$, and these vertices do not connect before time $\min(\rr_u, \rr_v)$.
\end{proof}

Next, we show that the prefix-time assignment obtained from the first local search is strictly smaller than the previous assignment.

\begin{lemma}
\label{lm:rr-le-rplus}
    For any vertex \(v \in V\),
    $$
        \rr_v \le \rplus_v.
    $$
\end{lemma}
\begin{proof}
    Since \(\rr\) is a prefix-time assignment, the first moment not assigned to \(v\) is precisely \(\rr_v\). Let \(S\) be the last active set in Boosted Execution that assigns its growth to \(v\), and let this happen at time \(\currenttime=\rr_v\). Since \(v\) is in the base phase at this moment, we have \(\tplus_v > \currenttime\), which implies that \(v\) must also belong to an active set \(S'\) in Legacy Execution at time \(\currenttime\).

    By Corollary~\ref{cor:localsearch_activeset_refinement}, active sets in Legacy Execution form a refinement of those in Boosted Execution, which means \(S' \subseteq S\). Then, by Lemma~\ref{lm:base-set-subset}, we have
    $
        \BaseSet(S', \currenttime) \subseteq \BaseSet(S, \currenttime).
    $
    Since \(v\) is in the base phase for both \(S\) and \(S'\), and \(v\) belongs to \(\priorityset(\BaseSet(S, \currenttime))\), Lemma~\ref{lm:smaller-priorityset} implies that \(v \in \priorityset(\BaseSet(S', \currenttime))\) as well.

    By Definition~\ref{def:rplus}, this means that \(S'\) assigns its growth to \(v\) at time \(\currenttime\) in Legacy Execution. Therefore, \(\rplus_v\), which accumulates growth assigned to \(v\) over time, must be at least \(\rr_v\).
    Hence, \(\rr_v \le \rplus_v\), as claimed.
\end{proof}

We now show that, in each connected component of the optimal solution, the maximum assigned value (see  Definition~\ref{def:assignment}) is equal in both executions.

\begin{lemma}
\label{lm:rplus_vone_rr_vone}
    For any connected component \(\opti\) in \(\OPT\),
    $$
        \rplus_{\max}(\opti) = \rr_{\max}(\opti).
    $$
\end{lemma}
\begin{proof}
    Let \(v \in \opti\) be a vertex such that \(\rplus_v = \rplus_{\max}(\opti)\), and among all such vertices, assume \(v\) has the highest priority in \(\opti\).

    Consider a moment \(\currenttime\) just before \(\tplus_v\), where \(v\) is in an active set \(S\) in Legacy Execution. 
    Since \(v \in \BaseSet(S,\currenttime)\), \(S\) is assigning growth to a vertex in \(\opti\). Furthermore, as \(\rplus\) is a prefix-time assignment, \(S\) must assign growth to \(v\), as otherwise the vertex in \(\opti\) to which growth is assigned would have a higher value of \(\rplus\) than \(v\).
    By Corollary~\ref{cor:localsearch_activeset_refinement}, there exists an active set $S'$ in Boosted Execution at moment \(\currenttime\) such that \(S \subseteq S'\). 

    We now show that \(S'\) also assigns its growth to \(v\). By Lemma \ref{lm:base-set-subset}, \(\BaseSet(S,\currenttime)\subseteq \BaseSet(S',\currenttime)\) so \(v\in\BaseSet(S',\currenttime)\). 
    This implies that there is a vertex \(u\in\opti\) in \(\BaseSet(S',\currenttime)\) to which \(S'\) assigns growth.
    Suppose for contradiction \(u\neq v\).
    This would mean that \(u\) has a higher priority than \(v\).
    Since \(u\) is active in Legacy Execution at time \(\currenttime\), and has higher priority than \(v\), this contradicts our assumption that \(v\) is the highest-priority vertex among those with maximum \(\rplus\) in \(\opti\).

    Therefore, \(S'\) must assign its growth to \(v\), implying \(\rplus_v \le \rr_v\). 
    Thus, we can conclude that
    \begin{align*}
        \rplus_{\max}(\opti) &= \rplus_v \\
        &\le \rr_v \\
        &\le \rr_{\max}(\opti).
    \end{align*}

    Now, let \(w \in \opti\) be a vertex such that \(\rr_w = \rr_{\max}(\opti)\). Then:
    \begin{align*}
        \rplus_{\max}(\opti) &\le \rr_{\max}(\opti) \\
        &= \rr_w \\
        &\le \rplus_w \tag{Lemma~\ref{lm:rr-le-rplus}} \\
        &\le \rplus_{\max}(\opti).
    \end{align*}
    Thus, all inequalities must be equalities, and we conclude that
    \(
        \rplus_{\max}(\opti) = \rr_{\max}(\opti).
    \)
\end{proof}

\paragraph{An Exclusive Assignment.}
We now define our first exclusive assignment, corresponding to Boosted Execution, and establish several of its key properties.

\begin{definition}
\label{def:rstar}
    Consider the moment $\currenttime$ during Boosted Execution. For each active set $S \in \Aa_{\currenttime}$ in the \emph{base phase}, we assign this moment of growth \emph{proportionally} to all vertices in $\priorityset(\BaseSet(S, \currenttime))$. Specifically, each vertex in $\priorityset(\BaseSet(S, \currenttime))$ receives a fraction of the growth equal to $1 / |\priorityset(\BaseSet(S, \currenttime))|$. We denote the total growth assigned to a vertex $v$ by $\rstar_v$.
\end{definition}

Since each active set divides its growth equally among its assigned vertices, this defines an exclusive assignment (see Definition~\ref{def:ex-assignment}).  
Moreover, since for any active set only the growth during its base phase is assigned, and the total fraction of assignment during those moments is $1$, we conclude the following:

\begin{corollary}
\label{cor:sum-rstar-yyb}
    Given a set \( S \subseteq V \), we have
    $$
        \sum_{v \in S} \rstar_{S,v} = \yyb_{S}.
    $$
\end{corollary}
We can extend the above lemma and sum it over all sets to obtain the following lemma.

\begin{lemma}
\label{lm:sum-rstar-y-base}
    For the exclusive assignment in Boosted Execution, we have
    $$
        \sum_{v \in \V} \rstar_v = \yy_{\base}.
    $$
\end{lemma}
\begin{proof}
We can prove this by expanding the total assignment over all vertices:
\begin{align*}
    \sum_{v \in \V} \rstar_v 
    &= \sum_{v \in \V} \sum_{S \subseteq \V} \rstar_{S, v} \tag{Definition~\ref{def:assignment}} \nonumber \\
    &= \sum_{S \subseteq \V} \sum_{v \in S} \rstar_{S, v}
    \tag{$\rstar_{S, v} = 0$ for all $v \notin S$}\\
    &= \sum_{S \subseteq \V} \yyb_{S}
    \tag{Corollary~\ref{cor:sum-rstar-yyb}}\\
    &= \yy_{\base}.
    \tag{Definition~\ref{def:yy-base-boost}}
\end{align*}
\end{proof}

The next lemma, which bounds the sum of $\rstar$ over all vertices, follows directly from the above lemma.

\begin{lemma} \label{lm:sum-rstar-less-copt-minus-win}
    For the exclusive assignment in Boosted Execution, we have
    \[
        \sum_{v \in \V} \rstar_v \leq \cc(\OPT) - \winone.
    \]
\end{lemma}
\begin{proof}
    We have
    \begin{align*}
        \sum_{v\in \V} \rstar_v &= \yy_{\base}  \tag{Lemma \ref{lm:sum-rstar-y-base}} \\
        & = \sum_{S\subseteq \V}\yplus_S - \winone \tag{Lemma \ref{lm:win-equal-ybase-diff}}\\
        & \leq \cc(\OPT) - \winone.\tag{Lemma \ref{lm:yys_opt}}
    \end{align*}
\end{proof}

We can also use Lemma~\ref{lm:sum-rstar-y-base} to show that the difference between the total $\yy$ and the total $\rstar$ is exactly $\lossone$.

\begin{lemma}
\label{lm:yy-using-rstar-and-loss}
    For Boosted Execution, we have
    $$
        \sum_{v \in \V} \rstar_v + \lossone = \sum_{S \subseteq V} \yy_S.
    $$
\end{lemma}
\begin{proof}
The result follows directly from previous lemmas:
    \begin{align*}
        \sum_{v \in \V} \rstar_v + \lossone
        &= \yy_{\base} + \lossone \tag{ Lemma~\ref{lm:sum-rstar-y-base}} \\
        &= \yy_{\base} + \yy_{\add} \tag{ Lemma~\ref{lm:loss-equal-yadd}} \\
        &= \sum_{S \subseteq V} \yy_S. \tag{ Lemma~\ref{lm:ys-sum-ybase-yadd}}
    \end{align*}
\end{proof}
We now compare the exclusive and prefix-time assignments of Boosted Execution.

\begin{lemma}
\label{lm:rstar-smaller-rr}
    For any vertex \(v\),
    $$
        \rstar_v \le \rr_v.
    $$
\end{lemma}
\begin{proof}
    By Definitions~\ref{def:rr} and~\ref{def:rstar}, at any moment when growth is assigned to \(v\), the amount contributed to \(\rstar_v\) is less than or equal to the amount contributed to \(\rr_v\), since in the exclusive assignment the growth is divided among potentially more vertices.
    As \(\rstar_v\) accumulates less (or equal) growth than \(\rr_v\) over time, %
    \(
        \rstar_v \le \rr_v.
    \)
\end{proof}

We can use the above lemma to show that for each connected component of the optimal solution, the total assigned value $\rstar$ to its vertices is at most the cost of its corresponding tree.

\begin{lemma} 
\label{lm:opt-lowerbound_rstar}  
    For any connected component of the optimal solution $\opti$ with tree $\otree$, we have
    $$  
     \sum_{v\in \opti} \rstar_v \le \cc(\otree).
    $$
\end{lemma}
\begin{proof}
    It can be concluded as follows:
    \begin{align*}
        \sum_{v\in \opti} \rstar_v &\le \sum_{v\in \opti} \rr_v \tag{Lemma~\ref{lm:rstar-smaller-rr}}\\
        &\le \sum_{v\in \opti} \rplus_v \tag{Lemma~\ref{lm:rr-le-rplus}}\\
        &\le \cc(\otree). \tag{Lemma~\ref{lm:upper_bound_sum_rplus}}
    \end{align*}
\end{proof}

The next lemma is independent of the assignment definitions and relies only on the basic setup of Boosted Execution. 
It shows that the total growth of active sets that do not color any edge of the optimal solution is bounded by the sum of $\winone$ and $\lossone$.

\begin{lemma}
\label{lm:ys_not_cut_opt_win_loss}
    We can bound the total growth of active sets that do not cut $\OPT$ by $\winone+\lossone$, that is,
    $$\sum_{\substack{S \subseteq \V \\ \deltaS\cap \OPT = \emptyset}} \yys \le \winone + \lossone.$$
\end{lemma}
\begin{proof}
    First, we expand the left-hand side:
    \begin{align*}
        \sum_{\substack{S \subseteq \V \\ \deltaS\cap \OPT = \emptyset}} \yys
        &= \sum_{\substack{S \subseteq \V \\ \deltaS\cap \OPT = \emptyset}} \yys 
        - \sum_{S \subseteq \V} \yys 
        + \sum_{S \subseteq \V} \yys\\
        &= \sum_{\substack{S \subseteq \V \\ \deltaS\cap \OPT = \emptyset}} \yys 
        - \sum_{S \subseteq \V} \yys 
        + \sum_{S \subseteq \V} \yyb_{S} + \yadd
        \tag{Lemma~\ref{lm:ys-sum-ybase-yadd}}\\
        &= -\sum_{\substack{S \subseteq \V \\ \deltaS\cap \OPT \ne \emptyset}} \yys 
        + \sum_{S \subseteq \V} \yyb_{S} + \yadd\\
        &= -\sum_{\substack{S \subseteq \V \\ \deltaS\cap \OPT \ne \emptyset}} \yys 
        + \sum_{S \subseteq \V} \yyb_{S} + \lossone
        \tag{Lemma~\ref{lm:loss-equal-yadd}}\\
        &\le -\sum_{\substack{S \subseteq \V \\ \deltaS\cap \OPT \ne \emptyset}} \yyb_{S} 
        + \sum_{S \subseteq \V} \yyb_{S} + \lossone
        \tag{Corollary~\ref{cor:yyb-smaller-yy}}\\
        &= \sum_{\substack{S \subseteq \V \\ \deltaS\cap \OPT = \emptyset}} \yyb_{S} + \lossone.
    \end{align*}
    Next, we expand the right-hand side:
    \begin{align*}
        \winone + \lossone &= 
        \sum_{S \subseteq \V} \yplus_{S} - \sum_{S \subseteq \V} \yyb_{S} + \lossone\tag{Lemma~\ref{lm:win-equal-ybase-diff}}\\
        &= \sum_{S \subseteq \V} \yplus_{S} -\sum_{\substack{S \subseteq \V \\ 
         \deltaS\cap \OPT = \emptyset}} \yyb_{S} - \sum_{\substack{S \subseteq \V \\ \deltaS\cap \OPT \ne \emptyset}} \yyb_{S} + \lossone.
    \end{align*}
    It suffices to show:
    $$
        \sum_{\substack{S \subseteq \V \\ \deltaS\cap \OPT = \emptyset}} \yyb_{S} + \lossone \le \sum_{S \subseteq \V} \yplus_{S} -\sum_{\substack{S \subseteq \V \\ 
         \deltaS\cap \OPT = \emptyset}} \yyb_{S} - \sum_{\substack{S \subseteq \V \\ \deltaS\cap \OPT \ne \emptyset}} \yyb_{S} + \lossone.
    $$
    Simplifying, we get:
    $$
        2\sum_{\substack{S \subseteq \V \\ \deltaS\cap \OPT = \emptyset}} \yyb_{S} +\sum_{\substack{S \subseteq \V \\ \deltaS\cap \OPT \ne \emptyset}} \yyb_{S} \le \sum_{S \subseteq \V} \yplus_{S}
    $$
    To prove this, consider an active set $S$ in Boosted Execution at some moment $\currenttime$, assumed to be in its base phase.  
    By Definition~\ref{def:base-phase}, there exists $v \in S$ such that $\tplus_v > \currenttime$.
    We aim to show that the contribution of $S$ to the left-hand side has a corresponding contribution to the right-hand side.
    We distinguish two cases:

    \begin{itemize}
        \item If $\deltaS\cap \OPT \ne \emptyset$, then at the same moment in Legacy Execution, $v$ is in an active set $S' \subseteq S$ by Corollary~\ref{cor:localsearch_activeset_refinement}. Thus, we say $S$ corresponds to $S'$.
        
        \item If $\deltaS \cap \OPT = \emptyset$, then $\pairv \in S$; otherwise, there must be a path between $v$ and $\pairv$ in the optimal solution, and $S$ would cut that path.
        In Legacy Execution, at the same moment $\currenttime$, $v$ and $\pairv$ belong to different active sets $S'$ and $S''$, respectively, by Definition~\ref{def:legacy_execution} and the fact that $\tplus_v > \currenttime$. Moreover, both $S'$ and $S''$ are subsets of $S$ by Corollary~\ref{cor:localsearch_activeset_refinement}.
        In this case, we say $S$ corresponds to $S'$ and $S''$.
    \end{itemize}

    Therefore, at any given moment, each active set in Boosted Execution either corresponds to one active set in Legacy Execution if it cuts $\OPT$, or corresponds to two active sets in Legacy Execution if it does not cut $\OPT$.

    Since the corresponding active sets in Legacy Execution are subsets of distinct active sets in Boosted Execution at a fixed moment, it follows that Legacy active sets themselves are all distinct. This completes the proof.
\end{proof}

\paragraph{Classifying Connected Components of the Optimal Solution.}  
Next, using the exclusive assignment from Boosted Execution, we partition the connected components of $\OPT$ into two distinct families.
This partition helps us analyze and bound the cost of our solutions, starting with $\solone$.

\begin{definition}
\label{def:classify_A_B}
    For a constant $\cone$, we classify a connected component $\opti$ with tree $\otree$ in the optimal solution as belonging to family $\A$ if
    $$
    \sum_{v \in \opti} \rstar_v \leq (1 - \cone) \cdot \cc(\otree),
    $$
    and to family $\B$ if
    $$
    \sum_{v \in \opti} \rstar_v > (1 - \cone) \cdot \cc(\otree).
    $$
    It follows that $\OPT = \A \cup \B$. We also refer to the forests of connected components in $\A$ and $\B$ as $\OPTA$ and $\OPTB$, respectively. Consequently, $\cc(\OPT) = \cc(\OPTA) + \cc(\OPTB)$.
\end{definition}

It is important to note that in Section~\ref{sec:final}, we determine an appropriate value for $\cone$ to ensure that the algorithm achieves the desired approximation guarantee.

We further partition the connected components in $\B$ into two families.

\begin{definition}
\label{def:classify_B}
    For a constant $\cfive$, a connected component $\opti \in \B$ with tree $\otree$ in the optimal solution is classified into family $\Bone$ if
    $$
        \rmax(\opti) \le \cfive \cdot \cc(\otree),
    $$
    and into family $\Btwo$ if
    $$
        \rmax(\opti) > \cfive \cdot \cc(\otree).
    $$
    Consequently, $\B = \Bone \cup \Btwo$. We denote by $\OPTBone$ and $\OPTBtwo$ the forests consisting of the connected components in $\Bone$ and $\Btwo$, respectively. It follows that $\cc(\OPTB) = \cc(\OPTBone) + \cc(\OPTBtwo)$.
\end{definition}

The value of $\cfive$ will be defined in terms of other parameters in Definition~\ref{def:ctwo_cfive}, and its exact value will be determined in Section~\ref{sec:final}.

By combining Lemma~\ref{lm:rstar-smaller-rr} and Definition~\ref{def:classify_A_B}, we have the following property for connected components in $\B$.

\begin{corollary}  
\label{corollary:B_bound_r}  
    For any connected component $\opti \in \B$, we have
    $$  
     \sum_{v\in \opti} \rr_v > (1-\cone) \cdot \cc(\otree).
    $$  
\end{corollary} 

The next three lemmas focus on bounding the cost of the solution $\solone$, the forest returned by Boosted Execution.
First, we show that if this solution does not meet the desired approximation factor, then $\lossone$ can be bounded in terms of the cost of the optimal solution.

\begin{lemma}\label{lm:lcl-loss}
    Consider the solution $\solone$. Then either:
    \begin{itemize}
        \item $\cc(\solone) \le (2 - 2\alf) \cdot \cc(\OPT)$, or
        \item $\lossone \le \frac{\alf}{\bta} \cdot \cc(\OPT)$.
    \end{itemize}     
\end{lemma}
\begin{proof}
    If $\cc(\solone) \le (2 - 2\alf) \cdot \cc(\OPT)$, we are done. Otherwise, assume $\cc(\solone) > (2 - 2\alf) \cdot \cc(\OPT)$. Then,
    \begin{align*}
        \cc(\solone) 
        &\le 2 \sum_{S \subseteq \V} \yy_S 
        \tag{Lemma~\ref{lm:monotonic-forest-twice-growth}} \\
        &= 2\left(\sum_{S \subseteq \V} \yplus_S - \winone + \lossone \right)
        \tag{Corollary~\ref{cl:total-ys-based-on-win-loss}} \\
        &\le 2\left(\sum_{S \subseteq \V} \yplus_S - \bta \lossone \right)
        \tag{Definition~\ref{def:total-win-loss}} \\
        &\le 2\left(\cc(\OPT) - \bta \lossone \right)
        \tag{Lemma~\ref{lm:yys_opt}}.
    \end{align*}
    Combining this with the assumption $\cc(\solone) > (2 - 2\alf) \cdot \cc(\OPT)$, we get
    $$
        (2 - 2\alf) \cdot \cc(\OPT) < 2\left(\cc(\OPT) - \bta \lossone \right).
    $$
    Dividing both sides by $2$ and rearranging terms, we obtain
    $$
        \lossone < \frac{\alf}{\bta} \cdot \cc(\OPT).
    $$
    This completes the proof.
\end{proof}

Next, we show that if $\solone$ fails to achieve the desired approximation guarantee, then $\winone$ can be bounded in terms of the cost of the optimal solution.

\begin{lemma}\label{lm:lcl-win}
    Consider the solution $\solone$. Then either:
    \begin{itemize}
        \item $\cc(\solone) \le (2 - 2\alf) \cdot \cc(\OPT)$, or
        \item $\winone \le (\alf + \frac{\alf}{\bta}) \cdot \cc(\OPT)$.
    \end{itemize}     
\end{lemma}
\begin{proof}
    If $\cc(\solone) \le (2 - 2\alf) \cdot \cc(\OPT)$, we are done. Otherwise, assume $\cc(\solone) > (2 - 2\alf) \cdot \cc(\OPT)$. Then:
    \begin{align*}
        (2 - 2\alf) \cdot \cc(\OPT) &< \cc(\solone) \\
        &\le 2 \sum_{S \subseteq \V} \yy_S 
        \tag{Lemma~\ref{lm:monotonic-forest-twice-growth}} \\
        &= 2 \left( \sum_{S \subseteq \V} \yplus_S - \winone + \lossone \right) 
        \tag{Corollary~\ref{cl:total-ys-based-on-win-loss}} \\
        &\le 2 \left( \cc(\OPT) - \winone + \lossone \right) 
        \tag{Lemma~\ref{lm:yys_opt}}.
    \end{align*}
    Dividing both sides by $2$ and rearranging terms gives:
    \begin{align*}
        \winone &< \lossone + \alf \cdot \cc(\OPT) \\
        &\le \left( \alf + \frac{\alf}{\bta} \right) \cdot \cc(\OPT), 
        \tag{Lemma~\ref{lm:lcl-loss}}
    \end{align*}
    which completes the proof.
\end{proof}

Finally, we bound the cost of $\solone$ by separately analyzing the assigned values of vertices in classes $\A$ and $\B$, as defined in Definition~\ref{def:classify_A_B}, and incorporating the bound on $\lossone$.

\begin{lemma}
\label{lm:solone-ub}
    For the cost of the solution $\solone$, either
    $$
    \cc(\solone) \le 2(1-\cone) \cdot \cc(\A) + 2\cc(\B) + 2\dfrac{\alf}{\bta} \cdot \cc(\OPT),
    $$
    or $\solone$ is a $(2 - 2\alf)$ approximation of the optimal solution.
\end{lemma}
\begin{proof}
    If $\solone$ is a $(2 - 2\alf)$ approximation, the proof is complete. Otherwise, we proceed as follows:
    \begin{align*}
        \cc(\solone) 
        &\le 2\sum_{S \subseteq \V} \yy_S \tag{Lemma~\ref{lm:monotonic-forest-twice-growth}} \\
        &= 2\sum_{v \in \V} \rstar_v + 2 \lossone \tag{Lemma~\ref{lm:yy-using-rstar-and-loss}} \\
        &\le 2\sum_{\opti \in \A} \sum_{v \in \opti} \rstar_v + 2\sum_{\opti \in \B} \sum_{v \in \opti} \rstar_v + 2 \lossone \\
        &\le 2(1 - \cone) \cdot \cc(\OPTA) + 2\sum_{\opti \in \B} \sum_{v \in \opti} \rstar_v + 2 \lossone \tag{Definition~\ref{def:classify_A_B}} \\
        &\le 2(1 - \cone) \cdot \cc(\OPTA) + 2 \cc(\OPTB) + 2 \lossone \tag{Lemma~\ref{lm:opt-lowerbound_rstar}} \\
        &\le 2(1 - \cone) \cdot \cc(\OPTA) + 2 \cc(\OPTB) + 2\dfrac{\alf}{\bta} \cdot \cc(\OPT) \tag{Lemma~\ref{lm:lcl-loss}}.
    \end{align*}
    This completes the proof.
\end{proof}

\section{Structural Results on Local Minima}
\label{sec:steiner_tree}
In this section, we analyze the properties of the minimal moat growing algorithms resulting from our local search. Our main property, called the \emph{claw property}, states that for any triple of actively connected vertices, the total growth of active sets separating them lower bounds the cost of any tree connecting them by a constant factor. We generalize this property to larger sets of vertices and use it to bound the cost of our algorithm for the Steiner Tree problem.

\subsection{Claw Property}
\label{sec:claw_property}

We begin by stating the claw property in the following lemma. Given the result of our local search algorithm, it provides bounds on how long triples of vertices that connect actively can remain separated based on their distances to any other vertex.

\begin{lemma}[Claw Property]\label{lm:xyz}
    Consider the final execution of \BoostedModGW{} during a call to~\LocalSearch{}~and let $u, v, w$ be vertices that connect actively in this execution. 
    Define $\sttm$ as the first moment when at least two of $u, v, w$ belong to the same active set,  
    and let $\sttmp$ be the first moment when all three belong to the same active set.  
    Then, 
    the following inequality holds for any vertex 
    $q$
    \[
    \sttm + \sttmp \leq \frac{\stx+\sty+\stz}{2} + \bta\frac{\min(\stx,\sty,\stz)}{2}. 
    \]
\end{lemma}

\begin{proof}
    Let $\stins$ be the final instance on which \BoostedModGW{} is run during the local search. Consider the boost action $\Boost=(q,\stbstm)$, where $\stbstm$ is the smallest value for which $q$ reaches $u$, $v$, or $w$ actively in the boosted instance.
    If no such value exists or if $\stbstm>\sttmp$, then $q$ does not reach $u$, $v$, or $w$ while all four grow actively for $\sttmp$. This implies that
    \[
    \stx \geq 2\sttmp, \quad \sty \geq 2\sttmp, \quad \text{and} \quad \stz \geq 2\sttmp.
    \]  
    Consequently, we obtain
    \begin{align*}
        \sttm + \sttmp &\leq 2\sttmp \\
        &\leq \frac{\stx}{2} + \frac{\sty}{2}
        \\
        &\leq \frac{\stx+\sty+\stz}{2} + \bta\frac{\min(\stx,\sty,\stz)}{2}.
    \end{align*}
    From this point forward, we consider the case where $\stbstm \leq \sttmp$.
    We observe that by construction, \[\stbstm \leq \frac{\min(\stx,\sty,\stz)}{2}.\] 
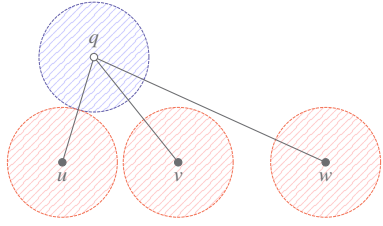
\begin{figure}[t]
    \centering
        \begin{tikzpicture}[scale=0.7]

\draw[White] (0,-2.9) -- (0, 1.4);
\def\dem{Red!70}
\def\demf{Red!20}
\def\demi{Blue!70}
\def\demif{Blue!20}
\def\demii{Green!70}
\def\demiif{Green!20}
\def\ter{White}
\def\col{Black!70}
\def\treeone{Purple}
\def\treetwo{RubineRed!70!Black}
\def\noder{0.07cm}

\def\tercl#1{%
  \ifcase#1
    White\or
    Goldenrod!20\or
    olive!20\or
    Green!20\or
    Sepia!20\or
    Plum!20\or
  \fi
}

\def\len{2}
\def\r{1.044}

\coordinate (A) at (0, 0);
\coordinate (B) at (-0.3*\len, -1*\len);
\coordinate (C) at (0.8*\len, -1*\len);
\coordinate (D) at (2.2*\len, -1*\len);

\foreach \i/\c/\cc in {A/\demi/\demif, B/\dem/\demf, C/\dem/\demf, D/\dem/\demf} {
    \draw[\c, 
    line width=0.3pt, 
    dash pattern=on 1.2pt off 0.4pt,
    pattern={
        Lines[angle=45, distance=2pt,  line width=0.3pt]%
    },
    pattern color=\cc
    ] (\i) circle (\r*\len/2);
}
\newcommand{\sml}[1]{\scalebox{0.7}{#1}}
\draw[\col] (A) -- (B);
\draw[\col] (A) -- (C);
\draw[\col] (A) -- (D);
\foreach \i/\c in {A/\ter, B/\col, C/\col, D/\col} {
    \draw[\col, fill= \c] (\i) circle(\noder);
}
\node[\col, above=-0pt] at (A) {\sml{$q$}};
\node[\col, below=-0pt] at (B) {\sml{$u$}};
\node[\col, below=-0pt] at (C) {\sml{$v$}};
\node[\col, below=-0pt] at (D) {\sml{$w$}};
    \end{tikzpicture}
    \caption{
An illustration of the moat-growing process on a  subrgraph in a boosted instance,  
where a boost is applied to vertex $q$. Before the boost, vertices $u$, $v$,  
and $w$ connect while active.  
The boost enables $q$ to reach $u$ while active, imposing constraints on the  
growth of sets containing $u$, $v$, or $w$ but not $q$, determined by the distance of $q$ to $u$, $v$, and $w$.
}
    \label{fig:claw}
\end{figure}     
Let $\stinsp$ be the boosted instance obtained after applying the boost $\Boost$.  
Figure \ref{fig:claw} illustrates a moment in the execution of \BoostedModGW{} on $\stinsp$, when $q$ reaches one of $u,v,w$. We compare the executions of \BoostedModGW{} on $\stins$ and $\stinsp$.

We use $\YIn$ to denote the $\ys$ values in the execution on $\stins$ and $\YOut$ for the values in the execution on $\stinsp$. 
To analyze the effect of the boost, we track the difference  
\[
\deltaY = \sum_{S \subseteq \V} \YIn_S - \sum_{S \subseteq \V} \YOut_S
\]  
throughout the executions. This allows us to demonstrate that $\Boost$ is a valuable boost action unless the statement of the lemma holds. Since $\Boost$ is included in the boost actions space, and the local search has already terminated, it must not be valuable, completing the proof.  
    
    First, we show that at any moment after \( \sttmp \), the increase in \(\sum_{S \subseteq \V} \YIn_S\) is at least as large as the increase in \(\sum_{S \subseteq \V} \YOut_S\), implying that \( \deltaY \) does not decrease after the moment \( \sttmp \). 
    To establish this,
    we observe that by Lemma~\ref{lm:boost-not-involved-active-set}, at any moment after \( \sttmp\), 
    each active set \( S \) in the run on \( \stinsp \) can be mapped 
    to a distinct active set \( S' \subseteq S\) at the same moment in the run on \( \stins \).

    Since $\deltaY$ does not decrease after moment $\sttmp$, we use $\styi_S$ and $\styo_S$ to denote the values of $\YIn$ and $\YOut$ at moment $\sttmp$, and bound $\deltaY$ using the fact that 
    \begin{align}
    \label{ineq:z-bound-deltay}
    \deltaY \geq \sum_{S\subseteq\V}\styi_S - \sum_{S\subseteq\V} \styo_S.    
    \end{align}
    
     We consider active sets including one of $\{u,v,w,q\}$
    separately from those that do not. 
    In the run on $\stins$, $u$, $v$, and $w$ are contained in three different active sets until moment $\sttm$, and two different active sets from moment $\sttm$ to $\sttmp$. Therefore, we have 
    \begin{align}
        \sum_{S\subseteq\V}\styi_S 
        &\geq \sum_{\substack{S\subseteq\V\\S\cap\{u,v,w,q\}=\emptyset}}\styi_S + \sum_{\substack{S\subseteq\V\\\lvert S\cap\{u,v,w\}\rvert=1}} \styi_S + \sum_{\substack{S\subseteq\V\\\lvert S\cap\{u,v,w\}\rvert=2}} \styi_S \nonumber
        \\
        &\geq \sum_{\substack{S\subseteq\V\\S\cap\{u,v,w,q\}=\emptyset}}\styi_S + 3\sttm + 2(\sttmp - \sttm)
        \nonumber\\
&=\sum_{\substack{S\subseteq\V\\S\cap\{u,v,w,q\}=\emptyset}}\styi_S + 2\sttmp + \sttm. \label{ineq:zin-lower-bound}
    \end{align}
    On the other hand, in the run on $\stinsp$, $q$ connects actively with $u$, $v$, or $w$ at a moment before $\sttmp$. Then, $u$, $v$, and $w$ are guaranteed to be in the same active set as $q$ after the moments $\frac{\stx}{2}$, $\frac{\sty}{2}$ and $\frac{\stz}{2}$ respectively. Sets containing $q$ can grow for at most $\sttmp$ which is the entire duration we consider. Overall, we get
\begin{align}
\sum_{S \subseteq \V} \styo_S 
&\leq 
\sum_{\substack{S \subseteq \V \\ S \cap \{u,v,w,q\} = \emptyset}} \styo_S 
+ \sum_{\substack{S \subseteq \V \setminus \{q\} \\ u \in S}} \styo_S 
+ \sum_{\substack{S \subseteq \V \setminus \{q\} \\ v \in S}} \styo_S 
+ \sum_{\substack{S \subseteq \V \setminus \{q\} \\ w \in S}} \styo_S 
+ \sum_{\substack{S \subseteq \V \\ q \in S}} \styo_S \nonumber\\
&\leq 
\sum_{\substack{S \subseteq \V \\ S \cap \{u,v,w,q\} = \emptyset}} \styo_S 
+ \frac{\stx}{2} 
+ \frac{\sty}{2} 
+ \frac{\stz}{2} 
+ \sttmp. \label{ineq:zout-upper}
\end{align}        

    Furthermore, we can once again use Lemma \ref{lm:boost-not-involved-active-set} to show that if $S$ is an active set in the run on $\stinsp$ that does not include $q$, there must exist a subset of $S$ that is active at the same moment in the run on $\stins$.
    This implies that 
    \begin{align}
    \sum_{\substack{S\subseteq\V\\S\cap\{u,v,w,q\}=\emptyset}}\styo_S \leq \sum_{\substack{S\subseteq\V\\S\cap\{u,v,w,q\}=\emptyset}}\styi_S \label{ineq:z-no-q-bound}
    \end{align}
    as at any moment before $\sttmp$, the active sets not including $u$, $v$, $w$, or $q$
    in the run on $\stinsp$ can be mapped to disjoint active sets not including these vertices in the run on $\stins$. 

    Now, we can combine the inequalities to show that
    \begin{align*}
        \sum_{S\subseteq\V} \YIn_S - \sum_{S\subseteq\V} \YOut_S &\geq \sum_{S\subseteq\V} \styi_S - \sum_{S\subseteq\V} \styo_S \tag{Equation \ref{ineq:z-bound-deltay}}
        \\
        &\geq \left(\sum_{\substack{S\subseteq\V\\S\cap\{u,v,w,q\}=\emptyset}}\styi_S + 2\sttmp + \sttm\right) - \left(\sum_{\substack{S\subseteq\V\\S\cap\{u,v,w,q\}=\emptyset}}\styo_S + \frac{\stx+\sty+\stz}{2} + \sttmp\right)
        \tag{Equations \ref{ineq:zin-lower-bound} and \ref{ineq:zout-upper}}
        \\
        &\geq (2\sttmp + \sttm) - (\frac{\stx+\sty+\stz}{2} + \sttmp) \tag{Equation \ref{ineq:z-no-q-bound}}\\
        &= \sttmp + \sttm - \frac{\stx+\sty+\stz}{2}.
    \end{align*}    
    On the other hand, we have 
        \begin{align*}
            \sum_{S\subseteq\V} \YIn_S - \sum_{S\subseteq\V} \YOut_S &= \Win(\stins,\Boost) - \Loss(\stins,\Boost)
        \end{align*}
        by Lemma \ref{lm:new-ys-based-on-win-loss}. Since $\Boost$ cannot be a valuable boost, we have
        \begin{align*}
            \sum_{S\subseteq\V} \YIn_S - \sum_{S\subseteq\V} \YOut_S &= \Win(\stins,\Boost) - \Loss(\stins,\Boost) \leq \bta\Loss(\stins,\Boost).
        \end{align*}

        Combining the two inequalities results in
        \begin{align*}
            \sttm+\sttmp-\frac{\stx+\sty+\stz}{2} \leq \bta\Loss(\stins,\Boost).
        \end{align*}
        which is equivalent to 
        \begin{align*}
            \sttm+\sttmp \leq \frac{\stx+\sty+\stz}{2} + \bta\Loss(\stins,\Boost).
        \end{align*}
        Finally, $\Loss(\stins,\Boost)$ can be upper bounded by \[\frac{\min(\stx,\sty,\stz)}{2}\]
        using Lemma \ref{lm:loss-less-than-new-time} and the fact that \[\stbstm \leq \frac{\min(\stx,\sty,\stz)}{2}
        .\] 
        Therefore, we have
        \[
        \sttm+\sttmp \leq \frac{\stx+\sty+\stz}{2} + \bta \cdot \frac{\min(\stx,\sty,\stz)}{2}.
        \]
\end{proof}

\subsection{Generalizing the Claw Property to Larger Sets}
\label{sec:claw_property_extension}
Our main result in this section is Lemma~\ref{lm:steiner-tree}, 
which gives a bound on the cost of the tree connecting 
any set of vertices that actively reach one another
in a minimal moat growing algorithm. 
This bound can be stated for any assignment satisfying the following definition:

\begin{definition}\label{def:priority-based}
    Given an execution of a monotonic moat growing and $S \subseteq \V$ such that vertices in $S$ connect while active, we say an assignment $\gr$ is ``priority-based on $S$'' if for any active set $S'$ with $\gr_{S', v} > 0$ for a vertex $v \in S$, $v=\argmax_{u\in S\cap S'}\priority_u$.
\end{definition}
\begin{lemma} \label{lm:steiner-tree}
    Consider the final execution of \BoostedModGW{} during a call to \LocalSearch{}. 
    Let $S$ be a subset of vertices such that $S$ is connected in $\OPT$ and the vertices in $S$ are connected while active in this execution.
    Then, for any assignment $\rstree$ for this execution that is priority-based on $S$, we have
    \[\frac{6}{5+\bta} \left(\sum_{v\in S}\rstree_v - \rstree_{\max}(S)\right) \le \cc(\acttree)\]
    where $\acttree$ is the minimal subtree of $\OPT$ connecting $S$.
\end{lemma}

To simplify our arguments, we normalize the structure of the tree \( \acttree \) without loss of generality. 
We can transform \( \acttree \) into a rooted perfect binary tree using the following operations:
\begin{itemize}
    \item adding zero-cost edges and splitting high-degree vertices,
    \item duplicating vertices and connecting each duplicate to its original via a zero-cost edge, and
    \item bypassing degree-2 vertices by replacing their incident edges with a single edge.
\end{itemize}

These operations allow us to ensure that the leaves of the transformed tree form a set \( S' \) consisting of \( S \) along with possible duplicates of vertices in \( S \). 
Moreover, all leaves can be assumed to be attached to their parents via zero-cost edges. 
These modifications do not affect the behavior of \BoostedModGW{}, since duplicates are immediately connected to their originals during the moat growing process. Additionally, vertices in $S'$ would connect actively. 

Therefore, throughout the remainder of this section, we assume without loss of generality that \( \acttree \) is a rooted perfect binary tree with leaf set \( S \), and that each leaf is attached to its parent via a zero-cost edge.

Let \( \vi \) denote the set of internal vertices in \( \acttree \), and let \( \Root \) be its root. 
For any non-root vertex \( v \), we use \( \pare_v \) to represent the cost of the edge connecting \( v \) to its parent.
For each vertex \( v \) in \( \acttree \), let \( \rep_v \) be the closest leaf within its subtree, 
and let \( \repd_v \) be the distance from \( v \) to \( \rep_v \) in \( \acttree \).
Furthermore, for each internal vertex \( v \), we define \( \reptim_v \) as the first moment during the final execution of \BoostedModGW{} 
in the local search when a leaf from its left subtree and a leaf from its right subtree become connected.

As a first step toward proving Lemma~\ref{lm:steiner-tree}, we establish a basic bound on the cost of \( T_S \) in terms of the \( \repd_v \) values associated with its internal vertices.

\begin{lemma} \label{lm:st-dsum}
    We have
    $$\sum_{v\in\vi\setminus\RSET} \repd_v \leq \cc(\acttree).$$
\end{lemma}
\begin{proof}
    To prove the claim, we use induction to show a slightly stronger statement: 
    for each vertex \( v \), let \( T_v \) be the subtree of \( \acttree \) rooted at \( v \), and let \( \fard_v \) denote the length of the path from \( v \) to the farthest leaf from \( v\) in \( T_v \).  
    Then, 
    \[
    \sum_{u\in T_V} \repd_u \leq \cc(T_v) - \fard_v.\]
        
    In the base case, tree $T_v$ is a leaf and both sides of the inequality are zero, so the claim holds trivially.
    
Suppose the claim holds for the left and right children \( \ell \) and \( r \) of a vertex \( v \).  
Then, we have:
    \begin{align*}
    \sum_{u \in T_v} \repd_u &= \sum_{u \in T_\ell} \repd_u + \sum_{u \in T_r} \repd_u + \repd_v \\
    &\leq \cc(T_\ell) - \fard_\ell + \cc(T_r) - \fard_r + \repd_v 
    \tag{Induction Hypothesis} \\
    &= \cc(T_v) - \fard_\ell - \fard_r - \pare_\ell - \pare_r + \repd_v 
    \tag{\( \cc(T_v) = \cc(T_\ell) + \cc(T_r) + \pare_\ell + \pare_r \)}
\end{align*}

Now, note that the farthest leaf from \( v \) lies in either the left or right subtree.  
Without loss of generality, assume it lies in the left subtree, so that \( \fard_v = \fard_\ell + \pare_\ell \).  
Then
\begin{align*}
\sum_{u \in T_v} \repd_u 
&\leq \cc(T_v) - \fard_\ell - \pare_\ell - \fard_r - \pare_r + \repd_v \\
&= \cc(T_v) - \fard_v - \fard_r - \pare_r + \repd_v \\
&\leq \cc(T_v) - \fard_v
\end{align*}
where the last inequality uses the fact that \( \repd_v \leq \fard_r + \pare_r \), 
since \( \repd_v \) is the distance from \( v \) to its closest leaf in \( T_v \), and \( \fard_r + \pare_r \) is the distance from \( v \) to some leaf in \( T_r \subseteq T_v \).

Applying this bound at the root of \( \acttree \), and noting that \( \fard_{\Root} \geq 0 \), we obtain:
\[
\sum_{v \in \vi \setminus \RSET} \repd_v \leq \sum_{v \in T_{\Root}} \repd_v \leq \cc(\acttree).
\]
\end{proof}

To prove Lemma~\ref{lm:steiner-tree}, we first establish a bound that relates the sum of \( \reptim_v \) values over the internal vertices of \( \acttree \) to the cost of the tree \( \cc(\acttree) \) in Lemma~\ref{lm:timetree}.  
We then show that this same sum also serves as an upper bound on the sum of assignment values \( \rstree \) in Lemma~\ref{lm:st-rsum}.

\begin{lemma} \label{lm:timetree}
    We have
    $$3\sum_{v\in\vi} \reptim_v \leq \frac{5+\bta}{2} \cc(\acttree).$$
\end{lemma}
\begin{proof}
\begin{figure}[t]
    \centering
        \begin{tikzpicture}[scale=0.7]
\def\dem{Red!70}
\def\demf{Red!20}
\def\demi{Blue!70}
\def\demif{Blue!20}
\def\demii{Green!70}
\def\demiif{Green!20}
\def\ter{White}
\def\col{Black!70}
\def\treeone{Purple}
\def\treetwo{RubineRed!70!Black}
\def\noder{0.1cm}

\def\tercl#1{%
  \ifcase#1
    White\or
    Goldenrod!20\or
    olive!20\or
    Green!20\or
    Sepia!20\or
    Plum!20\or
  \fi
}

\def\edge{1}
\def\height#1{%
  \ifcase#1
    1\or
    1\or
    0.5\or
    0.25\or
    0.25\or
    0.25\or
    0.25\or
  \fi
}
\def\total{3.5}
\def\width#1{%
  \ifcase#1
    1.118\or
    1.118\or
    0.577\or
    0.288\or
    3\or
    3.25\or
    3.5\or
  \fi
}
\pgfmathsetseed{34}
\def\num{%
  \pgfmathrandominteger{\temp}{1}{2}
  \pgfmathsetmacro{\rand}{\temp*2-3}
}

\def\edge{2}
\def\r{1.044}

\coordinate (P) at (0, 0);

\pgfmathsetmacro{\dx}{sqrt(4*\height{0}*\height{0}/3)/2}
\coordinate (Q) at ($(P) + (-\dx*\edge, -
\height{0}*\edge)$);
\coordinate (W) at ($(P) + (\dx*\edge, -
\height{0}*\edge)$);

\pgfmathsetmacro{\dx}{sqrt(4*\height{1}*\height{1}/3)/2}
\coordinate (U) at ($(Q) + (-\dx*\edge, -
\height{1}*\edge)$);
\coordinate (V) at ($(Q) + (\dx*\edge, -
\height{1}*\edge)$);
\coordinate (W1) at ($(W) + (\dx*\edge, -
\height{1}*\edge)$);

\coordinate (U1) at (U);
\coordinate (V1) at (V);

\foreach \i in {2,...,6} {
    \pgfmathsetmacro{\dx}{sqrt(4*\height{\i}*\height{\i}/3)/4}
    \pgfmathsetmacro{\prev}{\i-1}
    \foreach \j in {U, V, W} { 
        \num
        \coordinate (\j\i) at ($(\j\prev) + (\rand*\dx*\edge,-\height{\i}*\edge)$);
    }
}

\def\mar{0.2}
\coordinate (UVMID) at ($(U)!0.5!(V)$);
\coordinate (UVMID) at ($(UVMID) - \edge*(0, \total-\height{0}-\height{1}+\mar)$);

\coordinate (UMID) at ($(U) - \edge*(0, \total-\height{0}-\height{1}+\mar)$);
\coordinate (URIGHT) at ($(UVMID) - (0.6, 0)$);
\coordinate (ULEFT) at ($(URIGHT)!2!(UMID) - (1.2, 0)$);

\coordinate (VMID) at ($(V) - \edge*(0, \total-\height{0}-\height{1}+\mar)$);
\coordinate (VLEFT) at ($(UVMID) - (0.5, 0)$);
\coordinate (VRIGHT) at ($(VLEFT)!2!(VMID) - (0.1, 0)$);

\coordinate (WMID) at ($(W1) - \edge*(0, \total-\height{0}-\height{1}+\mar)$);
\coordinate (WLEFT) at ($(VRIGHT) + (0.1, 0)$);
\coordinate (WRIGHT) at ($(WLEFT)!2!(WMID) + (3.5, 0)$);

\def\UC{Red}
\def\WC{Blue}
\def\VC{Green}
\def\normal{1.5pt}
\def\thin{1.2pt}
\def\verythin{0.1pt}

\foreach \j/\c in {W/\WC, V/\VC, U/\UC} {
       \draw[\c, line width=\thin] (\j1) \foreach \i in {2,...,6} { -- (\j\i)};
}

\draw[\col] (P) -- ++(+0.5, +0.5) node[above right] {$\iddots$};
\draw[line width=\verythin] (URIGHT) -- (U) -- (ULEFT) -- cycle;
\draw[line width=\verythin] (VRIGHT) -- (V) -- (VLEFT) -- cycle;
\draw[line width=\verythin] (WRIGHT) -- (W) -- (WLEFT) -- cycle;
\draw[\WC, line width=\normal] (W) -- (P) -- (Q);
\draw[\UC, line width=\normal] (U) -- (Q);
\draw[\VC, line width=\normal] (Q) -- (V);
\draw[\WC, line width=\thin] (W) -- (W1) -- (W2);

\foreach \i in {P, Q, W, U, V} {
    \draw[\col, fill= \ter] (\i) circle(\noder);
}
\foreach \i in {W6, U6, V6} {
    \draw[\col, fill= \col] (\i) circle(\noder);
}
\node[\col, above left=-2pt] at (P) {$p$};
\node[\col, above left=-2pt] at (Q) {$q$};
\node[\col, above left=-2pt] at (U) {$u$};

\node[\col, above right=-2pt] at (V) {$v$};
\node[\col, above right=-2pt] at (W) {$w$};

\node[\col, left=-1pt] at (W6) {$\rep_{w}$};
\node[\col, left=-1pt] at (V6) {$\rep_{v}$};
\node[\col, left=-1pt] at (U6) {$\rep_{u}$};

    \end{tikzpicture}
    \caption{
An illustration of a section of the tree \( \acttree \) connecting its leaves in \( S \), focused on the internal vertex \( q \). 
Leaves \( \rep_u \) and \( \rep_v \) cannot connect before time \( \reptim_q \), while \( u \) and \( w \) cannot connect before \( \reptim_p \).  
Applying the claw property to \( q \), \( \rep_u \), \( \rep_v \), and \( \rep_w \) yields a bound on \( \reptim_p + \reptim_q \) in terms of the costs of the paths from the leaves to \( q \).  
    }
    \label{fig:claw-binary}
\end{figure}
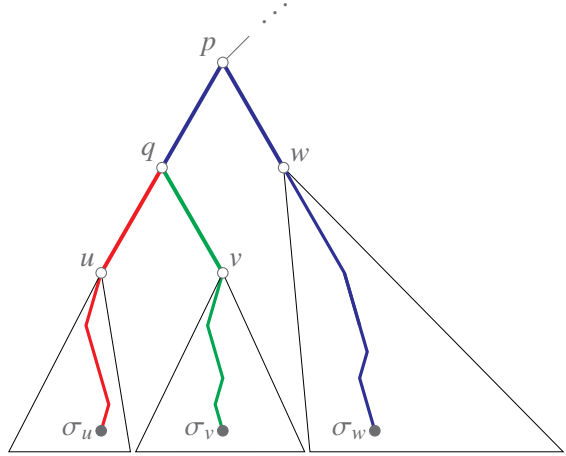 For each non-root internal vertex \( q \), let \( u \) and \( v \) be its children, \( w \) its sibling, and \( p \) its parent, as illustrated in Figure~\ref{fig:claw-binary}.  
Considering the leaves \( \rep_u \), \( \rep_v \), and \( \rep_w \), and applying Lemma~\ref{lm:xyz}, we obtain the bound:
\begin{align*}
    \sttm + \sttmp 
    &\leq \frac{d(q, \rep_u) + d(q, \rep_v) + d(q, \rep_w)}{2} 
    + \bta \cdot \frac{\min\big(d(q, \rep_u), d(q, \rep_v), d(q, \rep_w)\big)}{2},
\end{align*}
where \( \sttm \) denotes the first moment when any pair in \( \{\rep_u, \rep_v, \rep_w\} \) becomes connected and
\( \sttmp \) denotes the first moment when all three are connected.

Now, before the moment \( \min(\reptim_q, \reptim_p) \), none of the leaves in the set \( \{\rep_u, \rep_v, \rep_w\} \) can be connected, so \( \sttm \geq \min(\reptim_q, \reptim_p) \).  

Additionally, since all three of \( \rep_u \), \( \rep_v \), and \( \rep_w \) are connected at moment \( \sttmp \), it follows that \( \reptim_q \leq \sttmp \) and \( \reptim_p \leq \sttmp \), which means \( \sttmp \geq \max(\reptim_q, \reptim_p) \).  
Combining these bounds shows that
\[
\sttm + \sttmp \geq \min(\reptim_q, \reptim_p) + \max(\reptim_q, \reptim_p) = \reptim_q + \reptim_p.
\]
Substituting this into the earlier bound yields
\begin{align*}
    \reptim_q + \reptim_p 
    &\leq \frac{d(q, \rep_u) + d(q, \rep_v) + d(q, \rep_w)}{2} 
    + \bta \cdot \frac{\min\big(d(q, \rep_u), d(q, \rep_v), d(q, \rep_w)\big)}{2}.
\end{align*}
Next, we use the following bounds on the distances from \( q \) to leaves \( \{\rep_u, \rep_v, \rep_w\} \)
\[
d(q, \rep_v) \leq \repd_v + \pare_v, \quad
d(q, \rep_u) \leq \repd_u + \pare_u, \quad
d(q, \rep_w) \leq \repd_w + \pare_q + \pare_w,
\]
to show that
\begin{align*}
    \reptim_q + \reptim_p 
    &\leq \frac{(\repd_u + \pare_u) + (\repd_v + \pare_v) + (\repd_w + \pare_q + \pare_w)}{2} 
    + \bta \cdot \frac{\min\big( \repd_u + \pare_u,\, \repd_v + \pare_v,\, \repd_w + \pare_q + \pare_w \big)}{2}.
\end{align*}
Lastly, since \( \min(a_1, a_2, a_3) \leq \frac{a_1 + a_2}{2} \) for any values \( a_1, a_2, a_3 \), we obtain the following bound:
\begin{align}\label{ineq:tpq}
    \reptim_q + \reptim_p 
    &\leq \frac{\repd_u + \pare_u + \repd_v + \pare_v + \repd_w + \pare_q + \pare_w}{2} 
    + \bta \cdot \frac{\repd_u + \pare_u + \repd_v + \pare_v}{4}.
\end{align}

Let \( \ell \) and \( r \) be the children of the root and 
consider the leaves \(\rep_\ell\) and \(\rep_r\). 
These leaves belong to separate active sets and grow independently until at least moment \( \reptim_\Root \),  
during which they both contribute to coloring the path between them.  
Therefore, 
\[
\dis(\rep_\ell, \rep_r) \geq 2\reptim_\Root.
\]
On the other hand, their distance is at most the length of the path connecting them in the tree $\acttree$:
\[
\dis(\rep_\ell, \rep_r) \leq \repd_\ell + \pare_\ell + \repd_r + \pare_r.
\]
Consequently, we obtain
\begin{align*}
    \reptim_\Root 
    &\leq \frac{\repd_\ell + \pare_\ell + \repd_r + \pare_r}{2} \\
    &\leq \frac{\repd_\ell + \pare_\ell + \repd_r + \pare_r}{2} 
    + \bta \cdot \frac{\repd_\ell + \pare_\ell + \repd_r + \pare_r}{4},
\end{align*}
where the additional term with coefficient \( \bta \) is included to align with the structure of inequality~\eqref{ineq:tpq}.

Additionally, for each leaf \( v \) with parent \( p \), we have \( \reptim_p \leq 0 \),  
since leaves in \( \acttree \) are connected to their parent by zero-cost edges and are thus immediately connected to their siblings.

Now, summing inequality~\eqref{ineq:tpq} over all non-root internal vertices \( q \), and including the bounds for the root and leaf vertices described above, yields the following overall bound:
\begin{align*}
    3 \sum_{v \in \vi} \reptim_v
    &\leq \sum_{v \in \V(\acttree) \setminus \RSET} \frac{2 + \bta}{4} (\repd_v + \pare_v)
    + \sum_{v \in \vi \setminus \RSET} \frac{\repd_v + 2\pare_v}{2}.
\end{align*}
Here, each term \( \reptim_v \) appears three times on the left-hand side of the summed inequalities: once in the inequality associated with vertex \( v \) itself, and once in the inequality for each of its two children.

On the right-hand side, for
each non-root vertex \(v\), the term \(\repd_v+\pare_v\) is summed with a coefficient of \(\frac{1}{2}+\frac{\bta}{4}\) in the inequality corresponding to its parent.
Additionally, for each non-root internal vertex \(v\), \(\frac{\repd_v + \pare_v}{2}\) is counted in the inequality for its sibling and \(\frac{\pare_v}{2}\) appears in its own inequality.  

Now, since \( \repd_v = \pare_v = 0 \) for all \( v \notin \vi \), we can simplify the previous bound to:
\begin{align*}
    3 \sum_{v \in \vi} \reptim_v 
    &\leq \frac{4 + \bta}{4} \sum_{v \in \vi \setminus \RSET} \repd_v  
    + \frac{6 + \bta}{4} \sum_{v \in \vi \setminus \RSET} \pare_v.
\end{align*}
Note that \( \sum_{v \in \vi \setminus \RSET} \pare_v \leq \cc(\acttree) \),  
since each edge in \( \acttree \) appears in this sum at most once.  
In addition, Lemma~\ref{lm:st-dsum} implies that \( \sum_{v \in \vi \setminus \RSET} \repd_v \leq \cc(\acttree) \).

Combining these bounds, we conclude that

\begin{align*}
    3 \sum_{v \in \vi} \reptim_v 
    &\leq \frac{4 + \bta}{4} \sum_{v \in \vi \setminus \RSET} \repd_v  
    + \frac{6 + \bta}{4} \sum_{v \in \vi \setminus \RSET} \pare_v \\
    &\leq \frac{10 + 2\bta}{4} \cc(\acttree) \\
    &= \frac{5 + \bta}{2} \cc(\acttree).
\end{align*}
\end{proof}

Next, we relate the \(\reptim_v\) values to the assignment \(\rstree\) to complete the proof.
\begin{lemma} \label{lm:st-rsum}
    For assignment $\rstree$ that is priority-based on $S$, we have 
    $$\sum_{v\in S}\rstree_v - \rstree_{\max}(S) \leq \sum_{v\in \vi} \reptim_v.$$
\end{lemma}
\begin{proof}
Let \( \stlsttime \) denote the final moment during which the moat-growing procedure runs.
 It is clear that \( \rstree_v \leq \stlsttime \) and \( \reptim_v \leq \stlsttime \) for all \( v \).  
We now express both sides of the inequality in equivalent integral forms.
For the right-hand side, we have
\begin{align*}
    \sum_{v \in \vi} \reptim_v 
    &= \sum_{v \in \vi} \int_0^{\reptim_v} 1 \, \diff\currenttime \\
    &= \int_0^{\stlsttime} \left\lvert \{ v \in \vi \mid \reptim_v \geq \currenttime \} \right\rvert \, \diff\currenttime.
\end{align*}
On the other hand, for the left-hand side, we can write
\begin{align*}
    \sum_{v \in S} \rstree_v - \rstree_{\max}(S)
    &= \sum_{v \in S} \int_0^{\rstree_v} 1 \, \diff\currenttime - \int_0^{\rstree_{\max}(S)} 1 \, \diff\currenttime\\
    &= \int_0^{\stlsttime} \left\lvert \{ v \in S \mid \rstree_v \geq \currenttime \} \right\rvert \, \diff\currenttime 
    - \int_0^{\rstree_{\max}(S)} 1 \, \diff\currenttime \\
    &= \int_0^{\stlsttime} \max\left( 0, \left\lvert \{ v \in S \mid \rstree_v \geq \currenttime \} \right\rvert - 1 \right) \, \diff\currenttime.
\end{align*}
In the final step, we use the fact that 
\(\left\lvert \{ v \in S \mid \rstree_v \geq \currenttime \} \right\rvert=0\) for any \(\currenttime > \rstree_{\max}(S)\).

It now suffices to show that, for any moment \( \currenttime \leq \stlsttime \), the following inequality holds:
\[
\max\left(0, \left\lvert \{ v \in S \mid \rstree_v \geq \currenttime \} \right\rvert - 1 \right) 
\leq \left\lvert \{ v \in \vi \mid \reptim_v \geq \currenttime \} \right\rvert.
\]
Furthermore, since the right-hand side is always non-negative, it suffices to prove that
\[
\left\lvert \{ v \in S \mid \rstree_v \geq \currenttime \} \right\rvert - 1 
\leq \left\lvert \{ v \in \vi \mid \reptim_v \geq \currenttime \} \right\rvert.
\]
Fix a moment \( \currenttime \), and let \( \stacts \) denote the family of active sets at moment \( \currenttime \) that intersect \( S \) non-trivially, with \( k = \lvert \stacts \rvert \).
Now consider any vertex \(v \in S\) with \( \rstree_v \geq \currenttime \). Then, at some moment no earlier than \( \currenttime \), there must exist an active set \( S' \) that assigns growth to \( v \); otherwise, we would have \(\rstree_v < \currenttime\), contradicting the assumption. Since \( v \) is still assigned growth after \( \currenttime \), it must be active at moment \( \currenttime \) and therefore belongs to an active set \(S^*\in \stacts\). Since active sets only grow larger during the moat growing process, we must have \( S^* \subseteq S'\).

Moreover, since the assignment \(\rstree\) is priority-based on \(S\), \(v\) must have the highest priority in \(S\cap S'\) by Definition \ref{def:priority-based} and hence also in $S\cap S^*$. 
This implies that any vertex \(v\) with \(\rstree_v \geq \currenttime\) is the maximum priority vertex in the intersection of some active set \(S^*\in \stacts\) with \(S\). Since there are exactly \(k\) such active sets, and the maximum priority vertex is unique, it follows that
\[
\big\lvert\{v \in S \mid \rstree_v \geq \currenttime \}\big\rvert \leq k.
\]

On the other hand, for each value \( \reptim_v \), consider the pair of leaves \( u, w \in S \) that define it—namely, a pair of leaves from the left and right subtrees of \( v \) that become connected at moment \( \reptim_v \).  
We interpret each such pair as an edge between \( u \) and \( w \), forming a graph on the vertex set \( S \) with one edge per internal vertex \( v \in \vi \).  
This graph is connected, which can be shown by induction on the structure of the subtree rooted at each internal vertex of \( \acttree \).

Now consider the family of active sets \( \stacts \) at moment \( \currenttime \), and merge all vertices in each set \( S' \in \stacts \) into a single vertex.  
Since merging vertices preserves connectivity, the resulting graph has \( k = |\stacts| \) vertices and remains connected.  
Thus, it must contain at least \( k - 1 \) non-self-loop edges.  

Each such edge represents a pair \( u, w\) 
that lie in distinct active sets at time \( \currenttime \).  
Since these leaves are not yet connected at moment \( \currenttime \), the internal vertex \( v \) corresponding to this pair satisfies \( \reptim_v \geq \currenttime \).
Therefore, there are at least \( k - 1 \) such vertices \( v \in \vi \), implying:
\[
\left\lvert \{ v \in \vi \mid \reptim_v \geq \currenttime \} \right\rvert \geq k - 1.
\]
Combining this with the earlier bound on \( \rstree \), we obtain:
\begin{align*}
\left\lvert \{ v \in S \mid \rstree_v \geq \currenttime \} \right\rvert - 1 
&\leq k - 1 \\
&\leq \left\lvert \{ v \in \vi \mid \reptim_v \geq \currenttime \} \right\rvert,
\end{align*}
which completes the proof.
\end{proof}
 
Combining Lemmas \ref{lm:timetree} and \ref{lm:st-rsum} completes the proof of Lemma \ref{lm:steiner-tree}:
\begin{align*}
    \sum_{v\in S}\rstree_v - \rstree_{\max}(S) &\leq \sum_{v\in \vi} \reptim_v\\
    &\leq\frac{5+\bta}{6} \cc(\acttree).
\end{align*}

\subsection{Approximation for Steiner Tree: Proof of Theorem~\ref{thm:main_steiner_tree}}
\label{apx:steiner_tree}
We can now prove that Algorithm \ref{alg:steiner-tree} achieves a better than $2$ approximation for the Steiner Tree problem. 

\begin{proof}[Proof of Theorem~\ref{thm:main_steiner_tree}]
    We show that solution $\solone$ obtained in Algorithm \ref{alg:steiner-tree} achieves approximation factor $\stapx$. Since Algorithm \ref{alg:steiner-tree} is identical to the initial part of Algorithm \ref{alg:main}, we reuse the terminology and lemmas of Sections \ref{sec:mod-goemans} and \ref{sec:local_search} and Legacy Execution and Boosted Execution. 
    
    By Lemma \ref{lm:monotonic-forest-twice-growth}, we can bound the cost of solution $\solone$ as
    \begin{align*}
        \cc(\solone) \leq 2\left(\sum_{S\subseteq\V} \yy_S - \sum_{S\subseteq\V;v\in S} \yy_S\right)
    \end{align*}
    for any vertex $v$. Now, consider the vertex $v^*$ that maximizes $\rstar_v$. We can plug $v^*$ in to get
    \begin{align*}
        \sum_{S\subseteq\V} \yy_S - \sum_{S\subseteq\V;v^*\in S} \yy_S 
        &\leq \sum_{v\in \V} \rstar_v + \lossone  -\sum_{S\subseteq\V;v^*\in S} \yy_S \tag{ Lemma \ref{lm:yy-using-rstar-and-loss}}\\
        &\leq \sum_{v\in \V} \rstar_v + \lossone  - \rstar_{v^*} \tag{$\rstar_{v^*} \leq \sum_{S\subseteq \V;v^*\in S}\yy_S$}\\
        &= \sum_{v\in \V} \rstar_v + \lossone  - \max_{v\in\V}\rstar_{v} \tag{$\rstar_{v^*}=\max_{v\in\V}\rstar_{v}$}.
    \end{align*}

    Now, let $\terminals$ be the set of terminals for the Steiner Tree instance. Since vertices in $\terminals$ must be connected by the algorithm, they are not deactivated before reaching each other and are actively connected in Legacy Execution.
    By Lemma \ref{lm:large_fingerprint_superactive}, they must also actively connect in Boosted Execution, since it has a larger fingerprint.
    Additionally, $\rstar$ is priority-based on $\terminals$ and satisfies the conditions for Lemma \ref{lm:steiner-tree}. Therefore, we can apply this lemma using $S=\terminals$ and $\rstree = \rstar$ to get
    \begin{align*}
        \sum_{S\subseteq\V} \yy_S - \sum_{S\subseteq\V;v^*\in S} \yy_S 
        &\leq \sum_{v\in \V} \rstar_v + \lossone  - \max_{v\in\V}\rstar_{v}\\
        &\leq \frac{5+\bta}{6}\cc(\OPT) + \lossone 
    \end{align*}
    since the cost of the minimal tree $T_\terminals$ connecting $\terminals$ is $\cc(\OPT)$. Furthermore, we can use Lemma \ref{lm:lcl-loss} to achieve the following bound
    \begin{align*}
        \cc(\solone)
        &\leq 2\left(\frac{5+\bta}{6}\cc(\OPT) + \lossone \right)\\
        &\leq 2\left(\frac{5+\bta}{6}\cc(\OPT) + \frac{\alf}{\bta} \cc(\OPT)\right)\\
        &\leq2\left(\frac{5+\bta}{6} + \frac{\alf}{\bta}\right) \cc(\OPT)
    \end{align*}
    when $\cc(\solone)\geq (2-2\alf)\cc(\OPT)$. Equivalently, we can bound the cost of $\solone$ as follows
    \begin{align*}
        \cc(\solone)\leq \max\big(2-2\alf,2(\frac{5+\bta}{6} + \frac{\alf}{\bta})\big)\cc(\OPT).
    \end{align*}
    Finally, minimizing the coefficient of $\cc(\OPT)$ gives us an approximation ratio of $\stapx\approx\strgh$ when $\alf=\stalfval$ and $\bta=\stbtaval$.
\end{proof}

\section{Extended Moat Growing}
\label{sec:extension}

We show how our local search can yield a better-than-2 approximation for the Steiner Tree problem.
However, this approach does not directly extend to the Steiner Forest problem.
More precisely, the bound derived from our local search applies only to vertices that are actively connected, a property that does not necessarily hold for each connected component of the optimal solution in the Steiner Forest case.

To address this, we introduce a new monotonic moat growing algorithm, which extends an existing moat growing algorithm.
In our case, we use Boosted Execution as the base algorithm. In the extended version, a small fraction of each active set’s growth in the original algorithm is assigned as potential to one of its vertices. During the extended process, when an active set is about to become inactive, it can consume this potential to remain active for a longer period.
This additional step tends to result in the majority of vertices within each connected component of the optimal solution becoming actively connected, which allows us to apply the same bound as in the Steiner Tree case to each component.
Therefore, after performing the extension, we run our local search again to apply that bound.

Note that in some components, a significant portion of vertices may still fail to become actively connected. In these cases, we show that most of the growth introduced by the extension step ends up coloring more than one edge of the optimal solution, which in turn leads to a new bound for those components.

We now describe in detail how this extended moat growing procedure works.

\subsection{Algorithm}
In \ExtendedGW{}, described in Algorithm~\ref{alg:extend}, each vertex $v$ is provided with a fingerprint value $\tin_v$ and a potential value $\pot_v$ as input. Initially, potentials are assigned to singleton connected components. Whenever two components merge, their potentials are combined and assigned to the new component.

In this moat growing algorithm, active sets grow until they reach the maximum $\tin$ value among their vertices. After that, they continue growing for an extended period by consuming their potential until it is fully exhausted. The output fingerprint of this method is the time at which each vertex becomes deactivated, which is larger than fingerprint $\tin$.

The algorithm uses the notation $\setpot_S$ to store the remaining potential of each active set $S$. In Line~\ref{line:decrease-setpot}, this value is reduced as active sets consume their potential, and in Line~\ref{line:merge_setpot}, the potentials of merging active sets are combined. Finally, in Line~\ref{line:setpot_zero}, we use $\setpot_S$ to identify active sets that have both reached their maximum $\tin$ and exhausted their potential, allowing us to deactivate them.

\begin{algorithm}[H]
  \caption{Extended Moat Growing}
  \label{alg:extend}
  \hspace*{\algorithmicindent} \textbf{Input:} An undirected graph $\G=(\V, \E, \cc)$ with edge costs $\cc: \E \rightarrow \mathbb{R}_{\ge 0}$, a function $\tin: V \rightarrow \mathbb{R}_{\ge 0}$ indicating the input fingerprint, and a potential $\pot: V \rightarrow \mathbb{R}_{\ge 0}$ for each vertex.\\
  \hspace*{\algorithmicindent} \textbf{Output:} A function $\tout: V \rightarrow \mathbb{R}_{\ge 0}$ indicating the latest moment each vertex was actively growing during the algorithm.
  \begin{algorithmic}[1]
    \Procedure{\EGWr}{$\G,\tin,\pot$}
      \label{func:egw}
      \State $\tau \gets 0$
      \State $\currentsets \gets \{\{v\} \mid v \in \V\}$
      \State $\deactivesets \gets \{v \mid v\in \V, \tin_v=0, \pot_v=0\}$
      \State $\activesets \gets \{\{v\} \mid v \in \V, v \notin \deactivesets\}$ 
      \State $\tout_v \gets 0 \text{ for all } v\in\V$
      \State $\setpot_{\{v\}} \gets \pot_v$ for all $v\in V$
      \State Implicitly set $\ys \gets 0$ for all $S \subseteq \V$
      \label{line:init_setpot}
      \While{$\activesets \neq \emptyset$}
      \Comment{While there exists an active set}
        \State $\Delta_e \gets \min_{e = uv \in E} \frac{c_e - \sum_{S \ni e} y_S}{|\{S_u, S_v\} \cap \activesets|}$, where $u\in S_u\in \currentsets$, $v\in S_v\in \currentsets$, and $S_u \ne S_v$
        \State $\Delta_t \gets \min_{v\in \V, \tin_v > \currenttime} (\tin_v - \currenttime)$
        \State $\Delta_p \gets \min_{S \in \activesets, \setpot_S>0} \setpot_S$
        \State $\Delta \gets \min(\Deltae, \Delta_t, \Delta_p)$
        \For{$S \in \activesets$}
            \State $\ys \gets \ys + \Delta$
            \If {$\max_{v\in S} \tin_v \le \currenttime$}
                \State $\setpot_{S} \gets \setpot_{S} - \Delta$ 
            \label{line:decrease-setpot}
            \EndIf
        \EndFor
        \State $\currenttime \gets \currenttime + \Delta$
        \For{$e\in E$} 
          \State Let $S_v, S_u \in \currentsets$ be the sets that contains each endpoint of $e$
          \If{$\sum_{S: e \in \deltaS} \ys = \ce$ \textbf{and} $S_v \neq S_u$} 
          \Comment{Edge $(v, u)$ become fully colored}
            \State $\currentsets \gets (\currentsets \setminus \{S_u, S_v\}) \cup \{S_u \cup S_v\}$
            \State $\activesets \gets (\activesets \setminus \{S_u, S_v\}) \cup \{S_u \cup S_v\}$
            \State $\setpot_{S_u\cup S_v} \gets \setpot_{S_u}+\setpot_{S_v}$
            \label{line:merge_setpot}
          \EndIf
        \EndFor
        \For{$S \in \activesets$} 
          \If{$\setpot_S = 0$ \textbf{and} $\max_{v\in S}\tin_v\leq \currenttime$} 
          \Comment{$S$ become inactive}
          \label{line:setpot_zero}
            \State $\activesets \gets \activesets \setminus \{S\}$ 
            \For{$v \in S \setminus \deactivesets$}
               \State $\tout_v \gets \tau$
                \State $\deactivesets \gets \deactivesets \cup \{v\}$
            \EndFor
          \EndIf
        \EndFor
      \EndWhile
      \State \Return $\tout$
    \EndProcedure
  \end{algorithmic}
\end{algorithm}

In Algorithm~\ref{alg:main}, we do not directly call \ExtendedGW{}.  
Instead, we call \Extend{} in Line~\ref{exe:egw}, which internally uses \ExtendedGW{}.  
Within this call, we make use of the output of Boosted Execution. Based on the growth of active sets during their base phase in Boosted Execution—denoted by \(\yyb\) in Definition~\ref{def:yy-base-boost}—we construct an exclusive assignment \(\rex\) such that for any set \(S \subseteq \V\), it holds that  
\[
\sum_{v \in S} \rex_{S, v} = \yyb_S.
\]  
Note that, according to Definition~\ref{def:assignment}, we assign only to active vertices in \(S\).  
We then run \ExtendedGW{} with fingerprint \(\ttt\) and potential values \(\pot_v = \eps \rex_v\) for each vertex \(v\), producing a new fingerprint \(\tp\).  
We will later show that the output \(\tp\) remains unchanged as long as the assignment \(\rex\) used to define \(\pot\) satisfies the condition above. This allows us to substitute another assignment in our analysis without affecting the resulting fingerprint.
The pseudocode of \Extend{} is provided in Algorithm~\ref{alg:extend-wrapper}.  

Finally, we apply the \LocalSearch{} algorithm to the instance with fingerprint \(\tp\), resulting in the solution \(\solext\) (Line~\ref{exe:Xlocalsearch} in Algorithm~\ref{alg:main}).

\begin{algorithm}[H]
  \caption{Extend}
  \label{alg:extend-wrapper}
  \hspace*{\algorithmicindent} \textbf{Input:} 
  Graph $\G = (\V, \E, \cc)$ with edge costs $\cc: \E \to \mathbb{R}_{\geq 0}$, function $\tplus: V \to \mathbb{R}_{\geq 0}$ specifying the fingerprint of Legacy Execution, function $\ttt: V \to \mathbb{R}_{\geq 0}$ specifying the fingerprint of Boosted Execution,
  function $\yyb: 2^V \to \mathbb{R}_{\geq 0}$ specifying the growth in the base phase of Boosted Execution for each subset of vertices, and parameter $\eps > 0$.\\
  \hspace*{\algorithmicindent} \textbf{Output:} 
  $\tp$ is the fingerprint resulting from the extension.
  \begin{algorithmic}[1]
    \Procedure{\Extendr}{$\G, \tplus, \ttt, \yyb, \eps$}
        \label{func:extend}
        \State $\rex_v \gets 0 \text{ for all } v\in\V$
        \For{$S \subseteq V$ \textbf{ such that } $\yyb_S > 0$}
            \State $v \gets \arg\max_{v \in S} \tplus_v$
            \State $\rex_v \gets \rex_v + \yyb_S$
        \EndFor
        \State $\tp \gets \ExtendedGW(\G, \ttt, \eps\rex)$
        \label{line:egw}
        \State \Return $\tp$
    \EndProcedure
  \end{algorithmic}
\end{algorithm}

We observe that Algorithm~\ref{alg:extend} runs in polynomial time for any given input, since in each iteration of the loop, one of the following occurs: two active sets are merged, an active set becomes inactive, or $\currenttime$ passes the maximum $\tin$ value among the vertices of an active set. Each of these three events can occur at most a linear number of times in $\lvert \V \rvert$, ensuring the algorithm's overall polynomial runtime.
Consequently, Algorithm~\ref{alg:extend-wrapper} runs in polynomial time.

\begin{corollary}\label{cor:extend-polynomial}
    The \Extend{} procedure runs in polynomial time.
\end{corollary}

\subsection{Properties and Definitions}
Our first step is to show that the behavior of \ExtendedGW{} is the same if potential $\pot_v$ is defined using any exclusive assignment where $\sum_{v \in S} \rex_{S,v} = \yyb_S$ for all sets $S \subseteq \V$. 
Then, we will assume that $\rex_v=\rstar_v$ is used in the remainder of this section, as this assignment satisfies the required condition due to Corollary~\ref{cor:sum-rstar-yyb}. 

In proving this, we use the following corollary, which follows from Lemma \ref{lm:large_fingerprint_refinement} and the fact that the fingerprint of any run of \ExtendedGW{} is larger than its input fingerprint.
\begin{corollary}\label{cor:ext-any-pot-refine}
    Consider an execution of \ExtendedGW{} given fingerprint $\ttt_v$ and any potential $\pot_v \geq 0$. Then, at any moment $\currenttime$, the active sets of Boosted Execution form a refinement of the active sets in this execution at the same moment. Additionally, the connected components at the same moment in Boosted Execution are a refinement of the connected components in this execution.
\end{corollary}

\begin{lemma} \label{lm:ext-rex-no-matter}
    The behavior of \ExtendedGW{} given fingerprint \(\tp\) is identical for all choices of exclusive assignment $\rex_v$ used to define potential values $\pot_v=\eps \rex_v$, as long as the assignment satisfies \[\sum_{v \in S} \rex_{S, v} = \yyb_S\] for any set $S\subseteq \V$. In particular, the output $\tout$ will be the same.
\end{lemma}
\begin{proof}
Assume otherwise that the output differs for two exclusive assignments $\rxone$ and $\rxtwo$ satisfying the condition. Then, consider the first moment $\currenttime$ where active sets in the two runs of \ExtendedGW{} are not the same. 
There are two possible cases:
\begin{enumerate}
    \item An edge $e$ becomes fully colored in one instance, merging two active sets while it does not become fully colored in the other. This is not possible, as active sets have been the same in the two runs until $\currenttime$, so the sum $\sum_{S;e\in\deltaS} \ys$ is the same between the two runs and edge $e$ can become fully colored in either both runs or neither.  
    \item 
    An active set $S$ becomes inactive in one run, while it remains active in the other. In this case, since $\max_{v\in S}\tin_v$ is the same between the two runs, we must have $\setpot_S=0$ in one run while $\setpot_S\neq0$ in the other, where $\setpot_S$ shows the remaining potential of set $S$. Now, $\setpot_S$ for any set $S$ is $\sum_{v\in S}\pot_v$ minus the potential used by subsets of $S$. Since active sets at any moment before $\currenttime$ are the same in the two runs and use their potential at the same moments, the potential used by subsets of $S$ is the same between runs. Therefore, it suffices to show that $\sum_{v\in S}\pot_v$ is the same in the different runs. Furthermore, since $\pot_v = \eps \rex_v$, we need to show that $\sum_{v\in S} \rxone_v = \sum_{v\in S} \rxtwo_v.$

    Now, consider any vertex $v \in S$. 
    Since $v$ can become inactive at moment $\currenttime$ in the \ExtendedGW{} runs, it cannot be in an active set at the same moment in Boosted Execution by Corollary \ref{cor:ext-any-pot-refine} and becomes inactive at $\currenttime$ at the latest. 
    Therefore, any active set $S'$ assigning growth to $v$ in Boosted Execution must include $v$ at a moment $\currenttime' \leq \currenttime$.
    Corollary \ref{cor:ext-any-pot-refine} implies that $S'\subseteq S^*$ where $S^*$ is the active set including $v$ at moment $\currenttime'$ in the \ExtendedGW{} run.
    Furthermore, $S^*$ must be a subset of $S$ since $\currenttime'\leq\currenttime$. Therefore, for any active set $S'$ that assigns growth to a vertex $v$ in $S$, we have $S' \subseteq S$.
    Now, we have
    \begin{align*}
        \sum_{v\in S} \rxone_v &= \sum_{v\in S} \sum_{S'\subseteq\V} \rxone_{S',v}\\
        &= \sum_{v\in S} \sum_{S'\subseteq S} \rxone_{S',v} \tag{$S' \subseteq S$ if $\rxone_{S',v} >0$ for any $v\in S$} \\
        &=\sum_{S'\subseteq S} \sum_{v\in S} \rxone_{S',v} \tag{Change summation order}\\
        &=\sum_{S'\subseteq S} \sum_{v\in S'} \rxone_{S',v} \tag{$S' \subseteq S$, $\rxone_{S',v} = 0$ for all $v \notin S'$}\\
        &=\sum_{S'\subseteq S} \yyb_{S'}.
    \end{align*}
    Similarly, we can prove the same equality for $\rxtwo$. Therefore, we have
    \[
    \sum_{v\in S} \rxone_v = \sum_{S'\subseteq S} \yyb_{S'} = \sum_{v\in S} \rxtwo_v. 
    \]
    Therefore, set $S$ cannot become inactive in one run and not in the other.
\end{enumerate}
As neither case for active sets differing is possible, the active sets must be the same at any moment for different initial assignments used for calculating potential values. This means that the behavior of the algorithm will be the same, and the same output will be generated. 
\end{proof}

Given Lemma \ref{lm:ext-rex-no-matter}, we assume in our analysis that \ExtendedGW{} is called with potentials defined using $\rstar$ as opposed to the the assignment used in Algorithm~\ref{alg:extend-wrapper}.

\subsection{Extended Execution}

We now introduce the notations for the calls to \ExtendedGW{}, beginning with Extended Execution:

\begin{definition}[Extended Execution]
\label{def:ep}
The execution of \ExtendedGW{} (Line~\ref{line:egw} of Algorithm~\ref{alg:extend-wrapper}) in the call of \Extend{} in Line~\ref{exe:egw} of Algorithm~\ref{alg:main} is referred to as \emph{Extended Execution}. In addition, we define the following notation:
\begin{itemize}
    \item $\yp_S$ represents the growth of set $S$,
    \item $\Ap_{\currenttime}$ is the family of active sets at time $\currenttime$,
    \item $\corep(S)$ is the superactive set derived from the active set $S$,
    \item $\maxactsp$ denotes the family of maximal superactive sets,
    \item $\tp$ is the output fingerprint such that $\tp_v \ge \ttt_v$ for all $v \in V$.
\end{itemize}
\end{definition}

Based on this fingerprint, we define the potential phase as follows:

\begin{definition}[Potential Phase]
\label{def:poten-phase}
    In any monotonic moat growing,
    vertex $v$ is in \emph{potential phase} at moment $\currenttime$ if $\ttt_v \le \currenttime < \tp_v$.

    An active set $S$ is considered to be in the \emph{potential phase} if:
    \begin{itemize}
        \item no vertex in $S$ is in the base phase or boost phase, and
        \item at least one vertex in $S$ is in the potential phase.
    \end{itemize}

\end{definition}

Since the fingerprint $\tp$ is larger than the fingerprint $\ttt$, we can use Lemma~\ref{lm:large_fingerprint_refinement} to obtain the following.

\begin{corollary}
\label{cor:extend_refinement}
    The active sets at any moment $\currenttime$ during Boosted Execution form a {\em refinement} of the active sets at the same moment during Extended Execution.
    Similarly, the connected components at the same moment in Boosted Execution form a refinement of the connected components in Extended Execution.
\end{corollary}

Now, we define assignments for Extended Execution.

\begin{definition}
\label{def:rp}
We simulate Extended Execution to compute this assignment. First, assume each vertex $v$ has a \emph{potential assignment capacity}, $\potcap_v$, which initially is set to $\eps\rstar_v$.

Whenever two active sets, $A$ and $B$, merge, for any pair of vertices $u$ and $v$ in $A \cup B$ such that $\optcom(u) = \optcom(v)$ and $\priority_u < \priority_v$, with both $u$ and $v$ still active, we transfer the remaining potential assignment capacity of $u$ to $v$, i.e., $\potcap_v \gets \potcap_v + \potcap_u$ and $\potcap_u \gets 0$. After all transfers are completed, at most one vertex in each connected component of the optimal solution in the active set has a positive remaining potential assignment capacity.

Consider the moment $\currenttime$ during Extended Execution. For each active set $S \in \Ap_{\currenttime}$:
\begin{itemize}
    \item If $S$ is in the \emph{base phase}, we assign its growth at this moment \emph{proportionally} to all vertices in $\priorityset(\BaseSet(S, \tau))$, where the fraction of growth assigned to each vertex is $1/|\priorityset(\BaseSet(S, \tau))|$. We refer to the total growth assigned to vertex $v$ in the base phase by $S$ as $\rpb_{S,v}$.
    \item If $S$ is in the \emph{boost phase}, no growth is assigned to any vertex.
    \item If $S$ is in the \emph{potential phase}, we assign this moment of growth \emph{proportionally} to all vertices with positive remaining potential assignment capacity in $S$ (with fraction $1/k$ if there are $k$ such vertices). We then decrease the potential assignment capacity of these vertices accordingly. We refer to the total growth assigned to vertex $v$ in the potential phase by $S$ as $\rpp_{S,v}$.
\end{itemize}

We denote the vertices to which active set $S$ assigns its growth at moment $\currenttime$ by $\asp(S, \currenttime)$.

We also refer to the total amount of growth assigned to vertex $v$ as $\rp_v$, and to the total amount of growth assigned to vertex $v$ by active set $S$ as $\rp_{S,v} = \rpb_{S,v} + \rpp_{S,v}$.
\end{definition}

We first show that throughout the simulation described above, the potential remaining for an active set in Algorithm~\ref{alg:extend} matches the remaining potential assignment capacity. This ensures the validity of the assignment, and shows that any active set in the potential phase has capacity to which it can assign growth.

\begin{lemma}  
\label{lm:potcap-equal-setpot}
    In the simulation of Extended Execution (Definition~\ref{def:rp}), at any time $\currenttime$, for any active set $S$,  
    \[
        \sum_{v \in S} \potcap_v = \setpot_S.
    \]
\end{lemma}  
\begin{proof}  
    We have $\pot_v = \eps \rstar_v$. 
    Initially, for each vertex $v$ in Extended Execution, $\setpot_{\{v\}} = \pot_v = \eps \rstar_v$ (Line~\ref{line:init_setpot}). Similarly, $\potcap_v = \eps \rstar_v$ for all $v \in \V$ (Definition~\ref{def:rp}), establishing the base case at $\currenttime = 0$.  

    At any later time, when two sets $S_u$ and $S_v$ merge, their potentials sum:  
    \(
        \setpot_{S_u \cup S_v} \gets \setpot_{S_u} + \setpot_{S_v}
    \) (Line ~\ref{line:merge_setpot}).
    Similarly, $\potcap$ can undergo capacity transfers that preserve the total capacity (Definition~\ref{def:rp}), ensuring that the condition of the lemma remains valid.  

    During a \emph{potential phase} of duration $\Delta$ for active set $S$, we update  
    \(
        \setpot_S \gets \setpot_S - \Delta
    \)  
    (Line~\ref{line:decrease-setpot}). 
    In the beginning of this period, the total remaining capacity in $S$ is equal to $\setpot_S$, which is at least $\Delta$. 
    During this period, $S$ assigns growth with fraction one (Definition~\ref{def:rp}), and it follows that the total capacity decreases by $\Delta$. Therefore, the condition of the lemma is preserved.  

    Since every modification to $\setpot$ and $\potcap$ maintains
    \[
        \sum_{v \in S} \potcap_v = \setpot_S
    \]  
    for any active set $S$, the proof is complete.  
\end{proof}  

Next, we focus on the relationship of assignment and growth of active sets in Extended Execution. 
\begin{lemma}
    \label{lm:ext-yps}
    The total growth in Extended Execution can be bounded as follows:
    \begin{align*}
        \sum_{S\subseteq \V} \yp_{S} \leq \sum_{v\in \V} \rp_v + \lossone.
    \end{align*}
\end{lemma}
\begin{proof}
    Consider any active set $S \in \Ap_\currenttime$ growing at any moment $\currenttime$ of Extended Execution. There are two cases for $S$:

    In the first case, suppose $S$ is in the base phase or the potential phase. By Definition~\ref{def:rp}, $S$ assigns its growth at this moment, with total fraction $1$, to a subset of its vertices. Consequently, the growth for this moment is calculated in both $\yp_S$ and $\sum_{v \in \V} \rp_v$.

    In the second case, consider $S$ is in the boost phase. For any $v \in S$, we have $\tplus_v \leq \currenttime$, and there exists a vertex $u$ such that $\currenttime < \ttt_u$. 
    At the same moment in Boosted Execution, $u$ is in an active set $S'$ since $\currenttime < \ttt_u$. By Corollary~\ref{cor:extend_refinement}, we have $S' \subseteq S$. Any vertex $v \in S' \subseteq S$ satisfies $\tplus_v \leq \currenttime$, meaning all vertices of $S'$ are in the boost phase. Therefore, at this moment, there exists a subset $S' \subseteq S$ whose growth is calculated in $\yy_{\add}$, while the growth of $S$ is calculated in $\yp_S$.

    Thus, at any given moment during Extended Execution, the growth of each active set is accounted for in either $\sum_{v \in \V} \rp_v$ or $\yy_{\add}$. Additionally, since active sets at any moment are mutually exclusive, no calculation is considered multiple times, which leads to:
    \begin{align*}
        \sum_{S \subseteq \V} \yp_S & \leq \sum_{v \in \V} \rp_v + \yy_{\add} \\
        &= \sum_{v \in \V} \rp_v + \lossone. \tag{Lemma \ref{lm:loss-equal-yadd}}
    \end{align*} 
\end{proof}

In the following lemma, we prove that the total growth assigned to a vertex by active sets in the base phase of Extended Execution is bounded by the total growth assigned to that vertex in Boosted Execution.
\begin{lemma}
\label{lm:rp-basephase-less-rstar}
    In assignment $\rp$, for any vertex $v$, 
    $$
        \sum_{S \subseteq \V} \rpb_{S, v} \le \rstar_v.
    $$
\end{lemma}
\begin{proof}
    In Extended Execution, assume that at moment $\currenttime$, active set $S \in \Ap_{\currenttime}$ in the base phase assigns its growth to vertex $v$. This implies that $v \in \BaseSet(S, \currenttime)$, and therefore $\currenttime < \tplus_v$. Consequently, $v$ is active in Boosted Execution at moment $\currenttime$, meaning there exists an active set $S' \in \Aa_{\currenttime}$ such that $v \in S'$. 

By Corollary~\ref{cor:extend_refinement}, we have $S' \subseteq S$. Additionally, by Lemma~\ref{lm:base-set-subset}, it follows that $\BaseSet(S', \currenttime) \subseteq \BaseSet(S, \currenttime)$. Since $v \in \priorityset(\BaseSet(S, \currenttime))$, Lemma~\ref{lm:smaller-priorityset} tells us that $v \in \priorityset(\BaseSet(S', \currenttime))$. Therefore, $S'$ assigns its growth to $v$ in assignment~$\rstar$.

To prove the inequality, it suffices to show that, at any given moment, the growth that set $S$ assigns to vertex $v$ in the assignment $\rp$ is no greater than the growth that $S'$ assigns to $v$ at the same moment. This is because the growth assigned to $v$ by $S$ is calculated in $\rpb_{S, v}$, while the growth assigned to $v$ by $S'$ is calculated in $\rr_v$. 

To proceed, since $\BaseSet(S', \currenttime) \subseteq \BaseSet(S, \currenttime)$, by Lemma~\ref{lm:smaller-priorityset-size}, we have
$$
|\priorityset(\BaseSet(S', \currenttime))| \le |\priorityset(\BaseSet(S, \currenttime))|.
$$
Since both $S$ and $S'$ assign growth proportionally (Definitions~\ref{def:rstar} and~\ref{def:rp}), it follows that $S'$ assigns a larger fraction to $v$ than $S$, completing the proof.
\end{proof}

Next, we aim to establish a similar bound on the growth assigned to a vertex 
by active sets during the potential phase in Extended Execution, 
this time using the potential the vertex receives based on growth assigned to 
in Boosted Execution. 
While such a bound may not hold for individual vertices, 
we prove an analogous inequality for certain subsets of vertices—
specifically, the intersection of any maximal superactive set 
in Extended Execution with any connected component of the optimal solution.
\begin{lemma}
\label{lm:rp-potphase-exatcly-eps-rstar}
    In assignment $\rp$, for the intersection of any connected component $\opti$ of the optimal solution and any maximal superactive set $S$ of Extended Execution,
    $$
        \sum_{v \in S \cap \opti} \sum_{S' \subseteq V} \rpp_{S', v} = \eps \sum_{v \in S \cap \opti} \rstar_v.
    $$
\end{lemma}
\begin{proof}
    In Extended Execution, let $S'$ be the active set $S$ is derived from, and let $\currenttime$ be the time when $S'$ becomes deactivated. We observe that $\setpot_{S'} = 0$ at this moment (Line~\ref{line:setpot_zero}). By Lemma~\ref{lm:potcap-equal-setpot}, we also have $\sum_{v \in S'} \potcap_v = \setpot_{S'} = 0$, indicating that no vertex in $S'$, and therefore none in $S \subseteq S'$, has any remaining capacity.  

    Now, we analyze the behavior of the assignment $\rp$ based on Definition~\ref{def:rp}.  
    First, no vertex transfers its capacity to any vertex outside its connected component of the optimal solution.  
    Second, an active set is deactivated precisely when its $\setpot$ value reaches zero, meaning no assignment capacity remains among its vertices. Thus, capacity is transferred only between vertices that have not been deactivated yet within the same active set.
    This means that whenever capacity is exchanged between two vertices, they must belong to the same superactive set in Extended Execution.
    It follows from laminarity of superactive sets that for any maximal superactive set in Extended Execution, capacity is not given to or received from vertices outside the superactive set, and its capacity remains confined to its vertices.  

    Combining these observations, we conclude that at moment $\currenttime$, the assignment capacity of vertices in $S \cap \opti$ is used solely through growth assigned to vertices within the same set. Since these vertices reach a final capacity of zero, the total growth assigned to vertices in $S \cap \opti$ during the potential phase of active must be exactly 
    \[\eps \cdot \sum_{v \in S \cap \opti} \rstar_v.\]  
\end{proof}

Now, we can combine the bounds obtained in Lemma~\ref{lm:rp-basephase-less-rstar} and Lemma~\ref{lm:rp-potphase-exatcly-eps-rstar} for the total growth assigned to each vertex in Extended Execution.
\begin{lemma}\label{lm:ext-rpeps}
    For the intersection of any connected component $\opti$ of the optimal solution and any maximal superactive set $S$ in Extended Execution, 
    $$
        \sum_{v \in S \cap \opti} \rp_v \le 
        (1+\eps) \sum_{v \in S \cap \opti} \rstar_v.
    $$
\end{lemma}
\begin{proof}
    We derive the inequality as follows:
    \begin{align*}
        \sum_{v \in S \cap \opti} \rp_v 
        &= \sum_{v \in S \cap \opti} \sum_{S' \subseteq V} \rp_{S', v}\\
        &= \sum_{v \in S \cap \opti} \sum_{S' \subseteq V} \rpb_{S', v}
        + \sum_{v \in S \cap \opti} \sum_{S' \subseteq V} \rpp_{S', v} \tag{Definition~\ref{def:rp}}\\
        &\le \sum_{v \in S \cap \opti} \rstar_{v}
        + \sum_{v \in S \cap \opti} \sum_{S' \subseteq V} \rpp_{S', v} \tag{Lemma~\ref{lm:rp-basephase-less-rstar}}\\
        &= \sum_{v \in S \cap \opti} \rstar_{v}
        + \eps \sum_{v \in S \cap \opti} \rstar_v \tag{Lemma~\ref{lm:rp-potphase-exatcly-eps-rstar}}\\
        &= (1 + \eps) \sum_{v \in S \cap \opti} \rstar_v.
    \end{align*}
\end{proof}

\subsection{Extended-Boosted Execution}

Next, we focus on the execution of \LocalSearch{} after extension.

\begin{definition}[Extended-Boosted Execution]
\label{def:ez}
We refer to the corresponding execution of \BoostedModGW{} (Line~\ref{line:corres-mg}; Algorithm~\ref{alg:local-search}) for the \LocalSearch{} in Line~\ref{exe:Xlocalsearch} of Algorithm~\ref{alg:main} as Extended-Boosted Execution. Consequently, we define the following:
\begin{itemize}
    \item The growth of set $S$ is denoted by $\yz_S$,
    \item The family of active sets at moment $\currenttime$ is $\Az_{\currenttime}$,
    \item The superactive set derived from the active set $S$ is $\corez(S)$, 
    \item The family of maximal superactive sets is $\maxactsz$,
    \item The solution found by this execution is $\solext$,
    \item The total loss in the \LocalSearch{} is $\losstwo$,
    \item The total win in the \LocalSearch{} is $\wintwo$,
    \item The output fingerprint of \LocalSearch{} is $\tz$, where $\tz_v \ge \tp_v$ for all $v \in \V$.
\end{itemize}  
\end{definition}

For any connected component $\opti$ in $\OPT$, we define the following notation.

\begin{definition}
\label{def:largest}
     For any component $\opti\in\OPT$, let $S\in \maxactsz$ be the maximal superactive set with a non-empty intersection with $\opti$ that is deactivated the latest in Extended-Boosted Execution. Then, we define $\largest(\opti)$ as the intersection $S \cap \opti$ and use $\xtree$ to denote the minimal subtree of $\otree$ that connects the vertices in $\largest(\opti)$.
\end{definition}

Next, we define the final boost phase as follows:
\begin{definition}(Final Boost Phase)
    A vertex $v$ is in the final boost phase at moment $\currenttime$ if $\tp_v \le \currenttime < \tz_v$.

    A set $S$ is in the final boost phase at moment $\currenttime$ if it has no vertex in base, boost, or potential phases while it has at least one vertex in the final boost phase. 
\end{definition}

We refer to the total growth of active sets in Extended-Boosted Execution during the final boost phase as $\yz_{\add}$. This corresponds to the $\yadd$ in the output of $\BoostedModGW(\G, \tp, \tz)$ in the last \LocalSearch{} call in Line~\ref{exe:Xlocalsearch} of Algorithm~\ref{alg:main}, since in Line~\ref{line:BoostedMG_compute_boost}, the condition $\tp_v \le \currenttime$ holds for all vertices in an active set, which is equivalent to the final boost phase.

We next define the function $\Deputy$, which is conceptually similar to the $\priorityset$ function. This allows us to define an assignment for Extended-Boosted Execution that follows priority values. 

\begin{definition}[Deputy]
\label{def:deputy}
    Let $U$ and $S$ be sets of vertices such that $S \subseteq U \subseteq \V$. Then, $\Deputy(U, S)$ is the set of vertices in $U$ that have the highest priority within the intersection of $U$ and their respective connected components in $\OPT$, for those components that intersect $S$ non-trivially, i.e.,
    $$
        \Deputy(U, S) = \{u \in U \mid \exists v \in S: u = \argmax_{w \in U \cap \optcom(v)} \priority_w\}.
    $$
\end{definition}
Then, we prove the following for $\Deputy$.
\begin{lemma}
\label{lm:deputy-large}
For any sets of vertices \( S \subseteq U \subseteq V \),
\[
    |\Deputy(U, S)| \ge \left| \{ \optcom(v) \mid v \in S \} \right|.
\]
\end{lemma}

\begin{proof}
Consider any vertex \( u \in S \), and let \( \opti = \optcom(u) \). Since \( u \in S \subseteq U \), it follows that \( U \cap \opti \ne \emptyset \).

Let \( v \in U \cap \opti \) be the vertex with the highest priority among those in the intersection. We claim that \( v \in \Deputy(U, S) \), which we will verify shortly. Assuming this claim, it follows that each distinct optimal component \( \opti \) corresponding to some \( u \in S \) maps to a vertex \( v \in \Deputy(U, S) \) such that \( v \in \opti \). As $v$ is in $\opti$ and the components are disjoint, these vertices are distinct. This implies the number of vertices in $\Deputy(U, S)$ is at least the number of distinct optimal components over \( S \), as desired.

To justify the claim that \( v \in \Deputy(U, S) \), we check that it satisfies the conditions of Definition~\ref{def:deputy}:
\begin{itemize}
    \item \( v \in U \) since \( v \in U \cap \opti \)
    \item For the mentioned \( u \), we have \( \opti = \optcom(u) \), and \( v = \arg\max_{w \in U \cap \optcom(u)} \priority_w \).
\end{itemize}
Hence, \( v \in \Deputy(U, S) \), completing the proof.
\end{proof}

Now, we define the growth assignment to vertices for this execution.

\begin{definition}
\label{def:rz}
    Consider the moment $\currenttime$ of Extended-Boosted Execution. For each active set $S$:
    \begin{itemize}
        \item If $S$ is in the base phase, we divide the growth of this moment proportionally among all vertices in $\priorityset(\BaseSet(S, \currenttime))$, meaning each vertex in this set is assigned a fraction of $1/|\priorityset(\BaseSet(S, \currenttime))|$. We refer to the total growth assigned by active set $S$ to vertex $v$ in its base phase as $\rzb_{S,v}$.
        \item If $S$ is in the boost phase, no growth is assigned.
        \item If $S$ is in the potential phase, define the family $\Fam$ as:
        $$
            \Fam = \{S' \in \Ap_{\currenttime} \mid S' \subseteq S\}.
        $$
        We divide the growth of this moment of $S$ proportionally among all vertices in 
        $$\Deputy(\corez(S), \bigcup_{S' \in \Fam} \asp(S', \currenttime)).$$
        Each vertex in this set receives $1/k$ if the size of this set is $k$. We will later show that this set is non-empty, ensuring that any active set in the potential phase assigns its growth.
        
        We refer to the total growth assigned by active set $S$ to vertex $v$ in its potential phase as $\rzp_{S,v}$.
        \item If $S$ is in the final boost phase, similar to the boost phase, no growth is assigned.
    \end{itemize}
    We denote $\rz_v$ as the total growth assigned to vertex $v$ and $\rz_{S,v}$ as the total growth assigned to vertex $v$ by set $S$.
\end{definition}

In the following lemma, we demonstrate that each active set in the potential phase assigns its growth to a non-empty set.
\begin{lemma}
\label{lm:rz-pot-features}
    For active set $S$ in the potential phase at moment $\currenttime$ of Extended-Boosted Execution, given $\Fam = \{S' \in \Ap_{\currenttime} \mid S' \subseteq S\}$,
    \[\bigcup_{S' \in \Fam} \asp(S', \currenttime)\]
    is a non-empty subset of $\corez(S)$. Consequently, \[\Deputy(\corez(S), \bigcup_{S' \in \Fam} \asp(S', \currenttime))\] is defined and non-empty.
\end{lemma}
\begin{proof}
    Since $S$ is in the potential phase, there must exists a vertex $v \in S$ that is in the potential phase.
    Since $v$ is in the potential phase, it is in an active set at moment $\currenttime$ in Extended Execution. 
    Consider the set $S'\in\Ap_\currenttime$ that contains $v$. Then, by Corollary~\ref{cor:localsearch_activeset_refinement}, $S' \subseteq S$ and therefore $S'\in \Fam$. Additionally, any $S'\in \Fam$ cannot be in the base or boost phase, and must be in the potential phase.  

    since $\Fam$ is non-empty and each set $S'\in \Fam$ is in the potential phase, 
    \[
    \bigcup_{S' \in \Fam} \asp(S', \currenttime)
    \]  
    is non-empty. 
    
    Additionally, for any $S'\in \Fam$, $\asp(S',\currenttime)\subseteq\corep(S')$. Furthermore, Lemma \ref{lm:large_fingerprint_superactive_refinement}, 
    implies that $\corep(S')\subseteq\corez(S)$ since Extended-Boosted Execution has a larger fingerprint. Therefore, 
    \begin{align*}
        \bigcup_{S' \in \Fam} \asp(S', \currenttime) \subseteq  \bigcup_{S' \in \Fam} \corep(S') \subseteq \corez(S).
    \end{align*}
    Therefore, $\Deputy(\corez(S),\bigcup_{S' \in \Fam} \asp(S', \currenttime))$ is a valid use of $\Deputy$.
    Finally, Lemma \ref{lm:deputy-large} shows that since $\bigcup_{S' \in \Fam} \asp(S', \currenttime)$ is non-empty, $\Deputy(\corez(S),\bigcup_{S' \in \Fam} \asp(S', \currenttime))$ is also non-empty.
\end{proof}

We use the following bounds on the total growth during Extended-Boosted Execution.
\begin{lemma}
    \label{lm:ext-yzs}
    The total growth in Extended-Boosted Execution can be bounded as follows:
    \begin{align*}
        \sum_{S\subseteq \V} \yz_{S} = \sum_{S\subseteq \V} \yp_{S} - \wintwo + \losstwo.
    \end{align*}
\end{lemma}
\begin{proof}
    Since $\yz$ values are obtained after a local search, this follows from Corollary \ref{cl:total-ys-based-on-win-loss} with $\YIn=\yp$ and $\YOut=\yz$. 
\end{proof}

\begin{lemma}
    \label{lm:ext-yzrz}
    The total growth in Extended-Boosted Execution can be bounded as follows:
    \begin{align*}
        \sum_{S\subseteq \V} \yz_{S} \leq \sum_{v\in \V} \rz_{v} + \lossone + \losstwo.
    \end{align*}
\end{lemma}
\begin{proof}
    First, we claim that
\begin{align}
    \sum_{S \subseteq \V} \yz_S \leq \sum_{v \in \V} \rz_v + \yy_{\add} + \yz_{\add}.
    \label{eq:yz-to-rz-lossone-losstwo}
\end{align}

Consider the moment $\currenttime$ of Extended-Boosted Execution. Let $S \in \Az_{\currenttime}$ be an active set. There are three cases for $S$.

\begin{itemize}
    \item $S$ is in the base or potential phase. In this case, $S$ assigns its growth with fraction one to a subset of its vertices. Therefore, the growth of $S$ is calculated in both $\yz_S$ and $\sum_{v \in \V} \rz_v$.
    
    \item $S$ is in the boost phase. For any vertex $v \in S$, we have $\tplus_v \le \currenttime$ and there exists a vertex $u \in S$ such that $\currenttime \le \ttt_u$. This means that $u$ is also active at the same time in Boosted Execution, in an active set $S'$. Since the fingerprint $\tz$ is larger than $\ttt$, by Lemma~\ref{lm:large_fingerprint_refinement}, and because $u$ appears in both $S$ and $S'$, we conclude that $S' \subseteq S$. For any vertex $v \in S' \subseteq S$, we also have $\tplus_v \le \currenttime$, and for $u$, we have $\currenttime \le \ttt_u$. Thus, $S'$ is in the boost phase. 
    
    On one hand, we know that the growth of this moment for $S$ is calculated in $\yz_S$. On the other hand, the growth of active set $S'$ is calculated in $\yy_{\add}$. Considering that active sets at a given moment are mutually exclusive and for all $S$, $S' \subseteq S$, this growth is calculated in both $\yz_S$ and $\yy_{\add}$, with no growth being calculated multiple times.
    
    \item $S$ is in the final boost phase. The growth of $S$ is calculated in both $\yz_S$ and $\yz_{\add}$.
\end{itemize}

Therefore, since the growth of each active set at any moment is calculated in one term on the right-hand side of inequality~\ref{eq:yz-to-rz-lossone-losstwo}, we have the claim proved.

Finally, applying Lemma~\ref{lm:loss-equal-yadd} for both $\yy_{\add}$ and $\yz_{\add}$ proves that
$$
    \sum_{S \subseteq \V} \yz_S \leq \sum_{v \in \V} \rz_v + \lossone + \losstwo.
$$
\end{proof}

We now prove the following lemmas regarding the relation between the assignment of Extended Execution and Extended-Boosted Execution.

\begin{lemma}
\label{lm:rz-basephase-less-rp}
    In assignment $\rz$, and any vertex $v$, 
    $$
        \sum_{S \subseteq \V} \rzb_{S, v} \le \sum_{S \subseteq \V} \rpb_{S, v}.
    $$
\end{lemma}
\begin{proof}
    At time $\currenttime$ during Extended-Boosted Execution, consider active set $S \in \Az_{\currenttime}$ in its base phase that assigns a fraction $0 < x \leq 1$ of its growth to vertex $v$. To prove the inequality, we show that for any such set $S$, there exists a corresponding active set $S' \in \Ap_{\currenttime}$ in Extended Execution, also in its base phase, that assigns a fraction $y \geq x$ of its growth to $v$ at moment $\currenttime$. Moreover, all such sets $S'$ must be distinct members of $\Ap_{\currenttime}$ at this time.

We know that $S$ is in its base phase. At the same moment in Extended Execution, vertex $v$ is also in the base phase, so there exists an active set $S'$ in its base phase that includes $v$. By Corollary~\ref{cor:localsearch_activeset_refinement}, we have $S' \subseteq S$. Additionally, by Lemma~\ref{lm:base-set-subset}, $\BaseSet(S', \currenttime) \subseteq \BaseSet(S, \currenttime)$. 

Since $S$ assigns growth to $v$ in the base phase, $v$ must be in $\priorityset(\BaseSet(S, \currenttime))$. Then, by Lemma~\ref{lm:smaller-priorityset}, we know that $v \in \priorityset(\BaseSet(S', \currenttime))$, so $S'$ also assigns growth to $v$ in Extended Execution.

Furthermore, according to Lemma~\ref{lm:smaller-priorityset-size}, we have
\[
|\priorityset(\BaseSet(S', \currenttime))| \le |\priorityset(\BaseSet(S, \currenttime))|.
\]
Since the growth is assigned proportionally across vertices, this implies that the fraction assigned to $v$ by $S'$ is larger than or equal to the fraction assigned by $S$.
\end{proof}

Next, we prove a similar lemma for the growth assigned by active sets during their potential phase. The difference here is that instead of considering a single vertex, we consider the intersection of a maximal superactive set and a connected component of the optimal solution.
\begin{lemma}
\label{lm:rz-sum-potphase-less-rp}
    For the intersection of any connected component $\opti$ of the optimal solution and any maximal superactive set $S \in \maxactsz$ in Extended-Boosted Execution, we have 
    $$
        \sum_{v \in S \cap \opti} \sum_{S' \subseteq \V} \rzp_{S,v} \le \sum_{v \in S \cap \opti} \sum_{S' \subseteq \V} \rpp_{S,v}.
    $$
\end{lemma}
\begin{proof}
    At any moment $\currenttime$, for any active set $S'\in\Az_{\currenttime}$ in the potential phase that assigns a fraction $x$ of growth to a vertex $v\in S\cap\opti$ in Extended-Boosted Execution, we find an active set $S''\in\Ap_{\currenttime}$ in the potential phase such that $S''\subseteq S'$, and $S''$ assigns a fraction $y\geq x$ of growth to a vertex $u\in S\cap\opti$ in Extended Execution. Since each active set only assigns growth to one vertex in $\opti$, and active sets at moment $\currenttime$ are disjoint, this is sufficient to prove our claim. 

    Since $S'$ is in the potential phase and assigns growth to $v$, Definitions~\ref{def:deputy} and \ref{def:rz} imply that
    $v \in \corez(S')$ and
    there exists a vertex $u$ in $\asp(S'', \currenttime)$ such that:
    \begin{itemize}
        \item $u \in \opti$,
        \item $S'' \subseteq S$,
        \item $S'' \in \Ap_{\currenttime}$ is an active set in the potential phase of Extended Execution, and
        \item $S''$ assigns a fraction of its growth to $u$. It follows that $u$ is in its potential phase. 
    \end{itemize}
    We observe that both \( u \) and \( v \) must be active at moment \(\currenttime\) in Extended-Boosted Execution, since \( u \) is in its potential phase and \(v\in \corez(S')\).  
    Combined with the fact that \( S \) is the maximal superactive set containing \( v \), it follows that \( u \in S \).  
    Now, it remains to show that \( S'' \) assigns a larger fraction of its growth to \( u \) compared to the fraction that \( S' \) assigns to \( v \).

    Let  
    \[
    U = \corez(S')
    \]  
    and  
    \[
    H = \bigcup_{S^* \in \Fam} \asp(S^*, \currenttime),
    \]  
    where  
    \[
    \Fam = \{S^* \in \Ap_{\currenttime} \mid S^* \subseteq S'\}.
    \]  
    
    By Lemma~\ref{lm:rz-pot-features}, we know that \( H \subseteq U \). Also, by Lemma~\ref{lm:deputy-large}, we have  
    \[
    |\Deputy(U, H)| \geq |\{\optcom(w) \mid w \in H\}|.
    \]  
    Since in assignment $\rp$, $S''$ assigns its growth to at most one vertex of each component,  
    \[
    |\asp(S'', \currenttime)| = |\{\optcom(w) \mid w \in \asp(S'', \currenttime)\}|.
    \]  
    Since \(\asp(S'', \currenttime) \subseteq H\), it follows that  
    \[
    |\asp(S'', \currenttime)| \leq |\Deputy(U, H)|.
    \]  
    Therefore, at time \(\currenttime\), the growth fraction assigned by \( S'' \subseteq S'\) to \( u \in S \cap \opti \) is larger than the growth fraction assigned by \( S' \) to \( v \), as both assign their growth proportionally. This completes the proof.
\end{proof}

We proceed by combining Lemma~\ref{lm:rz-basephase-less-rp} and Lemma~\ref{lm:rz-sum-potphase-less-rp}.
\begin{lemma} \label{lm:ext-rzrp}
    For the intersection of any connected component $\opti$ of the optimal solution and any maximal superactive set $S \in \maxactsz$ in Extended-Boosted Execution, we have 
    $$
        \sum_{v \in S \cap \opti} \rz_v \leq \sum_{v \in S \cap \opti} \rp_v.
    $$
\end{lemma}
\begin{proof}
    We expand the left-hand side of the inequality as follows:
    \begin{align*}
        \sum_{v \in S \cap \opti} \rz_v 
        &= \sum_{v \in S \cap \opti} \sum_{S' \subseteq V} \rz_{S', v}\\
        &= \sum_{v \in S \cap \opti} \sum_{S' \subseteq V} \rzb_{S', v}
        + \sum_{v \in S \cap \opti} \sum_{S' \subseteq V} \rzp_{S', v} \tag{Definition~\ref{def:rz}}\\
        &\le \sum_{v \in S \cap \opti} \sum_{S' \subseteq V} \rpb_{S', v} 
        + \sum_{v \in S \cap \opti} \sum_{S' \subseteq V} \rzp_{S', v} 
        \tag{Lemma~\ref{lm:rz-basephase-less-rp} for all $v \in S \cap \opti$}\\
        &\le \sum_{v \in S \cap \opti} \sum_{S' \subseteq V} \rpb_{S', v} 
        + \sum_{v \in S \cap \opti} \sum_{S' \subseteq V} \rpp_{S', v} \tag{Lemma~\ref{lm:rz-sum-potphase-less-rp}}\\
        &= \sum_{v \in S \cap \opti} \rp_v .\tag{Definition~\ref{def:rp}}
    \end{align*}
    Consequently, the proof is complete.
\end{proof}

Next, we prove that the family of maximal superactive sets in Extended Execution is a refinement of the family of superactive sets in Extended-Boosted Execution. This also holds when we observe the intersection of the sets in each family with any connected component of the optimal solution.
\begin{lemma} \label{lm:ext-boost-superactive-union-of-ext-superactive}
    Let $U$ be the intersection of any connected component $\opti$ of the optimal solution and any maximal superactive set $S \in \maxactsz$ in Extended-Boosted Execution. Then, $U$ can be decomposed into intersections of maximal superactive sets in Extended Execution and $\opti$, i.e., there exists a family $\Fam \subseteq \maxactsp$ such that
    $$
        U = \bigcup_{S' \in \Fam} (\opti \cap S').
    $$
    Furthermore, $\Fam$ includes all sets in $\maxactsp$ that intersect $U$ non-trivially. 
\end{lemma}
\begin{proof}
    Let $\Fam \subseteq \maxactsp$ be the family of maximal superactive sets in Extended Execution that intersect $S$. We claim that $\Fam$ is a partition of $S$. To see this, note that $\Fam$ is pairwise disjoint and covers $S$ since $\maxactsp$ forms a partition over $\V$.

    Assume, for contradiction, that there exists \( S' \in \Fam \) such that \( S' \) contains a vertex \( v \in S \) and a vertex \( u \notin S \). By Corollary~\ref{lm:large_fingerprint_superactive_refinement}, the superactive sets of Extended Execution are a refinement of the superactive sets of Extended-Boosted Execution, as this execution has a larger fingerprint. Therefore, at some moment, \( v \) and \( u \) must both belong to the same superactive set in Extended-Boosted Execution. The laminarity of superactive sets then leads to a contradiction with the maximality of $S$.

    Now, since 
    \[
    S = \bigcup_{S' \in \Fam} S'
    \]
    and $U = \opti \cap S$, we have
    \[
    U = \bigcup_{S' \in \Fam} (\opti \cap S'),
    \]
    which completes the proof.
\end{proof}

We observe that since $\maxactsz$ forms a partition over $\V$, we can decompose any component $\opti$ of $\OPT$ into intersections of $\opti$ with maximal superactive sets in $\maxactsz$.
Noting that one of these intersections is $\LG$, we can also write $\NLG$ as a union of intersections of $\opti$ with with maximal superactive sets in $\maxactsz$ by putting aside the superactive set corresponding with $\LG$.
Additionally, the intersection of each superactive set in the $\maxactsz$ with $\opti$ can be decomposed into intersections of superactive sets in $\maxactsp$ with $\opti$ by Lemma \ref{lm:ext-boost-superactive-union-of-ext-superactive}. This implies that $\opti$ and $\NLG$ can be further decomposed into intersections of $\opti$ and superactive sets in $\maxactsp$.
 
Based on this observation, we can derive the following corollaries from Lemma \ref{lm:ext-rpeps}, Lemma~\ref{lm:rp-potphase-exatcly-eps-rstar}, and Lemma \ref{lm:ext-rzrp}.

\begin{corollary} \label{cor:ext-comp-eps}
    For any component $\opti$ of the optimal solution, we have
    \[
    \sum_{v\in \opti} \rp_v \leq (1+\eps)\sum_{v \in \opti} \rstar_v.
    \]
\end{corollary}

\begin{corollary} \label{cor:ext-nlg-eps}
    For any component $\opti$ of the optimal solution, we have
    \[
    \sum_{v\in \NLG} \rp_v \leq (1+\eps)\sum_{v \in \NLG} \rstar_v.
    \]
\end{corollary}

\begin{corollary} \label{cor:ext-nlg-potphase-eps}
    For any component $\opti$ of the optimal solution, the total growth assigned to vertices of $\NLG$ in the potential phase of active sets in Extended Execution is exactly
    \[
    \eps\sum_{v \in \NLG} \rstar_v.
    \]
\end{corollary}

\begin{corollary} \label{cor:ext-opti-rz-less-rp}
    For any component $\opti$ of the optimal solution, we have 
    \[
    \sum_{v\in \opti} \rz_V \leq \sum_{v\in \opti} \rp_v.
    \]
\end{corollary}

\begin{corollary} \label{cor:ext-nlg-rz-less-rp}
    For any component $\opti$ of the optimal solution, we have 
    \[
    \sum_{v\in \NLG} \rz_V \leq \sum_{v\in \NLG} \rp_v.
    \]
\end{corollary}

Next, we provide bounds for the maximum assigned value to vertices in $\LG$ in both Extended Execution and Extended-Boosted Execution.
\begin{lemma}\label{lm:ext-rpmax}
    For any connected component $\opti$ of the optimal solution,
    $$
        \rp_{\max}(\largest(\opti)) \le \rstar_{\max}(\opti) + 
        \sum_{v \in \largest(\opti)} \eps\rstar_v
    $$
\end{lemma}
\begin{proof}
Let vertex $v^*$ be the vertex with the maximum $\rp$ value in $\largest(\opti)$. By Lemma~\ref{lm:rp-basephase-less-rstar}, the total growth assigned to $v^*$ in the base phase of Extended Execution is at most $\rstar_{v^*}$, and since $v^* \in \opti$, we have $\rstar_{v^*} \le \rstar_{\max}(\opti)$.

We know that $\largest(\opti)$ is the intersection of a maximal superactive set in $\maxactsz$ with $\opti$. By Lemma~\ref{lm:ext-boost-superactive-union-of-ext-superactive}, we can decompose $\LG$ into intersections of maximal superactive sets in $\maxactsp$ with $\opti$. Therefore, applying Lemma~\ref{lm:rp-potphase-exatcly-eps-rstar} to the superactive sets in this decomposition, the total growth assigned to vertices in $\LG$ during the potential phase is
\[
\eps \sum_{v \in \LG} \rstar_v.
\]
This is also an upper bound for the growth assigned to $v^*$ during the potential phase, since $v^* \in \LG$. Therefore, we have
\[
    \rp_{\max}(\largest(\opti)) = \rp_u \le \rstar_{\max}(\opti) + \sum_{v \in \largest(\opti)} \eps \rstar_v.
\]
\end{proof}

\begin{lemma}\label{lm:ext-rzmax}
    For any connected component $\opti$ of the optimal solution,
    $$
        \rz_{\max}(\largest(\opti)) \le \rstar_{\max}(\opti) + 
        \sum_{v \in \largest(\opti)} \eps\rstar_v
    $$
\end{lemma}
\begin{proof}
    Consider a vertex $v^*$ in $\largest(\opti)$ such that $\rz_{v^*} = \rz_{\max}(\LG)$. We can bound the growth assigned to $v^*$ in the base phase as follows:
     \begin{align*}
     \sum_{S\subseteq \V} \rzb_{S,v^*} &\leq \sum_{S\subseteq \V } \rpb_{S,v^*} \tag{Lemma \ref{lm:rz-basephase-less-rp}}\\
     &\leq \rstar_{v^*} \tag{Lemma \ref{lm:rp-basephase-less-rstar}}.
     \end{align*}
     On the other hand, we can upper bound the growth assigned to $v^*$ in the potential phase by the total growth assigned to any vertex in $\LG$. Then, we have
     \begin{align*}
        \sum_{S\subseteq\V} \rzp_{S,{v^*}} 
         &\leq \sum_{v\in \LG} \sum_{S\subseteq \V} \rzp_{S,v}.
        \end{align*}
    Since $\LG$ is the intersection of a superactive set in Extended-Boosted Execution, we can apply Lemma~\ref{lm:rz-sum-potphase-less-rp} to obtain
        \[
        \sum_{S\subseteq\V} \rzp_{S,{v^*}} 
         \leq \sum_{v\in \LG} \sum_{S\subseteq \V} \rpp_{S,v}. 
         \]
         Furthermore, by Lemma~\ref{lm:ext-boost-superactive-union-of-ext-superactive}, we can decompose $\LG$ into intersections of superactive sets in Extended Execution and $\opti$, so we can apply Lemma~\ref{lm:rp-potphase-exatcly-eps-rstar}
         \begin{align*}
             \sum_{S\subseteq\V} \rzp_{S,{v^*}} 
         &\leq \sum_{v\in \LG} \eps \rstar_{v}\\ 
         &=\eps \sum_{v\in \LG} \rstar_{v}.
         \end{align*}
     
     Finally, we have: 
     \begin{align*}
         \rz_{\max}(\LG)&=\rz_{v^*}\\
         &=\sum_{S \subseteq \V} \rz_{S,v^*}\\
         &=\sum_{S \subseteq \V} \rzb_{S,v^*} + \sum_{S \subseteq \V} \rzp_{S,v^*}\\
         &\leq \rstar_{v^*} + \eps \sum_{v\in \LG} \rstar_v\\
         &\leq \rstar_{\max}(\opti) + \eps \sum_{v\in \LG} \rstar_v. \tag{$v^* \in \opti$}
     \end{align*}
\end{proof}

The following lemma asserts that any vertex assigned growth by an active set 
in Extended Execution has the highest priority among all vertices in the intersection 
of the active set, the connected component of the optimal solution it belongs to, 
and the maximal superactive set it is contained in.
\begin{lemma} \label{lm:rp-max-priority}
    For any connected component $\opti$ of the optimal solution and superactive set $S'$ in Extended Execution, if $\rp_{S, v} > 0$ for active set $S$ and vertex $v \in\opti \cap S'$, we have
    $$
        v = \argmax_{u \in S \cap \opti \cap S'} \priority_u.
    $$
\end{lemma}
\begin{proof}
    Consider a moment $\currenttime$ when active set $S$ assigns its growth to vertex $v$. Since $S$ is assigning growth at moment $\currenttime$, it is either in the base phase or in the potential phase.

    First, assume $S$ is in the base phase. In this case, $S$ assigns its growth to vertices in $\priorityset(\BaseSet(S, \currenttime))$. By Lemma \ref{lm:opti_cap_priorityset}, it follows that $v = \argmax_{w \in S \cap \opti} \priority_w$. Consequently, we have $v = \argmax_{w \in S \cap \opti \cap S'} \priority_w$.

    Now, assume that $S$ is in the potential phase. For contradiction, suppose there exists a vertex $u \in S \cap \opti \cap S'$ such that $\priority_u > \priority_v$. Since $u \in S'$, $v$ and $u$ are first connected while they are both active. Therefore, the potential capacity of $v$ must be transferred to $u$, as $\priority_u > \priority_v$. However, this implies that $\potcap_v = 0$ at moment $\currenttime$, meaning that $v$ cannot be assigned growth by $S$. This leads to a contradiction, so our assumption must be false.

    Thus, $v = \argmax_{u \in S \cap \opti \cap S'} \priority_u$, as required.
\end{proof}

\begin{lemma}
\label{lm:two-v-in-same-act-comp-assigns-to-higher-priority}
    For any connected component $\opti$ of the optimal solution, considering assignment $\rz$, for an active set $S$, if $\rz_{S, v} > 0$ for some $v \in \largest(\opti)$, we have
    $$
        v = \argmax_{u \in S \cap \largest(\opti)} \priority_u.
    $$
\end{lemma}
    Consider a moment $\currenttime$ when active set $S$ assigns its growth to $v$. Since $S$ is assigning growth at moment $\currenttime$, it is either in the base phase or the potential phase. 
    
    First, if $S$ is in the base phase, it assigns its growth to vertices in $\priorityset(\BaseSet(S,\currenttime))$. 
    Therefore, by Lemma \ref{lm:opti_cap_priorityset}, $v=\argmax_{w\in S\cap \opti}\priority_w$ and therefore $v=\argmax_{w\in S\cap\opti\cap S'}\priority_w$. 

    Otherwise, $S$ is in the potential phase. Let $\Fam=\{S'\mid S'\subseteq S, S'\in\Ap_\currenttime\}.$ Then, growth is assigned to vertices in 
    \[
    \Deputy(\corez(S), \bigcup_{S' \in \Fam} \asp(S', \currenttime)).
    \]
    Now, assume for contradiction that there exists $u\in S\cap \largest(\opti)$ such that $\priority_u>\priority_v$. Since $u$ and $v$ are both in $\largest(\opti)$, they are connected actively in Extended-Boosted Execution. Since $v\in \corez(S)$, $u$ must also be in $\corez(S)$ because $u$ and $v$ are deactivated at the same moment. However, if $u\in \corez(S)$, then \[\Deputy(\corez(S), \bigcup_{S' \in \Fam} \asp(S', \currenttime))\] 
    cannot include $v$, because $\priority_u > \priority_v$ and $u\in\corez(S)\cap\optcom(v)$. 

\subsection{Cost Analysis of $\solext$}

\paragraph{Component Specific Bounds.}
To bound the cost of the solution after extension, $\solext$, we analyze components of the optimal solution $\OPT$ one by one. 
For any component $\opti$ of $\OPT$, we recall that $\otree$ is used to refer to the tree in $\OPT$ that contains $\opti$. 

First, we provide the following bounds for components in $\A$.
\begin{lemma} \label{lm:ext-a-bound}
    For any component $\opti$ of $\OPT$, if $\opti \in \A$, we have
\[
\sum_{v\in\opti} \rp_v\leq (1+\eps)\sum_{v\in\opti}\rstar_v \leq (1+\eps)(1-\cone)\cc(\otree).
\]
\end{lemma}
\begin{proof}
    This follows from Corollary \ref{cor:ext-comp-eps} and Definition \ref{def:classify_A_B}.
\end{proof}

The second bound follows from the previous bound and Corollary \ref{cor:ext-opti-rz-less-rp}.

\begin{corollary}
    \label{cor:ext-a-rz-bound}
    For any component $\opti$ of $\OPT$, if $\opti \in \A$, we have
\[
\sum_{v\in\opti} \rz_v\leq (1+\eps)\sum_{v\in\opti}\rstar_v \leq (1+\eps)(1-\cone)\cc(\otree).
\]
\end{corollary}

Next, we focus on components in $\B$. We provide several different lower bounds for these components and combine them to produce our final bound. 

For each component $\opti$ in $\B$, our bounds utilize the partitioning of $\opti$ into $\largest(\opti)$ and $\NLG$. For vertices in $\NLG$, we define the following value, which is utilized in several bounds. 
\begin{definition} \label{def:ext-X}
    We use $\X(\opti)$ to denote the total growth assigned to vertices in $\NLG$ in assignment $\rp$ corresponding to active sets $S$ cutting at least two edges of $\otree$. That is, 
    \begin{align*}
\X(\opti)=\sum_{v\in\NLG}\sum_{\substack{S\subseteq\V\\\lvert\deltaS\cap\otree\rvert>1}} \rp_{S,v}.
    \end{align*}
\end{definition}

The following lemma establishes that vertices in $\NLG$ cannot connect to vertices in $\LG$ while they are active in Extended Execution. This is later used to identify active sets that cut $\otree$.
\begin{lemma} \label{lm:ext-nlg}
Consider any component $\opti$ of $\OPT$, and
    let $S$ be any active set assigning growth to a vertex $v$ in $\NLG$ during Extended Execution, i.e., $\rp_{S,v}>0$ for some $v \in \NLG$.
    Then, $S$ does not contain any vertex in $\largest(\opti)$.
\end{lemma}
\begin{proof}
Assume, for contradiction, that $S$ contains a vertex $u \in \largest(\opti)$ while $\rp_{S, v} > 0$ for some $v \in \opti \setminus \largest(\opti)$. Since $\rp_{S,v}>0$, we must have $v\in\corep(S)$, and $v$ is active at moment $\currenttime$. By Lemma \ref{lm:large_fingerprint_vertex_remains_active}, $v$ is also active at moment $\currenttime$ in Extended-Boosted Execution, since fingerprint $\tz$ is larger than $\tp$.

Let $S'$ be the active set containing $v$ at moment $\currenttime$ in Extended-Boosted Execution. By Corollary \ref{cor:localsearch_activeset_refinement}, we know that $S\subseteq S'$, and therefore $u\in S'$. Since $u\in \largest(\opti)$, $u$ cannot be deactivated before $v$ in Extended-Boosted Execution. Consequently, $u$ is also active at moment $\currenttime$. This implies that $u$ and $v$ are connected while active, meaning they will be deactivated at the same moment. Hence, $u$ can only be in $\largest(\opti)$ if $v\in\largest(\opti)$. However, we assumed that $v \in \NLG$ while $u\in \LG$, which leads to a contradiction.
\end{proof}

In the following lemma, during Extended Execution, we identify a class of active sets that cut a connected component of the optimal solution.
\begin{lemma} \label{lm:ext-lone-vertex-lg-always-cut}
    For any component $\opti$ of $\OPT$, there exists a vertex $v^*$ in $\largest(\opti)$, such that for any active set $S$ assigning growth to a vertex $v\in\opti\setminus\{v^*\}$ during Extended Execution, $S$ cuts $\otree$.  
\end{lemma}
\begin{proof}
    Consider a vertex $u$ in $\largest(\opti)$ that remains active the longest in Extended Execution. Let $S'$ be the maximal superactive set containing $u$, 
    and let $v^*$ be the vertex in $S' \cap \opti$ with the highest priority, i.e.,
    \[
    v^* = \argmax_{w \in S' \cap \opti} \priority_w.
    \]
    Since both $u$ and $v^*$ belong to $S'$, they are deactivated at the same moment 
    during Extended Execution. As a result, whenever any vertex in $\LG$ is active 
    in Extended Execution, $v^*$ must also be active.
    Additionally, it follows from Lemma \ref{lm:ext-boost-superactive-union-of-ext-superactive} that $S'\cap\opti\subseteq\LG$ and consequently $v^*\in\largest(\opti)$. 

    Now, consider any active set $S$ such that $\rp_{S,v}>0$ for a vertex $v\in\opti\setminus\{v^*\}$. We show that $S$ cannot include $v^*$. 
    If $v\in \NLG$, this results from Lemma \ref{lm:ext-nlg}. 
    Otherwise, if $v\in\largest(\opti)$, assume for contradiction that $v^*\in S$.
    Since $v$ is assigned growth by $S$, it must be active at this moment.  
    Since $v^*$ is active as long as any other vertex in $\LG$, it follows that $v$ and $v^*$ would connect while active.
    Therefore, $v$ must also belong to the maximal superactive set $S'$.
    Now, by definition, $v^*=\argmax_{w\in S'\cap \opti}\priority_w$ and thus 
    \[
    v^*=\argmax_{w\in S\cap \opti\cap S'}\priority_w.
    \]
    Thus, it follows from Lemma \ref{lm:rp-max-priority} that 
    $S$ 
    cannot assign growth to $v$ since
    $v\in S\cap\opti\cap S'$ and \(v \neq v^*\).
    This is a contradiction, so $v^*$ must not belong to $S$.
    
    Finally, since $S$ includes $v$ but not $v^*$, and $v$ and $v^*$ are connected by $\otree$, it follows that $S$ must cut $\otree$. 
\end{proof}

The following lemma presents our first bound for the cost of components $\opti$ in $\B$. In this bound, we consider how the sets contributing to $\X(\opti)$ color at least two edges of the tree $\otree$. 
\begin{lemma} \label{lm:ext-x-bound}
    For any component $\opti\in\B$, we have
    \[
    \sum_{v\in\opti}\rp_v + \X(\opti) - \rp_{\max}(\largest(\opti)) \leq \cc(\otree).
    \]
\end{lemma}
\begin{proof}
    
    We can lower bound $\cc(\otree)$ as follows:
    \begin{align*}
        \cc(\otree) &\geq \sum_{S\subseteq \V} \lvert \deltaS\cap \otree\rvert \cdot \yp_S. \tag{Lemma~\ref{lm:tree-bound-by-activeset-num-cuts}}
    \end{align*}
    Breaking down the sum based on $\lvert \deltaS\cap \otree\rvert$, we get
    \begin{align*}
        \cc(\otree) &\geq \sum_{\substack{S\subseteq \V\\
        S \odot \otree}} \yp_S + \sum_{\substack{S\subseteq \V\\\lvert \deltaS\cap \otree\rvert>1}} \yp_S.
    \end{align*}
    Next, we lower bound $\yp_S$ values by their contribution to $\rp$ values for subsets of $\V$. 
    Take the vertex $v^*\in \largest(\opti)$ considered in Lemma \ref{lm:ext-lone-vertex-lg-always-cut}. 
    For the first summation term, we use the subset $\opti\setminus\{v^*\}$, since by Lemma \ref{lm:ext-lone-vertex-lg-always-cut}, any set assigning growth to a vertex in this subset in Extended Execution must cut $\otree$. In the second summation, we use $\NLG$ to find a term matching the definition of $\X(\opti)$.
    \begin{align*}    
        \cc(\otree)&\geq \sum_{\substack{S\subseteq \V\\
        S \odot \otree}} \sum_{v\in\opti\setminus\{v^*\}} \rp_{S,v} + \sum_{\substack{S\subseteq \V\\\lvert \deltaS\cap \otree\rvert>1}} \sum_{v\in \NLG} \rp_{S,v}  \tag{$\yp_S \geq \sum_{v\in V_1} \rp_{S,v}$ for any subset $V_1$}\\
        &\geq \sum_{\substack{S\subseteq \V}} \sum_{v\in \opti\setminus\{v^*\}} \rp_{S,v} + \sum_{\substack{S\subseteq \V\\\lvert \deltaS\cap \otree\rvert>1}} \sum_{v\in \NLG} \rp_{S,v} \tag{Lemma \ref{lm:ext-lone-vertex-lg-always-cut}}\\
        &\geq \sum_{v\in \opti\setminus\{v^*\}}\sum_{\substack{S\subseteq \V}}  \rp_{S,v} + \sum_{v\in \NLG} \sum_{\substack{S\subseteq \V\\\lvert \deltaS\cap \otree\rvert>1}} \rp_{S,v} \tag{Change order of summations}\\
        &= \sum_{v\in \opti\setminus\{v^*\}} \rp_v + \sum_{v\in \NLG} \sum_{\substack{S\subseteq \V\\\lvert \deltaS\cap \otree\rvert>1}} \rp_{S,v} \tag{$\sum_{S\subseteq V} \rp_{S,v} = \rp_v$ for any vertex $v$}\\
        &\geq \sum_{v\in \opti\setminus\{v^*\}} \rp_v + \X(\opti) \tag{Definition \ref{def:ext-X}}\\
        &=\sum_{v\in \opti} \rp_v + \X(\opti) - \rp_{v^*}.
    \end{align*}
    
    Finally, since $v^* \in \largest(\opti)$, we have $\rp_{v^*} \leq \rp_{\max}(\largest(\opti))$, and we can conclude that 
    \[
    \sum_{v\in\opti}\rp_v + \X(\opti) - \rp_{\max}(\largest(\opti)) \leq \cc(\otree)
    \]
    as desired.
\end{proof}

In the following lemma, during Extended Execution, we identify a class of active sets that cut either zero or at least two edges of $\xtree$ for a connected component $\opti$ of the optimal solution.
\begin{lemma} \label{lm:ext-cut-lg}
    Consider any component $\opti$ of $\OPT$, and
    let $S$ be any active set assigning growth to a vertex $v$ in $\NLG$ during Extended Execution, i.e., $\rp_{S,v}>0$ for some $v \in \NLG$. Then, $S$ cannot cut exactly one edge of $\xtree$.
\end{lemma}
\begin{proof}
    Assume, for contradiction, that $S$ cuts exactly one edge of $\xtree$. By Lemma \ref{lm:cut_one_edge_cut_leaves}, cutting exactly one edge of $\xtree$ implies that $S$ must cut one of the leaves of $\xtree$. Therefore, $S$ must include at least one of the leaves of $\xtree$.

    Since $\xtree$ is the minimal tree connecting the leaves in $\otree$, the leaves of $\xtree$ must be in the connected component $\largest(\opti)$ of $\otree$. However, by Lemma \ref{lm:ext-nlg}, an active set assigning growth to a vertex in $\NLG$ during Extended Execution cannot include any vertex from $\largest(\opti)$. This leads to a contradiction.
    Therefore, $S$ cannot cut exactly one edge of $\xtree$.
\end{proof}

We can now prove our second lower bound on the cost of component $\opti\in \B$, which utilizes Lemma \ref{lm:steiner-tree} to bound the cost of $\xtree$, and lower bounds the coloring on the rest of $\otree$ using the fact that growth assigned to vertices in $\NLG$ cannot contribute to coloring exactly one edge of $\xtree$. Figure \ref{fig:extend-tree} illustrates this partition of $\otree$.

\tikzset{
  square/.style args={#1}{
    insert path={
        ++(45:1.1*#1)
      -- ++($(45:-1.2*#1)+(135:1.1*#1)$)
      -- ++($(135:-1.2*#1)+(225:1.1*#1)$)
      -- ++($(225:-1.2*#1)+(315:1.1*#1)$)
      -- cycle %
    }
  }
}
\tikzset{
  star/.style args={#1}{
    insert path={
      ++(90:1.1*#1)
      \foreach \a in {90,162,234,306,378} {
        -- ++($(\a:-1.4*#1) + (\a+36:0.7*#1)$)
        -- ++($(\a+36:-0.7*#1) + (\a+72:1.4*#1)$)
      }
      -- cycle
    }
  }
}
\begin{figure}[ht]
    \centering
\begin{tikzpicture}[
scale=0.7
]

\def\dem{Red!70}
\def\ter{White}
\def\col{Black!70}
\def\noder{0.1cm}
\def\inoder{0.1cm}
\def\enoder{0.1cm}

\def\dblue{Black!100}

\def\n{4}
\def\m{4}
\pgfmathtruncatemacro{\last}{\n+1}

\def\tercl#1{%
  \ifcase#1
    White\or
    olive!20\or
    Green!20\or
    Sepia!20\or
    Plum!20\or
    Goldenrod!20\or
  \fi
}
\def\tersh#1{%
  \ifcase#1
    none\or
    star\or
    square\or
    pentagon\or
    triangle\or
    hexagon\or
  \fi
}

\tikzset{
  treenode/.style={
    circle,
    draw,
    \col,
    minimum size=2*\noder,
    inner sep=0pt,
    outer sep=0pt,
    fill=\ter
  },
  rootnode/.style={
    treenode,
  },
  leafnode/.style={
    treenode,
    fill=\col
  },
  treeedge/.style={
    color=\col,
  },
  tickedge/.style={
    treeedge,
    color=\dblue,
    line width=1.5pt
  }
}

\newcommand{\sml}[1]{\scalebox{0.9}{#1}}
\def\lmar{-2pt}
\def\edge{1.2}

\coordinate (A) at (0*\edge,0*\edge);

\coordinate (B) at (-2*\edge,-1.3*\edge);
\coordinate (C) at (0*\edge,-1.1*\edge);
\coordinate (D) at (2*\edge,-1.6*\edge);
\draw[tickedge] (A) -- (B) node[pos=0.2, above left] {\sml{$\xtree$}};
\draw[tickedge] (A) -- (C);
\draw[tickedge] (A) -- (D);

\coordinate (E) at (-2.8*\edge,-2.6*\edge);
\coordinate (F) at (-2.1*\edge,-2.4*\edge);
\coordinate (G) at (-1.4*\edge,-2.2*\edge);
\draw[tickedge] (B) -- (E);
\draw[treeedge] (B) -- (F);
\draw[tickedge] (B) -- (G);

\coordinate (H) at (-0.5*\edge,-2.3*\edge);
\draw[tickedge] (C) -- (H);

\coordinate (I) at (1.3*\edge,-2.1*\edge);
\coordinate (J) at (2.5*\edge,-2.5*\edge);
\draw[tickedge] (D) -- (I);
\draw[treeedge] (D) -- (J);

\coordinate (K) at (-2.9*\edge,-3.7*\edge);
\coordinate (L) at (-2.4*\edge,-3.4*\edge);
\draw[tickedge] (E) -- (K);
\draw[treeedge] (E) -- (L);

\coordinate (M) at (-1.3*\edge,-3.5*\edge);
\draw[tickedge] (G) -- (M);

\coordinate (N) at (1.0*\edge,-3.0*\edge);
\coordinate (O) at (1.5*\edge,-3.4*\edge);
\coordinate (P) at (2.2*\edge,-3.1*\edge);
\draw[treeedge] (I) -- (N);
\draw[tickedge] (I) -- (O);
\draw[treeedge] (I) -- (P);

\foreach \i in {A, B, C, D, E, G, I} {
    \draw[\col, fill= \ter] (\i) circle (\inoder);
}
\foreach \i in {F, J, L, N, P} {
    \draw[\col, fill=\col] (\i) [square=\enoder];
}
\foreach \i in {H, K, M, O} {
    \draw[\col, fill=\col] (\i) [star=\enoder];
}

\end{tikzpicture}
    \caption{An illustration of tree $\otree$ for a component $\opti$ of $\OPT$ with the bolded subtree representing $\xtree$. We use the fact that vertices in $\LG$ connect actively to provide a stronger bound on the cost of $\xtree$ based on our local search algorithm, while also bounding the cost of the rest of the tree $\otree$ using growth assigned to vertices outside $\LG$.
    }
    \label{fig:extend-tree}
\end{figure} 
\begin{lemma} \label{lm:ext-tree-bound}
    For any component $\opti\in\B$, we have
    \[
    \frac{6}{5+\bta} \left[\sum_{v\in\largest(\opti)}\rz_v - \rz_{\max}(\largest(\opti))\right] + \sum_{v\in\opti\setminus\largest(\opti)} \rz_v -  \X(\opti)  \leq \cc(\otree)
    \]
\end{lemma}
\begin{proof}
    First, consider the set $\largest(\opti)$ and the tree $\xtree$ connecting these vertices. By definition, the vertices in $\largest(\opti)$ belong to the same maximal superactive set and reach each other actively in Extended-Boosted Execution. Additionally, by Lemma \ref{lm:two-v-in-same-act-comp-assigns-to-higher-priority}, assignment $\rz$ is priority-based on $\largest(\opti)$. Thus, we can apply Lemma \ref{lm:steiner-tree} with $S=\largest(\opti)$ and $\rstree=\rz$ to get
    \begin{align*}
        \frac{6}{5+\bta} \left[\sum_{v\in\largest(\opti)}\rz_v - \rz_{\max}(\largest(\opti))\right]\leq \cc(\xtree).
    \end{align*}
    On the other hand, we can bound the cost of $\otree\setminus\xtree$ based on the total coloring as follows:
    \begin{align*}
    \cc(\otree\setminus\xtree) &\geq \sum_{\substack{S\subseteq\V\\S\odot (\otree\setminus\xtree)}}  \yp_S \tag{Lemma \ref{lm:tree-bound-by-activeset-num-cuts}}\\ &\geq \sum_{\substack{S\subseteq\V\\S\odot \otree \\\deltaS\cap\xtree=\emptyset}}  \yp_S \\
    &\geq
    \sum_{\substack{S\subseteq\V\\S\odot \otree \\\deltaS\cap\xtree=\emptyset}} \sum_{v\in\NLG} \rp_{S,v}. \tag{$\yp_S \geq \sum_{v\in v_1} \rp_{S,v}$ for any $V_1\subseteq \V$}
    \end{align*}
    Next, we can break down the outer summation into two parts by taking all sets that cut $\otree$ and subtracting those that also cut $\xtree$, resulting in
    \begin{align*}
    \cc(\otree\setminus\xtree) &\geq \sum_{\substack{S\subseteq\V\\S\odot \otree}} \sum_{v\in\NLG} \rp_{S,v} - \sum_{\substack{S\subseteq\V\\S\odot \xtree}} \sum_{v\in\NLG} \rp_{S,v}.
    \end{align*}
    By Lemma \ref{lm:ext-cut-lg}, no active set $S$ with $\rp_{S,v}>0$ for some vertex $v\in \NLG$ can cut exactly one edge of $\xtree$. Therefore, any active set $S$ contributing to the second summation with a non-zero summand $\sum_{v\in\NLG} \rp_{S,v}$ must cut at least two edges of $\xtree$. Additionally, any active set cutting at least two edges of $\xtree$ cuts at least two edges of $\otree$, and we have 
    \begin{align*}
\cc(\otree\setminus\xtree)&\geq\sum_{\substack{S\subseteq\V\\S\odot \otree}} \sum_{v\in\NLG} \rp_{S,v} - \sum_{\substack{S\subseteq\V\\\lvert\deltaS\cap\xtree\rvert>1}} \sum_{v\in\NLG} \rp_{S,v}\\
&\geq\sum_{\substack{S\subseteq\V\\S\odot \otree}} \sum_{v\in\NLG} \rp_{S,v} - \sum_{\substack{S\subseteq\V\\\lvert\deltaS\cap\otree\rvert>1}} \sum_{v\in\NLG} \rp_{S,v}. \tag{$\xtree\subseteq\otree$}
    \end{align*}
    Furthermore, any active set with $\sum_{v\in\NLG}\rp_{S,v}>0$ must include a vertex in $\NLG$, and no vertex in $\largest(\opti)$ by Lemma \ref{lm:ext-nlg}, and must therefore cut $\otree$, which connects $\NLG$ and $\largest(\opti)$. 
    Thus, we can simplify the bound for the first summation to get
    \begin{align*}    
    \cc(\otree\setminus\xtree)& \geq \sum_{S\subseteq\V} \sum_{v\in\NLG} \rp_{S,v} - \sum_{\substack{S\subseteq\V\\\lvert\deltaS\cap\otree\rvert>1}} \sum_{v\in\NLG} \rp_{S,v}\\
    &= \sum_{S\subseteq\V} \sum_{v\in\NLG} \rp_{S,v} - \X(\opti) \tag{Definition \ref{def:ext-X}}\\
    &\geq  \sum_{v\in\NLG}\sum_{S\subseteq\V} \rp_{S,v} - \X(\opti) \tag{Change order of summation}\\
    &=\sum_{V\in\NLG} \rp_v  - \X(\opti) \tag{$\sum_{S\subseteq\V}\rp_{S,v}=\rp_v$ for any $v$}\\
    &\geq \sum_{V\in\NLG} \rz_v  - \X(\opti). \tag{Corollary \ref{cor:ext-nlg-rz-less-rp}}
    \end{align*}
    Finally, the desired inequality is found by combining the bounds on $\cc(\xtree)$ and $\cc(\otree\setminus\xtree)$ since $\cc(\otree)=\cc(\otree\setminus\xtree) + \cc(\xtree)$.
\end{proof}

The following lemma includes our third bound for the cost of component $\opti \in \B$, which does not depend on $\X(\opti)$.
\begin{lemma}\label{lm:ext-simp-bound}
    For any component $\opti\in \B$, we have 
    \[
    \sum_{v\in\opti}\rz_v - \sum_{v\in\largest(\opti)}\rz_v \leq (1+\eps)\cc(\otree) -(1+\eps)\sum_{v\in\largest(\opti)}\rstar_v.
    \]
\end{lemma}
\begin{proof}
    We have
    \begin{align*}
        \sum_{v\in\opti}\rz_v - \sum_{v\in\largest(\opti)}\rz_v  
        &=\sum_{v\in\NLG} \rz_v\\
        &\leq\sum_{v\in\NLG} \rp_v\tag{Lemma \ref{cor:ext-nlg-rz-less-rp}}\\
        &\leq (1+\eps)\sum_{v\in\NLG} \rstar_v \tag{Corollary \ref{cor:ext-nlg-eps}}\\
        &=(1+\eps)\left(\sum_{v\in\opti}\rstar_v - \sum_{v\in\largest(\opti)}\rstar_v\right)\\
        &\leq (1+\eps)\cc(\otree)- (1+\eps)\sum_{v\in\largest(\opti)}\rstar_v \tag{Lemma \ref{lm:opt-lowerbound_rstar}}.
    \end{align*}
\end{proof}

The next set of lemmas helps us bound the value of $\X(\opti)$ for any component $\opti$ of $\OPT$.

\begin{lemma} \label{lm:ext-pot-satisfied}
    For any active set $S$ in the potential phase of Extended Execution,
    we have 
    \[
    \unsatisfied(S)=\emptyset.
    \]
\end{lemma}
\begin{proof}
    Assume otherwise that there exists a vertex $v \in S$ such that $\pairv \notin S$. 
    Let $\currenttime$ be a moment when $S$ is in the potential phase. 
    Since $S$ is in the potential phase, we must have \(\ttt_v \leq \currenttime\).
    Let $S'$ be the set containing $v$ in Legacy Execution at moment $\currenttime$.
    By Corollary \ref{cor:localsearch_activeset_refinement} and Corollary  \ref{cor:extend_refinement},  it follows that $S' \subseteq S$.
    However, since $\tplus_v \leq \ttt_v \leq \currenttime$, $S'$ must contain both $v$ and $\pairv$, because $v$ and $\pairv$ are connected by $\tplus_v$ in Legacy Execution. This implies that $\pairv \in S' \subseteq S$, which is a contradiction.
\end{proof}

\begin{lemma} \label{lm:ext-pot-grow}
    For any component $\opti$ of $\OPT$, any active set in the potential phase of Extended Execution cannot cut exactly one edge of $\otree$.
\end{lemma}
\begin{proof}
    Assume for contradiction that an active set $S$ in the potential phase at moment $\currenttime$ 
    cuts exactly one edge of $\otree$. By Lemma~\ref{lm:remove_one_edge}, 
    removing this single edge from $\otree$ can only disconnect demand pairs in $\opti$ 
    that are cut by $S$, which correspond precisely to $\unsatisfied(S)$. 
    However, since $S$ is in the potential phase, Lemma~\ref{lm:ext-pot-satisfied} implies that 
    $\unsatisfied(S) = \emptyset$.
    Therefore, removing this edge from $\OPT$ does not disconnect any demand pairs, yielding a valid solution with either a strictly lower cost or fewer edges than $\OPT$. This contradicts the definition of $\OPT$ as a minimal optimal solution (see Section~\ref{sec:prelim}).
\end{proof}

\begin{lemma} \label{lm:ext-epsx}
    For any component $\opti$ of $\OPT$, we have
    \[
    \X(\opti) \geq \eps \sum_{v \in \NLG} \rstar_v.
    \]
\end{lemma}
\begin{proof}
    This is trivially true if $\NLG$ is empty. Otherwise, consider any active set in the potential phase assigning growth to a vertex in $\NLG$ during Extended Execution. This active set cannot include any vertex in $\largest(\opti)$ by Lemma \ref{lm:ext-nlg}. Therefore, it must cut $\otree$.

    Moreover, by Lemma \ref{lm:ext-pot-grow}, it cannot cut exactly one edge of $\otree$. Thus, by Corollary~\ref{cor:ext-nlg-potphase-eps}, the growth assigned to vertices in $\NLG$ by sets in the potential phase is at least $\eps \sum_{v \in \NLG} \rstar_v$. Since all of this growth is counted in $\X(\opti)$, we conclude that
    \[
    \X(\opti) \geq \eps \sum_{v \in \NLG} \rstar_v.
    \]
\end{proof}

\paragraph{Overall Bounds.}
Building on the bounds for the cost of individual components of $\OPT$, we now introduce bounds for the total cost of $\OPT$. To enhance readability, we begin by defining the following notations. 
\begin{definition} \label{def:ext-short}
    Given the optimal solution $\OPT$, we define the following:
    \begin{align*}
        &\XX = \sum_{\opti \in \B} \X(\opti)\\
        &\lgr = \sum_{\opti \in \B} \sum_{v \in \largest(\opti)} \rstar_v\\
        &\lgrz = \sum_{\opti \in \B} \sum_{v \in \largest(\opti)} \rz_v.
    \end{align*}
\end{definition}
\begin{lemma} \label{lm:ext-bound1}
    We can bound the sum of $\yz$ values as
    $$
        \sum_{S\subseteq\V} \yz_S 
        \le (1 + \eps)(1 - \cone) 
        \cc(\OPTA) 
        + \cc(\OPTB) - 
        \XX + 
        \eps \lgr+ 
        \rmaxsum - \bta \losstwo + \lossone.
    $$
\end{lemma}
\begin{proof}
    We use Lemma \ref{lm:ext-a-bound} for components in $\A$ and Lemma \ref{lm:ext-x-bound} for components in $\B$. Then, summing these inequalities over all components of $\OPT$, we get
    \begin{align*}
        \sum_{\opti \in \A} \sum_{v\in\opti} \rp_v + \sum_{\opti \in \B} \left[\X(\opti) -\rp_{\max}(\largest(\opti)) + \sum_{v\in\opti} \rp_v\right] 
        &\leq \sum_{\opti \in \A} (1+\eps)(1-\cone)\cc(\otree) + \sum_{\opti\in\B} \cc(\otree)\\
        &= (1+\eps)(1-\cone)\cc(\OPTA) + \cc(\OPTB)
    \end{align*}
    Next, gather the sum $\sum_{v\in\V}\rp_v$ on the left-hand side and move the other terms to the right-hand side, which results in
    \begin{align*}
        \sum_{v\in \V} \rp_v &\leq (1+\eps)(1-\cone)\cc(\OPTA) + \cc(\OPTB) - \sum_{\opti\in\B} \X(\opti) +\sum_{\opti \in \B}\rp_{\max}(\largest(\opti))\\
        &\leq (1+\eps)(1-\cone)\cc(\OPTA) + \cc(\OPTB) - \XX +\sum_{\opti \in \B}\rp_{\max}(\largest(\opti)) \tag{Definition \ref{def:ext-short}}\\
        &\leq (1+\eps)(1-\cone)\cc(\OPTA) + \cc(\OPTB) -\XX +\sum_{\opti \in \B}(\rstar_{\max}(\opti) + \eps \sum_{v \in \largest(\opti)} \rstar_v) \tag{Lemma \ref{lm:ext-rpmax}}\\
        &\leq (1+\eps)(1-\cone)\cc(\OPTA) + \cc(\OPTB) -\XX + \eps \lgr + \sum_{\opti \in \B}\rstar_{\max}(\opti)\tag{Definition \ref{def:ext-short}}\\
        & \leq (1+\eps)(1-\cone)\cc(\OPTA) + \cc(\OPTB) -\XX + \eps \lgr + \sum_{\opti \in \B}\rr_{\max}(\opti) \tag{Lemma \ref{lm:rstar-smaller-rr}}.
    \end{align*}
    Additionally, we have
    \begin{align*}
        \sum_{S\subseteq\V} \yz_S &\leq \sum_{S\subseteq\V} \yp_S - \wintwo + \losstwo \tag{Lemma \ref{lm:ext-yzs}}\\&\leq \sum_{S\subseteq\V} \yp_S - (1+\bta)\losstwo + \losstwo\tag{Definition \ref{def:total-win-loss}}\\
        &\leq \sum_{v\in V} \rp_v - \bta\losstwo + \lossone\tag{Lemma \ref{lm:ext-yps}}.
    \end{align*}
    Finally, we can combine the inequalities to get
    \begin{align*}
        \sum_{S\subseteq\V} \yz_S
        & \leq (1+\eps)(1-\cone)\cc(\OPTA) + \cc(\OPTB) -\XX + \eps \lgr + \rmaxsum - \bta\losstwo + \lossone.
    \end{align*}
\end{proof}

\begin{lemma}\label{lm:ext-bound2}
    We can bound the sum of $\yz$ values as
\begin{align*} 
        \sum_{S\subseteq\V} \yz_S 
        &\le 
        (1 + \eps)(1 - \cone) \cc(\OPTA) 
        + \cc(\OPTB)- \frac{1-\bta}{5+\bta} \lgrz
        + \frac{6\eps}{5+\bta} 
        \lgr
        + \frac{6}{5+\bta} 
        \rmaxsum +  \XX + \lossone + \losstwo.
\end{align*}
\end{lemma}
\begin{proof}
    By Lemma \ref{lm:ext-yzrz}, we have
    \begin{align*}
        \sum_{S\subseteq\V} \yz_S \leq \sum_{v\in\V} \rz_v + \lossone + \losstwo.
    \end{align*}
    Therefore, to prove the desired bound, it suffices to show that 
    \begin{align*}
        \sum_{v\in\V} \rz_v \leq (1 + \eps)(1 - \cone) \cc(\OPTA) 
        + \cc(\OPTB)- \frac{1-\bta}{5+\bta} \lgrz
        + \frac{6\eps}{5+\bta} 
        \lgr
        + \frac{6}{5+\bta} 
        \rmaxsum +  \XX.
    \end{align*}
    We use Corollary \ref{cor:ext-a-rz-bound} for components in $\A$ and Lemma \ref{lm:ext-tree-bound} for components in $\B$. Then, summing these inequalities over all components of $\OPT$, we get
    \begin{align*}
        \sum_{\opti \in \A} \sum_{v\in\opti} \rz_v + \sum_{\opti \in \B} \left(\frac{6}{5+\bta} \left[\sum_{v\in\largest(\opti)}\rz_v - \rz_{\max}(\largest(\opti))\right] + \sum_{v\in\opti\setminus\largest(\opti)} \rz_v -  \X(\opti)  \right)
        &\leq (1+\eps)(1-\cone)\sum_{\opti \in \A}\cc(\otree) + \sum_{\opti \in \B}\cc(\otree)\\
        &\leq (1+\eps)(1-\cone)\cc(\OPTA) + \cc(\OPTB)
    \end{align*}
    
    Then, gathering the sum $\sum_{v\in V} \rz_v$ on the left-hand side and moving the other terms to right-hand side results in
    \begin{align*}
        \sum_{v\in V} \rz_v &\leq (1+\eps)(1-\cone)\cc(\OPTA) + \cc(\OPTB) - \frac{1-\bta}{5+\bta}\sum_{\opti \in \B} \sum_{v\in\largest(\opti)}\rz_v + \frac{6}{5+\bta}\sum_{\opti \in \B} \rz_{\max}(\largest(\opti)) + \sum_{\opti\in \B} \X(\opti)
    \end{align*}
    which can be rewritten as
    \begin{align*}
        \sum_{v\in V} \rz_v&\leq (1+\eps)(1-\cone)\cc(\OPTA) + \cc(\OPTB) - \frac{1-\bta}{5+\bta}\lgrz + \frac{6}{5+\bta}\sum_{\opti \in \B} \rz_{\max}(\largest(\opti)) + \XX
    \end{align*}
    using Definition \ref{def:ext-short}. Now, it suffices to show that
    \[
    \frac{6}{5+\bta}\sum_{\opti \in \B} \rz_{\max}(\largest(\opti)) 
    \leq \frac{6\eps}{5+\bta} 
        \lgr
        + \frac{6}{5+\bta} 
        \rmaxsum
    \]
    to prove our desired bound. This can be established as follows to complete the proof:
    \begin{align*}
        \frac{6}{5+\bta}\sum_{\opti \in \B} \rz_{\max}(\largest(\opti)) 
        &\leq \frac{6}{5+\bta}\sum_{\opti \in \B} \left[\rstar_{\max}(\opti) + \eps \sum_{v\in \largest(\opti)}\rstar_v \right]\tag{ Lemma \ref{lm:ext-rzmax}}\\
        &\leq \frac{6\eps}{5+\bta} \lgr +\frac{6}{5+\bta}\sum_{\opti \in \B} \rstar_{\max}(\opti)\tag{Definition \ref{def:ext-short}}\\
        &\leq \frac{6\eps}{5+\bta} \lgr +\frac{6}{5+\bta}\sum_{\opti \in \B} \rr_{\max}(\opti)\tag{Lemma \ref{lm:rstar-smaller-rr}}.
    \end{align*}
\end{proof}
\begin{lemma} \label{lm:ext-bound-3}
    We can bound the sum of $\yz$ values as
    \begin{align*}
        \sum_{S\subseteq V} \yz_S \leq (1 + \eps)(1 - \cone) \cc(\OPTA) 
        + (1+\eps) \cc(\OPTB) - (1+\eps)\lgr + \lgrz + \lossone + \losstwo.
    \end{align*}
\end{lemma}
\begin{proof}
    By Lemma \ref{lm:ext-yzrz}, we have
    \begin{align*}
        \sum_{S\subseteq\V} \yz_S \leq \sum_{v\in\V} \rz_v + \lossone + \losstwo.
    \end{align*}
    Therefore, to prove the desired bound, it suffices to show that 
    \begin{align*}
        \sum_{v\in\V} \rz_v \leq (1 + \eps)(1 - \cone) \cc(\OPTA) 
        + (1+\eps) \cc(\OPTB) - (1+\eps)\lgr + \lgrz.
    \end{align*}
    We use Corollary \ref{cor:ext-a-rz-bound} for components in $\A$ and Lemma \ref{lm:ext-simp-bound} for components in $\B$ and sum up the inequalities to arrive at the following:
    \begin{align*}
        \sum_{\opti\in\A}\sum_{v\in\opti}\rz_v + \sum_{\opti\in\B} \left(\sum_{v\in\opti}\rz_v - \sum_{v\in\largest(\opti)}\rz_v\right) 
        &\leq (1+\eps)(1-\cone)\sum_{\opti \in \A}\cc(\otree) + \sum_{\opti \in \B}\left((1+\eps)\cc(\otree) - (1+\eps)\sum_{v\in\largest(\opti)}\rstar_v\right)\\
        &= (1+\eps)(1-\cone)\cc(\OPTA) + (1+\eps) \cc(\OPTB) - (1+\eps) \sum_{\opti \in \B}\sum_{v\in\largest(\opti)}\rstar_v\\
        &= (1+\eps)(1-\cone)\cc(\OPTA) + (1+\eps) \cc(\OPTB) - (1+\eps) \lgr \tag{Definition \ref{def:ext-short}}.
    \end{align*}
    Then, gathering the sum $\sum_{v\in V} \rz_v$ on the left-hand side and moving the other terms to right-hand side results in
    \begin{align*}
        \sum_{v\in V}\rz_v &\leq (1+\eps)(1-\cone)\cc(\OPTA) + (1+\eps) \cc(\OPTB) - (1+\eps) \lgr + \sum_{\opti\in\B}\sum_{v\in\largest(\opti)}\rz_v\\
        &= (1+\eps)(1-\cone)\cc(\OPTA) + (1+\eps) \cc(\OPTB) - (1+\eps) \lgr + \lgrz \tag{Definition \ref{def:ext-short}},
    \end{align*}
    which completes the proof.
\end{proof}

Taking a weighted average of the bounds in Lemmas \ref{lm:ext-bound1} and \ref{lm:ext-bound2} with weights $1-\wvar$ and $\wvar$ for $0\leq \wvar \leq 1$ respectively, we arrive at the following bound.
\begin{corollary} \label{cor:ext-w}
    For any $0 \leq \wvar \leq 1$, we have
    \begin{align*}
        \sum_{S\subseteq V} \yz_S &\leq (1+\eps)(1-\cone)\cc(\OPTA) + \cc(\OPTB) - \wvar \frac{1-\bta}{5+\bta} \lgrz+ (2\wvar-1)\XX
        \\&\phantom{\leq}+ (1+\wvar\frac{1-\bta}{5+\bta})(\eps\lgr + \rmaxsum)  + \lossone + (\wvar-\bta + \wvar\bta) \losstwo.
    \end{align*}
\end{corollary}

Now, taking $\wvar \leq \frac{1}{2}$ and using the bound on $\X$ from Lemma \ref{lm:ext-epsx}, we arrive at our next bound.

\begin{lemma} \label{lm:ext-w-bound}
    For any $0 \leq \wvar \leq \frac{1}{2}$, we have
    \begin{align*}
        \sum_{S\subseteq V} \yz_S &\leq (1+\eps)(1-\cone)\cc(\OPTA) + (1 - (1-2\wvar)(1-\cone)\eps) \cc(\OPTB) - \wvar \frac{1-\bta}{5+\bta} \lgrz
        \\&\phantom{\leq}+ (2 - \wvar\frac{9+3\bta}{5+\bta})\eps\lgr +(1+\wvar\frac{1-\bta}{5+\bta})\rmaxsum  + \lossone + (\wvar-\bta + \wvar\bta) \losstwo.
    \end{align*}
\end{lemma}
\begin{proof}
    By Corollary \ref{cor:ext-w}, we have
    \begin{align*}
        \sum_{S\subseteq V} \yz_S &\leq (1+\eps)(1-\cone)\cc(\OPTA) + \cc(\OPTB) - \wvar \frac{1-\bta}{5+\bta} \lgrz+ (2\wvar-1)\XX
        \\&\phantom{\leq}+ (1+\wvar\frac{1-\bta}{5+\bta})(\eps\lgr + \rmaxsum)  + \lossone + (\wvar-\bta + \wvar\bta) \losstwo.
    \end{align*}
    Rewriting $\XX$ as a sum by Definition \ref{def:ext-short}, we can use Lemma \ref{lm:ext-epsx} to show that
    \begin{align*}
        \XX &= \sum_{\opti\in\B} \X(\opti)\\
        &\geq \sum_{\opti\in\B}\sum_{v\in\NLG} \eps\rstar_v\\
        &=\eps\sum_{\opti\in\B}\left[\sum_{v\in\opti} \rstar_v - \sum_{v\in\largest(\opti)} \rstar_v\right]\\
        &=\eps\sum_{\opti\in\B}\sum_{v\in\opti} \rstar_v - \eps\lgr \tag{Definition \ref{def:ext-short}}
        . 
    \end{align*}
    Additionally, for any $\opti\in \B$, \(
    \sum_{v\in\opti} \rstar_v \geq (1-\cone)\cc(\otree)
    \) by Definition \ref{def:classify_A_B}, which implies 
    \begin{align*}
    \XX &\geq \eps\sum_{\opti\in\B}(1-\cone)\cc(\opti) - \eps\lgr \\
        &= \eps(1-\cone)\cc(\OPTB) - \eps \lgr.
    \end{align*}
    Since $2\wvar -1\leq 0$, it follows that
    \begin{align*}
        (2\wvar -1)\XX &\leq (2\wvar -1)\eps(1-\cone)\cc(\OPTB) - (2\wvar-1)\eps \lgr.
    \end{align*}
    Combining this with the inequality from Corollary \ref{cor:ext-w} completes the proof.
\end{proof}

Next, we introduce a bound that does not include $\lgr$ and $\lgrz$ terms by taking a convex combination of Lemma \ref{lm:ext-w-bound} and Lemma $\ref{lm:ext-bound-3}$. 
\begin{lemma}\label{lm:ext-wp-bound}
    For any $\wvar\leq \frac{1}{2}$ and 
    $0\leq \wpvar \leq 1$ such that
    \(1/(1+\frac{1-\bta}{5+\bta}\wvar)\leq\wpvar\leq (1+\eps)/(1 + 3\eps - \eps\wvar\frac{9+3\bta}{5+\bta})
    ,\) we have
    \begin{align*}
        \sum_{S\subseteq V} \yz_S 
        &\leq (1+\eps)(1-\cone)\cc(\OPTA)
        + (1 + \eps - \eps\wpvar (1+(1-2\wvar)(1-\cone))) \cc(\OPTB) \\&\phantom{\leq}+\wpvar(1+\wvar\frac{1-\bta}{5+\bta})\rmaxsum  + \lossone +(1 - \wpvar(1+\bta)(1-\wvar)) \losstwo.
    \end{align*}
\end{lemma}
\begin{proof}
    Taking a convex combination of the bounds in Lemma $\ref{lm:ext-bound-3}$ and Lemma \ref{lm:ext-w-bound}, with $\wpvar$ as the coefficient for the second inequality results in the following bound:
    \begin{align*}
        \sum_{S\subseteq V} \yz_S
        &\leq (1+\eps)(1-\cone)\cc(\OPTA)
        + (1 + \eps - \eps\wpvar (1+(1-2\wvar)(1-\cone))) \cc(\OPTB) \\&\phantom{\leq}+\wpvar(1+\wvar\frac{1-\bta}{5+\bta})\rmaxsum  + \lossone+(1 - \wpvar(1+\bta)(1-\wvar)) \losstwo\\
        &\phantom{\leq}+(1-\wpvar - \wpvar\wvar \frac{1-\bta}{5+\bta}) \lgrz + \left(\wpvar\eps(2 - \wvar\frac{9+3\bta}{5+\bta})-(1-\wpvar)(1+\eps)\right)\lgr.
    \end{align*}
    It suffices for the coefficients of $\lgr$ and $\lgrz$ terms to be non-positive to arrive at our desired inequality. We can verify that the coefficient for $\lgrz$ is non-positive when
    \begin{align*}
    &1-\wpvar - \wpvar\wvar \frac{1-\bta}{5+\bta} \leq 0.
    \end{align*}
    This can be rearranged into
    \[
     \wpvar  \geq 1/(1+\frac{1-\bta}{5+\bta}\wvar),
    \]
    observing that $1+\frac{1-\bta}{5+\bta}\wvar$ is always positive.
    
    On the other hand, for the coefficient of $\lgr$ to be non-positive, we must have
    \begin{align*}
        &\wpvar\eps(2 - \wvar\frac{9+3\bta}{5+\bta})-(1-\wpvar)(1+\eps) \leq 0.
    \end{align*}
    This is equivalent to
        \[
        \wpvar(1 + 3\eps - \eps\wvar\frac{9+3\bta}{5+\bta}) \leq (1+\eps).
        \]
    Noting that $1+3\eps-\eps\wvar\frac{9+3\bta}{5+\bta}$  is always positive since $\eps\wvar\frac{9+3\bta}{5+\bta} \leq \frac{1}{2}\cdot3\cdot\eps <3\eps$, we find the following upper bound for $\wpvar$:
    \[
    \wpvar \leq (1+\eps)/(1 + 3\eps - \eps\wvar\frac{9+3\bta}{5+\bta}).
    \]
    Therefore, under the condition that $1/(1+\frac{1-\bta}{5+\bta}\wvar)\leq\wpvar\leq (1+\eps)/(1 + 3\eps - \eps\wvar\frac{9+3\bta}{5+\bta})$, the terms containing $\lgr$ and $\lgrz$ vanish and we arrive at the desired inequality.
\end{proof}

We use the next two lemmas to upper bound the $\lossone$ and $\losstwo$ terms in Lemma \ref{lm:ext-wp-bound}.

 \begin{lemma} \label{lm:ext-losstwo-bound}
 If $\solext$ is not a $2-2\alf$ approximation for the optimal solution $\OPT$, then we must have
     \[\losstwo <= (\alf + \eps)/\bta \cdot \cc(\OPT) - (1 + \eps + \eps/\bta) \cdot \lossone.\]
 \end{lemma}

\begin{proof}
By Lemma \ref{lm:monotonic-forest-twice-growth}, the solution found by the algorithm after the local search has cost at most
\begin{align*}
\cc(\solext)&\leq 2\sum_{S\subseteq \V} \yz_S\\
&=2\sum_{S\subseteq \V} \yp_S - 2\wintwo + 2\losstwo \tag{ Lemma \ref{lm:ext-yzs}}\\&\leq 2\sum_{v\in \V}\rp_v -2\wintwo + 2\lossone + 2\losstwo \tag{ Lemma \ref{lm:ext-yps}}.
\end{align*}
Furthermore, we know by Corollary \ref{cor:ext-comp-eps} that $\sum_{v\in \V} \rp_v \leq \sum_v(1+\eps)\rstar_v$. Then, we can say 
\begin{align*}
    \cc(\solext)&\leq 2\sum_{v\in \V}(1+\eps) \rstar_v + 2\lossone -2\wintwo +2\losstwo\\
    &\leq 2 (1 + \eps)\cc(\OPT) - 2(1+\eps)\winone + 2\lossone -2\wintwo +2\losstwo \tag{Lemma~\ref{lm:sum-rstar-less-copt-minus-win}}\\
    &\leq 2 (1 + \eps)\cc(\OPT) - 2(1+\eps)(1+\bta)\lossone + 2\lossone -2\wintwo +2\losstwo \tag{Definition \ref{def:total-win-loss}}\\
    &\leq 2 (1 + \eps)\cc(\OPT) - 2(1+\eps)(1+\bta)\lossone + 2\lossone -2(1+\bta) \losstwo +2\losstwo \tag{Definition \ref{def:total-win-loss}}\\
    &= 2 (1 + \eps)\cc(\OPT) - 2(\eps+\bta + \eps\bta)\lossone - 2\bta\losstwo.
\end{align*}
Finally, since we assumed that $\cc(\solext)\geq (2-2\alf)\cc(\OPT)$, we must have
\begin{align*}
    &2 (1 + \eps)\cc(\OPT) - 2(\eps+\bta + \eps\bta)\lossone - 2\bta\losstwo \geq (2-2\alf)\cc(\OPT)
\end{align*}
and therefore
\begin{align*}
    \losstwo \leq (\alf + \eps)/\bta \cdot \cc(\OPT) - (1+\eps+\eps/\bta)\cdot\lossone.
\end{align*}
\end{proof}

\begin{lemma} \label{lm:ext-lscoeff}
    If neither $\solone$ nor $\solext$ is a $2-2\alf$ approximation for a given instance with optimal solution $\OPT$, then for any coefficient $\losscoeff$, we have
    \[\lossone + \losscoeff\cdot\losstwo \leq 
    \max\left(
    \frac{\alf}{\bta},
    \frac{\losscoeff(\alf+\eps)}{\bta}
    ,\frac{\alf+\losscoeff(1-\alf-\alf/\bta)\eps}{\bta}\right)\cc(\OPT).
    \]
\end{lemma}
\begin{proof}
    We begin by observing that $\lossone$ and $\losstwo$ are always non-negative.
    If $\losscoeff<0$,  we have
    \begin{align*}
        \lossone + \losscoeff\cdot\losstwo &\leq \lossone \leq \frac{\alf}{\bta}\cc(\OPT)
    \end{align*}
    by Lemma \ref{lm:lcl-loss}. Otherwise, we can use Lemma \ref{lm:ext-losstwo-bound} and get
    \begin{align*}
        \lossone + \losscoeff\cdot\losstwo &\leq \lossone + \losscoeff\left(\frac{\alf+\eps}{\bta}\cc(\OPT) - (1+\eps+\eps/\bta)\lossone\right)\\
        &\leq \frac{\losscoeff(\alf+\eps)}{\bta}\cc(\OPT)+ (1-\losscoeff(1+\eps+\eps/\bta))\lossone.
    \end{align*}
    Then, either $1-\losscoeff(1+\eps+\eps/\bta)<0$, in which case we have
    \begin{align*}
        \lossone + \losscoeff\cdot\losstwo 
        &\leq \losscoeff\frac{\alf+\eps}{\bta}\cc(\OPT) + (1-\losscoeff(1+\eps+\eps/\bta))\lossone\\
        &\leq \frac{\losscoeff(\alf+\eps)}{\bta}\cc(\OPT)
    \end{align*}
    or $1-\losscoeff(1+\eps+\eps/\bta)\geq 0$, in which case we can again use Lemma \ref{lm:lcl-loss} to get
    \begin{align*}
        \lossone + \losscoeff\cdot\losstwo 
        &\leq \losscoeff\frac{\alf+\eps}{\bta}\cc(\OPT) + (1-\losscoeff(1+\eps+\eps/\bta))\lossone\\
        &\leq \frac{\losscoeff(\alf+\eps)}{\bta}\cc(\OPT)+ (1-\losscoeff(1+\eps+\eps/\bta))\frac{\alf}{\bta}\cc(\OPT)\\
        &=\frac{\alf+\losscoeff(1-\alf-\alf/\bta)\eps}{\bta}\cc(\OPT).
    \end{align*}
    Since one of these three cases must apply, the maximum of the three values is always an upper bound for $\lossone+\losscoeff\cdot\losstwo$.

\end{proof}

\begin{lemma}\label{lm:ext-coeff-bound}
    For any $\wvar\leq \frac{1}{2}$ and 
    $0\leq \wpvar \leq 1$ such that
    \(1/(1+\frac{1-\bta}{5+\bta}\wvar)\leq\wpvar\leq (1+\eps)/(1 + 3\eps - \eps\wvar\frac{9+3\bta}{5+\bta}),\) either the following bound on $\cc(\solext)$ holds
        \begin{align*}
            \cc(\solext)&\leq 2\biggl((1+\eps)(1-\cone)\cc(\OPTA)
            + (1 + \eps - \eps\wpvar(1+(1-2\wvar)(1-\cone))) \cc(\OPTB) \\
            &\phantom{\leq}+\wpvar(1+\wvar\frac{1-\bta}{5+\bta})\rmaxsum  \\
            &\phantom{\leq}+\max(\frac{\alf}{\bta},\frac{\left(1 - \wpvar(1+\bta)(1-\wvar)\right)(\alf+\eps)}{\bta},\frac{\alf+(1-\alf-\alf/\bta)\left(1 - \wpvar(1+\bta)(1-\wvar)\right)\eps}{\bta})\cc(\OPT)\biggr).
        \end{align*}
    or one of $\solone$ and $\solext$ is a $2-2\alf$ approximation of the optimal solution $\OPT$.
\end{lemma}

\begin{proof}
    We know by Lemma \ref{lm:monotonic-forest-twice-growth} that $\cc(\solext)\leq 2\sum_{S\subseteq \V} \yz_S$. Then, it follows from Lemma \ref{lm:ext-wp-bound} that
    \begin{align*}
        \cc(\solext)&\leq 2\biggl((1+\eps)(1-\cone)\cc(\OPTA)
        + (1 + \eps - \eps\wpvar(1 + (1-2\wvar)(1-\cone))) \cc(\OPTB) \\&\phantom{\leq}+\wpvar(1+\wvar\frac{1-\bta}{5+\bta})\rmaxsum  + \lossone + (1 - \wpvar(1+\bta)(1-\wvar)) \losstwo\biggr).
    \end{align*}
    Now, using Lemma \ref{lm:ext-lscoeff} with $\losscoeff=1 - \wpvar(1+\bta)(1-\wvar)$ completes the proof.
\end{proof}

\begin{lemma}\label{lm:ext-final-bound}
    For any $\wvar\leq \frac{1}{2}$ and 
    $0\leq \wpvar \leq 1$ such that
    \(1/(1+\frac{1-\bta}{5+\bta}\wvar)\leq\wpvar\leq (1+\eps)/(1 + 3\eps - \eps\wvar\frac{9+3\bta}{5+\bta}),\) either the following bound on $\cc(\solext)$ holds
        \begin{align*}
            \cc(\solext)&\leq 2\biggl((1+\eps)(1-\cone)\cc(\OPTA)
            + (1 + \eps - \eps\wpvar(1 + (1-2\wvar)(1-\cone))) \cc(\OPTB) \\
            &\phantom{\leq}+ \cfive\wpvar(1+\wvar\frac{1-\bta}{5+\bta}) \cc(\OPTBone)+\wpvar(1+\wvar\frac{1-\bta}{5+\bta})\sum_{\opti\in\Btwo}\rmax(\opti)  \\&\phantom{\leq}+\max(\frac{\alf}{\bta},\frac{\left(1 - \wpvar(1+\bta)(1-\wvar)\right)(\alf+\eps)}{\bta},\frac{\alf+(1-\alf-\alf/\bta)\left(1 - \wpvar(1+\bta)(1-\wvar)\right)\eps}{\bta})\cc(\OPT)\biggr).
        \end{align*}
    or one of $\solone$ and $\solext$ is a $2-2\alf$ approximation of the optimal solution $\OPT$.
\end{lemma}
\begin{proof}
    We start with the bound given in Lemma \ref{lm:ext-coeff-bound}. We can decompose $\rmaxsum$ into \[\sum_{\opti\in\Bone}\rmax(\opti) + \sum_{\opti\in\Btwo}\rmax(\opti).\] Now, by Definition \ref{def:classify_B}, for any set $\opti\in\Bone$, we have $\rmax(\opti)\leq \cfive\cc(\otree)$.
    Observing that the coefficient of this term is always non-negative, we can apply this upper bound for every component $\opti$ in $\Bone$ to derive the desired bound.
\end{proof}

\section{Autarkic Pairs}
\label{sec:candidate}
In this section, we introduce a new approach designed specifically for the Steiner Forest problem, in contrast to the more general local search method discussed earlier.
This method centers on the notion of \emph{autarkic pairs}, which we identify and connect directly, while the remaining demands are handled by Legacy Moat Growing.
Each autarkic pair consists of two subsets of vertices such that the vertices in one subset are the pairs of those in the other, and both subsets are expected to belong to the same connected component of the optimal solution.
Within each subset, the vertices connect to each other significantly earlier than the time it takes for the two subsets to connect.

The motivation behind detecting and directly connecting autarkic pairs is that, in some connected components of the optimal solution, certain vertices may be merged into the same group early in a moat growing algorithm, but it may take a long time for these groups to connect to each other. 
This delay increases the maximum assigned value for the vertices in that component, leading to a large bound in Lemma~\ref{lm:ext-coeff-bound}. 
As a result, $\solext$ may not perform well.

To address this issue, we attempt to identify such connected components in the optimal solution. 
In these components, instead of allowing the last two active sets to grow further to reach each other, we select and add the shortest path between a pair of vertices, each belonging to one of these two active sets, to our solution. 
This reduces the maximum assigned value of that connected component since the last two active sets do not need to grow to reach each other, as they can use the selected shortest path.

Since we do not have access to the optimal solution, we make rational guesses about these pairs to ensure our selection covers all relevant ones.
However, it may include pairs that do not belong to the intended connected components of the optimal solution, as we lack explicit knowledge of it.
Nonetheless, we carefully select autarkic pairs to ensure that false guesses do not interfere with our analysis.

In Section~\ref{sec:candidate_algorithm}, we describe the \candidate{} procedure and provide its pseudocode. 
Then, in Section~\ref{sec:candidate_properties}, we identify properties for connected components in $\Btwo$ and show that they have autarkic pairs. 
Finally, in Section~\ref{sec:candidate_analysis}, we use these properties to analyze the solution produced by executing \candidate{} in Line~\ref{exe:candid}, given that the algorithm previously calls \LocalSearch{} in Line~\ref{exe:localsearch} and we are finding autarkic pairs based on Boosted Execution.

\subsection{Algorithm}
\label{sec:candidate_algorithm}
In this section, we describe the \candidate{} procedure, with its pseudocode provided in Algorithm~\ref{alg:candidate}. This method aims to identify specific pairs of subsets of vertices, referred to as autarkic pairs, and connect them directly. See also Figure~\ref{fig:ten} for an illustration of how this procedure works.

To formally define the selection of autarkic pairs, we first introduce the following notation.

\begin{definition}
\label{def:dis}
    Recall that for any vertices $v, u \in \V$, $\dis(v, u)$ denotes their distance in the graph $G$.
    For any subset of vertices $S \subseteq \V$, we define $\maxdisS$ as the maximum distance between a vertex $v \in S$ and its pair:
    \[
        \maxdisS = \max_{v \in S} \dis(v, \pairv).
    \]
\end{definition}

First, for each set that was active at any point during Boosted Execution, we accumulate its growth duration $\yys$ into $\YUS$.
At the end, $\YS$ represents the total time during which the vertices in $S$ were the only unsatisfied vertices in an active set.
Importantly, if there exists a vertex $v \in S$ such that $\pairv \in S$, then $\YS = 0$, as no set $S' \subseteq \V$ satisfies $\unsatisfied(S') = S$ in this case.

Then, for any set $S \neq \emptyset$ with $\YS > 0$, we compare $\YS + \YPS$ to $\maxdisS$, using the threshold parameter $\ccnd \in (0,1)$ provided by the main algorithm.
If the condition in Line~\ref{line:autarkic_condition} is satisfied, we select $S$ and $\pairS$ as an autarkic pair.
Next, we choose an arbitrary vertex $\vap \in S$, add the shortest path between $\vap$ and its pair $\pair_{\vap}$ to the solution, and insert a zero-cost edge between them in $\EAP$, since they are now connected.

Finally, we run \BBG{} on $\Gp = \G \cup \EAP$ to satisfy the remaining demands.
The vertices in autarkic pairs are connected using the paths selected earlier, while the remaining pairs are handled by \BBG{}.
Note that, instead of \BBG{}, other Steiner Forest algorithms—such as Algorithm~\ref{alg:main}, which achieves a better approximation factor—could be used on the new instance. 
However, for simplicity in our analysis, we continue to use \BBG{}.

\begin{algorithm}[ht]
  \caption{Autarkic Pairs}
  \label{alg:candidate}
  \hspace*{\algorithmicindent} \textbf{Input:} 
  A graph $\G=(\V, \E, \cc)$ with edge costs $\cc: \E \rightarrow \mathbb{R}_{\ge 0}$, demand function $\pair : \V \to \V$,
  a function $\yy: 2^{\V} \rightarrow \mathbb{R}_{\ge 0}$ indicating the growth of active sets in Boosted Execution, and parameter $0 < \ccnd < 1$.\\
  \hspace*{\algorithmicindent} \textbf{Output:} 
  A forest $F$ that satisfies all demands. 
  \begin{algorithmic}[1]
    \Procedure{$\candidater$}{$\CInsExp$}
      \label{func:candid}
      \State $F \gets \emptyset$
      \State $\EAP \gets \emptyset$
      \State Implicitly set $\YS \gets 0$ for all $S \subseteq \V$
      \For{$S \subseteq \V$} \label{line:autarkic_for_yys}
        \State $\YUS \gets \YUS + \yys$ \label{line:define_YS}
      \EndFor
      \For{$S \subseteq \V, S \neq \emptyset$} \label{line:autarkic_for_YS}
        \If{$(1+\ccnd)\left(\YS + \YPS\right) > \maxdisS$}  \label{line:autarkic_condition}
          \State Select arbitrary vertex $\vap \in S$
          \State Add the shortest path between $\vap$ and $\pair_{\vap}$ to $F$ \label{line:autarkicpair_add_edge}
          \State Insert a zero-cost edge $(\vap, \pair_{\vap})$ into $\EAP$ \label{line:contract_vertices}
        \EndIf
      \EndFor
      \State $\Gp \gets \G \cup \EAP$ \label{line:define_Gp}
      \State $F \gets F \cup \BBG(\Gp, \pair) \setminus \EAP$ \label{line:autarkic_legacy} 
      \State \Return $F$
    \EndProcedure
  \end{algorithmic}
\end{algorithm}

Note that the for loop in Line~\ref{line:autarkic_for_yys} only needs to consider subsets $S \subseteq \V$ with $\yys > 0$, and the for loop in Line~\ref{line:autarkic_for_YS} only needs to consider subsets $S \subseteq \V$ with $\YS + \YPS > 0$. Since the number of such sets is polynomial, and by applying Corollary~\ref{cor:legacy_polynomial_time}, we obtain the following corollary.

\begin{corollary}
    \label{cor:autarkic_pair_polynomial_time}
    The \candidate{} procedure runs in polynomial time.
\end{corollary}

Moreover, if $S$ is selected as an autarkic pair, we should not check $\pairS$ as another subset in the for loop in Line~\ref{line:autarkic_for_YS}.
This prevents selecting a pair twice as an autarkic pair. 

Furthermore, in Line~\ref{exe:candid}, we pass $\yyb$ instead of $\yy$ to \candidate{}. This does not affect the outcome, since for sets with nonempty unsatisfied vertex sets, we have $\yyb = \yy$. Additionally, as we ignore $\Y_{\emptyset}$ (see Line~\ref{line:autarkic_for_YS}), we can assume that the input to \candidate{} is effectively $\yy$.

To analyze our algorithm, we first need to identify properties of the connected components of the optimal solution in $\Btwo$.
These properties demonstrate that these connected components contain autarkic pairs, allowing us to bound $\solcnd$.

\subsection
[Properties of Connected Components in $\Btwo$]
{Properties of Connected Components in $\boldsymbol{\Btwo}$}
\label{sec:candidate_properties}
In this section, we establish several properties for the connected components of the optimal solution in the class $\Btwo$.
These properties rely on the fact that \LocalSearch{} has already been executed before invoking \candidate{}, and that the values $\yys$ in the input to \candidate{} represent the growth of active sets from Boosted Execution, which does not involve any valuable boost actions.

Recall that the class $\Btwo$, introduced in Definition~\ref{def:classify_B}, depends on the parameter $\cfive$, whose exact form was not specified earlier.
Here, we define $\cfive$ together with another parameter $\ctwo$, both expressed in terms of $\bta$, $\cone$, and $\ccnd$.
Their exact numerical values will be fixed in Section~\ref{sec:final}, once the remaining parameters have been determined.

\begin{definition}
\label{def:ctwo_cfive}
    We define two constants $\ctwo$ and $\cfive$ based on $\beta$, $\cone$, and $\ccnd$ as follows:
    \begin{align*}
        \ctwo &= 4\cone \dfrac{1+\bta}{1-\bta}, \\
        \cfive &= \dfrac{\cone}{\ccnd} \left(4 + 3\ccnd + 4 (1 + \ccnd) \dfrac{1+\bta}{1 - \bta}\right).
    \end{align*}
\end{definition}

To simplify the presentation, we focus on a fixed connected component $\opti \in \Btwo$ and denote its corresponding tree in the optimal solution by $T_{\opti}$.
The following notations and definitions are specific to this component, starting with the selection of the two highest-priority vertices within $\opti$.

\begin{definition}
\label{def:vone_vtwo}
    We define $\vone$ as the vertex with the maximum priority in $\opti$, and $\vtwo$ as the pair of $\vone$.
    That is,
    \begin{align*}
        \vone &= \argmax_{v \in \opti} \priority_v\\
        \vtwo &= \pair_{\vone}.
    \end{align*} 
\end{definition}

Note that $\vtwo \in \opti$, since $\vone \in \opti$ and is connected to its pair $\vtwo$ in any valid solution.
One straightforward observation from this definition is that $\vone$ has the maximum assigned value as stated formally below.

\begin{lemma}
\label{lm:vone_has_rmax}
    Vertex $\vone$ has the maximum assigned value for $\rplus$ and $\rr$. That means
    $$\rplus_{\vone} = \rplus_{\max}(\opti)$$
    $$\rr_{\vone} = \rmax(\opti).$$
\end{lemma}

\begin{proof}
    Given Definition~\ref{def:vone_vtwo}, $\vone$ has the maximum priority in $\opti$.
    Therefore, we can use Corollary~\ref{cor:priority-tplus} to conclude that $\tplus_{\vone}$ is the maximum value between $\tplus$ of vertices in $\opti$. 
    That means for any vertex $v\in \opti$ we have
    \begin{align}
    \label{eq:tvone_tv}
        \tplus_{\vone} \ge \tplus_v.
    \end{align}

    Let us fix any time $\currenttime < \tplus_{\vone}$.
    According to Corollary~\ref{cor:tt-ge-tplus}, we also have $\currenttime < \ttt_{\vone}$.
    If $S$ is the connected component containing $\vone$ in Legacy Execution (and similarly in Boosted Execution), then $S$ is an active set since $\vone$ remains in an active set until time $\tplus_{\vone}$ (and similarly $\ttt_{\vone}$), as per Definition~\ref{def:fingerprint} and the fact that $\tplus$ (similarly $\ttt$) is the fingerprint of the execution.
    Furthermore, we have $\vone \in \BaseSet(S, \currenttime)$.
    Therefore, $\vone \in \priorityset(\BaseSet(S, \currenttime))$ because $\max_{u \in \BaseSet(S, \currenttime) \cap \opti} \priority_u \le \max_{u \in \opti} \priority_u = \priority_{\vone}$.
    Thus, we can conclude that we are assigning to $\vone$ at any time $\currenttime < \tplus_{\vone}$.
    Using Lemmas~\ref{lm:rplus_prefix_time_assignment} and~\ref{lm:rr_prefix_time_assignment}, we know that both $\rplus$ and $\rr$ are prefix-time assignments, respectively.
    Therefore, we can conclude that 
    \begin{align}
    \label{eq:r_vone_tplus_vone}
        \rplus_{\vone}, \rr_{\vone} \ge \tplus_{\vone}.
    \end{align}

    Moreover, for any vertex $v \in \opti$ and anytime $\currenttime \ge \tplus_v$ and any subset of vertices $S \subseteq \V$, we have $v \notin \BaseSet(S, \currenttime)$, which means, $v \notin \priorityset(\BaseSet(S, \currenttime))$.
    Therefore we can conclude that
    \begin{align}
    \label{eq:r_v_tplus_v}
        \rplus_v, \rr_v \le \tplus_v.
    \end{align}

    By combining Equations~\ref{eq:tvone_tv},~\ref{eq:r_vone_tplus_vone}, and~\ref{eq:r_v_tplus_v}, for any vertex $v\in \opti$ we have
    \begin{align*}
        \rplus_{\vone} \ge \tplus_{\vone} \ge \tplus_v \ge \rplus_v \\
        \rr_{\vone} \ge \tplus_{\vone} \ge \tplus_v \ge \rr_v
    \end{align*}
    which proves the lemma.
\end{proof}

Given Lemma~\ref{lm:rplus_vone_rr_vone}, we have $\rplus_{\max}(\opti) = \rmax(\opti)$.  
Using the above lemma, we can conclude the following corollary.

\begin{corollary}
\label{cor:rplus_vone_rr_vone}
    The assignment to $\vone$ is the same in both $\rplus$ and $\rr$.  
    In other words,  
    $$
        \rplus_{\vone} = \rr_{\vone}.
    $$
\end{corollary}

Now, we aim to identify important properties that demonstrate the presence of autarkic pairs within connected component $\opti \in \Btwo$.
The first property, detailed in Lemma~\ref{lm:bound_rone-rtwo}, shows that $\rr_{\vone}$ and $\rr_{\vtwo}$ have almost identical values.
Then, in Lemma~\ref{lm:bound_rr_vthree}, we demonstrate that other vertices of $\opti$, which are actively connected to $\vone$ or $\vtwo$, connect with them early in Boosted Execution.
Next, Lemma~\ref{lm:bound_twoopt} reveals that most of the growth time of active sets containing $\vone$ and $\vtwo$ does not involve unsatisfied vertices from other connected components of the optimal solution.
These three properties lead to the identification of two subsets of vertices, one containing $\vone$ and the other containing $\vtwo$.
We establish a lower bound for the total $\Y$ of these subsets in Lemma~\ref{lm:bound_YS_candidate}, provide an upper bound for the $\maxdis$ of these two subsets in Lemma~\ref{lm:bound_dis_candidate}, and finally demonstrate that they are autarkic pairs in Lemma~\ref{lm:Btwo_has_candidate}.
To facilitate the proof of these properties and lemmas, we present auxiliary lemmas throughout to simplify the process.
Here is one such auxiliary lemma.

\begin{lemma}
\label{lm:rplus_vtwo_rr_vone}
    We can relate the assigned values of $\vtwo$ and $\vone$ for Legacy Execution and Boosted Execution as follows: 
    $$\rplus_{\vtwo} = \rr_{\vone}.$$
\end{lemma}
\begin{proof}
    Based on Corollary~\ref{cor:rplus_vone_rr_vone}, we have $\rplus_{\vone} = \rr_{\vone}$.
    It remains to show that $\rplus_{\vtwo} = \rplus_{\vone}$, which completes the proof.

    To prove that $\rplus_{\vtwo} = \rplus_{\vone}$, first note that $\tplus$ returned by \BBG{} indicates the time at which each vertex connects to its pair, which is the same for $\vone$ and $\vtwo$. In other words, $\tplus_{\vone} = \tplus_{\vtwo}$.
    Additionally, by Lemma~\ref{lm:rplus_prefix_time_assignment}, $\rplus$ is a prefix-time assignment.
    Therefore, if we show that for any time $\currenttime_- < \tplus_{\vone}$ both vertices are assigned, and for any time $\currenttime_+ > \tplus_{\vone}$ neither is assigned, it follows that these two vertices receive the same $\rplus$ value.

    For $\currenttime_+ > \tplus_{\vone}$, both $\vone$ and $\vtwo$ are not members of $\BaseSet(S, \currenttime_+)$ for any subset $S \subseteq \V$, and consequently, are not members of $\priorityset(\BaseSet(S, \currenttime_+))$.

    For $\currenttime_- < \tplus_{\vone}$, these two vertices have not yet reached each other. Thus, each active set containing one of them does not contain the other. Let $S \subseteq \V$ be a subset of vertices containing one of them at time $\currenttime_-$. 
    We know that this vertex is also in $\BaseSet(S, \currenttime_-)$, since $\currenttime_- < \tplus_{\vone}, \tplus_{\vtwo}$. 
    Moreover, the vertex is in $\priorityset(\BaseSet(S, \currenttime_-))$. This is immediate for $\vone$, as it has the maximum priority in $\opti$.
    For $\vtwo$, note that since $\vone$ has the maximum priority in $\opti$, its $\tplus$ is also maximal among vertices in $\opti$. Given that $\tplus_{\vtwo} = \tplus_{\vone}$ and by Definition~\ref{def:priority}, $\vtwo$ has the second-highest priority in $\opti$. 
    Therefore, if $\vtwo \in S$ at time $\currenttime_-$, which implies $\vone \notin S$, then $\vtwo$ is the vertex with the maximum priority from $\opti$ in $S$, and thus belongs to $\priorityset(\BaseSet(S, \currenttime_-))$.

    This completes the proof that $\rplus_{\vtwo} = \rplus_{\vone}$, and consequently proves the lemma.
\end{proof}

Now we provide the first important property, showing that $\rr_{\vone}$ and $\rr_{\vtwo}$ have roughly the same value.

\begin{lemma}
\label{lm:bound_rone-rtwo}
    The difference between the assigned values $\rr$ of vertices $\vone$ and $\vtwo$ can be bounded by
    $$\rr_{\vone} - \rr_{\vtwo} \le \cone \cc(\otree).$$
\end{lemma}
\begin{proof}
    We prove this by showing that $\rr_{\vone} - \rr_{\vtwo}$ is a lower bound on the difference between the total $\rplus$ of vertices in $\opti$ and their total $\rr$, and then show that $\cone \cc(T_{\opti})$ is an upper bound for this value.
    \begin{align*}
        \rr_{\vone} - \rr_{\vtwo} &= \rplus_{\vtwo} - \rr_{\vtwo} \tag{Lemma~\ref{lm:rplus_vtwo_rr_vone}} \\
        &\le \rplus_{\vtwo} - \rr_{\vtwo} + \sum_{v\in \opti \setminus \{\vtwo\}} (\rplus_v - \rr_v) \tag{Lemma~\ref{lm:rr-le-rplus}} \\
        &= \sum_{v\in \opti} \rplus_v - \sum_{v\in \opti} \rr_v \\
        &\le \cc(\otree) - (1-\cone) \cc(\otree) \tag{Lemma~\ref{lm:upper_bound_sum_rplus}, Corollary~\ref{corollary:B_bound_r}} \\
        &= \cone \cc(T_{\opti})
    \end{align*}
\end{proof}

The next three lemmas establish connections between $\extraboostedotree$ and the total assigned value $\rr$ of vertices in $\opti$, and derive bounds on $\extraboostedotree$.
Here, $\extraboosted$ refers to $\extra$ as defined in Definition~\ref{def:uncolored_single_multi_colored}, with respect to Boosted Execution. 
These lemmas play a key role in the proofs of Lemmas~\ref{lm:bound_rr_vthree} and~\ref{lm:bound_twoopt}.

\begin{lemma}
\label{lm:bound_rr_with_ys}
    For any subset of vertices $S' \subseteq \opti$ with $\vone \notin S'$, we have
    $$\sum_{v \in S'} \rr_v \le \sum_{\substack{S \subseteq \V \\ S \cap S' \neq \emptyset \\ \vone \notin S}} \yys.$$
\end{lemma}
\begin{proof}
    First, we show that each $\yys$ value from sets on the right-hand side can be assigned to at most one $\rr_v$ of a vertex on the left-hand side. Then we show that only $\yys$ values from those sets can be assigned to $\rr_v$ of vertices in $S'$, completing the proof.

    For any time $\currenttime$ and any subset of vertices $S \subseteq \V$, according to Lemma~\ref{lm:opti_cap_priorityset}, we have
    $$(\priorityset(\BaseSet(S, \currenttime)) \cap \opti) \in \{\emptyset, \{\argmax_{v \in S \cap \opti} \priority_v\}\}.$$
    Since $S' \subseteq \opti$, we conclude that
    $$|\priorityset(\BaseSet(S, \currenttime)) \cap S'| \leq 1,$$ 
    meaning the value of $\yys$ can contribute to $\rr_v$ of at most one vertex in $S'$.

    Furthermore, for any subset of vertices $S \subseteq \V$ and time $\currenttime$, if $|\priorityset(\BaseSet(S, \currenttime)) \cap S'| = 1$, then
    $$(\priorityset(\BaseSet(S, \currenttime)) \cap S') = \{\argmax_{v \in S \cap \opti} \priority_v\}.$$
    This happens only if $\argmax_{v \in S\cap \opti} \priority_v$ is in $S'$.
    Therefore, $(S\cap \opti) \cap S' \neq \emptyset$, and consequently $S \cap S' \neq \emptyset$.
    Additionally, since $\vone$ has the maximum priority in $\opti$, if $\vone \in S$, then $\argmax_{v \in S\cap \opti} \priority_v = \vone \notin S'$, contradicting the assignment to $S'$.
    Thus, a $\yys$ value can be assigned to $\rr_v$ of some vertex $v \in S'$ only if $S\cap S'\neq \emptyset$ and $\vone \notin S$.
\end{proof}

\begin{lemma}
\label{lm:bound_r_with_coloring_opti}
    The total assigned value $\rr$ to vertices in $\opti$ is bounded by
    $$ \sum_{v \in \opti} \rr_v \le (\rr_{\vone} - \rr_{\vtwo}) + \cc(\otree) - \dfrac{\extraboostedotree}{2}.$$
\end{lemma}
\begin{proof}
    We begin by bounding $\rr_{\vone}$ using Lemma~\ref{lm:bound_rr_with_ys} with $S' = \{\vtwo\}$:
    \begin{align*}
        \rr_{\vone} &= \rr_{\vtwo} + \left(\rr_{\vone} - \rr_{\vtwo}\right)\\
        &\le \sum_{\substack{S\subseteq \V \\ \vtwo \in S \\ \vone \notin S}} \yys + \left(\rr_{\vone} - \rr_{\vtwo}\right). \tag{Lemma~\ref{lm:bound_rr_with_ys}}
    \end{align*}
    Since $\vone$ and $\vtwo$ form a pair, they remain in active sets until they are connected in Legacy Execution. As a consequence, Corollary~\ref{cor:localsearch_activeset_refinement} ensures that $\vone$ and $\vtwo$ also remain active until they connect in Boosted Execution.
    Consequently, at any moment in Boosted Execution, a connected component containing $\vtwo$ but not $\vone$ is active, and another component containing $\vone$ but not $\vtwo$ is also active. 
    It follows that $\sum_{\substack{S\subseteq \V \\ \vtwo \in S \\ \vone \notin S}} \yys  = \sum_{\substack{S\subseteq \V \\ \vone \in S \\ \vtwo \notin S}} \yys$, which combining by above inequality lead to
    \begin{align}
    \label{eq:bound_rr_vone}
        \rr_{\vone} \le \sum_{\substack{S\subseteq \V \\ \vone \in S \\ \vtwo \notin S}} \yys + \left(\rr_{\vone} - \rr_{\vtwo}\right).
    \end{align}

    Next, we bound the contribution of all vertices in $\opti$ except for $\vone$.
    By Lemma~\ref{lm:bound_rr_with_ys}, we have
    \begin{align}
    \label{eq:bound_rr_opti_except_vone}
        \sum_{\substack{v \in \opti \\ v\neq \vone}} \rr_v \le \sum_{\substack{S\subseteq \V \\ S\cap \opti \neq \emptyset \\ \vone \notin S}} \yys.
    \end{align}

    We now define the following families of sets:
    \begin{align*}
        \mathcal{S}_1 &= \{S \subseteq \V \mid \vone \in S, \vtwo \notin S\}, \\
        \mathcal{S}_2 &= \{S \subseteq \V \mid S \cap \opti \neq \emptyset, \vone \notin S\}, \\
        \mathcal{S}_3 &= \{S \subseteq \V \mid S \odot \opti\}.
    \end{align*}
    Observe that $\mathcal{S}_1$ and $\mathcal{S}_2$ are disjoint, since each set in $\mathcal{S}_1$ contains $\vone$ while each set in $\mathcal{S}_2$ does not. Furthermore, every set in $\mathcal{S}_1 \cup \mathcal{S}_2$ cuts $\opti$, and thus $(\mathcal{S}_1 \cup \mathcal{S}_2) \subseteq \mathcal{S}_3$. These observations imply
    $$\sum_{S\in \mathcal{S}_1}\yys + \sum_{S\in \mathcal{S}_2}\yys \le \sum_{S\in \mathcal{S}_3}\yys,$$
    which, in expanded form, is
    \begin{align}
    \label{eq:sum_ys_of_sets_cut_opti}
        \sum_{\substack{S\subseteq \V \\ \vone \in S \\ \vtwo \notin S}} \yys + \sum_{\substack{S\subseteq \V \\ S\cap \opti \neq \emptyset \\ \vone \notin S}} \yys \le \sum_{\substack{S\subseteq \V \\ S\odot \opti}} \yys.
    \end{align}
    
    Finally, we combine all these inequalities to complete the proof.
    \begin{align*}
        \sum_{v \in \opti} \rr_v &= \rr_{\vone} + \sum_{\substack{v \in \opti \\ v\neq \vone}} \rr_v \\
        &\le \left(\sum_{\substack{S\subseteq \V \\ \vone \in S \\ \vtwo \notin S}} \yys + \left(\rr_{\vone} - \rr_{\vtwo}\right)\right) + \sum_{\substack{S\subseteq \V \\ S\cap \opti \neq \emptyset \\ \vone \notin S}} \yys \tag{Equations~\ref{eq:bound_rr_vone} and~\ref{eq:bound_rr_opti_except_vone}}\\
        &\le (\rr_{\vone} - \rr_{\vtwo}) + \sum_{\substack{S\subseteq \V \\ S\odot \opti}} \yys \tag{Equation~\ref{eq:sum_ys_of_sets_cut_opti}}\\
        &\le (\rr_{\vone} - \rr_{\vtwo}) + \cc(\otree) - \dfrac{\extraboostedotree}{2}. \tag{Lemma~\ref{lm:ys_le_T_extra}}
    \end{align*}
\end{proof}

\begin{lemma}
\label{lm:extra_coloring}
    The quantity $\extraboostedotree$ can be bounded as follows:
    $$\extraboostedotree \le 4\cone \cc(\otree).$$
\end{lemma}
\begin{proof}
    The bound follows easily from previous lemmas:
    \begin{align*}
        (1-\cone) \cdot \cc(\otree) &< \sum_{v\in \opti} \rr_v \tag{$\opti \in \B$, Corollary~\ref{corollary:B_bound_r}}\\
        &\le (\rr_{\vone} - \rr_{\vtwo}) + \cc(\otree) - \dfrac{\extraboostedotree}{2} \tag{Lemma~\ref{lm:bound_r_with_coloring_opti}}\\
        &\le (1+\cone) \cdot \cc(\otree) - \dfrac{\extraboostedotree}{2}. \tag{Lemma~\ref{lm:bound_rone-rtwo}}
    \end{align*}
    After reordering terms, we have
    \begin{align*}
        \extraboostedotree \le 4\cone \cc(\otree).
    \end{align*}
\end{proof}

Using Definition~\ref{def:ctwo_cfive}, we show a lower bound on the value of $\rr_{\vtwo}$.

\begin{lemma}
\label{lm:rr_vtwo_lower_bound}
    The value $\rr_{\vtwo}$ satisfies the following lower bound:
    $$\rr_{\vtwo} > \ctwo \cc(\otree).$$
\end{lemma}
\begin{proof}
    \begin{align*}
        \rr_{\vtwo} &\ge \rr_{\vone} - \cone \cc(\otree) \tag{Lemma~\ref{lm:bound_rone-rtwo}} \\
        &= \rmax(\opti) - \cone \cc(\otree) \tag{Lemma~\ref{lm:vone_has_rmax}} \\
        &> (\cfive - \cone) \cc(\otree) \tag{Definition~\ref{def:classify_B}}
    \end{align*}
    Therefore, it suffices to show that $\cfive - \cone \ge \ctwo$. Using Definition~\ref{def:ctwo_cfive}, this inequality reduces to:
    $$\dfrac{\cone}{\ccnd} \left(4 + 3\ccnd + 4 (1 + \ccnd)\dfrac{1+\bta}{1 - \bta}\right) - \cone \ge 4\cone \dfrac{1+\bta}{1-\bta}.$$
    Multiplying both sides by the positive factor $\frac{\ccnd}{\cone}$, we obtain:
    $$4 + 3\ccnd + 4 (1 + \ccnd)\dfrac{1+\bta}{1 - \bta} - \ccnd \ge 4\ccnd \dfrac{1+\bta}{1-\bta},$$
    which simplifies to:
    $$4 + 2\ccnd + 4\dfrac{1+\bta}{1 - \bta} \ge 0.$$
    This inequality clearly holds, as all parameters are positive and $\bta < 1$.
\end{proof}

\begin{definition} 
\label{def:vthree} 
    Let $\iota_u$ denote the moment when a vertex $u$ becomes connected to either $\vone$ or $\vtwo$, and let $S$ be the set of vertices $u \in \opti \setminus \{\vone, \vtwo\}$ for which $\iota_u$ is defined—that is, $u$ connects to one of them—and $u$ is unsatisfied just before $\iota_u$ (regardless of whether it becomes satisfied at $\iota_u$). Then, we define
    \[
    \vthree = \argmax_{u \in S} \iota_u.
    \]
    Intuitively, $\vthree$ is the last vertex in $\opti \setminus \{\vone, \vtwo\}$ that was not previously connected to either $\vone$ or $\vtwo$, and then becomes connected to one of them for the first time while still unsatisfied just before that moment in Boosted Execution.
    
    If no such vertex exists, we assume $\vthree$ to be a dummy vertex in $\opti$ that is connected to $\vone$ by a zero-cost edge. 
\end{definition}

Next, we aim to use the above lemmas to establish another important property, which gives an upper bound for $\rr_{\vthree}$.

\begin{lemma}
\label{lm:bound_rr_vthree}
    The value of $\rr_{\vthree}$ can be bounded as 
    $$
    \rr_{\vthree} \le \ctwo \cc(\opti).
    $$
\end{lemma}
\begin{proof}
    We prove this by contradiction, assuming that $\rr_{\vthree} > \ctwo \cc(\opti)$.
    
    Since $\vone$, $\vtwo$, and $\vthree$ belong to $\opti$, they are all connected in the optimal solution. 
    This means the subgraph of the optimal solution induced by these three vertices and the paths between them forms a star with three leaves, with a vertex $q \in \opti$ as the center. 
    We refer to this star as $\textit{ST}$.
    Note that it is possible that $q$ coincides with one of these vertices.
    In this case, we can assume that they are distinct vertices connected by a zero-cost edge.

    Let $\currenttime$ represent the first moment that two of these vertices are connected, and $\currenttime'$ the first time that all three are connected in Boosted Execution. 
    Given Lemma~\ref{lm:vone_has_rmax}, we have $\rr_{\vone} \ge \rr_{\vtwo}$. 
    Moreover, by Lemma~\ref{lm:rr_vtwo_lower_bound}, we have $\rr_{\vtwo} > \ctwo \cc(\opti)$, and by the contradiction assumption, $\rr_{\vthree} > \ctwo \cc(\opti)$.
    Since the $\rr$ values of all these vertices are greater than $\ctwo \cc(\opti)$, using Lemma~\ref{lm:connect_after_min_r}, we can conclude that no pair among these three vertices connects until time $\ctwo \cc(\opti)$. 
    Therefore,
    \begin{align}
    \label{eq:t_t'_ctwo_opti}
        \currenttime \ge \min(\rr_{\vone}, \rr_{\vtwo}, \rr_{\vthree}) > \ctwo \cc(\opti).
    \end{align} 

    Next, we aim to contradict the above inequality.

    Since $\vone$ and $\vtwo$ are paired, they are actively connected in Legacy Execution, which implies they are also actively connected in Boosted Execution (by Lemma~\ref{lm:large_fingerprint_superactive}).
    Additionally, according to Definition~\ref{def:vthree}, $\vthree$ remains unsatisfied before connecting to the other two vertices.
    Thus, by Corollary~\ref{cor:localsearch_activeset_refinement}, we can similarly conclude that $v_3$ is actively connected to $v_1$ and $v_2$.
    Therefore, we can apply Lemma~\ref{lm:xyz} as follows:
    \begin{align}
    \label{eq:ST_bound_for_add_new_ball}
        \currenttime + \currenttime'  \leq \frac{\stxp+\styp+\stzp}{2} + \bta\frac{\min(\stxp,\styp,\stzp)}{2},
    \end{align}
    where $\stxp$, $\styp$, and $\stzp$ are the distances from $\vone$, $\vtwo$, and $\vthree$ to $q$, respectively.

    Next, we bound $\stxp+\styp+\stzp$ and $\min(\stxp,\styp,\stzp)$ based on $\currenttime$, $\currenttime'$, and $\extraboostedotree$.

    Let $S' = \{\vone, \vtwo, \vthree\}$ be the set of leaves of $\textit{ST}$. 
    We aim to show that $\sum_{S \odot S'} \yys \le \currenttime + 2\currenttime'$.
    In Boosted Execution, until time $\currenttime$, there are three connected components containing at least one of these vertices, contributing at most $3\currenttime$ to the sum.
    After time $\currenttime$, two vertices become connected, forming a single component. 
    Between $\currenttime$ and $\currenttime'$, this contributes an additional $2(\currenttime' - \currenttime)$ to the sum.
    After $\currenttime'$, all vertices are in the same component, and the sum remains unchanged.
    Hence,
    \begin{align}
    \label{eq:bound_yss_cutting_leaves}
        \sum_{\substack{S \subseteq \V \\ S \odot S'}} \yys \le 3\currenttime + 2(\currenttime' - \currenttime) = \currenttime + 2\currenttime'.
    \end{align}

    Now, since $\textit{ST}$ is a subgraph of $\otree$, Lemma~\ref{lm:extra_subgraph} implies:
    \begin{align*}
        \extraboostedotree &\ge \extraboosted(\textit{ST}) \tag{Lemma~\ref{lm:extra_subgraph}}\\
        &\ge \cc(\textit{ST}) - \sum_{\substack{S\subseteq \V \\ S\odot S'}} \yys \tag{Lemma~\ref{lm:coloring_exactly_one_edge}}\\
        &\ge \stxp+\styp+\stzp - (\currenttime + 2\currenttime'). \tag{Equation~\ref{eq:bound_yss_cutting_leaves}}
    \end{align*}
    Rearranging gives:
    \begin{align}
    \label{eq:bound_sum_x_y_z}
        \stxp+\styp+\stzp \le \currenttime + 2\currenttime' + \extraboostedotree.
    \end{align}

    Next, we bound $\min(\stxp,\styp,\stzp)$.  
    We know that two vertices of $S'$ are connected at time $\currenttime$ in Boosted Execution.  
    Without loss of generality, let us assume that these are $\vone$ and $\vtwo$.  
    We refer to the path between these vertices in the optimal solution as $\textit{PT}$.  
    We know that $\cc(\textit{PT}) \ge \stxp+\styp$.
    
    Moreover, since these two vertices reach each other at time $\currenttime$, the total $\yys$ of active sets containing exactly one of $\vone$ and $\vtwo$ is at most $2\currenttime$.  
    This means
    \begin{align}
        \label{eq:bound_yss_cutting_vone_vtwo}
        \sum_{\substack{S\subseteq \V \\ S\odot \{\vone,\vtwo\}}} \yys \le 2\currenttime.
    \end{align}
    
    Since $\textit{PT}$ is a subgraph of $\otree$, we can use Lemma~\ref{lm:extra_subgraph} to conclude
    \begin{align*}
        \extraboostedotree &\ge \extraboosted(\textit{PT}) \tag{Lemma~\ref{lm:extra_subgraph}}\\
        &\ge \cc(\textit{PT}) - \sum_{\substack{S\subseteq \V \\ S\odot \{\vone,\vtwo\}}} \yys \tag{Lemma~\ref{lm:coloring_exactly_one_edge}}\\
        &\ge \stxp+\styp - 2\currenttime \tag{Equation~\ref{eq:bound_yss_cutting_vone_vtwo}}.
    \end{align*}
    
    By simple reordering, we obtain:
    \[
        2\currenttime + \extraboostedotree \ge \stxp+ \styp \ge 2\min(\stxp, \styp).
    \]
    
    Therefore,
    \[
        \currenttime + \extraboostedotree \ge \min(\stxp,\styp),
    \]
    and it follows that
    \begin{align}
        \label{eq:bound_min_x_y_z}
        \min(\stxp,\styp,\stzp) \le \currenttime + \extraboostedotree.
    \end{align}

    Substituting inequalities from \eqref{eq:bound_sum_x_y_z} and \eqref{eq:bound_min_x_y_z} into \eqref{eq:ST_bound_for_add_new_ball} yields:
    \begin{align*}
        \currenttime + \currenttime' &\le \dfrac{\stxp+\styp+\stzp}{2} + \bta \dfrac{\min(\stxp,\styp,\stzp)}{2} \tag{Equation~\ref{eq:ST_bound_for_add_new_ball}} \\
        &\le  \dfrac{\currenttime + 2\currenttime' + \extraboostedotree}{2} + \bta \dfrac{\currenttime + \extraboostedotree}{2}. \tag{Equations~\ref{eq:bound_sum_x_y_z} and \ref{eq:bound_min_x_y_z}}
    \end{align*}
    Rearranging this results in
    \begin{align*}
        \currenttime &\le \dfrac{1+\bta}{1-\bta}\extraboostedotree \\
        &\le 4\cone \dfrac{1+\bta}{1-\bta} \cc(\otree) \tag{Lemma~\ref{lm:extra_coloring}}\\
        &= \ctwo \cc(\otree) \tag{Definition~\ref{def:ctwo_cfive}}
    \end{align*}
    This contradicts inequality \eqref{eq:t_t'_ctwo_opti}, completing the proof.
\end{proof}

\begin{definition}
\label{def:twoopt}
    Let $\twoopt$ denote the total growth of active sets that contain unsatisfied vertices from $\opti$ as well as vertices outside of it. Specifically, 
    $$\twoopt = \sum_{\substack{S\subseteq \V \\ \unsatisfied(S) \cap \opti \neq \emptyset \\ \unsatisfiedS \nsubseteq \opti}} \yys.$$
    Note that according to Line~\ref{line:define_YS}, this is equivalent to the following:
    $$\twoopt = \sum_{\substack{S\subseteq \V \\ S \cap \opti \neq \emptyset \\ S \nsubseteq \opti}} \YS.$$
\end{definition}

First, we establish a relation between the total assignment to vertices in $\opti$ under $\rr$ and $\rstar$ using $\twoopt$. 
Then, we use this relation to bound $\twoopt$ in Lemma~\ref{lm:bound_twoopt}, which plays a key role in proving that $\opti$ contains an autarkic pair.

\begin{lemma}
\label{lm:rr_to_rstar}
    We can bound the sum of $\rstar$ values for vertices in $\opti$ in terms of their $\rr$ values and $\twoopt$ as follows:
    $$\sum_{v \in \opti}\rstar_v \le \sum_{v \in \opti}\rr_v - \dfrac{\twoopt}{2}.$$
\end{lemma}
\begin{proof}
    First, we define two families of subsets of vertices:
    \begin{align*}
        \mathcal{S}_1 &= \{S \subseteq \V \mid \unsatisfiedS \cap \opti \neq \emptyset, \unsatisfiedS \nsubseteq \opti\}, \\
        \mathcal{S}_2 &= \{S \subseteq \V \mid \unsatisfiedS \cap \opti = \emptyset \text{ or } \unsatisfiedS \subseteq \opti\}.
    \end{align*}
    Note that
    \begin{align}
    \label{eq:family_S_1_S_2}
        \mathcal{S}_1 \cap \mathcal{S}_2 = \emptyset \quad \text{and} \quad \mathcal{S}_1 \cup \mathcal{S}_2 = 2^{\V}.
    \end{align}
    Intuitively, we partition all sets into those that contribute to $\twoopt$ (the sets in $\mathcal{S}_1$) and those that do not (the sets in $\mathcal{S}_2$). 
    For sets in $\mathcal{S}_1$, their growth is fully assigned to the $\rr$ values of vertices in $\opti$. 
    However, due to the structure of $\rstar$ assignments, each such set contributes at most half of that growth to the $\rstar$ values of these vertices. 
    This gap between the full assignment in $\rr$ and the partial assignment in $\rstar$ is what yields the factor of $\twoopt / 2$ in the inequality.

    For any set $S \in \mathcal{S}_1$ active at time $\currenttime$ during Boosted Execution (i.e., $S \in \Aa_{\currenttime}$), we have $\unsatisfiedS \cap \opti \neq \emptyset$. Therefore, there exists a vertex $v \in S \cap \opti$ with $\pairv \notin S$ (by Definition~\ref{def:unsatisfied}).  
    In Legacy Execution, vertex $v$ and $\pairv$ reach each other at time $\tplus_v$ (see Definition~\ref{def:legacy_execution}). Moreover, by Corollary~\ref{cor:localsearch_activeset_refinement}, they are in the same connected component at time $\tplus_v$ in Boosted Execution.  
    Since in Boosted Execution they are not yet connected at time $\currenttime$, but are connected by time $\tplus_v$, it follows that $\currenttime < \tplus_v$.
    This implies $v \in \BaseSet(S, \currenttime)$ (see Definition~\ref{def:base-phase}), and hence $\BaseSet(S, \currenttime) \cap \opti \neq \emptyset$.
    Applying Lemma~\ref{lm:opti_cap_priorityset} then gives:  
    \begin{align}
    \label{eq:S_1_contain_opti}
        |\priorityset(\BaseSet(S, \currenttime))\cap \opti| = 1.
    \end{align}
    By Definition~\ref{def:rr}, the growth from $S$ is assigned entirely to exactly one vertex in $\opti$, meaning:
    \begin{align}
    \label{eq:r_s_v_y_s_for_S_1}
        \quad \sum_{v \in \opti} \rr_{S, v} = \yys \quad \text{ for all } S \in \mathcal{S}_1.
    \end{align}
    
    Similarly, for $S \in \mathcal{S}_1$, since $\unsatisfiedS \nsubseteq \opti$, there exists $v \in S \setminus \opti$ with $\pairv \notin S$. Using the same reasoning, $v \in \BaseSet(S, \currenttime)$, and hence $\BaseSet(S, \currenttime) \setminus \opti \neq \emptyset$. 
    Consequently, by Lemma~\ref{lm:opti_cap_priorityset}, we obtain that $|\priorityset(\BaseSet(S, \currenttime))\setminus \opti| \ge 1$.
    Combining this with Equation~\ref{eq:S_1_contain_opti}, we conclude that
    $$|\priorityset(\BaseSet(S, \currenttime))| \ge 2.$$
    By Definition~\ref{def:rstar}, the growth from such sets is distributed across multiple components, with vertices in $\opti$ collectively receiving a fraction of the growth equal to $1/|\priorityset(\BaseSet(S, \currenttime))|$. Therefore:
    \begin{align}
    \label{eq:rstar_s_v_y_s_for_S_1}
          \sum_{v \in \opti} \rstar_{S, v} \leq \dfrac{1}{2}\yys \quad \text{ for all } S \in \mathcal{S}_1.
    \end{align}

    We can now complete the proof:
    \begin{align*}
        \sum_{v \in \opti}\rr_v - \dfrac{\twoopt}{2} &= \sum_{v \in \opti}\sum_{S \subseteq \V} \rr_{S, v} - \dfrac{1}{2} \sum_{S\in \mathcal{S}_1} \yys \tag{Definitions~\ref{def:assignment} and~\ref{def:twoopt}} \\
        &= \sum_{S\in \mathcal{S}_1} \left( \sum_{v \in \opti}\rr_{S,v} - \dfrac{1}{2}\yys\right) + \sum_{S\in \mathcal{S}_2} \sum_{v \in \opti}\rr_{S,v} \tag{Equation~\ref{eq:family_S_1_S_2}} \\
        &= \sum_{S\in \mathcal{S}_1} \left( \yys - \dfrac{1}{2}\yys\right) + \sum_{S\in \mathcal{S}_2} \sum_{v \in \opti}\rr_{S,v} \tag{Equation~\ref{eq:r_s_v_y_s_for_S_1}} \\
        &\ge \sum_{S\in \mathcal{S}_1} \dfrac{1}{2}\yys + \sum_{S\in \mathcal{S}_2} \sum_{v \in \opti}\rstar_{S,v} \tag{Lemma~\ref{lm:rstar-smaller-rr}} \\
        &\ge \sum_{S\in \mathcal{S}_1} \sum_{v \in \opti}\rstar_{S,v} + \sum_{S\in \mathcal{S}_2} \sum_{v \in \opti}\rstar_{S,v} \tag{Equation~\ref{eq:rstar_s_v_y_s_for_S_1}}\\
        &= \sum_{v \in \opti}\rstar_v. \tag{Equation~\ref{eq:family_S_1_S_2}}
    \end{align*}
\end{proof}

\begin{lemma}
\label{lm:bound_twoopt}
    We can bound the value of $\twoopt$ as follows:
    $$\twoopt \le 4 \cone \cc(\otree).$$
\end{lemma}
\begin{proof}
    The proof follows from the previous lemmas as outlined below:
    \begin{align*}
        (1-\cone) \cc(\otree) &< \sum_{v \in \opti} \rstar_v \tag{$\opti \in \B$, Definition~\ref{def:classify_A_B}}\\
        &\le \sum_{v \in \opti}\rr_v  - \dfrac{\twoopt}{2}\tag{Lemma~\ref{lm:rr_to_rstar}} \\
        &\le (\rr_{\vone} - \rr_{\vtwo}) + \cc(T_{\opti})  - \dfrac{\twoopt}{2} \tag{Lemma~\ref{lm:bound_r_with_coloring_opti}, $\extraboostedotree\ge 0$}\\
        &\le (1 + \cone) \cc(\otree)  - \dfrac{\twoopt}{2}\text{.} \tag{Lemma~\ref{lm:bound_rone-rtwo}} 
    \end{align*}
    Reordering the terms results in
    \begin{align*}
        \twoopt \le 4\cone \cc(\otree).
    \end{align*}
\end{proof}

\begin{definition}
\label{def:Sone_Stwo}
    Let $\currenttime_{\vthree}$ denote the moment immediately after $\rr_{\vthree}$, meaning all merges and deactivations at time $\rr_{\vthree}$ have already occurred in Boosted Execution.  
    Let $\sone', \stwo' \in \Aa_{\currenttime_{\vthree}}$ be the active sets containing $\vone$ and $\vtwo$, respectively, at this moment.  
    We define subsets $\sone, \stwo \subseteq \V$ as the sets of unsatisfied vertices in $\opti$ that lie in $\sone'$ and $\stwo'$, respectively.  
    That is,  
    \begin{align*}
        \sone &= \unsatisfied(\sone') \cap \opti, \\
        \stwo &= \unsatisfied(\stwo') \cap \opti.
    \end{align*}
\end{definition}

We use the notation $\sone$, $\stwo$, $\sone'$, and $\stwo'$, and first show that $\sone$ and $\stwo$ form pairs.  
Subsequently, in Lemma~\ref{lm:bound_YS_candidate}, we provide a lower bound for $\Y_{\sone} + \Y_{\stwo}$, and in Lemma~\ref{lm:bound_dis_candidate}, we give an upper bound for $\maxdis(\sone)$.  
Finally, using the value of $\cfive$ from Definition~\ref{def:ctwo_cfive} and the autarkic condition in Line~\ref{line:autarkic_condition} of \candidate, we conclude in Lemma~\ref{lm:Btwo_has_candidate} that $\sone$ and $\stwo$ are selected as autarkic pairs.
 
\begin{lemma}
\label{lm:pair_sone_stwo}
    Sets $\sone$ and $\stwo$ are pairs of each other, meaning
    $$
        \pair(\sone) = \stwo.
    $$
\end{lemma}
\begin{proof}
    Recall the moment $\currenttime_{\vthree}$ from Definition~\ref{def:Sone_Stwo}.  
    Let $u \in \sone$. We claim that $\pair_u \in \stwo$ at this moment.

    Since $u$ is unsatisfied at time $\currenttime_{\vthree}$, $\pair_u$ cannot belong to $\sone'$.  
    On the other hand, $\pair_u$ must eventually be connected to $u$ and thus to $\vone$ by the end of the algorithm.  
    By Definition~\ref{def:vthree}, $\pair_u$ cannot become connected to $\vone$ or $\vtwo$ for the first time after $\currenttime_{\vthree}$, as it is unsatisfied at this moment.  
    Therefore, $\pair_u$ must already be in $\stwo'$ at moment $\currenttime_{\vthree}$.

    Moreover, since $u$ is unsatisfied, so is $\pair_u$, meaning $\pair_u \in \unsatisfied(\stwo')$.  
    And since $u \in \sone \subseteq \opti$, we also have $\pair_u \in \opti$.  
    Hence, $\pair_u \in \unsatisfied(\stwo') \cap \opti = \stwo$.

    By symmetry, for any $u \in \stwo$, we similarly have $\pair_u \in \sone$.  
    This completes the proof.
\end{proof}

\begin{lemma}
\label{lm:bound_YS_candidate}
    There is a lower bound for the value that leads us to consider $S_1$ and $S_2$ as an autarkic pair as follows:
    $$\Y_{S_1} + \Y_{S_2} \ge 2\rr_{\vone} - (6\cone + 2\ctwo)\cc(\otree).$$
\end{lemma}
\begin{proof}
    By Lemma~\ref{lm:connect_after_min_r}, the vertices $\vone$ and $\vtwo$ do not reach each other until time $\min(\rr_{\vone}, \rr_{\vtwo})$, which is equal to $\rr_{\vtwo}$ by Lemma~\ref{lm:vone_has_rmax}, in Boosted Execution.
    Moreover, since these vertices are pairs of each other, they are actively connected in Legacy Execution and, by Lemma~\ref{lm:large_fingerprint_superactive}, also in Boosted Execution.
    Therefore, from the beginning of Boosted Execution until time $\rr_{\vtwo}$, they are both unsatisfied and belong to different active sets, which implies
    $$\sum_{\substack{S\subseteq \V\\S\odot \{\vone, \vtwo\}}} \YS = \sum_{\substack{S\subseteq \V\\ \unsatisfiedS\odot \{\vone, \vtwo\}}} \yys \ge 2\rr_{\vtwo}.$$
    
    Furthermore, since $\rr_{\vthree}$ is the last time an unsatisfied vertex from $\opti$ reaches either $\vone$ or $\vtwo$, after time $\rr_{\vthree}$ no additional unsatisfied vertex from $\opti$ reaches them. 
    This means that for any time $\currenttime \in (\rr_{\vthree}, \rr_{\vtwo})$, if $S'_{\currenttime}$ is the active set containing $\vone$ at time $\currenttime$ and $S''_{\currenttime} = \unsatisfied(S'_{\currenttime}) \cap \opti$, then $S''_{\currenttime} \subseteq S_1$.
    Additionally, no vertex in $S_1$ or $S_2$ becomes satisfied during the interval $(\rr_{\vthree}, \rr_{\vtwo})$, because as mentioned earlier, $\vone$ and $\vtwo$ do not meet until time $\rr_{\vtwo}$, meaning $\sone$ and $\stwo$ do not meet either.
    By Lemma~\ref{lm:pair_sone_stwo}, the pair of every vertex in $\sone$ lies in $\stwo$ and vice versa, so they cannot become satisfied until they meet. 
    Therefore, $S_1 \subseteq S''_{\currenttime}$.
    We conclude that $S''_{\currenttime} = S_1$, which shows that throughout the interval $(\rr_{\vthree}, \rr_{\vtwo})$, the unsatisfied vertices of $\opti$ in the active set containing $\vone$ are exactly $\sone$.
    A symmetric argument holds for $\vtwo$ and $\stwo$.
    Thus,
    \begin{align}
    \label{eq:Y_S1_S2_with_to}
        \sum_{\substack{S\subseteq \V \\ S\cap \opti \in \{S_1, S_2\}}} \YS \ge 2 \rr_{\vtwo} - 2 \rr_{\vthree}.
    \end{align}
    
    Moreover, by Definition~\ref{def:twoopt}, the total amount of time during which a set $S$ containing unsatisfied vertices from $\opti$ also includes vertices from other connected components of the optimal solution is at most $\twoopt$. Therefore,
    \begin{align*}
        \Y_{S_1} + \Y_{S_2} &\ge \sum_{\substack{S\subseteq \V \\ S\cap \opti \in \{S_1, S_2\}}} \YS - \sum_{\substack{S\subseteq \V \\ S\cap \opti \in \{S_1, S_2\} \\ S\nsubseteq \opti}} \YS \\
        &\ge (2 \rr_{\vtwo} - 2 \rr_{\vthree}) - \sum_{\substack{S\subseteq \V \\ S\cap \opti \neq \emptyset \\ S\nsubseteq \opti}} \YS \tag{ Equation~\ref{eq:Y_S1_S2_with_to}} \\
        &\ge \left(2\rr_{\vone} - 2(\rr_{\vone} - \rr_{\vtwo}) - 2 \rr_{\vthree}\right) - \twoopt \tag{ Definition~\ref{def:twoopt}} \\
        &\ge 2\rr_{\vone} - 2\cone \cc(\opti) - 2\ctwo \cc(\opti) - 4\cone \cc(\opti) \tag{ Lemmas~\ref{lm:bound_rone-rtwo},~\ref{lm:bound_rr_vthree}, and~\ref{lm:bound_twoopt}}\\
        &= 2\rr_{\vone} - (6\cone + 2\ctwo)\cc(\opti).
    \end{align*}
\end{proof}

The next lemmas establish an upper bound on $\extralegacyotree$ in terms of the cost of $\otree$.
Here, $\extralegacy$ refers to $\extra$ as defined in Definition~\ref{def:uncolored_single_multi_colored}, with respect to Legacy Execution. 

\begin{lemma}
\label{lm:legacy_extra_coloring}
    The following inequality holds:
    $$
    \extralegacyotree \le 2\cone \cc(\otree).
    $$
\end{lemma}
\begin{proof}
    First, note that for any moment $\currenttime$ and any active set $S \subseteq \V$ at that moment in Legacy Execution, we have  
    $|\priorityset(\BaseSet(S, \currenttime)) \cap \opti| \leq 1$
    by Lemma~\ref{lm:opti_cap_priorityset}.
    This implies that the value of $\yplus_S$ can contribute to $\rplus_v$ for at most one vertex $v \in \opti$.
    Furthermore, by Lemma~\ref{lm:rplus_assigning_condition}, $\yplus_S$ can be assigned to $\rplus_v$ for some vertex $v \in \opti$ only if $v \in S$ and $\pairv \notin S$. In other words, $S \odot \opti$. Therefore:
    \begin{align}
    \label{eq:rplus_opti_yplus_opti}
        \sum_{v\in \opti} \rplus_v \le \sum_{\substack{S\subseteq \V \\ S\odot \opti}} \yplus_S
    \end{align}

    Now we have:
    \begin{align*}
        (1 - \cone) \cc(\otree) &< \sum_{v \in \opti} \rr_v \tag{Corollary~\ref{corollary:B_bound_r}} \\
        &\leq \sum_{v \in \opti} \rplus_v \tag{Lemma~\ref{lm:rr-le-rplus}} \\
        &\leq \sum_{\substack{S \subseteq \V \\ S \odot \opti}} \yplus_S \tag{Equation~\ref{eq:rplus_opti_yplus_opti}} \\
        &\leq \cc(\otree) - \dfrac{\extralegacyotree}{2}, \tag{Lemma~\ref{lm:ys_le_T_extra}}
    \end{align*}
    Rearranging terms yields:
    $$
    \extralegacyotree \leq 2\cone \cc(\otree),
    $$
    as desired.
\end{proof}

\begin{lemma}
\label{lm:bound_dis_candidate}
    We can bound $\maxdis(\sone)$ as follows:
    $$\maxdis(\sone) \le 2\cone\cc(\otree) + 2\rr_{\vone}.$$
\end{lemma}
\begin{proof}
    First, we want to show that for every vertex $v\in \opti$, we can bound the distance between $v$ and $\pairv$ by $\dis(v, \pairv) \le 2\cone\cc(\otree) + 2\rr_{\vone}$.

    By Lemma~\ref{lm:legacy_connected_at_rmax}, the vertices $v$ and $\pairv$ are connected in Legacy Execution at time $\rplus_{\max}(\opti)$. 
    This implies that after time $\rplus_{\max}(\opti)$, no set cuts the pair $\{v, \pairv\}$. 
    Moreover, at each moment before time $\rplus_{\max}(\opti)$, there are at most two active sets containing $v$ or $\pairv$ since active sets at each moment are disjoint.
    Thus, we have
    \begin{align}
    \label{eq:yplus_pathv}
        \sum_{\substack{S\subseteq \V \\ S \odot \{v, \pairv\}}} \yplus_S \le 2\rplus_{\max}(\opti).
    \end{align}
    
    In addition, since $v$ and $\pairv$ are connected in the optimal solution, there exists a path between them; let us denote this path by $path_v$. 
    Then:
    \begin{align*}
        \dis(v, \pairv) &\le \cc(path_v) \\
        &\le  \extralegacy(path_v) + \sum_{\substack{S\subseteq \V \\ S \odot \{v, \pairv\}}} \yplus_S \tag{Lemma~\ref{lm:coloring_exactly_one_edge}} \\
        &\le \extralegacy(path_v) + 2\rplus_{\max}(\opti) \tag{Equation~\ref{eq:yplus_pathv}}\\
        &\le \extralegacyotree + 2\rplus_{\vone} \tag{Lemmas~\ref{lm:extra_subgraph} and~\ref{lm:vone_has_rmax}} \\
        &\le 2\cone\cc(\otree) + 2\rr_{\vone}. \tag{Lemma~\ref{lm:legacy_extra_coloring}, Corollary~\ref{cor:rplus_vone_rr_vone}}
    \end{align*}
    Since this holds for every $v \in \opti$, by Definition~\ref{def:dis}, we conclude that
    $$
    \maxdis(\sone) = \max_{v \in \sone} \dis(v, \pairv) \leq 2\cone \cc(\otree) + 2\rr_{\vone},
    $$
    completing the proof.
\end{proof}

\begin{lemma}
\label{lm:Btwo_has_candidate}
    Sets $\sone$ and $\stwo$ are selected as an autarkic pair.
\end{lemma}
\begin{proof}
    For contradiction,  assume that $\sone$ and $\stwo$ are not selected as an autarkic pair.
    Given Lemma~\ref{lm:pair_sone_stwo}, they are pairs of each other.
    Therefore, not getting selected as autarkic pairs would imply that they do not satisfy the condition in Line~\ref{line:autarkic_condition}.
    This means that
    \begin{align}
    \label{eq:not_selected_autarkic}
        (1+\ccnd) (\Y_{\sone} + \Y_{\stwo}) \le \maxdis(\sone)
    \end{align}
    We can use this to show that
    \begin{align*}
        (1+\ccnd) (2\rr_{\vone} - (6\cone + 2\ctwo)\cc(\otree)) &\le (1+\ccnd) (\Y_{\sone} + \Y_{\stwo}) \tag{Lemma~\ref{lm:bound_YS_candidate}} \\
        &\le \maxdis(\sone) \tag{Equation~\ref{eq:not_selected_autarkic}} \\
        &\le 2\cone\cc(\otree) + 2\rr_{\vone}. \tag{Lemma~\ref{lm:bound_dis_candidate}}
    \end{align*}
    By moving $\rr_{\vone}$ to one side and other terms to the other side, we get
    \begin{align*}
        2 \ccnd\rr_{\vone} \le (2\cone + (6\cone + 2\ctwo)(1+\ccnd))\cc(\otree).
    \end{align*}
    Then, dividing both sides by $2\ccnd$ results in
    \begin{align*}
        \rr_{\vone} &\le \dfrac{1}{\ccnd}(\cone + (3\cone + \ctwo)(1+\ccnd))\cc(\otree) \\
        &=\dfrac{1}{\ccnd}\left(4\cone + 3\ccnd \cone + 4(1+\ccnd)\cone\dfrac{1+\bta}{1-\bta}\right)\cc(\otree) \tag{Definition~\ref{def:ctwo_cfive}}\\
        &=\dfrac{\cone}{\ccnd}\left(4+3\ccnd+4(1+\ccnd)\dfrac{1+\bta}{1-\bta}\right)\cc(\otree) \\
        &=\cfive \cc(\otree) \tag{Definition~\ref{def:ctwo_cfive}}
    \end{align*}
        However, since $\opti \in \Btwo$, given Definition~\ref{def:classify_B}, we have $\rmax(\opti) > \cfive \cc(\otree)$ which, combined with Lemma~\ref{lm:vone_has_rmax}, gives $\rr_{\vone} > \cfive \cc(\otree)$.
        This contradicts the above inequality and proves that $\sone$ and $\stwo$ are selected as an autarkic pair.
\end{proof}

\subsection{Cost Analysis of $\solcnd$}
\label{sec:candidate_analysis}
Now that we have established properties for connected components in $\Btwo$, we first provide a lower bound for $\Y$ of all autarkic pairs in Lemma~\ref{lm:total_candidate_time}.
Then, in Lemma~\ref{lm:bound_sol_cnd}, we derive an upper bound for the solution $\solcnd$, which is obtained by executing \candidate{} in Line~\ref{exe:candid}, and a new upper bound for the solution $\solone$ in Lemma~\ref{lm:bound_sol_one_for_cnd}.
Finally, since the cost of our final solution in Algorithm~\ref{alg:main} is at most equal to the minimum cost between $\solcnd$ and $\solone$, we combine their bounds to establish an upper bound on the minimum cost between them in Lemma~\ref{lm:sol-cnd-one-ub}.

First, we start by defining a notation for the total $\Y$ of autarkic pairs and bounding it.

\begin{definition}[Total Autarkic Growth]
    We define $\Ycnd$ as the total value of $\YS + \Y_{\pairS}$ for subsets of vertices $S, \pairS \subseteq \V$ that are selected as autarkic pairs in Line~\ref{line:autarkic_condition}.
\end{definition}

\begin{lemma}
\label{lm:total_candidate_time}
    The total autarkic growth has a lower bound as follows:
    $$
    \Ycnd \ge 2\left(1 - \dfrac{3\cone + \ctwo}{\cfive}\right) \sum_{\opti \in \Btwo} \rmax(\opti).
    $$
\end{lemma}
\begin{proof}
    For any connected component of the optimal solution $\opti \in \Btwo$, let $\sone^{\opti}$ and $\stwo^{\opti}$ be defined as $\sone$ and $\stwo$ in Definition~\ref{def:Sone_Stwo}.
    According to Lemma~\ref{lm:Btwo_has_candidate}, we know that they are selected as autarkic pairs.
    Furthermore, we have a lower bound on the value of $\Y_{\sone^{\opti}} + \Y_{\stwo^{\opti}}$ using Lemma~\ref{lm:bound_YS_candidate}.
    Therefore, the total autarkic growth is at least the sum of these lower bounds for all connected components in $\Btwo$, which is
    \begin{align*}
        \Ycnd &\ge \sum_{\opti \in \Btwo} \left(\Y_{\sone^{\opti}} + \Y_{\stwo^{\opti}}\right) \tag{Lemma~\ref{lm:Btwo_has_candidate}} \\
        &\ge \sum_{\opti \in \Btwo} \left(2\rmax(\opti) - (6\cone + 2\ctwo)\cc(\otree)\right) \tag{Lemmas~\ref{lm:bound_YS_candidate} and~\ref{lm:vone_has_rmax}} \\
        &> \sum_{\opti \in \Btwo} \left(2\rmax(\opti) - \dfrac{6\cone + 2\ctwo}{\cfive}\rmax(\opti)\right) \tag{Definition~\ref{def:classify_B}} \\
        &= 2\left(1 - \dfrac{3\cone + \ctwo}{\cfive}\right) \sum_{\opti \in \Btwo} \rmax(\opti)
    \end{align*}
\end{proof}

\begin{definition}
\label{def:tone_ttwo}
    Let $\tone$ be the total value of $\yys$ for active sets such that $\unsatisfiedS$ is selected as an autarkic pair and $S$ is a single-edge set.
    Similarly, let $\ttwo$ be the total value of $\yys$ for active sets such that $\unsatisfiedS$ is selected as an autarkic pair and $S$ is a multi-edge set.
\end{definition}
The following lemma establishes that $
    \Ycnd = \tone + \ttwo$ by showing any set $S$ contributing to $\Ycnd$ must be a single-edge or multi-edge set.  
\begin{lemma}
    We have 
    $$
    \Ycnd = \tone + \ttwo.
    $$
\end{lemma}
\begin{proof}
    Given Line~\ref{line:define_YS}, for any subset of vertices $S \subseteq \V$ such that $\yys$ is added to $\Y_{S'}$, the vertices in $S'$ are unsatisfied in $S$. 
    According to Line~\ref{line:autarkic_for_YS}, if $S'$ is an autarkic pair, it is not an empty set.
    Therefore, there exists a vertex $v \in S' = \unsatisfiedS$ such that $v$ is unsatisfied, and hence $\pairv \notin S$.
    Since the optimal solution must connect $v$ and $\pairv$, there is a path between them, and since $S\odot \{v, \pairv\}$, the set $S$ cuts the path between them and consequently cuts the optimal solution.
    This implies that $S$ either colors exactly one edge of the optimal solution, meaning $S$ is a single-edge set and $\yys$ is added to $\tone$, or colors at least two edges of the optimal solution, meaning $S$ is a multi-edge set and $\yys$ is added to $\ttwo$.
    Hence, we conclude that
    $$
    \Ycnd = \tone + \ttwo.
    $$
\end{proof}

Recall that $\Gp$ is defined in Line~\ref{line:define_Gp}. We now analyze the cost of the optimal solution in $\Gp$.

\begin{definition}
    We define $\OPT'$ as the optimal solution for $\Gp$.
\end{definition}

Since we run \BBG{} on $\Gp$ in Line~\ref{line:autarkic_legacy}, and \BBG{} has an approximation factor of 2, we can find a solution with cost at most $2\cc(\OPT')$ for $\Gp$.
Therefore, we first aim to bound $\cc(\OPT')$.
Note that instead of using \BBG{}, we could call our algorithm recursively on $\Gp$ to obtain a better-than-2 approximation.
However, we use \BBG{} to simplify the analysis.

The next lemma is necessary to establish an upper bound for $\cc(\OPT')$.

\begin{lemma}
\label{lm:only_one_autarkic}
    For any vertex $v \in \V$, it can only belong to at most one autarkic pair.
\end{lemma}
\begin{proof}
    We will prove this by contradiction. 
    Assume that there are two autarkic pairs, $(S', \pair(S'))$ and $(S'', \pair(S''))$, such that $v \in S'$ and $v \in S''$, with $S' \neq S''$.

    According to the autarkic pair condition in Line~\ref{line:autarkic_condition}, we have:
    \begin{align*}
        \Y_{S'} + \Y_{\pair(S')} &> \dfrac{1}{1+\ccnd} \maxdis(S') \\
        &\ge \dfrac{1}{1+\ccnd}\dis(v, \pairv) \tag{$v\in S$, Definition~\ref{def:dis}} 
    \end{align*}
    Similarly, we derive the same inequality for $S''$:
    $$
        \Y_{S''} + \Y_{\pair(S'')} > \dfrac{1}{1+\ccnd} \dis(v, \pairv).
    $$
    Summing these inequalities, we get:
    \begin{equation}
    \label{eq:two_autarkic_dis_v_Y_S}
        \Y_{S'} + \Y_{\pair(S')} + \Y_{S''} + \Y_{\pair(S'')} > \dfrac{2}{1+\ccnd} \dis(v, \pairv).
    \end{equation}

    Let $\mathcal{S} = \{S', \pair(S'), S'', \pair(S'')\}$.
    We know that $\mathcal{S}$ contains four different subsets of vertices, and according to Line~\ref{line:define_YS}, the growth $\yys$ of any set $S \subseteq \V$ is assigned to at most one of the sets in $\mathcal{S}$.
    Moreover, since every set in $\mathcal{S}$ cuts the pair $\{v, \pairv\}$, any subset $S \subseteq \V$ for which $\unsatisfiedS\in\mathcal{S}$ (meaning its growth $\yys$ is assigned to the $\Y$ of one of the sets in $\mathcal{S}$) must also cut $\{v, \pairv\}$.
    If $S$ does not cut $\{v, \pairv\}$, then $\unsatisfiedS \cap \{v, \pairv\} = \emptyset$, implying that $\unsatisfiedS \notin \mathcal{S}$.
    Therefore,
    \begin{align*}
        \sum_{\substack{S\subseteq \V \\ S\odot \{v,\pairv\}}}\yys& \geq \sum_{\substack{S\subseteq\V \\ \unsatisfied(S)\in\mathcal{S}}}\yys\\&= \sum_{S\in \mathcal{S}}\YS\\
        &> \dfrac{2}{1+\ccnd}\dis(v, \pairv) \tag{Equation~\ref{eq:two_autarkic_dis_v_Y_S}}\\
        &> \dis(v, \pairv) \tag{$0 < \ccnd < 1$}
    \end{align*}

    However, this contradicts the fact that the sets $S \subseteq \V$ that cut $\{v, \pairv\}$ in Legacy Execution color the shortest path between $v$ and $\pairv$ for the duration of $\yys$, and each portion of an edge can only be colored once. 
    Therefore, it is not possible for the sum to exceed the distance $\dis(v, \pairv)$.

    Thus, our assumption is false, and we conclude that $v$ can belong to at most one autarkic pair.
\end{proof}

\begin{lemma}
\label{lm:bound_cost_opt'}
    We can bound the cost of $\OPT'$ as follows:
    $$
    \cc(\OPT') \le \cc(\OPT) - \tone.
    $$
\end{lemma}
\begin{proof}
    To prove the lemma, we first provide a valid solution for $\Gp$ and show that its cost is at most $\cc(\OPT) - \tone$. 
    Then, since $\OPT'$ is the optimal solution for $\Gp$, its cost must be less than or equal to the cost of any valid solution for $\Gp$.
    Let $\solp$ denote the constructed solution for $\Gp$.

    Initially, we set $\solp$ equal to $\OPT \cup \EAP$, which clearly satisfies all demands.
    Then, for any single-edge set $S$—that is, $S$ cuts exactly one edge of $\OPT$—such that $\unsatisfiedS$ is selected as an autarkic pair, we remove from $\solp$ the edge of $\OPT$ that is cut by $S$.
    Note that while $S$ is a single-edge set with respect to $\OPT$, it may cut more than one edge in $\solp$ since $\solp$ also includes the edges of $\EAP$.
    Moreover, for each such removed, we have added a corresponding edge to $\EAP$ in Line~\ref{line:contract_vertices}, and thus to $\solp$, and will later show that these additions preserve feasibility.
    Clearly, the cost of $\solp$ is at most $\cc(\OPT) - \tone$, since the total cost of removed edges is at least $\tone$, and all edges in $\EAP$ have zero cost.

    Next, we need to prove that $\solp$ satisfies all demands, given that some edges from $\OPT$ have been removed but edges of $\EAP$ between autarkic pairs have been added.

    Let $e$ be an edge that was part of $\OPT$ and got removed from $\solp$ because a single-edge set $S$ cuts $e$ and $\unsatisfiedS$ is selected as an autarkic pair.
    Given Lemma~\ref{lm:remove_one_edge}, removing $e$ only disconnects vertices in $\unsatisfiedS$ from their pair and does not affect other vertices.
    Therefore, we need to show that those vertices are still satisfied in $\solp$.
    Since $S$ was an autarkic pair, we have added a zero-cost edge from a vertex $\vap\in S$ to $\pair_{\vap}$.
    If we show that all vertices in $\unsatisfiedS$ are connected to $\vap$ and all vertices in $\pair(\unsatisfiedS)$ are connected to $\pair_{\vap}$ in $\solp$, we can conclude that all vertices in $\unsatisfiedS$ are satisfied in $\solp$. We observe that the leaves of $T_1$ are contained in $\unsatisfiedS$, and leaves of $T_2$ are contained in $\pair(\unsatisfiedS)$.

    We observe that since $S$ is a single-edge set, $\unsatisfiedS$ can only include vertices from a single component of $\opti$. Therefore, we can define tree $T_1$ as the minimal subtree of $\OPT$ that connects vertices in $\unsatisfiedS$.
    Similarly, since $\pair(\unsatisfiedS)$ must lie within the same component, we define $T_2$ as the minimal subtree of $\OPT$ that connects the vertices in $\pair(\unsatisfiedS)$.  
    Figure~\ref{fig:autarkic_edge_removal} illustrates the structure of an autarkic pair, highlighting the aforementioned notation.

    Assume, for the sake of contradiction, that there is an edge in $T_1$ that is removed from $\solp$.
    This edge is removed by some single-edge set $S' \subseteq \V$ where $\unsatisfied(S')$ is selected as an autarkic pair. 
    Since $S'$ only cuts one edge of the optimal solution and also cuts $T_1$, it only cuts one edge of $T_1$.
    According to Lemma~\ref{lm:cut_one_edge_cut_leaves}, $S'$ must cut the leaves of $T_1$, which means $\unsatisfied(S') \subseteq S'$ cannot be $\unsatisfiedS$ as $\unsatisfiedS$ contains all leaves of $T_1$.
    Moreover, there must be a vertex $v \in S'$ that is a leaf of $T_1$.

    We can also show that $v \in \unsatisfied(S')$. Otherwise, $\pairv$ would need to be in $S'$.
    Since $v \in \unsatisfiedS$, $\pairv \notin S$, therefore $S$ would have some vertices of $S'$ but not all of them.
    Moreover, since $S'$ cuts leaves of $T_1$, it also has some vertices from $S$ but not all of them.
    This would imply that $S$ and $S'$ intersect but are not subsets of each other, which contradicts the laminarity of active sets as per Corollary~\ref{cor:active_sets_laminar}.
    Thus, $v \in \unsatisfied(S')$.

\tikzset{
  triangle/.style args={#1}{
    insert path={
        ++(90:1.2*#1)
      -- ++($(90:-1.2*#1)+(210:1.2*#1)$)
      -- ++($(210:-1.2*#1)+(330:1.2*#1)$)
      -- cycle %
    }
  }
}
\tikzset{
  hexagon/.style args={#1}{
    insert path={
      ++(90:#1)
      -- ++($(90:-#1)+(30:#1)$)
      -- ++($(30:-#1)+(-30:#1)$)
      -- ++($(-30:-#1)+(-90:#1)$)
      -- ++($(-90:-#1)+(-150:#1)$)
      -- ++($(-150:-#1)+(150:#1)$)
      -- cycle
    }
  }
}
\tikzset{
  pentagon/.style args={#1}{
    insert path={
      ++(90:1.05*#1)
      -- ++($(90:-1.05*#1)+(162:1.05*#1)$)
      -- ++($(162:-1.05*#1)+(234:1.05*#1)$)
      -- ++($(234:-1.05*#1)+(306:1.05*#1)$)
      -- ++($(306:-1.05*#1)+(18:1.05*#1)$)
      -- cycle
    }
  }
}
\tikzset{
  square/.style args={#1}{
    insert path={
        ++(45:1.1*#1)
      -- ++($(45:-1.1*#1)+(135:1.1*#1)$)
      -- ++($(135:-1.1*#1)+(225:1.1*#1)$)
      -- ++($(225:-1.1*#1)+(315:1.1*#1)$)
      -- cycle %
    }
  }
}
\tikzset{
  star/.style args={#1}{
    insert path={
      ++(90:1.1*#1)
      \foreach \a in {90,162,234,306,378} {
        -- ++($(\a:-1.1*#1) + (\a+36:0.5*#1)$)
        -- ++($(\a+36:-0.5*#1) + (\a+72:1.1*#1)$)
      }
      -- cycle
    }
  }
}

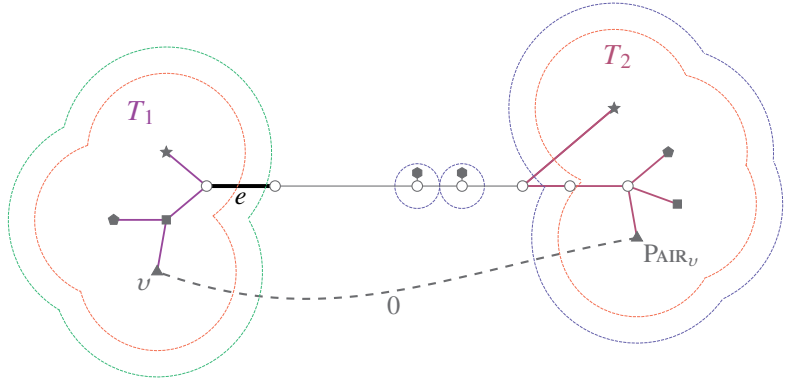
\begin{figure}[t]
    \centering
    \begin{subfigure}{0.9\textwidth}
        \centering
        \begin{tikzpicture}[scale=0.7]
\def\dem{Red!70}
\def\demi{Blue!70}
\def\demii{Green!70}
\def\ter{White}
\def\col{Black!70}
\def\treeone{Purple}
\def\treetwo{RubineRed!70!Black}
\def\noder{0.1cm}

\def\n{4}
\def\m{4}
\pgfmathtruncatemacro{\last}{\n+1}

\def\tercl#1{%
  \ifcase#1
    White\or
    olive!20\or
    Green!20\or
    Sepia!20\or
    Plum!20\or
    Goldenrod!20\or
  \fi
}
\def\tersh#1{%
  \ifcase#1
    none\or
    star\or
    square\or
    pentagon\or
    triangle\or
    hexagon\or
  \fi
}

\def\angle{40}
\def\smallRotate{0}

\def\petalLength{4}
\def\forkLength{6/6}
\def\shortLen{6/24}
\def\edgeLen{0.6} %
\def\cof{1.7}

\def\heightStep{1.2}
\def\widthStep{1.2}

\tikzset{
    center arc/.style args={#1:#2:#3}{
        insert path={+ (#1:#3) arc (#1:#1+#2:#3)}
    }
}

\coordinate (A) at (0, 0);
\coordinate (R) at ($(A) + (\smallRotate:\petalLength)$);
\coordinate (B) at ($(R) + (-\smallRotate:\petalLength)$);

\coordinate (Rf) at ($(R) + (-\smallRotate:0)$);
\coordinate (Rs) at ($(R) + (-\smallRotate:2*\cof*\shortLen)$);

\coordinate (Rf) at ($(R) + (-\smallRotate:0)$);
\coordinate (Rs) at ($(R) + (-\smallRotate:2*\cof*\shortLen)$);

\coordinate (R1) at ($(R) + (-\smallRotate:\petalLength/2)$);

\coordinate (Rf) at ($(R) + (-\smallRotate:0)$);
\coordinate (Rs) at ($(R) + (-\smallRotate:2*\cof*\shortLen)$);

\coordinate (Rff) at ($(Rf) + (-\smallRotate+90:\shortLen)$);
\coordinate (Rss) at ($(Rs) + (-\smallRotate+90:\shortLen)$);
\foreach \i in {Rf, Rs} {
    \draw[\demi, line width=0.3pt, dash pattern=on 1.2pt off 0.4pt] (\i) circle(\cof*\shortLen);
}

\coordinate (A1) at ($(A) + (180+\smallRotate-\angle:\forkLength)$);
\coordinate (A2) at ($(A) + (180+\smallRotate+\angle:\forkLength)$);
\coordinate (A3) at ($(A2) + 
(180+\smallRotate:\forkLength)$);
\coordinate (A4) at ($(A2) + 
(180+\smallRotate+2*\angle:\forkLength)$);

\coordinate (B1) at ($(B) + (-\smallRotate-2*\angle:\forkLength)$);
\coordinate (B2) at ($(B) + (-\smallRotate-0.5*\angle:\forkLength)$);
\coordinate (B3) at ($(B) + 
(-\smallRotate+1*\angle:\forkLength)$);
\coordinate (B4) at ($(B) + 
(-\smallRotate+2.5*\angle:1.5*\forkLength)$);

\coordinate (e2) at ($(A)!1.3cm!(R)$);

\coordinate (e3) at ($(R1) + (-\smallRotate:0.9cm)$);

\draw[\col] (A) -- (R) -- (R1);

\draw[\treeone, line width=0.7pt] (A1) -- (A) -- (A2);
\draw[\treeone, line width=0.7pt] (A3) -- (A2) -- (A4);

\draw[\treetwo, line width=0.7pt] (B1) -- (B) -- (B2);
\draw[\treetwo, line width=0.7pt] (B) -- (B3);
\draw[\treetwo, line width=0.7pt] (R1) -- (B4);

\draw[\col] (Rf) -- (Rff);
\draw[\col] (Rs) -- (Rss);

\draw[\treetwo, line width=0.7pt] (R1) -- (B);
\draw[line width=1.5pt] (A) -- node[below=-2pt] {\small{$e$}} (e2);
\draw[\treetwo, line width=0.7pt] (R1) -- (e3);
\draw[\treetwo, line width=0.7pt] (R1) --  (B4);

\draw[\col, thick, dashed] (B1) to[out=190, in=-20] node[below=-2pt] {\small{0}} (A4);

\def\r{2*\forkLength}

\draw[\demii, line width=0.3pt, dash pattern=on 1.2pt off 0.4pt] (A1) [center arc=\smallRotate-39:205:\r];
\draw[\demii, line width=0.3pt, dash pattern=on 1.2pt off 0.4pt] (A3) [center arc=\smallRotate+118:121:\r];
\draw[\demii, line width=0.3pt, dash pattern=on 1.2pt off 0.4pt] (A4) [center arc=\smallRotate+201:188:\r];

\draw[\demi, line width=0.3pt, dash pattern=on 1.2pt off 0.4pt] (B1) [center arc=-\smallRotate+152:172:\r];
\draw[\demi, line width=0.3pt, dash pattern=on 1.2pt off 0.4pt] (B2) [center arc=-\smallRotate+296:88:\r];
\draw[\demi, line width=0.3pt, dash pattern=on 1.2pt off 0.4pt] (B3) [center arc=-\smallRotate-4:74:\r];
\draw[\demi, line width=0.3pt, dash pattern=on 1.2pt off 0.4pt] (B4) [center arc=-\smallRotate+32:197:\r];

\def\r{1.5*\forkLength}
\draw[\dem, line width=0.3pt, dash pattern=on 1.2pt off 0.4pt] (A1) [center arc=\smallRotate-54:229:\r] node[pos=0.8, below right=-0pt, color=\treeone] {$T_1$};
\draw[\dem, line width=0.3pt, dash pattern=on 1.2pt off 0.4pt] (A3) [center arc=\smallRotate+109:136:\r];
\draw[\dem, line width=0.3pt, dash pattern=on 1.2pt off 0.4pt] (A4) [center arc=\smallRotate+194:211:\r];

\draw[\dem, line width=0.3pt, dash pattern=on 1.2pt off 0.4pt] (B1) [center arc=-\smallRotate+135:194:\r];
\draw[\dem, line width=0.3pt, dash pattern=on 1.2pt off 0.4pt] (B2) [center arc=-\smallRotate+290:100:\r];
\draw[\dem, line width=0.3pt, dash pattern=on 1.2pt off 0.4pt] (B3) [center arc=-\smallRotate-10:87:\r];
\draw[\dem, line width=0.3pt, dash pattern=on 1.2pt off 0.4pt] (B4) [center arc=-\smallRotate+25:220:\r] node[pos=0.2, below left=-0pt, color=\treetwo] {$T_2$};

\foreach \i in {
A, B, 
 R1, e2, e3, Rf, Rs} {
    \draw[\col, fill= \ter] (\i) circle (\noder);
}

\foreach \i/\c in {
A1/\tersh{1}, A2/\tersh{2}, A3/\tersh{3}, A4/\tersh{4}, 
B1/\tersh{4}, B2/\tersh{2}, B3/\tersh{3}, B4/\tersh{1}, 
Rff/\tersh{5}, Rss/\tersh{5}} {
    \draw[\col, fill=\col] (\i) [\c=\noder];
}

\node[\col, below left=-2pt] at (A4) {\small{$\vap$}};
\node[\col, below right=-2pt] at (B1) {\small{$\pair_{\vap}$}};
        \end{tikzpicture}
    \end{subfigure}
    \caption{An autarkic pair within a component of the optimal solution. The moats on the right are multi-edge sets and therefore do not trigger any edge removal. However, the moats on the left cut only the edge $e$, allowing us to remove it from $\solp$.
    This removal only impacts the unsatisfied vertices within that moat. To maintain feasibility, we add an edge between $\vap$ and $\pair_{\vap}$ to both $\EAP$ and our solution.
    By ensuring that all edges in $T_1$ and $T_2$—the subtrees of the optimal solution connecting each subset of the autarkic pair—are retained in $\solp$, we guarantee that all vertices in the autarkic pair remain satisfied.}
    \label{fig:autarkic_edge_removal}
\end{figure}     
    Now that we know $\unsatisfiedS$ and $\unsatisfied(S')$ are distinct autarkic pairs and $v$ is in both of them, we can apply Lemma~\ref{lm:only_one_autarkic}, which states that each vertex can belong to at most one autarkic pair. 
    This contradiction shows that no edge from $T_1$ can be removed.

    The same argument can be applied to $T_2$. 
    Therefore, no edge from $T_1$ and $T_2$ is removed from $\solp$, and all vertices in $\unsatisfiedS$ are connected to $\vap$ and all vertices in $\pair(\unsatisfiedS)$ are connected to $\pair_{\vap}$.
    Since $\vap$ and $\pair_{\vap}$ are connected in $\EAP$ and consequently in $\solp$, removing $e$ does not violate any pair demand, and therefore, $\solp$ satisfies all demands.
\end{proof}

\begin{lemma}
\label{lm:bound_sol_cnd}
    We can give an upper bound for $\solcnd$ as follows:
    $$
    \solcnd \le 2\cc(\OPT) - (1 - \ccnd) \Ycnd + 2\ttwo.
    $$
\end{lemma}
\begin{proof}
    Since we only select autarkic pairs that satisfy $(1+\ccnd)(\YS + \YPS) \ge \maxdisS$ based on Line~\ref{line:autarkic_condition}, and we add an edge between $\vap \in S$ and $\pair_{\vap} \in \pairS$ with a cost $\dis(\vap, \pair_{\vap}) \le \maxdisS$, the cost of the selected edge in Line~\ref{line:autarkicpair_add_edge} for an autarkic pair is at most $(1+\ccnd)(\YS + \YPS)$. This means the total cost of all selected edges is at most $(1+\ccnd)\Ycnd$.

    Since we run the 2-approximation algorithm on $\G'$, we obtain a solution with cost $2\cc(\OPT')$. Including the cost of the selected edges in Line~\ref{line:autarkicpair_add_edge}, the total cost is at most:
    \begin{align*}
        \solcnd &\le 2\cc(\OPT') + (1+\ccnd)\Ycnd \\
        &\le 2(\cc(\OPT) - \tone) + (1+\ccnd)\Ycnd \tag{Lemma~\ref{lm:bound_cost_opt'}}\\
        &= 2\cc(\OPT) - 2(\Ycnd - \ttwo) + (1+\ccnd)\Ycnd \tag{Definition~\ref{def:tone_ttwo}}\\
        &= 2\cc(\OPT) - (1 - \ccnd) \Ycnd + 2\ttwo 
    \end{align*}
\end{proof}

\begin{lemma}
\label{lm:bound_sol_one_for_cnd}
    We can either give an upper bound for $\solone$ as follows:
    $$
    \solone \le \left(2+2\alf + \dfrac{4\alf}{\bta}\right)\cc(\OPT) - 2\ttwo,
    $$
    or $\solone$ is a $2-2\alf$ approximation of the optimal solution $\OPT$.
\end{lemma}
\begin{proof}
    If $\cc(\solone) \le (2-2\alf)\cc(\OPT)$, the lemma is already proved.
    Therefore, we assume that this condition does not hold.
    In this case, we can bound $\lossone$ and $\winone$ using Lemmas~\ref{lm:lcl-loss} and~\ref{lm:lcl-win}.
    
    Let $\multicolored^+$ be $\multicolored$ as defined in Definition~\ref{def:uncolored_single_multi_colored}, corresponding to Boosted Execution.
    By Definition~\ref{def:tone_ttwo}, during the duration of $\ttwo$, at least two edges of the optimal solution are being colored by each active set.
    Thus, the total portion of $\OPT$ colored by these active sets is at least $2\ttwo$, since each active set colors at least two edges.
    Therefore,
    \begin{align}
    \label{eq:MC+OPT_2ttwo}
        2\ttwo \le \multicolored^+(\OPT).
    \end{align}

    Additionally, each active set either cuts the optimal solution (meaning $S\odot \OPT$) or does not cut it (meaning $\deltaS \cap \OPT=\emptyset$).
    Both cannot happen simultaneously.
    Hence, we divide all active sets and their corresponding $\yys$ between these two groups to bound $\solone$.
    \begin{align*}
        \solone &\le 2\sum_{S \subseteq \V} \yy_S \tag{Lemma~\ref{lm:monotonic-forest-twice-growth}}\\
        &=2\left(\sum_{\substack{S \subseteq \V \\ S\odot \OPT}} \yys + \sum_{\substack{S \subseteq \V \\ \deltaS\cap \OPT = \emptyset}} \yys\right) \\
        &\le 2\left(\cc(\OPT) - \dfrac{\multicolored^+(\OPT)}{2} + \sum_{\substack{S\subseteq \V \\ \deltaS \cap \OPT = \emptyset}}\yys\right) \tag{Lemma~\ref{lm:ys_le_T_extra}}\\
        &\le 2\left(\cc(\OPT) + \winone + \lossone\right) - \multicolored^+(\OPT) \tag{Lemma~\ref{lm:ys_not_cut_opt_win_loss}}\\
        &\le 2\cc(\OPT) + 2\left(\alf+\dfrac{\alf}{\bta}\right) \cc(\OPT) + 2\dfrac{\alf}{\bta}\cc(\OPT) - \multicolored^+(\OPT) \tag{Lemmas~\ref{lm:lcl-win} and~\ref{lm:lcl-loss}}\\
        &\le \left(2+2\alf + \dfrac{4\alf}{\bta}\right)\cc(\OPT) - 2\ttwo. \tag{Equation~\ref{eq:MC+OPT_2ttwo}}
    \end{align*}
\end{proof}

\begin{lemma}
\label{lm:sol-cnd-one-ub}
    We can either bound the best solution between $\solcnd$ and $\solone$ as follows:
    $$\min(\cc(\solcnd), \cc(\solone)) \le \left(2 + \alf + \dfrac{2\alf}{\bta}\right) \cc(\OPT) - (1-\ccnd)\left(1-\dfrac{3\cone+\ctwo}{\cfive}\right)\sum_{\opti \in \Btwo}\rmax(\opti),$$
    or $\solone$ is a $2-2\alf$ approximation of the optimal solution $\OPT$.
\end{lemma}
\begin{proof}
    If $\cc(\solone) \le (2-2\alf)\cc(\OPT)$, the proof is done.
    Therefore, we assume that this condition does not hold.
    In this case, we can use the bound in Lemma~\ref{lm:bound_sol_one_for_cnd} for $\solone$.
    Therefore, we can prove our claim as follows:
    \begin{align*}
        \min(\cc(\solcnd), \cc(\solone))
        &\le \dfrac{\solcnd + \solone}{2} \\
        &\le \dfrac{\left(2\cc(\OPT) - (1 - \ccnd)\Ycnd + 2\ttwo\right) + \left(\left(2+2\alf + \dfrac{4\alf}{\bta}\right)\cc(\OPT) - 2\ttwo\right)}{2} \tag{Lemmas~\ref{lm:bound_sol_cnd} and~\ref{lm:bound_sol_one_for_cnd}}\\
        &= \left(2 + \alf + \dfrac{2\alf}{\bta}\right) \cc(\OPT) - \dfrac{1-\ccnd}{2}\Ycnd \\
        &\le \left(2 + \alf + \dfrac{2\alf}{\bta}\right) \cc(\OPT) - (1-\ccnd)\left(1-\dfrac{3\cone+\ctwo}{\cfive}\right)\sum_{\opti \in \Btwo}\rmax(\opti). \tag{Lemma~\ref{lm:total_candidate_time}}
    \end{align*}
\end{proof}

\section{Approximation for Steiner Forest: Proof of Theorem~\ref{thm:main_steiner_forest}}
\label{sec:final}
This section establishes that Algorithm~\ref{alg:main} has an approximation factor of $(2 - 2\alf)$ in polynomial time.  

\begin{proof}[Proof of Theorem~\ref{thm:main_steiner_forest}]
We claim that Algorithm~\ref{alg:main} deterministically solves the Steiner Forest problem in polynomial time with an approximation factor of $2 - \appx$.

Recall that our final solution is chosen from three possible solutions, as specified in Line~\ref{line:best-sol}.  
Our goal is to show that the minimum cost among these three solutions is at most $(2 - 2\alf) \cdot \cc(\OPT)$ for some value of $\alf\ge\frac{\alfbound}{2}$.  
To this end, we define the forest $\ALG$ as the best choice among $\solone$, $\solext$, and $\solcnd$.  

First, all three solutions satisfy the demands as explained below:  
\begin{itemize}  
    \item $\solone$ is obtained from Boosted Execution, where the fingerprint is larger than that of Legacy Execution. By Lemma~\ref{lm:large_fingerprint_satisfies}, this ensures that the demands are met. 
    \item $\solext$ is obtained from Extended-Boosted Execution, where the fingerprint is larger as that of Legacy Execution. This follows from the fact that $\tz_v \ge \tp_v \ge \ttt_v \ge \tplus_v$ for all $v \in \V$. Thus, by Lemma~\ref{lm:large_fingerprint_satisfies}, the demands are satisfied.  
    \item $\solcnd$ selects some paths, adds a zero-cost edge between their endpoints, and then applies Legacy Moat Growing.
    Since Legacy Moat Growing satisfies all demands, replacing the zero-cost edges with the corresponding selected paths also satisfies all demands.
\end{itemize}  

Second, all these three solutions are obtained in polynomial time since \BBG{}, \LocalSearch{}, \candidate{}, and \Extend{} run in polynomial time (see Lemma~\ref{lm:localsearch_polynomial_time};  Corollaries~\ref{cor:legacy_polynomial_time},~\ref{cor:extend-polynomial}, and~\ref{cor:autarkic_pair_polynomial_time}; and Figure~\ref{fig:overview}).

In the remainder, we prove the following:
$$
    \cc(\ALG) \le (2 - 2\alf) \cdot \cc(\OPT)
$$    

In the previous sections, we analyzed each solution and established several bounds. Specifically, we have proven that:  
\begin{enumerate}
    \item either $\solone$ is a $(2-2\alf)$-approximation solution or the upper bound for $\solone$ in Lemma~\ref{lm:solone-ub} holds,  
    \item either $\solone$ or $\solext$ is a $(2-2\alf)$-approximation solution, or the upper bound for $\solext$ in Lemma~\ref{lm:ext-final-bound} holds, and  
    \item either $\solone$ is a $(2-2\alf)$-approximation solution or the upper bound given in Lemma~\ref{lm:sol-cnd-one-ub} holds for the better solution between $\solcnd$ and $\solone$.  
\end{enumerate}  

First, if for the chosen $\alf$, either $\solone$, $\solext$, or $\solcnd$ provides a $(2 - 2\alf)$-approximation, then the theorem follows immediately:
\begin{align*}
\cc(\ALG) \le \min(\cc(\solone), \cc(\solext), \cc(\solcnd)) \le (2 - 2\alf) \cdot \cc(\OPT).
\end{align*}

Consequently, from the listed lemmas, we obtain three upper bounds: one for $\cc(\solone)$, one for $\cc(\solext)$, and one for $\min\big(\cc(\solone), \cc(\solcnd)\big)$. Since  
\[
\cc(\ALG) \le \min(\cc(\solone), \cc(\solext), \cc(\solcnd)),
\]  
these bounds also apply to $\cc(\ALG)$.  

We index these bounds according to their positions in the list above. The $i$-th bound follows the general form  
\[
    \cc(\ALG) \le \kappa_{i,\A} \cdot \cc(\A) + \kappa_{i,\Bone} \cdot \cc(\Bone) + \kappa_{i,\Btwo} \cdot \cc(\Btwo) + \kappa_{i,r} \cdot \sum_{\opti \in B} \rmax(\opti),
\]  
where each coefficient \(\kappa_{i,j}\) is a function of the algorithm's parameters and is determined by the \(i\)-th bound. 

We form a weighted sum of these three upper bounds by assigning a weight of \(\omega_i\) to the \(i\)-th upper bound such that 
$\sum_{i = 1}^{3} \omega_i = 1$ and $\omega_i \ge 0$. It follows that
\begin{align*}
    \cc(\ALG)
    &\le  \sum_{i = 1}^{3} \omega_i \cdot \bigg(\kappa_{i,\A} \cdot \cc(\A) + \kappa_{i,\B} \cdot \cc(\B) + \kappa_{i,r} \cdot \sum_{\opti \in B} \rmax(\opti)\bigg).
\end{align*}
Therefore, it suffices to proof the existence of {\em proper values} for the parameters of the algorithm and the analysis ($\alf, \bta, \ccnd, \cone, \varepsilon$, and $\wvar$), along with $\omega_1, \omega_2$, and $\omega_3$, such that
\begin{align}
\label{eq:final-ineq}
    \sum_{i = 1}^{3} \omega_i \cdot \bigg(\kappa_{i,\A} \cdot \cc(\A) + \kappa_{i,\B} \cdot \cc(\B) + \kappa_{i,r} \cdot \sum_{\opti \in B} \rmax(\opti)\bigg) \le (2 - 2\alf) \cdot \cc(\OPT)
\end{align}
It is important to note that \emph{proper values} must fit within the constraints of the algorithm’s parameters, as explained in the lemmas' statements, as well as the constraints related to the weighted sum $\omega$.

Observing Definitions~\ref{def:classify_A_B} and~\ref{def:classify_B}, we obtain $\cc(\OPT) = \cc(\A) + \cc(\Bone) + \cc(\Btwo)$. To proceed, we show that there exist suitable values for the parameters and $\omega$ such that  
\begin{align}  
    &\sum_{i=1}^{3} \omega_i \cdot \kappa_{i,\A} \leq (2 - 2\alf), \label{eq:for-A}\\  
    &\sum_{i=1}^{3} \omega_i \cdot \kappa_{i,\Bone} \leq (2 - 2\alf), \label{eq:for-B1}\\  
    &\sum_{i=1}^{3} \omega_i \cdot \kappa_{i,\Btwo} \leq (2 - 2\alf), \text{ and} \label{eq:for-B2}\\  
    &\sum_{i=1}^{3} \omega_i \cdot \kappa_{i,r} \leq 0.\label{eq:for-rmax}
\end{align}  
Summing these inequalities with the factors $\cc(\A)$, $\cc(\Bone)$, $\cc(\Btwo)$, and $\sum_{\opti \in B} \rmax(\opti)$ leads directly to inequality~\ref{eq:final-ineq}, which completes the proof.  

\newcommand{\maxterm}{\max\big(\frac{\alf}{\bta},\frac{\left(1 - \wpvar(1+\bta)(1-\wvar)\right)(\alf+\eps)}{\bta},\frac{\alf+(1-\alf-\alf/\bta)\left(1 - \wpvar(1+\bta)(1-\wvar)\right)\eps}{\bta}\big)}
\newcommand{\cndterm}{\left(2 + \alf + \dfrac{2\alf}{\bta}\right)}

For the left hand side of inequality~\ref{eq:for-A}, we have
\begin{align*}
\label{eq:open-A}
    \sum_{i=1}^{3} \omega_i \cdot \kappa_{i,\A} 
    &= 2\omega_1 \big((1 - \cone) + \dfrac{\alf}{\bta}\big)  \tag{Lemma~\ref{lm:solone-ub}}\\
    &+ 2\omega_2 (1 + \eps)(1  - \cone) \tag{Lemma~\ref{lm:ext-final-bound}}\\
    &+ 2\omega_2 \cdot \maxterm \tag{Lemma~\ref{lm:ext-final-bound}} \\
    &+ \omega_3\cndterm \tag{Lemma~\ref{lm:sol-cnd-one-ub}}.
\end{align*}

For the left hand side of inequality~\ref{eq:for-B1}, we have
\begin{align*}
\label{eq:open-B1}
    \sum_{i=1}^{3} \omega_i \cdot \kappa_{i,\Bone}
    &= 2\omega_1 (1 + \dfrac{\alf}{\bta}) \tag{Lemma~\ref{lm:solone-ub}}\\
    &+ 2\omega_2 \big(1 + \eps - \eps\wpvar(1 + (1-2\wvar)(1-\cone)) + \cfive\wpvar(1+\wvar\frac{1-\bta}{5+\bta})\big) \tag{Lemma~\ref{lm:ext-final-bound}}\\
    &+ 2\omega_2 \cdot \maxterm \tag{Lemma~\ref{lm:ext-final-bound}} \\
    &+ \omega_3\cndterm \tag{Lemma~\ref{lm:sol-cnd-one-ub}}.   
\end{align*}

\newcommand\Tstrut{\rule{0pt}{3.5ex}}         %
\newcommand\Bstrut{\rule[-2.3ex]{0pt}{0pt}}   %
\begin{table}[ht]
    \centering
    \begin{tabular}{ll}
        \toprule
        Constraint & \\
        \midrule
            $\begin{aligned}[t]
                0 < \bta, \eps, \ccnd < 1
            \end{aligned}$
            \Tstrut\Bstrut
            & Algorithm~\ref{alg:main} Assumptions\\ 
        \hline 
            $\begin{aligned}[t]
                0 < \wvar \le \frac{1}{2}
            \end{aligned}$
            \Tstrut\Bstrut
            & Lemma~\ref{lm:ext-final-bound}\\
        \hline 
            $\begin{aligned}[t]
                \wpvar \ge 1/(1+\frac{1-\bta}{5+\bta}\wvar)
            \end{aligned}$
            \Tstrut\Bstrut
            & Lemma~\ref{lm:ext-final-bound}\\ 
        \hline
            $\begin{aligned}[t]
                \wpvar \le (1+\eps)/(1 + 3\eps - \eps\wvar\frac{9+3\bta}{5+\bta})
            \end{aligned}$
            \Tstrut\Bstrut
            & Lemma~\ref{lm:ext-final-bound}\\
        \hline
            $\begin{aligned}[t]
                \ctwo = 4\cone \dfrac{1+\bta}{1-\bta}
            \end{aligned}$
            \Tstrut\Bstrut
            & Definition~\ref{def:ctwo_cfive}\\ 
        \hline
            $\begin{aligned}[t]
                \cfive = \dfrac{\cone}{\ccnd} \left(4 + 3\ccnd + 4 (1 + \ccnd)\dfrac{1+\bta}{1 - \bta}\right) 
            \end{aligned}$
            \Tstrut\Bstrut
            & Definition~\ref{def:ctwo_cfive}\\ 
        \hline
            $\begin{aligned}[t]
                \sum_{i=1}^{3} \omega_i = 1
            \end{aligned}$
            \Tstrut\Bstrut
            & Weighted Sum\\ 
        \hline
            $\begin{aligned}[t]
                \sum_{i=1}^{3} \omega_i \cdot \kappa_{i,\A} \leq (2 - 2\alf)
            \end{aligned}$
            \Tstrut\Bstrut
            & Inequality~\ref{eq:for-A}, \hyperref[eq:open-A]{Explicit form}\\
        \hline
            $\begin{aligned}[t]
                \sum_{i=1}^{3} \omega_i \cdot \kappa_{i,\Bone} \leq (2 - 2\alf)
            \end{aligned}$
            \Tstrut\Bstrut
            & Inequality~\ref{eq:for-B1}, \hyperref[eq:open-B1]{Explicit form}\\
        \hline
            $\begin{aligned}[t]
                \sum_{i=1}^{3} \omega_i \cdot \kappa_{i,\Btwo} \leq (2 - 2\alf)
            \end{aligned}$
            \Tstrut\Bstrut
            & Inequality~\ref{eq:for-B2}, \hyperref[eq:open-B2]{Explicit form}\\
        \hline
            $\begin{aligned}[t]
                \sum_{i=1}^{3} \omega_i \cdot \kappa_{i,r} \leq 0
            \end{aligned}$
            \Tstrut\Bstrut
            & Inequality~\ref{eq:for-rmax}, \hyperref[eq:open-rmax]{Explicit form}\\
        \bottomrule
        
    \end{tabular}
    \caption{\emph{Conditions for a $(2-2\alf)$-Approximation.}
If we can determine values for the previously defined parameters that satisfy the given constraints, the existence of a $(2-2\alf)$-approximation for PCSF is guaranteed.}
    \label{tab:inequalities}
\end{table}

For the left hand side of inequality~\ref{eq:for-B2}, we have 
\begin{align*}
\label{eq:open-B2}
    \sum_{i=1}^{3} \omega_i \cdot \kappa_{i,\Btwo}
    &= 2\omega_1 (1 + \dfrac{\alf}{\bta}) \tag{Lemma~\ref{lm:solone-ub}}\\
    &+ 2\omega_2 (1 + \eps - \eps\wpvar(1 + (1-2\wvar)(1-\cone))) \tag{Lemma~\ref{lm:ext-final-bound}}\\
    &+ 2\omega_2 \cdot \maxterm \tag{Lemma~\ref{lm:ext-final-bound}} \\
    &+ \omega_3\cndterm \tag{Lemma~\ref{lm:sol-cnd-one-ub}}.   
\end{align*}

Finally, for the left hand side of inequality~\ref{eq:for-rmax}, we have
\begin{align*}
\label{eq:open-rmax}
    \sum_{i=1}^{3} \omega_i \cdot \kappa_{i,r}
    &= 2\omega_1 \cdot 0 \tag{Lemma~\ref{lm:solone-ub}}\\
    &+ 2\omega_2 \cdot \wpvar(1+\wvar\frac{1-\bta}{5+\bta}) \tag{Lemma~\ref{lm:ext-final-bound}}\\
    &- \omega_3 (1-\ccnd)\left(1-\dfrac{3\cone+\ctwo}{\cfive}\right) \tag{Lemma~\ref{lm:sol-cnd-one-ub}}.   
\end{align*}

Now, we just need to provide values for the algorithm's parameters and weights such that all these conditions (Table~\ref{tab:inequalities}) hold. This can be viewed as a numerical optimization problem with the goal of maximizing $\alf$. However, for the purpose of this theorem, we must provide values that satisfy the constraints while ensuring $\alf \ge \frac{\appx}{2}$.

To this end, we provide values for the algorithm's parameters and weights $\omega_i$ in Table~\ref{tab:parameters}.
\end{proof}
\begin{table}[H]
    \centering
    \renewcommand{\arraystretch}{1.15} %
    \begin{tabular}{ccc}
        \toprule
        Parameter & Value & Fractional Value \\
        \midrule
        $\alf$   & $5.000000000000 \times 10^{-12}$ & $\frac{1}{200000000000}$ \\ \hline
        $\bta$   & $1.000000000000 \times 10^{-1}$ & $\frac{1}{10}$ \\ \hline
        $\cone$  & $1.183079480386 \times 10^{-8}$ & $\frac{3575638326237933}{302231454903657293676544}$ \\ \hline
        $\eps$ & $3.438743241429 \times 10^{-7}$ & $\frac{6495602330607721}{18889465931478580854784}$ \\ \hline
        $\wvar$  & $1.000000000000 \times 10^{-2}$ & $\frac{1}{100}$ \\ \hline
        $\wpvar$  & $9.999993185227 \times 10^{-1}$ & $\frac{32112103126037549486258500}{32112125009721801303670549}$ \\ \hline
        $\ccnd$  & $5.000000000000 \times 10^{-1}$ & $\frac{1}{2}$ \\ \hline
        $\ctwo$  & $5.783944126333 \times 10^{-8}$ & $\frac{13110673862872421}{226673591177742970257408}$ \\ \hline
        $\cfive$  & $3.036570666325 \times 10^{-7}$ & $\frac{91774717040106947}{302231454903657293676544}$ \\ \hline
        $\omega_1$ & $4.960317460317 \times 10^{-1}$ & $\frac{125}{252}$ \\ \hline
        $\omega_2$ & $7.936507936508 \times 10^{-3}$ & $\frac{1}{126}$ \\ \hline
        $\omega_3$ & $4.960317460317 \times 10^{-1}$ & $\frac{125}{252}$ \\
        \bottomrule
    \end{tabular}
    \caption{\emph{Parameters and Values.} This table presents values for the algorithm's parameters that ensure the provided algorithm is an $(2-\appx)$-approximation. Scientific notation provides an estimate for each parameter, while fractional representation helps verify the correctness of constraints without losing precision.}
    \label{tab:parameters}
\end{table}

\section{Acknowledgements}
This work is partially supported by DARPA QuICC, ONR MURI 2024 award on Algorithms, Learning, and Game Theory, Army-Research Laboratory (ARL) grant W911NF2410052, NSF AF:Small grants 2218678, 2114269, 2347322.

\bibliographystyle{alpha}
\bibliography{references}

\appendix

\section{Tight Examples and Counterexamples}

Here, we provide some instances to give a better intuition on different components of our algorithm.

\subsection{The Strength of Local Search}
\label{eg:grid}
\newcommand{\supertiny}[1]{\scalebox{0.6}{#1}}
\tikzset{
  triangle/.style args={#1}{
    insert path={
        ++(90:1.2*#1)
      -- ++($(90:-1.2*#1)+(210:1.2*#1)$)
      -- ++($(210:-1.2*#1)+(330:1.2*#1)$)
      -- cycle %
    }
  }
}
\tikzset{
  pentagon/.style args={#1}{
    insert path={
      ++(90:1.05*#1)
      -- ++($(90:-1.05*#1)+(162:1.05*#1)$)
      -- ++($(162:-1.05*#1)+(234:1.05*#1)$)
      -- ++($(234:-1.05*#1)+(306:1.05*#1)$)
      -- ++($(306:-1.05*#1)+(18:1.05*#1)$)
      -- cycle
    }
  }
}
\tikzset{
  square/.style args={#1}{
    insert path={
        ++(45:1.1*#1)
      -- ++($(45:-1.1*#1)+(135:1.1*#1)$)
      -- ++($(135:-1.1*#1)+(225:1.1*#1)$)
      -- ++($(225:-1.1*#1)+(315:1.1*#1)$)
      -- cycle %
    }
  }
}
\tikzset{
  star/.style args={#1}{
    insert path={
      ++(90:1.1*#1)
      \foreach \a in {90,162,234,306,378} {
        -- ++($(\a:-1.1*#1) + (\a+36:0.5*#1)$)
        -- ++($(\a+36:-0.5*#1) + (\a+72:1.1*#1)$)
      }
      -- cycle
    }
  }
}
\begin{figure}[ht]
    \centering
    \begin{subfigure}{0.32\textwidth}
    \centering
\begin{tikzpicture}[scale=0.7]
\draw[White] (4,-5.8) -- (4, 0.85);
\def\dem{Red!70}
\def\ter{White}
\def\col{Black!70}
\def\noder{0.1cm}

\def\n{4}
\def\m{4}
\pgfmathtruncatemacro{\last}{\n+1}

\def\tercl#1{%
  \ifcase#1
    White\or
    Green!20\or
    Sepia!20\or
    Plum!20\or
    Goldenrod!20
  \fi
}
\def\tersh#1{%
  \ifcase#1
    none\or
    star\or
    square\or
    pentagon\or
    triangle\or
  \fi
}

\def\angle{5}
\def\smallRotate{40}

\def\petalLength{0.5}
\def\forkLength{\petalLength/3}
\def\heightStep{1.2}
\def\widthStep{1.2}

\foreach \i in {1,...,\last} {
    \foreach \j in {0,...,\m} {
        \coordinate (T\i\j) at (\i*\widthStep, -\j*\heightStep);
    }
}

\foreach \i in {1,...,\n} {
    \foreach \j in {1,...,\m} {
        \draw[\col] (T\i0) to[out=-90-(\j*\angle), in=100] node[name=MV\i\j]{} (T\i\j);
        \pgfmathtruncatemacro{\deg}{180+((\n-\i+1)*\angle)}
        \draw[\col] (T\last\j) to[out=\deg, in=-10] node[name=MH\i\j]{} (T\i\j);
    }
}
\foreach \j in {1,...,\m} {
    \draw[\col] (T\last0) to[out=-90+(\j*\angle), in=80] node[name=ML\j]{} (T\last\j);
}
\foreach \i in {1,...,\n} {
    \pgfmathtruncatemacro{\deg}{180-((\n-\i+1)*\angle)}
    \draw[\col] (T\i0) to[out=10, in=\deg] node[name=MF\i]{} (T\last0);
}

\node[\col, below=-2pt] at (MH1\m) {\footnotesize{$1+\xi$}};
\node[\col, left=-2pt] at (MV1\m) {\footnotesize{2}};

\foreach \i in {1,...,\n} {
    \foreach \j in {0,...,\m} {
        \draw[\col, fill=\col] (T\i\j) [\tersh{\i}=\noder];
    }
}
\foreach \j in {0,...,\m} {
    \draw[\col, fill=\ter] (T\last\j) circle(\noder);
}

\end{tikzpicture}
\caption{}
\label{fig:grid-inp}
\end{subfigure}
\hfill
\begin{subfigure}{0.32\textwidth}
    \centering
\begin{tikzpicture}[scale=0.7]
\draw[White] (4,-5.8) -- (4, 0.85);
\def\dem{Red!70}
\def\ter{White}
\def\col{Black!70}
\def\noder{0.1cm}

\def\n{4}
\def\m{4}
\pgfmathtruncatemacro{\last}{\n+1}

\def\tercl#1{%
  \ifcase#1
    White\or
    Green!20\or
    Sepia!20\or
    Plum!20\or
    Goldenrod!20
  \fi
}
\def\tersh#1{%
  \ifcase#1
    none\or
    star\or
    square\or
    pentagon\or
    triangle\or
  \fi
}

\def\angle{5}
\def\smallRotate{40}

\def\petalLength{0.5}
\def\forkLength{\petalLength/3}
\def\heightStep{1.2}
\def\widthStep{1.2}

\foreach \i in {1,...,\last} {
    \foreach \j in {0,...,\m} {
        \coordinate (T\i\j) at (\i*\widthStep, -\j*\heightStep);
    }
}

\def\margin{1mm}
\def\dashfill{0.25mm}
\def\dashspace{0.05mm}
\def\dashsize{0.3mm}

\foreach \i in {1,...,\n} {
    \foreach \j in {1,...,\m} {
        \pgfmathtruncatemacro{\deg}{90+(\j*\angle)}
        \draw[\dem, line width=1.5pt] (T\i0) to[out=-\deg, in=100] (T\i\j);
        \pgfmathtruncatemacro{\deg}{180+((\n-\i+1)*\angle)}
        \draw[\col,postaction={decorate, decoration={
    markings,%
    mark=between positions 0.5cm and 1 step 0.3mm with {\draw[\dem, line width=1.5pt] (0,0) -- (0.4mm, 0);}
}}] (T\last\j) to[out=\deg, in=-10] (T\i\j);
    }
}
\foreach \j in {1,...,\m} {
    \pgfmathtruncatemacro{\deg}{-90+(\j*\angle)}
    \draw[\col] (T\last0) to[out=\deg, in=80] (T\last\j);
}
\foreach \i in {1,...,\n} {
    \pgfmathtruncatemacro{\deg}{180-((\n-\i+1)*\angle)}
    \draw[\col,postaction={decorate, decoration={
    markings,%
    mark=between positions 0 and 1-0.5cm step 0.3mm with {\draw[\dem, line width=1.5pt] (0,0) -- (0.4mm, 0);}
}}] (T\i0) to[in=\deg, out=10] (T\last0);
}

\foreach \i in {1,...,\n} {
    \foreach \j in {1,...,\m} {
        \draw[\col, fill=\col] (T\i\j) [\tersh{\i}=\noder];
        \node[\dem, above right=-2pt] at (T\i\j) {\supertiny{1}};
    }
    \draw[\col, fill=\col] (T\i0) [\tersh{\i}=\noder];
    \node[\dem, below right=-2pt] at (T\i0) {\supertiny{1}};
}
\foreach \j in {0,...,\m} {
    \draw[\col, fill=\ter] (T\last\j) circle(\noder);
}
\end{tikzpicture}
\caption{}
\label{fig:grid-2}
\end{subfigure}
\hfill
\begin{subfigure}{0.32\textwidth}
    \centering
\begin{tikzpicture}[scale=0.7]
\draw[White] (4,-5.8) -- (4, 0.85);
\def\dem{Blue!70}
\def\ter{White}
\def\col{Black!70}
\def\noder{0.1cm}

\def\n{4}
\def\m{4}
\pgfmathtruncatemacro{\last}{\n+1}

\def\tercl#1{%
  \ifcase#1
    White\or
    Green!20\or
    Sepia!20\or
    Plum!20\or
    Goldenrod!20
  \fi
}
\def\tersh#1{%
  \ifcase#1
    none\or
    star\or
    square\or
    pentagon\or
    triangle\or
  \fi
}

\def\angle{5}
\def\smallRotate{40}

\def\petalLength{0.5}
\def\forkLength{\petalLength/3}
\def\heightStep{1.2}
\def\widthStep{1.2}

\foreach \i in {1,...,\last} {
    \foreach \j in {0,...,\m} {
        \coordinate (T\i\j) at (\i*\widthStep, -\j*\heightStep);
    }
}

\def\margin{1mm}
\def\dashfill{0.25mm}
\def\dashspace{0.05mm}
\def\dashsize{0.3mm}

\foreach \i in {1,...,\n} {
    \foreach \j in {1,...,\m} {
        \pgfmathtruncatemacro{\deg}{90+(\j*\angle)}
        \draw[\dem, line width=1.5pt] (T\i0) to[out=-\deg, in=100] (T\i\j);
        \pgfmathtruncatemacro{\deg}{180+((\n-\i+1)*\angle)}
        \draw[\dem, line width=1.5pt] (T\last\j) to[out=\deg, in=-10] (T\i\j);
    }
}
\foreach \j in {1,...,\m} {
    \pgfmathtruncatemacro{\deg}{-90+(\j*\angle)}
    \draw[\dem, line width=1.5pt] (T\last0) to[out=\deg, in=80] (T\last\j);
}
\foreach \i in {1,...,\n} {
    \pgfmathtruncatemacro{\deg}{180-((\n-\i+1)*\angle)}
    \draw[\dem, line width=1.5pt] (T\i0) to[in=\deg, out=10] (T\last0);
}

\foreach \i in {1,...,\n} {
    \foreach \j in {1,...,\m} {
        \draw[\col, fill=\col] (T\i\j) [\tersh{\i}=\noder];
        \node[\dem, above right=-2pt] at (T\i\j) {\supertiny{$\frac{5+3\xi}{8}$}};
    }
    \draw[\col, fill=\col] (T\i0) [\tersh{\i}=\noder];
    \node[\dem, below right=-2pt] at (T\i0) {\supertiny{$\frac{5+3\xi}{8}$}};
}
\draw[\col, fill=\ter] (T\last0) circle(\noder);%
\foreach \j in {1,...,\m} {
    \draw[\col, fill=\ter] (T\last\j) circle(\noder); %
}
\end{tikzpicture}
\caption{}
\label{fig:grid-one-step}
\end{subfigure}
\caption{
An illustration of an instance on a grid-like graph, which presents a challenging instance for local search algorithms 
that prioritize minimizing total cost. 
(a) The grid structure of the instance: edges within each column have length \(2\), and edges within each row 
have length \(1 + \xi\), with sufficiently small $\xi >0$. In the corresponding Steiner Forest instance, each column—except the rightmost one—
must be connected. 
(b) The outcome of \BBGr{} on this instance, which fails to achieve a better than \(2\) approximation. 
The vertical edges in each column (except the rightmost) are selected, while all horizontal edges remain 
partially uncolored. Vertex labels indicate a possible assignment of growth values. 
(c) The outcome of our local search algorithm, given appropriate values of \( \bta \). 
All horizontal edges are selected, and vertical edges are fully colored; however, only a subset forming a 
spanning tree is ultimately chosen. The growth assignments differ from part (b), reflecting the total 
\( \win \) value captured by our local search.
}
\label{fig:grid}
\end{figure}
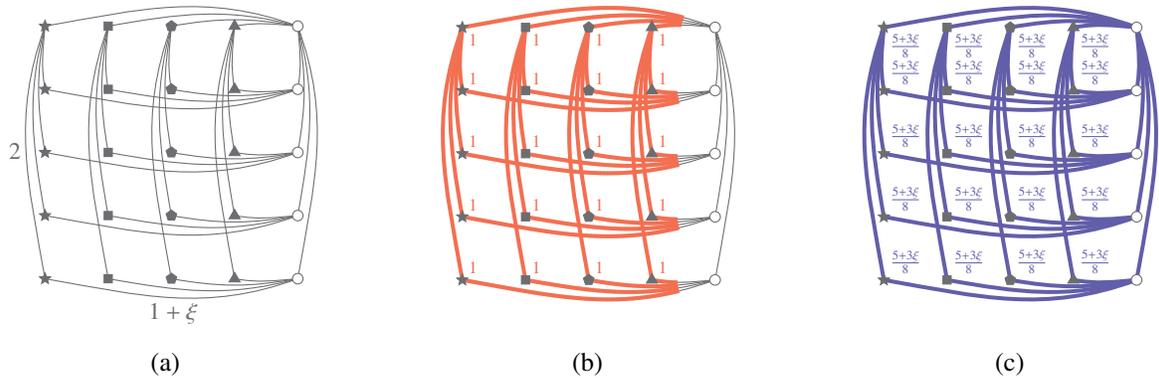
 
Intuitively, a local search algorithm for the Steiner Forest problem might start with Legacy Moat Growing algorithm and try possible moves to improve the overall cost of the solution. One intuitive set of moves for such an algorithm would be our choice of Boost actions, which keep vertices active for longer, even if they have no unsatisfied demands.
One important argument here is to determine whether keeping a vertex active for a longer period of time is beneficial. 
A natural idea is to compare the cost of the solution before and after a boost action. 
However, this example demonstrates that the idea does not work. Instead, it led us to observe a different objective: the total growth of active sets decreases after applying a boost action. 
Since twice this value serves as a bound for our solution, this observation resulted in an improvement that we leveraged to obtain our results.

\paragraph{No Cost-Based Boost Exists.}
For this example, we prove that Legacy Moat Growing results in a solution with a cost of at least $2-\xi$ times the optimal solution for any $\xi>0$ such that no boost action improving the cost is possible. 
Then, the boost-based local search algorithm focusing on the cost of the solution will fail to produce a better than $2$ approximation.

Consider the example illustrated in Figure  \ref{fig:grid} with $n$ rows and $m$ columns and parameter $\xi$. In this graph, vertices in each row are connected to the vertex in the last column with edges of length $1+\xi$, and vertices in each column are connected to the first vertex in their column using edges of length $2$. The demands for this instance require that each column except for the rightmost one form a connected component. 

In the initial moat growing run on this instance, each vertex in all columns except the last one is active for a duration of $1$, and each of these columns will form a connected component of the solution with cost $2(n-1)$, for a total cost of $2(n-1)(m-1)$. On the other hand, for small enough values of $\xi$ and large enough $n$ and $m$, the optimal solution chooses the edges in each row plus the edges of one column, achieving a total cost of $2(n-1)+(1+\xi)n(m-1)$. Therefore, this solution is a 
\[
\frac{2(n-1)(m-1)}{2(n-1)+(1+\xi)n(m-1)}=2 - \frac{4(n-1) + 2\xi(n)(m-1) + 2(m-1)}{2(n-1)+(1+\xi)n(m-1)}
\]
approximation of the optimal, which given appropriate values of $n$, $m$, and $\xi$, this solution will not be a better than $2-\xi'$ approximation. 

Next, we show that no boost action will lead to an improvement in terms of solution cost.
First, allowing any vertex not in the rightmost column to remain active longer cannot help achieve a lower solution, as the algorithm's behavior would not change before all demands are satisfied using the same forest, and additional growth afterward cannot decrease the cost of the solution. Additionally, for any vertex $v$ on the rightmost column, allowing it to grow for a duration of less than $\xi$ will not change the algorithm's behavior, as no vertex in the other columns will reach $v$ before its demand is satisfied. 

Now, if vertex $v$ in the rightmost column is allowed to grow for at least $\xi$, all the vertices in the same row as $v$ will reach $v$ before being connected to their demands and then grow together. However, the final solution will not improve. To see why this is true, note that any other vertex in the rightmost column is separated from any active set by edges of at least $1+\xi$ and, therefore, will not be connected to any vertex before this time. Since all demands are satisfied by time $1$, adding any additional edges would not affect the final solution produced by the moat growing algorithm. Now, in the subgraph containing only one vertex of the last column and all the vertices in the other columns, the lowest cost forest that satisfies the demands is the forest that connects each column separately, which is the same as the solution found by the moat growing algorithm. Therefore, any boost action cannot improve the total cost of the solution.

\paragraph{Effectiveness of Our Local Search.}
In contrast, our local search algorithm focusing on minimizing $\ybase$ values can find the optimal solution given an appropriate value of $\bta$. To see why, consider a boost action of $\frac{1+\xi}{2}$ on any vertex in the rightmost column. Then, all vertices in the same row are connected by time $\frac{1+\xi}{4}$. Then, this combined set grows until moment $1$, which is an additional duration of $\frac{3-\xi}{4}$. This means that the total growth of 
\(
m
\)
for this row is reduced to 
\[
m\frac{1+\xi}{4} + \frac{3-\xi}{4}=\frac{(m+3)+(m-1)\xi}{4}.
\]
This means we have a boost action with a win of
\[
\frac{(3-\xi)(m-1)}{4}
\]
while incurring a loss of 
\[
\frac{1+\xi}{4}.
\]
Now, for appropriate values of $\bta$, this will be a valuable boost action. Additionally, applying the boost action for each row does not detract from the value for the other rows, ensuring that the optimal solution can be found by the local search algorithm. Figure \ref{fig:grid} also illustrates assignments before and after applying our local search, showcasing how the moves are beneficial.

\subsection{A Lower Bound for the Local Search Algorithm}
\label{eg:binary15}

While our algorithm achieves a general approximation ratio better than 2 for the Steiner Tree problem, the example we have discovered demonstrates a lower bound of $3/2$, as described below.
Note that in this example, there is no boost action with positive win. However, the condition for a boost action to be valuable is $\Win \ge (1 + \bta) \cdot \Loss$. This suggests that, for a fixed $\bta$, one could construct a tighter example that yields a higher lower bound.

\paragraph{Constructing the Instance.}

Consider a complete binary tree with a clear distinction between left and right children. This tree has $2^h$ leaves at the lowest layer. The edge weights decrease exponentially from bottom to top:

\begin{itemize}
    \item The bottom-layer edges connecting leaves to their parents have weight $1$.
    \item Each upper-layer edge weight is exactly half the weight of the edges directly below it (i.e., the immediate parents have edges of weight $1/2$, their parents $1/4$, and so forth).
    \item Additionally, every two consecutive leaves (from leftmost to rightmost) are directly connected with edges of weight $(2 - \xi)$.
    \item All the leaves are the set of terminal vertices.
\end{itemize}

Clearly, the given structure provides a Steiner Tree instance (See Figure \ref{fig:binary-input}).

\begin{figure}
    \centering
    \begin{subfigure}{0.95\textwidth}
    \centering
    \begin{tikzpicture}[
  scale=0.7,
  level distance=1.5cm, 
  level 1/.style={sibling distance=8cm},
  level 2/.style={sibling distance=4cm},
  level 3/.style={sibling distance=2cm},
  edge from parent/.style={draw=\col} %
]
    \def\dem{Red!70}
            \def\ter{White}
            \def\col{Black!70}
            \def\noder{0.1cm}

  \node[circle,draw=\col,fill=\ter,inner sep=0pt,minimum size=2*\noder] (root) {}
    child {node[circle,draw=\col,fill=\ter,inner sep=0pt,minimum size=2*\noder] (A) {}
      child {node[circle,draw=\col,fill=\ter,inner sep=0pt,minimum size=2*\noder] (B) {}
        child {node[circle,draw=\col,fill=\col,inner sep=0pt,minimum size=2*\noder] (L1) {}}
        child {node[circle,draw=\col,fill=\col,inner sep=0pt,minimum size=2*\noder] (L2) {}}
      }
      child {node[circle,draw=\col,fill=\ter,inner sep=0pt,minimum size=2*\noder] (C) {}
        child {node[circle,draw=\col,fill=\col,inner sep=0pt,minimum size=2*\noder] (L3) {}}
        child {node[circle,draw=\col,fill=\col,inner sep=0pt,minimum size=2*\noder] (L4) {}}
      }
    }
    child {node[circle,draw=\col,fill=\ter,inner sep=0pt,minimum size=2*\noder] (D) {}
      child {node[circle,draw=\col,fill=\ter,inner sep=0pt,minimum size=2*\noder] (E) {}
        child {node[circle,draw=\col,fill=\col,inner sep=0pt,minimum size=2*\noder] (L5) {}}
        child {node[circle,draw=\col,fill=\col,inner sep=0pt,minimum size=2*\noder] (L6) {}}
      }
      child {node[circle,draw=\col,fill=\ter,inner sep=0pt,minimum size=2*\noder] (F) {}
        child {node[circle,draw=\col,fill=\col,inner sep=0pt,minimum size=2*\noder] (L7) {}}
        child {node[circle,draw=\col,fill=\col,inner sep=0pt,minimum size=2*\noder] (L8) {}}
      }
    };

  \draw[\col] (root) -- ++(+0.5, +0.5) node[above right] {$\iddots$};

  \def\threshold{1.55}
  \foreach \i in {1,2,3,4,5,6,7} {
    \pgfmathtruncatemacro{\j}{\i+1}
    \coordinate (C\i) at ($(L\i)!0.5!(L\j) + (270:\threshold)$);
    \draw[\col] (L\i) .. controls (C\i) .. (L\j);
  }
  
  \node[\col] at (C1) {$2 - \xi$};
  \node[\col, above left] at ($(L1)!0.5!(B)$) {1};
  \node[\col, above left] at ($(B)!0.5!(A)$) {$\frac{1}{2}$};
  \node[\col, above left] at ($(A)!0.5!(root)$) {$\frac{1}{4}$};
\end{tikzpicture}
\caption{}
\label{fig:binary-input}
\end{subfigure}
\begin{subfigure}{0.95\textwidth}
    \centering
    \begin{tikzpicture}[
  scale=0.7,
  level distance=1.5cm, 
  level 1/.style={sibling distance=8cm},
  level 2/.style={sibling distance=4cm},
  level 3/.style={sibling distance=2cm},
  edge from parent/.style={draw=\col} %
]
    \def\dem{Red!70}
            \def\ter{White}
            \def\col{Black!70}
            \def\noder{0.1cm}

  \node[circle,draw=\col,fill=\ter,inner sep=0pt,minimum size=2*\noder] (root) {}
    child {node[circle,draw=\col,fill=\ter,inner sep=0pt,minimum size=2*\noder] (A) {}
      child {node[circle,draw=\col,fill=\ter,inner sep=0pt,minimum size=2*\noder] (B) {}
        child {node[circle,draw=\col,fill=\col,inner sep=0pt,minimum size=2*\noder] (L1) {}}
        child {node[circle,draw=\col,fill=\col,inner sep=0pt,minimum size=2*\noder] (L2) {}}
      }
      child {node[circle,draw=\col,fill=\ter,inner sep=0pt,minimum size=2*\noder] (C) {}
        child {node[circle,draw=\col,fill=\col,inner sep=0pt,minimum size=2*\noder] (L3) {}}
        child {node[circle,draw=\col,fill=\col,inner sep=0pt,minimum size=2*\noder] (L4) {}}
      }
    }
    child {node[circle,draw=\col,fill=\ter,inner sep=0pt,minimum size=2*\noder] (D) {}
      child {node[circle,draw=\col,fill=\ter,inner sep=0pt,minimum size=2*\noder] (E) {}
        child {node[circle,draw=\col,fill=\col,inner sep=0pt,minimum size=2*\noder] (L5) {}}
        child {node[circle,draw=\col,fill=\col,inner sep=0pt,minimum size=2*\noder] (L6) {}}
      }
      child {node[circle,draw=\col,fill=\ter,inner sep=0pt,minimum size=2*\noder] (F) {}
        child {node[circle,draw=\col,fill=\col,inner sep=0pt,minimum size=2*\noder] (L7) {}}
        child {node[circle,draw=\col,fill=\col,inner sep=0pt,minimum size=2*\noder] (L8) {}}
      }
    };

  \draw[\col] (root) -- ++(+0.5, +0.5) node[above right] {$\iddots$};

  \def\threshold{1.55}
  \foreach \i in {1,2,3,4,5,6,7} {
    \pgfmathtruncatemacro{\j}{\i+1}
    \coordinate (C\i) at ($(L\i)!0.5!(L\j) + (270:\threshold)$);
    \draw[\dem, thick] (L\i) .. controls (C\i) .. (L\j);
  }

  \foreach \i/\j in {1/B,2/B,3/C,4/C,5/E,6/E,7/F,8/F} {
  \draw[\dem, dashed] (L\i) circle (\threshold);
  \draw[\dem, thick] (L\i) -- ($(L\i)!0.86!(\j)$);
  }
\end{tikzpicture}
    \caption{}
    \label{fig:binary-2}
\end{subfigure}
\begin{subfigure}{0.95\textwidth}
    \centering
    \begin{tikzpicture}[
  scale=0.7,
  level distance=1.5cm, 
  level 1/.style={sibling distance=8cm},
  level 2/.style={sibling distance=4cm},
  level 3/.style={sibling distance=2cm},
  edge from parent/.style={draw=\col} %
]
    \def\dem{Blue!70}
            \def\ter{White}
            \def\col{Black!70}
            \def\noder{0.1cm}

  \node[circle,draw=\col,fill=\ter,inner sep=0pt,minimum size=2*\noder] (root) {}
    child {node[circle,draw=\col,fill=\ter,inner sep=0pt,minimum size=2*\noder] (A) {}
      child {node[circle,draw=\col,fill=\ter,inner sep=0pt,minimum size=2*\noder] (B) {}
        child {node[circle,draw=\col,fill=\col,inner sep=0pt,minimum size=2*\noder] (L1) {}}
        child {node[circle,draw=\col,fill=\col,inner sep=0pt,minimum size=2*\noder] (L2) {}}
      }
      child {node[circle,draw=\col,fill=\ter,inner sep=0pt,minimum size=2*\noder] (C) {}
        child {node[circle,draw=\col,fill=\col,inner sep=0pt,minimum size=2*\noder] (L3) {}}
        child {node[circle,draw=\col,fill=\col,inner sep=0pt,minimum size=2*\noder] (L4) {}}
      }
    }
    child {node[circle,draw=\col,fill=\ter,inner sep=0pt,minimum size=2*\noder] (D) {}
      child {node[circle,draw=\col,fill=\ter,inner sep=0pt,minimum size=2*\noder] (E) {}
        child {node[circle,draw=\col,fill=\col,inner sep=0pt,minimum size=2*\noder] (L5) {}}
        child {node[circle,draw=\col,fill=\col,inner sep=0pt,minimum size=2*\noder] (L6) {}}
      }
      child {node[circle,draw=\col,fill=\ter,inner sep=0pt,minimum size=2*\noder] (F) {}
        child {node[circle,draw=\col,fill=\col,inner sep=0pt,minimum size=2*\noder] (L7) {}}
        child {node[circle,draw=\col,fill=\col,inner sep=0pt,minimum size=2*\noder] (L8) {}}
      }
    };

  \draw[\col] (root) -- ++(+0.5, +0.5) node[above right] {$\iddots$};

  \def\threshold{1.55}
  \foreach \i in {1,2,3,4,5,6,7} {
    \pgfmathtruncatemacro{\j}{\i+1}
    \coordinate (C\i) at ($(L\i)!0.5!(L\j) + (270:\threshold)$);
    \draw[\col] (L\i) .. controls (C\i) .. (L\j) coordinate[pos=0.25] (M\i) coordinate[pos=0.75] (N\i);
  }
  \foreach \i in {1,2,3} {
    \pgfmathtruncatemacro{\j}{\i+1}
                \draw[\dem, thick] (L\i) -- (M\i);
                \draw[\dem, thick] (L\j) -- (N\i);
            }
    \foreach \i in {4,5,6,7} {
    \pgfmathtruncatemacro{\j}{\i+1}
    \coordinate (C\i) at ($(L\i)!0.5!(L\j) + (270:\threshold)$);
    \draw[\dem, thick] (L\i) .. controls (C\i) .. (L\j);
  }
  \foreach \i/\j in {E/D, F/D} {
    \draw[\dem, thick] (\j) -- ($(\j)!0.92!(\i)$);
  }

  \foreach \i/\j in {5/E,6/E,7/F,8/F} {
  \draw[\dem, dashed] (L\i) circle (\threshold);
  \draw[\dem, thick] (L\i) -- ($(L\i)!0.86!(\j)$);
  }

    \def\r{1.1}
\newcommand{\DrawArcFromPoint}[5]{%
    \draw[#1] (#2) ++(#4:#3) arc (#4:#5:#3);
}
  \DrawArcFromPoint{\dem, dashed}{root}{3}{12}{166}
  \DrawArcFromPoint{\dem, dashed}{root}{3}{-105}{-51}
  \DrawArcFromPoint{\dem, dashed}{A}{2.5}{65}{186}
  \DrawArcFromPoint{\dem, dashed}{B}{1.4}{120}{200}
  \DrawArcFromPoint{\dem, dashed}{C}{\r}{-0}{12}
  \DrawArcFromPoint{\dem, dashed}{D}{2.3}{117}{-159}
  \DrawArcFromPoint{\dem, dashed}{L1}{\r}{110}{334}
  \DrawArcFromPoint{\dem, dashed}{L2}{\r}{-155}{-25}
  \DrawArcFromPoint{\dem, dashed}{L3}{\r}{-155}{-46}
  \DrawArcFromPoint{\dem, dashed}{L4}{\threshold}{-145}{89}
  \foreach \i/\j in {L1/B/,L2/B,L3/C,L4/C, B/A, C/A, A/root, D/root} {
  \draw[\dem, thick] (\i) -- (\j);
  }
  \foreach \i in {5,6,7,8} {
  \draw[\dem, dashed] (L\i) circle (\threshold);
  }
  \node[\col, above left] at (A) {$v$};
\end{tikzpicture}
    \caption{}
    \label{fig:binary-1.5}
\end{subfigure}
    \caption{Illustration of the $3/2$-approximation example. (A) The input graph: a complete binary tree with exponentially decreasing edge weights and additional horizontal edges of weight $2 - \xi$ between leaves. All leaves are terminals. (B) The initial forest constructed by \BBG{}, connecting leaves with horizontal $(2 - \xi)$ edges. (C) The effect of a boost action applied to an internal vertex $v$, which causes earlier merges within its subtree but does not decrease the total growth, and is thus not considered a valuable action.}
    \label{fig:binary}
\end{figure}
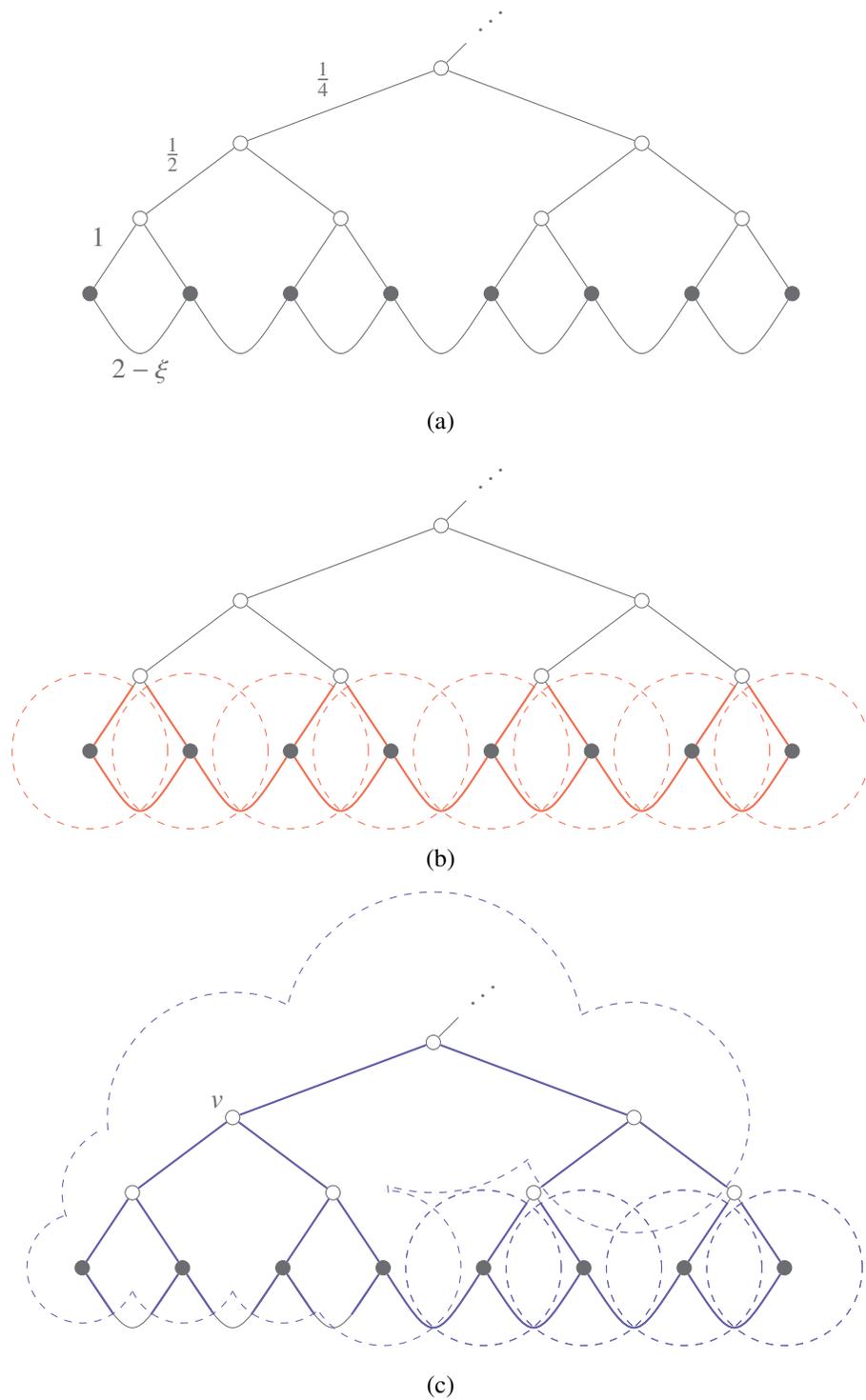 
\paragraph{Running \BBG{} and \LocalSearch{}.}

The local search begins with the execution of the \BBG{} algorithm. At time $(2 - \xi)/2$, each leaf vertex's active set intersects with neighboring active sets, fully coloring the extra edges at the bottom layer (See Figure \ref{fig:binary-2}). Hence, the fingerprint of the initial run sets the active times as follows:

\[
t_v = \begin{cases}
\frac{2 - \xi}{2}, & \text{if $v$ is a leaf},\\[6pt]
0, & \text{otherwise}.
\end{cases}
\]

This results in the initial forest with all additional edges of length $2 - \xi$ at the bottom of the input graph.

\paragraph{Analysis of Local Search.} 
We now aim to show that our local search is unable to find a better solution.

\begin{lemma}
Local search on this instance, starting from the fingerprint produced by Legacy Moat Growing, returns the same initial forest and fingerprint.
\end{lemma}

\begin{proof}
We show that for any vertex $v$, boosting its value $\tV$ to a certain threshold does not reduce the total growth $\sumYS$. By Lemma~\ref{lm:valuable-decrease-total}, such a boost is not considered valuable. We then argue that neither increasing nor decreasing the value further yields a valuable boost, and thus no valid local search move exists.

Figure~\ref{fig:binary-2} illustrates this boost on vertex \(v\).

Assume vertex \(v\) is at hop-distance \(d\) from the leaves in its subtree. Two observations are in order:

\begin{enumerate}
    \item The distance from \(v\) to any leaf outside its subtree is at least \(2\).
    \item The distance from \(v\) to any leaf within its subtree is
    \[
    2 - \left(\frac{1}{2}\right)^{d-1}.
    \]
\end{enumerate}

Now, suppose we boost \(v\) to:
\[
t_v = \frac{2 - \left(\frac{1}{2}\right)^{d-1}}{2}.
\]

At the moment of this boosted value, the active set containing \(v\) will merge with all active sets corresponding to the leaves of its subtree. And only in the next step, when time reaches \((2 - \xi)/2\) (assuming $\xi < (1/2)^{d-1}$), all active sets merge and moat growing ends. This implies only the subtree of \(v\) is affected by this boost (see Figure~\ref{fig:binary-1.5}).

Let’s compare the total growth for this subtree before and after the boost (denoted as $y$ and $y'$).

Before the boost:
\[
\sumYS = 2^d \cdot \frac{(2 - \xi)}{2}.
\]

After the boost:
\[
\sumYS' = \frac{2^d \cdot \left(2 - \left(\frac{1}{2}\right)^{d-1}\right)}{2} + \frac{2 - \xi}{2}.
\]

Subtracting the pre-boost from the post-boost expression, and using the assumption $\xi < (1/2)^{d-1}$, one verifies:
\[
2^d \cdot \frac{(2 - \xi)}{2} \le \frac{2^d \cdot \left(2 - \left(\frac{1}{2}\right)^{d-1}\right)}{2} + \frac{2 - \xi}{2}.
\]

Since the total growth does not decrease, this boost is not valuable. Any larger boost only increases the growth. A smaller boost may decrease the growth of $v$'s own active set, but increases it for all leaves in the subtree—which cancels out the gain.

Since the total growth does not decrease, this boost is not valuable. Any larger boost (i.e., increasing \(t_v\) beyond the value we just analyzed) only increases the total growth. On the other hand, any smaller boost may reduce the growth of \(v\)'s own active set, but increases the growth for all leaves in the subtree—compared to the analyzed boost value—thereby canceling out any potential decrease in total growth.

Thus, no boost action reduces the total cost. Local search halts without any change to the forest or fingerprint.

\end{proof}

\paragraph{Optimal Solution Analysis.}

The optimal solution clearly corresponds to selecting only the original binary tree edges, excluding any additional leaf-to-leaf edges. Its total cost, denoted $\cc(\OPT)$, is:

\[
\cc(\OPT) = 2^h \left(1 + \frac{1}{4} + \frac{1}{16} + \dots + \frac{1}{2^{2(h-1)}} \right) \le \frac{4}{3} \cdot 2^h
\]

\paragraph{Approximation Factor Analysis.}

The cost of our algorithm's forest, denoted as $\cc(F)$, is:
\[
\cc(F) = (2^h - 1) \cdot (2 - \xi)
\]
Hence, the approximation ratio is:
\[
\frac{\cc(F)}{\cc(\OPT)} \le \frac{(2^h - 1)(2 - \xi)}{2^h \cdot \frac{4}{3}} \le \frac{3(2^h - 1)}{2(2^h)}
\]
As $h$ grows large, this ratio approaches $3/2$:
\[
\lim_{h \to \infty} \frac{\cc(\F)}{\cc(\OPT)} = \frac{3}{2}
\]
This completes our analysis and demonstrates that the given instance achieves a $3/2$ approximation ratio.

Thus, this example establishes a lower bound of $3/2$ for the approximation ratio of our algorithm. Determining whether this bound is tight remains an intriguing open question.

\subsection{Necessity of Autarkic Pairs}
\label{eg:horseshoe}

In this section, we present an instance where the \BBG{} algorithm produces a solution with a 2-approximation, but neither local search with boost actions nor extension is able to improve it. This example highlights the limitations of these techniques in certain cases and illustrates the necessity of handling autarkic pairs separately using the \candidate{} algorithm.

\paragraph{Constructing the Instance.}

Consider a vertical path with \( n+1 \) edges of length 1, connecting a top vertex \( v \) to a bottom vertex \( u \). For each intermediate vertex along this path, attach a new vertex via an edge of length 1. Then, create multiple demand pairs over the endpoints of these added edges by duplicating the vertices and connecting each duplicate back with a near-zero-cost edge (with sufficiently small length \(\xi\)). Finally, add one more demand pair between \( v \) and \( u \), and place a direct edge of length 2 between them (see Figure~\ref{fig:ten-inp}).

\paragraph{Running \BBG{}, \LocalSearch{} and \EGW{}.}

When \BBG{} is executed on this instance, it assigns the same timestamp of \( \frac{1}{2} \) to all demand vertices (see Figure~\ref{fig:ten-norm}), and all are deactivated simultaneously at \( \tau = \frac{1}{2} \). No vertex remains active beyond this point. Consequently, any attempt to boost a vertex is ineffective—the algorithm’s mechanism for boosting cannot trigger a valid change, as any smaller boost is ignored once all vertices are halted. This uniform fingerprint represents an edge case in which the \LocalSearch{} fails to improve the solution.
Similarly, \EGW{} offers no improvement either, as it only increases the growth of active sets already finalized. Any additional growth preserves the previously selected forest without impacting the outcome.

\paragraph{Optimal Solution Analysis.}

The optimal solution in this instance connects each demand pair via its direct edge. Thus, the cost is:
\[
\cc(\OPT) = 2 + \sum_{i=1}^{n} 1 = n + 2.
\]

\paragraph{Approximation Factor Analysis.}

In contrast, the solution produced by the moat growing algorithm includes all the edges of the base vertical path (which sum to \( n+1 \)) plus the \( n \) additional edges connecting each \( v_i \) to \( u_i \). Therefore, the total cost of the forest \( F \) is:
\[
\cc(F) = (n+1) + n = 2n + 1.
\]
This yields an approximation ratio of:
\[
\frac{2n+1}{n+2} \approx 2 \quad \text{for sufficiently large } n.
\]

\paragraph{Key Observation and Role of the Autarkic Pairs Algorithm.}

The main observation from this example is that---even though the moat growing procedure yields a 2-approximation---this specific structure causes the boost action to be ineffective, specifically not choosing the direct edge between $v$ and $u$. In this edge case, the intermediate demands keep active sets growing over many unnecessary edges of the path and force the $(v, u)$ demand to connect through these edges.

This is precisely where the \candidate{} algorithm becomes essential: it detects these intermediate demands as autarkic pairs and connects them directly (see Figure~\ref{fig:ten-zero-edge}). This leads to immediate deactivation of their corresponding active sets and, as a result, allows the active sets of \( v \) and \( u \) to grow toward each other and connect via the direct edge of length 2 (see Figure~\ref{fig:ten-our}).

\subsection{A Lower Bound Worse than 2 for the Gluttonous Algorithm}
\label{eg:anupam83}
Given the high significance of Steiner Forest problem, long effort has been put even for finding another greedy algorithm that gives a constant approximation. Gupta and Kumar \cite{DBLP:conf/stoc/Gupta015} showed  the gluttonous algorithm preserves an $O(1)$-approximation. They also left this as an important open problem whether this algorithm is a 2-approximation.

In this section, we present a counterexample showing that the gluttonous algorithm is not a 2-approximation for the Steiner Forest problem; in fact, its approximation ratio is at least \(\frac{8}{3} \approx 2.666\). 

First, we provide an exact presentation of the gluttonous algorithm.

\paragraph{The Gluttonous Algorithm.}

The algorithm proceeds by iteratively selecting the minimum-cost edge that connects previously disconnected components of unsatisfied vertices. This process continues until all unsatisfied vertices are connected to their pairs and get satisfied. The gluttonous algorithm is notable for its simplicity and intuitive approach.

The gluttonous algorithm can be formally defined as follows:

\begin{enumerate}
  \item Initialize the forest $\cc(F)$ as an empty set.
  \item While there exists unsatisified (disconnected) demands:
  \begin{enumerate}
    \item Find the minimum distance pair $e = (u, v)$ such that $u$ and $v$ belong to different components in the current forest $\cc(F)$ and neither $v$ is connected to $\pairv$ nor $u$ is connected to $\pair_u$--both are unsatisifed.
    \item Add the edge $e$ to the forest $\cc(F)$ with cost of current shortest-path (equivalent to adding the edges of the current shortest-path to the forest).
    \item Connect $u$ and $v$ with a zero edge in the context graph.
  \end{enumerate}
\end{enumerate}

Here, the context graph models the evolving connectivity state among vertices by introducing zero-cost virtual connections between already-joined vertices.

Previous research has established that the gluttonous algorithm is a 96-approximation for the Steiner Forest problem. 

\paragraph{Counter Example.}

Here, we present a counterexample that demonstrates the gluttonous algorithm is not a 2-approximation for the Steiner Forest problem. The counterexample we provide exhibits a ratio of $\frac{8}{3}$ (approximately 2.666).

\paragraph{Description of the Counterexample.}

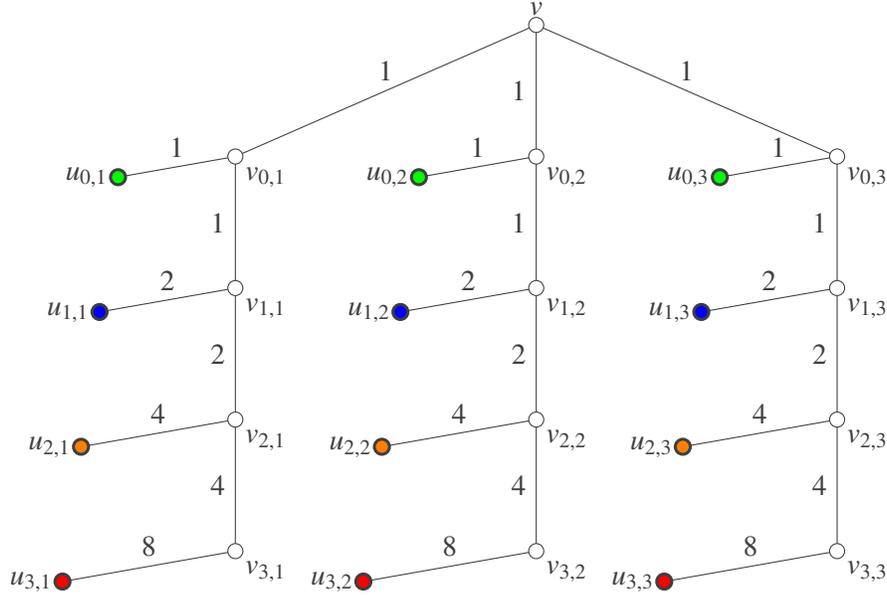
\begin{figure}[th]
\centering

\begin{tikzpicture}[scale = 0.5]
\def\dx{8}
\def\dy{3.5}
\def\r{0.7}
\foreach \i/\place in {-1/above, 0/left, 1/above} {
    \draw[Black!90] (0, 0) -- node[\place] {1} (\i * \dx, -\dy);
}
\foreach \i in {-1, 0, 1} {
    \foreach \j/\L in {2/1, 3/2, 4/4} {
        \draw[Black!90] (\i * \dx, -\j * \dy) -- node[left] {\L} (\i * \dx, -\j * \dy + \dy);
    }
}
\def\noder{0.2cm}
\draw[Black!90, fill=white] (0, 0) circle (\noder);
\node[Black!90,above] (a) at (0, 0) {$v$};
\foreach \i/\I in {-1/1, 0/2, 1/3} {
    \foreach \J/\j/\L/\col in {0/1/1/green, 1/2/2/blue, 2/3/4/orange, 3/4/8/red} {
        \draw[Black!90] (\i * \dx, -\j * \dy) -- node[above] {\L} ($ (\i * \dx, -\j * \dy) + (-170: \dx/3 + \j/2) $);
        \draw[Black!90, fill=white] (\i * \dx, -\j * \dy) circle (\noder);
        \node[Black!90, below right] (a) at (\i * \dx, -\j * \dy) {$v_{\J,\I}$};
        \draw[Black!90, line width=0.4mm,fill=\col] ($ (\i * \dx, -\j * \dy) + (-170: \dx/3 + \j/2) $) circle (\noder);
        \node[Black!90, left] (a) at ($ (\i * \dx, -\j * \dy) + (-170: \dx/3 + \j/2) $) {$u_{\J,\I}$};
    }
}
\end{tikzpicture}
\caption{The figure shows a counterexample for which the Gluttonous algorithm yields an approximation ratio of at least \(8/3\). Each column corresponds to one of the \(n = 3\) replicas of a base tree structure with \(k + 1 = 4\) layers and exponentially increasing edge weights. Each layer of leaves (i.e., vertices at the same depth across columns) forms a terminal group that must be pairwise connected. 
The Gluttonous algorithm connects each group of connectivity demands independently through directed edges between them, while the optimal solution reuses the underlying tree structure to connect all groups more efficiently. It is important to note that we can assume direct edges between demands are slightly smaller than the shortest path in the tree.}
\label{fig:ex83}

\end{figure}

Let us consider a graph, $H$, defined in the following manner. The graph encompasses a path that consists of vertices denoted as $v_0$, $v_1$, ..., $v_k$. Additionally, each index, $i$, has an associated vertex, $u_i$, which connects directly to the corresponding $v_i$ vertex. This configuration constructs a tree-like structure (Actually, $H$ corresponds to each column in Figure~\ref{fig:ex83}).

The costs of the edges within this graph adhere to a specific pattern. The edge cost connecting vertices $v_i$ and $v_{i+1}$ is defined as $2^i$. A similar cost pattern is observed for the edge linking vertices $v_i$ and $u_i$, wherein the cost is also defined as $2^i$.

Building on the previous step, consider creating \( n \) copies of the graph \( H \), labeled \( H_1, H_2, \dots, H_n \). In addition to these copies, add a new extra vertex called \( v \). In each copy \( H_j \), let \( u_{i,j} \) and \( v_{i,j} \) represent the counterparts of the original vertices \( u_i \) and \( v_i \), respectively.

For every $j$, form a connection between vertex $v_{0,j}$ and the extra vertex, $v$. The resultant structure upon these connections can then be defined as tree $T$ (See Figure \ref{fig:ex83}).

The graph $G$ is generated as a weighted metric graph, derived by computing the shortest paths within tree $T$. Subsequently, for each fixed $j$, $P_j$ is defined as a collection of terminal groups, where the $j$-th terminal group consists of the vertices $u_{i,j}$ for all possible values of $i$. Hence, there are $k + 1$ such terminal groups that are to be connected.

Here, a terminal group refers to a collection of vertices that must be pairwise connected. As discussed earlier, this setting can be transformed into our \((v, \pairv)\) formulation by duplicating each vertex and adding zero-cost edges between the duplicates.

In this proposed counterexample, it will be demonstrated in the subsequent sections that the ratio of the output from the Gluttonous algorithm to the optimal solution is no less than $\frac{8}{3}$. This assertion underscores the performance characteristics of the algorithm within this specific context.

\paragraph{Gluttonous Algorithm Performance.}

The following characteristics of distances between vertices in the graph are observed:

\begin{itemize}
    \item The distance from $u_{i,j}$ to $v$ is $2^{i+1}$.
    \item The distance between two vertices, $u_{i,j_1}$ and $u_{i,j_2}$, is $2 \cdot 2^i = 2^{i+2}$.
    \item The distance between two vertices, $u_{i_1,j}$ and $u_{i_2,j}$ (with $i_1 < i_2$), is $2^{i+1}$.
    \item Furthermore, the distance between two vertices, $u_{i_1,j_1}$ and $u_{i_2,j_2}$ (with $i_1 < i_2$, and $j_1 \neq j_2$), does not decrease when zero-edges are added between any pair of $u_{i,j_1}$ and $u_{i,j_2}$ (for $i < i_1$).
\end{itemize}

Based on the aforementioned observation, we notice that the gluttonous algorithm consecutively merges $u_{0,1}$ and $u_{0,2}$, then $u_{0,1}$ and $u_{0,3}$, and continues in this manner until it merges $u_{0,1}$ and $u_{0,n}$. In this progression, the zeroth group, $P_0$, becomes completely satisfied. Consequently, any later $u_{0,j}$ vertices are excluded from the requirements.

Subsequently, attention is turned to the remaining groups, namely $P_1$, $P_2$, through to $P_k$. Throughout this process, each group necessitates $n-1$ steps to reach full satisfaction. For each group $P_i$, each step incurs a cost of $2^{i+2}$. Therefore, the total cost incurred is expressed as $$(n-1)\sum_{i=0}^{k} 2^{i+2} = (n-1)(2^{k+3} - 4).$$

\paragraph{The Optimal Solution.}

It is apparent that tree $T$, as it stands, provides a feasible solution to the aforementioned requirements. Consequently, if we designate the optimal solution cost as $\cc(\OPT)$ and the total cost of $T$ as $\cc(T)$, then it can be inferred that $\cc(\OPT) \leq \cc(T)$.

More specifically, the upper bound of $\cc(\OPT)$ can be calculated as follows:

\[
\cc(\OPT) \leq n \cdot \left(1 + \sum_{i=0}^{k} 2^i + \sum_{i=0}^{k - 1} 2^i\right) = n \cdot \left(2^k + 2^{k + 1} - 1\right) = n \cdot \left(3 \cdot 2^k - 1\right).
\]

Thus,

\[
\cc(\OPT) \leq n \cdot \left(3 \cdot 2^k - 1\right).
\]

\paragraph{Approximation Factor Analysis.}

Let $\cc(F)$ denote the outcome cost of the Gluttonous algorithm. Hence,
\[\cc(F) = (n-1)(2^{k+3} - 4) = (n-1)(8 \cdot 2^k - 4).\]

Given that $\cc(\OPT) \leq n \cdot (3 \cdot 2^k - 1)$, the approximation ratio can be expressed as:

\[
\lim_{n,k \to \infty} \frac{\cc(F)}{\cc(\OPT)} = \lim_{n,k \to \infty} \frac{(n-1)(8 \cdot 2^k - 4)}{n \cdot (3 \cdot 2^k - 1)} = \frac{8}{3}
\]

Therefore, assuming sufficiently large values of $n$ and $k$, it can be asserted that the approximation ratio is at least $\frac{8}{3}$.
 
\end{document}